\documentclass[american,reqno,openany]{amsbook}
\usepackage{amsmath,amssymb}
\usepackage{graphics,epsfig,color}
\usepackage{booktabs,siunitx}  
\usepackage[mathscr]{euscript}
\usepackage{verbatim}
\usepackage{quoting}
\usepackage{mathrsfs}
\usepackage{hyperref,color}
\usepackage{appendix}
\usepackage{wasysym}
\usepackage[dvipsnames]{xcolor}
\usepackage[framemethod=default]{mdframed}
\usepackage[Lenny]{fncychap}
\usepackage{fancyhdr}
\usepackage{titletoc}
\usepackage{appendix}

\mdfsetup{skipabove=\topskip,skipbelow=\topskip}
\global\mdfdefinestyle{fig}{%
linecolor=gray,linewidth=3pt,align=center,
leftmargin=1cm,rightmargin=1cm}
\global\mdfdefinestyle{th}{%
linecolor=gray,linewidth=3pt}
\global\mdfdefinestyle{eq}{%
linecolor=black,linewidth=0.5pt,align=center}

\usepackage[nowrite,infront,suftesi]{frontespizio}

\usepackage{tikz}
\usepackage[all]{xy}
\xyoption{all}

\newtheoremstyle{theorem}{}{}{\itshape}{}{\bfseries}{.}{.5em}{}
\theoremstyle{theorem}
\newtheorem{theorem}{Theorem} 
\newtheorem{proposition}[theorem]{Proposition}
\newtheorem{lemma}[theorem]{Lemma}
\newtheorem{corollary}[theorem]{Corollary}

\newtheoremstyle{definition}{}{}{}{}{\bfseries}{.}{.5em}{}
\theoremstyle{definition}
\newtheorem{definition}[theorem]{Definition}
\newtheorem{assumption}[theorem]{Assumption}

\newtheoremstyle{remark}{}{}{}{}{\bfseries}{.}{.5em}{}
\theoremstyle{remark}
\newtheorem{remark}[theorem]{Remark}
\newtheorem{example}[theorem]{Example}

\makeatletter
\def\l@subsection{\@tocline{2}{0pt}{2.5pc}{5pc}{}}
\makeatother
\numberwithin{section}{chapter}
\numberwithin{subsection}{section}
\numberwithin{paragraph}{subsection}
\numberwithin{equation}{chapter}
\numberwithin{theorem}{chapter}
\numberwithin{figure}{chapter}
\numberwithin{table}{chapter}

\renewcommand{\div}{\textrm{\,\,div\,}}
\newcommand{\Ker}{\textrm{Ker\,}}
\newcommand{\Ima}{\textrm{Im\,}}
\newcommand{\us}{\underset}
\newcommand{\os}{\overset}
\newcommand{\mc}[1]{{\mathcal #1}}
\newcommand{\mtt}[1]{{\mathtt #1}}
\newcommand{\mf}[1]{{\mathfrak #1}}
\newcommand{\bb}[1]{{\mathbb #1}}

\newcommand{\rt}[1]{\textcolor{red}{#1}}
\newcommand{\upbar}[1]{\,\overline{\! #1}}

\newcommand{\vphi}{\varphi}
\newcommand{\p}{\partial}

\newcommand{\hrho}{\hat\rho}
\newcommand{\e}{\emph}
\newcommand{\D}{\Delta}
\newcommand{\n}{\nabla}
\newcommand{\veps}{\varepsilon}
\renewcommand{\tilde}{\widetilde}
\renewcommand{\hat}{\widehat}
\renewcommand{\k}{\kappa}
\renewcommand{\j}{\jmath}
\newcommand{\sign}{\text{sign}}

\definecolor{light}{gray}{.9}

\newcommand{\facciatabianca}{\newpage\shipout\null\stepcounter{page}}

\usepackage[some]{background}
\usepackage[absolute,overlay]{textpos}

\definecolor{color1}{RGB}{31,55, 61}
\definecolor{background}{RGB}{238,237,229}

\definecolor{comp}{RGB}{98,99,101}
\definecolor{phys}{RGB}{31,55,61}
\definecolor{math}{RGB}{133,63,72}
\definecolor{urban}{RGB}{51,95,60}

\definecolor{color1}{RGB}{133,63,72}

\usepackage{pgf,tikz}
\usetikzlibrary{arrows}

\usepackage{scrextend}

\backgroundsetup{
	scale=1,
	angle=0,
	opacity=1,
	contents={\begin{tikzpicture}[remember picture,overlay]
		\path [fill=color1] (-0.5\paperwidth,0.5\paperheight) rectangle (0.5\paperwidth,-0.0\paperheight);  
		\path [fill=background] (-0.5\paperwidth,-0.5\paperheight) rectangle (0.5\paperwidth,-0.0\paperheight);  
		\end{tikzpicture}}
}

\begin{document}

\begin{titlepage} 

	\BgThispage
	
	\fontfamily{phv}\selectfont{
		
		\begin{flushleft}
			\noindent
			{\Large \textcolor{background}{Gran Sasso Science Institute}}
			\vspace{4mm}
			\newline
			\fontsize{11}{1.12}\selectfont{
				\textbf{\textcolor{background}{PHD PROGRAMME IN}}
				\textbf{\textcolor{background}{MATHEMATICS IN NATURAL, SOCIAL AND LIFE SCIENCES}}
				\vspace{2mm} \newline
				\textcolor{background}{Cycle XXIX - Academic year 2016/2017-L'Aquila}

				\vspace{10mm}
			}
		\end{flushleft}
		
		\begin{flushleft}
			\fontsize{19}{1.12}\selectfont{
				\noindent
				
				\textbf{\textcolor{background}{\centering{Microscopic and macroscopic\\
							perspectives on\\
							\hspace{1.1cm}stationary nonequilibrium states}}}\newline
				\vspace{25mm}
			}
		\end{flushleft}
		
		\begin{flushright}
			\noindent
			\textcolor{background}{PHD CANDIDATE} \\
			{\large \color{background} \textbf{Leonardo De Carlo}}
		\end{flushright}
		\vspace{28mm}
		\begin{flushright}
			\noindent
			\textcolor{color1}{ADVISOR} \\
			{\large \color{color1} \textbf{Prof. Davide Gabrielli}} \\
			{\color{color1} Universit\`{a} degli studi  dell'Aquila}   
		\end{flushright}

		\begin{figure}[b]
			\begin{center}
				\begin{minipage}[c]{.40\textwidth}
					\centering
					\includegraphics[height=20mm]{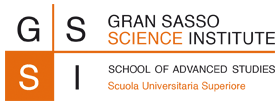}
				\end{minipage}
				\hspace{20mm}
				\begin{minipage}[c]{.40\textwidth}
					\centering
					\includegraphics[height=20mm]{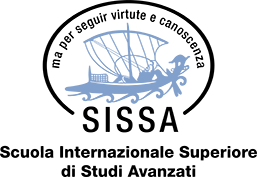}
				\end{minipage}
			\end{center}
		\end{figure}

	}
	
\end{titlepage}

\begin{frontespizio}

\end{frontespizio}

\facciatabianca

\thispagestyle{empty}
\vspace*{5cm}
\textbf{Thesis Jury Members} \newline

\noindent Cedric Bernardin ( Université de Nice Sophia-Antipolis)\newline
Anna De Masi (Università degli studi  dell'Aquila)\newline
Cristian Giardinà (Università degli studi di Modena e Reggio Emila)\newline
Mauro Piccioni (Università degli studi di Roma,  "La Sapienza")\newline
Errico Presutti (GSSI, L'Aquila)\newline

\textbf{Thesis Referees} \newline

\noindent Cedric Bernardin ( Université de Nice Sophia-Antipolis)\newline
Cristian Giardinà (Università degli studi di Modena e Reggio Emila)\newline

\pagestyle{fancy}
\fancyhf{}
\rhead{}
\lhead{}
\cfoot{\footnotesize \thepage}
\renewcommand{\headrulewidth}{0pt}

\facciatabianca

\thispagestyle{empty}
\vspace*{5cm}
\begin{center}
{\LARGE \emph{
To my grandparents\\
for all the time we spent together and  what they gave me since before the very beginning.\\}}
\end{center}

\facciatabianca

\setcounter{page}{1}

{\Large\sffamily  \tableofcontents}

\section*{Introduction: structure of the thesis}

The main subject of the thesis is the study of stationary nonequilibrium  states trough the use of microscopic stochastic models that encode the physical interaction in the rules of Markovian dynamics for particles configurations. These models are known as interacting particles systems and are simple enough to be treated analytically but also complex enough to capture essential physical behaviours.

The thesis is organized in two parts. The  part \ref{p:mt} is devoted to the microscopic theory of the stationary states. We characterize these states developing some general  structures  that have an interest in themselves.

The part \ref{p:mact} studies the problem  macroscopically. In particular we consider the large deviations asymptotic
behavior for a class of solvable one dimensional models of heat conduction.

Both part \ref{p:mt} and \ref{p:mact} begin with  an introduction of motivational character followed by an overview of the relevant results  and a summary explaining the  organization. Even tough the two parts are strictly connected  they can be  read independently  after chapter  \ref{ch:IPSmodels}. The material is presented in such a way to be self-consistent as much as possible.

The result of part \ref{p:mt} can be found in \cite{DG17a} while that ones of part \ref{p:mact} in \cite{DG17b}.

\vspace{2cm}

\subsection*{Acknowledgement}

First of all I thanks my advisor Davide Gabrielli to have led me with great attention during my research and for its careful reading of the manuscript. I acknowledge the Gran Sasso Science Istitute (GSSI) and the coordinators A. DeSimone, P. Marcati and E. Presutti for the opportunity to be involved in this research project. I especially thanks E. Presutti for having listened to me on many occasions.


{\sffamily\part{Microscopic theory}\label{p:mt}}

\section*{Why a microscopic theory?}
Idealization of the microscopical physics of the matter by simple models had a fundamental role to understand the macroscopic phenomenology of the equilibrium states \cite{Ma93, KTN91}. This efficacy of ideal models is due to the fact that to a reasonable extent the macroscopic behaviour is independent of the microscopic details, namely different systems exhibit qualitatively the same phenomenology at large scales. The same pictures is expected for nonequilibrium states, in these context using deterministic models is beyond the reach of the nowdays mathematics. Therefore, to study nonequilibrium states, it is common to assume a random motion at small scale. Even if at microscopical level particles move along Newton's equation, this assumption is physically reasonable, indeed  their motion is pretty chaotic in  a random fashion without exhibiting some aim or particular coordination and  the real reason of why we observe particles forming  macroscopic pattern like deterministic  field (i.e. density) equations (PDEs) is that locally there are some constraints imposed to their dynamics like local conservation of mass, momentum and energy. These constraints are not immediately visible on microscopic scale but manifest themselves on the large scale where a deterministic motion appears for macroscopic observables like density, temperature, charge, etc. . This is also the reason of why to a reasonable extent macroscopic behaviour is independent of the microscopic details. Since with stochastic models it is still possible to impose local conservation laws for the mass, energy, etc. that are exchanged by the microscopical dynamics  this approach allows to prove  rigorous autonomous macroscopic dynamics \cite{KipLan99}.
So studying different stochastic microscopical dynamics it is an important tool to comprehend which are the microscopical detail that are important to see different macroscopical behaviour. 
Here we are dealing both with models with discrete variables and models with continuous variables, we are calling the first ones interacting  particles systems and the second ones interacting energies-masses systems. The dynamics in these models are Markovian.

\vspace{0.8cm}

To get a satisfactory understanding of the nonequilibrium stationary states it is necessary to be able to describe the measures characterizing these states, that is the equivalent of the Gibbs measures in equilibrium  which lead to the thermodynamics ensembles.
Their construction starts from the understanding of the long time behaviour of the microscopic dynamics and they provide the bridge to pass from the microscopic to the macroscopic description of the system.  In nonequilibrium, the problem in the construction of these ensemble is presence of non-conservative terms in the Newton's equations describing the microscopic evolution of the molecules of a gas for example, for this reason classic arguments to justify the construction of the equilibrium statistical mechanics ensembles fail. For example, because of dissipation, the long time dynamics of a systems of many particles will be restricted to a subset of the space states of smaller dimension, i.e. of zero Lebesgue measure. An attempt of ensembles theory out of equilibrium is that based on the costruction of SRB measures for deterministic systems having  as asympotic motion an Anosov motion on a chaotic attractor \cite{Ga14}. This phenomenon can be  analytically studied in a  relatively simple way  in systems with low degrees of freedom, for example in \cite{DGG16}. In extended physical systems appear a certain number of "central directions" \cite{Ru12}, namely degrees of freedom  neither strictly contractive nor strictly expansive, under the "chaotic hypothesis", that extend "ergodic hypothesis" to nonequilibrium systems, 
these particular degrees of freedom "disappear" because the asymptotic motion is conceived as an Anosov motion on a chaotic attractor and the SRB measures are proposed as candidate  for nonequilibrium ensembles, see \cite{Ga14}.
To make things easier from the analytical point of view, another way is that to consider stochastic models with many microscopic degrees of freedom and studying their stationary states. Of course these  model have to be able to capture (hopefully) the essential peculiarities of real models in such a way to allow quantitative observable predictions. This last direction is  that one we are following.

In this first part we will be motivated by this framework to develop two different techniques  which will allow to construct invariant measures for Markov dynamics describing some interacting stochastic models, many of which are proved (more or less rigorously) to exhibit a collective behaviour described by a partial differential equation, as in mentioned in the previous paragraph. The microscopic distinction between equilibrium and non-equilibrium for our stochastic models will  be the condition to satisfy or not the detailed balance condition, i.e. reversibility. When it is not fulfilled we have  nonequilibrium models which microscopic currents are characterized by non-zero averages. The mathematical description of the microscopic  currents is another fundamental part in the comprehension of nonequilibrium states and it will be the second topic of our interest in this first part of the thesis.

\section*{Microscopic Results: overview}

First, we discuss two possible boundary driven totally asymmetric versions
of the KMP model. In one case the invariant measure is given by a product measure
of exponentials. For the other dynamics the invariant measure has not a simple
structure and we give a representation in terms of convex combinations of products
of Gamma distributions. This is done using a duality representation of the process.
The same could be done also for the dual model of the KMP (that we shortly call
KMPd). This is done in section \ref{sec:Amodels} where we  define
some totally asymmetric dynamics having the same weakly asymmetric diffusive
scaling limit, this will be discussed in chapter \ref{ch:HSL} jointly with subsection \ref{subsec:Ie-mc}.

A new functional Hodge decomposition for translational covariant discrete vector fields is proved in section \ref{sec: HD}. The discrete Hodge decomposition is the splitting of any discrete vector field into a gradient part, a curl part and an harmonic part. Consider a discrete vector field depending in a translational covariant  way (see \eqref{eq:jinv}) on  configurations of particles on lattice. We discuss that it is possible to do a similar splitting where the gradient component and the curl term are respectively a gradient and a a curl (an orthogonal gradient in two dimensions) with respect to the shift of  functions on configurations.

A motivation for the functional Hodge decomposition comes from the theory of hydrodynamic scaling limits. In the case of diffusive systems the derivation of the hydrodynamic scaling limit is much simplified in the case of models having the instantaneous current of gradient type. When also the invariant measure is known it is possible to have an explicit form of the transport coefficients.
In the non gradient cases the theory develops starting from an orthogonal decomposition \cite{VY97} and the transport coefficients have a representation in terms of the Green Kubo formulas \cite{KipLan99,Spo91}.
While the orthogonal splitting in \cite{VY97} is in general not explicit, our splitting is constructive. It is an interesting issue to study the information encoded by our decomposition, its relation with the variational representation of the transport coefficients \cite{KipLan99,Spo91} and in particular its relation with the finite dimensional approximation of the transport coefficients analyzed in \cite{AKM17}.

In chapter \ref{ch:SNS} we  propose and study two different ways to generate non reversible stochastic particle systems having a fixed Gibbsian invariant measure. The two methods are quite general and correspond to two different discrete geometric constructions that have an interest in themselves.

The construction of reversible Markovian dynamics having a fixed target invariant measure is a classic problem that is at the core  of the Monte Carlo method. This is very useful in several applicative problems and the literature concerning the method is huge. We suggest for example \cite{Dia09,LB15} and the references therein. Instead, the construction of non reversible Markovian dynamics is more difficult. This is an interesting problem since the violation of detailed balance is expected for example to speed up the convergence of the dynamics \cite{Bi15,KJZ16,SB15}. Another motivation come from the fact that the understanding of collective and macroscopic behaviour of interacting multi-component systems is based on the knowledge of the stationary state as discussed above in the introduction to the microscopic theory. Several references are dedicated to the problem of finding non reversible Markov chains having a fixed target invariant measure in specific models, see for example \cite{BB15,CT16,FGS16,Go13,GL15,LG06,PSS17}. We study the problem in the case of interacting particle systems and describe two somehow general approaches. We consider for simplicity the case of stochastic lattice gases but the methods can be generalized to more general models.

Our first approach of section \ref{sec:applications} is based on the equivalence between the problem of finding rates that have a fixed invariant measure and the problem of finding a divergence free flow on the transition graph. A flow is a map that associates a positive number, representing the mass flowed, to each oriented edge of a graph. The divergence of a flow is the amount of mass flowing outside of a vertex minus the amount of mass flowing into the vertex.   The usefulness of this equivalence is on the fact that the structure of divergence free flows is quite well understood. Any divergence free flow can be indeed represented as a superposition of  flows associated with elementary cycles \cite{BFG15,GV12,MacQ81} (see section 2 for precise definitions). A local non reversible dynamics is then generated considering elementary cycles on the configuration space for which the configuration of particles is frozen outside of a bounded region.

The second approach of section \ref{sec:sta&ort} is based on the  functional Hodge decomposition just described few lines above. 
The stationary condition for stochastic lattice gases can be naturally interpreted as an orthogonality condition with respect to a given harmonic discrete vector field. To satisfy this condition it is enough to consider discrete vector fields in the orthogonal complement and this can be done using the functional decomposition. 

With both methods, we obtain very general classes and families of models naturally parametrized by functions on the configuration space and having similar structures. We are not going to compare the results with the two methods but in principle this is possible. In particular for any model the stationary condition can be interpreted with the two geometric constructions. For simplicity, we show all our results in dimensions 1 and 2 and for stochastic lattice gases satisfying an exclusion rule. All the arguments can be directly generalized to higher dimensions and to more general state space. In particular in dimensions higher than 2 the same constructions can be done. The basic ideas are the same but more notation and more general geometric constructions are necessary. We concentrate on finite range translational invariant  measures but different situations with long range interactions and/or non translational invariant measures are also interesting. We consider only cycles obtained by local perturbations of configurations of particles but the case of non local cycles is also interesting \cite{PSS17}.

\subsection*{Structure of the part \ref{p:mt}}

The structure of part \ref{p:mt} is as follows.
In chapter \ref{ch:IPSmodels} we present the microscopic models that we are studying in the thesis. The problem of the stationary measures is introduced. Particular attention is given to the definition of  asymmetric (weakly and not weakly) models   and the description of the instantaneous current that is the basic object to understand the connection between the microscopic and the macroscopic description.

In chapter  \ref{ch:DEC} we set a general discrete exterior calculus suitable also for our purposes and here we prove the functional Hodge decomposition described before in the abstract.
In chapter \eqref{ch:SNS}, first we introduce  notation and discuss some basic tools in graph theory and  discrete geometry, later on in section \ref{sec:applications} we generate non reversible stochastic lattice gases starting from elementary cycles in the configuration space and in section \ref{sec:sta&ort} from our functional Hodge decomposition. 

 In appendix \ref{a:mic} are shown as complements  a functional Hodge decomposition for  vertex functions and a consequent  construction to generate invariant dynamics

\bigskip	

\noindent \textbf{\textit{Keywords:}} Markov chains, interacting particles systems, lattice gases, discrete exterior calculus, Hodge decomposition, nonequilibrium stationary states.


\clearpage
\pagestyle{fancy}
\fancyhf{}
\fancyhead[RE]{\sffamily \fontsize{10}{12} \selectfont \nouppercase{\leftmark}}
\fancyhead[RO]{\sffamily \fontsize{10}{12} \selectfont \nouppercase{\rightmark}}
\fancyhead[LE]{\sffamily \fontsize{10}{12} \selectfont Part \thepart: Microscopic theory}
\fancyhead[LO]{\sffamily \fontsize{10}{12} \selectfont Part \thepart: Microscopic theory}
\cfoot{\footnotesize \thepage}
\renewcommand{\headrulewidth}{0.01cm}

\chapter{Stochastic interacting systems}\label{ch:IPSmodels}

Before entering into detail of the study of microscopic models we spend few words to remind basic facts on continuous time                                                                                                                                                               Markov chains. Later in this chapter we introduce the basic notions to treat  stochastic interacting systems and the properties of the reference models. Moreover, we set the language we are carrying on trough all the text.

\section{Continuous time Markov chain }\label{sec:CTMC}

Here we treat Continuous Time Markov Chains (CTMCs) on a   finite  set  from  an infinitesimal (in time) description of the evolution. In this restriction to countable or finite state spaces  it is common to use the word "chain" rather than "process". A reference for general state space is \cite{Lig10} (chapter 3). CTMCs are very useful toy models to describe from a stochastic point a view phenomena with a fast loss of memory. They allow a very deep analytical treatment, in particular their long time behavior can be studied by some simple equations. 
CTMCs find application in many fields like electric networks \cite{DS06}, biology (epidemic) \cite{Nor97}, queuing theory \cite{Nor97}, information theory  and  web searching \cite{VL06}. The present context of our interest for the use of Markov Chains  is  the so called "Interacting particles systems", classical references are \cite{KipLan99}, \cite{Lig05} and \cite{Spo91}.  We  give a brief review of continuous time  Markov chains  with emphasis on stochastic interacting systems.
\vspace{1cm}

We consider a  squared domain $ \Lambda=(0,L]^n $  and its discrete approximation $\Lambda\cap\veps\mathbb{Z}^n $   with   scale separation  $ \veps $ such that the set of vertices $  V_N $ of the lattice obtained in this way     contains exactly $ N $ vertices separated by a distance of $ \veps $. At the boundary periodic conditions can be considered. In what follows we are using  two scale separations, one is $ \veps= 1 $ and the second one is $  \veps=\frac 1N $.
The first scale is called \e{microscopic} because the distance between sites (and then the minimum one between particles) is fixed to one and to consider lattice with many sites we have to take large domains $ \Lambda $. The second is called \e{macroscopic} because   when $ N $ is chosen larger and larger  the discretization  appears like a continuum and  at this scale  deterministic macroscopic behaviours   (quoted in the introduction to this  part I) will manifest in the models we are introducing. This phenomena is going to be the subject of part II.
For different models  we will choose a separation scale that is appropriate for the respective problems that follow.  The microscopic configurations are given by the collection of variable $ \eta(x) $, $ x\in V_{N} $, representing the number of particles, energy or mass at $ x $ along the model. When the variables $ \eta(x) $ are discrete we interpret them as number particles and  when continuous as quantity of mass-energy. Calling $\Sigma $ the state space at $ x $ we define   the configuration space as $\Sigma_N:= \Sigma^{ V_N}$. 
\begin{remark}\label{r:sum}
When a general formula contains an implicit sum over $ \Sigma_N $ the explicit counterpart is written for the particles case with a sum $ \us{\eta'\in \Sigma_N}{\sum} $,  for the energies-masses case the analogous equation is not going to be written apart for specific indicated cases and it is obtained replacing the sum with a suitable integral over $ \Sigma_N $.  
\end{remark}
The microscopic dynamics is given by the Markov process $ \{\eta_t\}_{t\in\mathbb{R}} $. Associated to each lattice site  there are some independent Poisson clocks of parameter depending on the local configuration nearby the site. When a clock rings, the configuration $ \eta $ changes into $ \eta' \neq\eta$ along the  \emph{transition rates}
\begin{equation}\label{eq:Mrates}
c(\eta,\eta')=\us{t\to 0}{\lim}\frac{\mathbb{P}^\eta(\eta_t=\eta')}{t}=\us{t\to 0}{\lim}\frac{p_t(\eta,\eta')}{t},
\end{equation}
where $ \mathbb{P}^\eta $ is the probability distribution over the trajectories on the configuration space starting from $ \eta $ at time $ 0 $. Let $ f:\Sigma_N\to \mathbb{R} $ be a real function on the configuration space, we define the action of an operator $ A $ on $ f $ as
\begin{equation}\label{eq:Aonf}
A f(\eta)=\us{\eta'\in\Sigma_N}{\sum}A(\eta,\eta')f(\eta').
\end{equation}
Then the dynamics is encoded in its \emph{infinitesimal generator $ \mc{L}_N $}, defined trough the equation
\begin{equation}\label{eq:Mgen}
\mathbb{E}(f(\eta_{t+h})|\eta_t)-f(\eta_t)=(\mc{L}_Nf)(\eta_t)h+o(h),
\end{equation}
so that the expected infinitesimal increment of $f(\eta_t)  $ is $ (\mc{L}_Nf)(\eta_t)dt $, with $ \mathbb{E} $ the expectation respect to the distribution $ \mathbb{P} $ over the trajectory on the configuration space. Taking $ \eta_t=\eta, $ from \eqref{eq:Mrates} one deduces  $ \mc{L}_N(\eta,\eta')=c(\eta,\eta') $ and $ c(\eta,\eta)=-\us{\eta'\neq\eta}{\sum}c(\eta,\eta') $.
The solution to the backwards equation 
$
\frac{d}{dt}f_t=\mc{L}_Nf_t,
$
with initial conditions $ f_0=f $ is the Markov semigroup $ e^{t\mc{L}_N} $ according to
\begin{equation}\label{eq:fev}
 f_t(\eta)=\us{\eta'\in\Sigma_N}{\sum}e^{t\mc{L}_N }(\eta,\eta')f(\eta').
\end{equation}
The transition probabilities of the Markov process $ \eta_t $ are then given by  the kernel of the semigroup generated by $ \mc{L}_N $, i.e.
$
p_t(\eta,\eta')=e^{t\mc{L}_N}(\eta,\eta')
$.
For an initial probability measure $ \mu_0=\mu $ on $ \Sigma_N $ the solution to
$
\frac{d}{dt}\mu_t=\mu_t\mc{L}_N
$
is
\begin{equation}\label{eq:muev}
\mu_t(\eta)=\sum_{\eta'\in\Sigma_N}\mu(\eta')e^{t\mc{L}_N}(\eta',\eta).
\end{equation}

\vspace{1cm}
We want to study  treatable and physically reasonable models, so we are treating  local dynamics.  To each $ x $ we associate a family $ \mathscr{A}_x$ of  bounded domains $ \Lambda_x\subset V_N $ such that  $ x \in \Lambda_x $ and in each of them we have a Poisson clock, when the one in $ \Lambda_x $ rings   $ \eta $  updates  inside   this domain  depending on the configuration in a  bounded region $ \tilde\Lambda_x \supset\Lambda_x$. We call $\eta^{\Lambda_x} $  the configuration with the values of $ \eta  $ updated in $ \Lambda_x $ and the same configuration outside of it. The updating is according to some rules that will be specified later on to define different microscopic models. The locality for our dynamics is defined as follows
\begin{definition}\label{d:locrate}
A Markov dynamics   is \e{local} if $c_{\Lambda_x}(\eta):= c(\eta,\eta^{\Lambda_x}) $ is function only of the configuration $ \eta  $ restricted to a bounded domain $ \tilde\Lambda_x\supset\Lambda_x $. 
The domains $ \Lambda_x $ that we are using are single vertexes or couples of vertexes called  \e{edges}, these ones can be considered oriented or non-oriented.  We should use two different notation for $ \Lambda_x $  to specify if we are using oriented or non-oriented objects,  for convenience we do this just when a model is specified explicitly.
\end{definition}
Let's to do a simple  example to clarify the notation.
\begin{example}\label{e:eta^x}
 For each $ x $ we can consider the collection of subset $ \mc A_x $ containing  the only set $ \Lambda_x={x} $. Let $ \Sigma $ be $ \{0,1\} $. When the clock in $ x  $ rings the configuration $ \eta $ is updated with $ \eta^{\Lambda_x}=\eta^x $ defined for example by $ \eta^{x}(z):=1-\eta(x)  \textrm{ if }\ z=x  $  and $\eta^{x}(z):= \eta(z)  \textrm{ if }\ z\neq x  $.
\end{example}
We focus on translational invariant systems. Let $ \tau_z $ be the shift by $ z $ on $ \mathbb{Z}^n $ defined by the relation $ \tau_z\eta(x):=\eta(x-z) $ with $ z\in\mathbb{Z}^n $ and for a function $ h:\eta\to h(\eta)\in \mathbb{R} $ we define $\tau_z h(\eta):=h(\tau_{-z}\eta)  $, moreover for a domain $ B\subseteq V_N $ we define  $ \tau_z B:=B+z $. The invariance property that we want is what we call covariance by translations, i.e.

\begin{definition}
Let $  V_N $ be the set of vertices of the   \e n-dimensional discrete torus $ \mathbb{T}^n_N = \mathbb{Z}^n/N \mathbb{Z}^n$.
We say that a dynamics $ c(\eta,\eta') $ is \emph{translational covariant} if    we have
\begin{equation}\label{eq:tcov}
c_{\Lambda_x}(\eta)=c_{\tau_z\Lambda_x}(\tau_z\eta)
\end{equation}
for all $ z\in\mathbb{Z}^n$, all $ \Lambda_x $  
and all $ \eta\in\Sigma_N $. The translation $ z $ is defined modulo $  N $.
\end{definition}

We introduce also the notion of local function. Given a configuration $\eta$ and a subset $B \subseteq  V_N$ we denote by $\eta_B$
the restriction of the configuration to the subset $B$. Given two configurations $\eta,\xi$ and two
subsets $B,B'$ such that $B\cap B'=\emptyset$, we denote by $\eta_B\xi_{B'}$ the configuration of particles on $B\cup B'$ that coincides with $\eta$ on $B$ and coincides with $\xi$ on $B'$.

\begin{definition}
A function $h:\Sigma_N\to \mathbb R$ is called \e{local} if there exists a finite subset $B$ (independent of $N$ large enough) such that $h\left(\eta_B\xi_{B^c}\right)=h(\eta)$ for any $\eta,\xi$. We denote by $B^c$ the \e{complementary set} of $B$. The minimal subset $B$ for which this holds is called the \e{domain} of dependence of the function $h$ and is denoted by $D(h)$.
\end{definition}

The domain $  V_N $ has to be taught as the physical domain where particles or masses interact.
The interaction is expressed stochastically with the transition rates $ c(\eta,\eta') $ from $ \eta $ to $ \eta' $ along the rules encoded in $ \mc{L}_N $. Dividing the right hand side  of \eqref{eq:Mgen} where $ \eta_t=\eta $ by an infinitesimal increment $ h=dt $ we get
\begin{equation}\label{eq:gen}
\mc{L}_Nf(\eta)=\sum_{\eta'\in\Sigma_N}c(\eta,\eta')[f(\eta')-f(\eta)],
\end{equation}
where $ f $ is an observable. Particles and masses are  considered to be indistinguishable. Collecting together the possible transitions for each  $ x $, we write the right hand side of \eqref{eq:gen} for the particles case as
\begin{equation}\label{eq:partgen}
(\mc{L}_Nf)(\eta)=\sum_{x\in V_N}\sum_{\Lambda_x\in\mathscr{A}_x}c_{\Lambda_x}(\eta)(f(\eta^{\Lambda_x})-f(\eta)),
\end{equation}
where $ \eta(x)\in\Sigma\subseteq\bb N $.
In  the case of example \ref{e:eta^x} dynamics \eqref{eq:partgen} becomes
\begin{equation}\label{eq:xgen}
(\mc{L}_Nf)(\eta)=\sum_{x\in V_N} c_{x}(\eta)(f(\eta^{x})-f(\eta)).
\end{equation}
While for interacting masses-energies, the  variables $ \eta $  range in a continuous subset  of $ \bb R $  and we write \eqref{eq:gen} as
\begin{equation}\label{eq:massgen}
(\mc{L}_Nf)(\eta)=\sum_{x\in V_N}\sum_{\Lambda_x\in\mathscr{A}_x}\int\Gamma^\eta_{\Lambda_x}(dj)(f(\eta^{\Lambda_x,j})-f(\eta)),
\end{equation} 
where $ j $ is a real parameter  distributed according to  $ \Gamma^\eta_{\Lambda_x}(dj) $ and determining the updated configuration $ \eta^{\Lambda_x,j} $ in $ \Lambda_x $. When $ \Lambda_x $ is a couple  of vertices $ j $ will be  a current flowing from a vertex to the other one along the rules in the dynamics. Rates  can be conservative, in this case  particles change the positions in $ \Lambda_x $ but the total number is the same and masses-energies are exchanged between the variables in $ \Lambda_x $ in such a way the conservation of mass-energy is respected. Of course the dynamics can be non-conservative, in this case particles and masses-energies can be created or annihilated in $ \Lambda_x $. 

Let's do few examples for the particles case after having introduced some notation. Other examples with more details are being treated hereafter.
We called the discrete physical domain of the dynamics with $  V_N $. When we are thinking about it as a graph with the edges connecting its lattice sites we are calling it  $  V_N $  and its elements  vertices. We denote with   $\mathcal E_N$ the set of all couples of vertices $\{x,y\}$ of $ V_N$ such that $ y=x\pm \veps\, e_i $ where $ e_i $ is the canonical versor in $ \mathbb{Z}^n $ along the direction $ i $. Elements of $ \mc E_N $ are named  non-oriented edges or simply edges. In this way we have an non-oriented graph $ ( V_N,\mc E_N) $. To every non-oriented graph $( V_N, \mathcal E_N)$ we associate canonically an oriented graph $( V_N,E_N)$ such that the set of oriented edges $E_N$ contains all the ordered pairs $(x,y)$ such that $\{x,y\}\in\mathcal E_N$. Note that if $(x,y)\in E_N$ then also $(y,x)\in E_N$. If $e=(x,y)\in E_N$ we denote
$e^-:=x$ and $e^+:=y$ and we call $ \mf e:=\{x,y\} $ the corresponding non-oriented edge.  
Accordingly we introduce the  notation $\eta^x$, $\eta^{x,y}$ and $\eta^{\mf e}$ where $x,y\in  V_N$ and $\mf e\in \mc E_N$. The notation $\eta^x$ was already introduced in example \ref{e:eta^x}, it is used when  $\Sigma=\{0,1\}$ and denotes a configuration of particles obtained changing the value in the single site $x$. The configuration $\eta^{x,y}$ is obtained moving one particle from $x$ to $y$. The configuration $\eta^{\mf e}$ is obtained exchanging the values at the two endpoints of $\mf e\in \mc E_N$ (note that $\eta^{x,y}$ and $\eta^{\{x,y\}}$ are different notation). More precisely we have
\begin{equation}\label{afterj1}
\eta^{x}(z):=\left\{
\begin{array}{ll}
1-\eta(x) & \textrm{if}\ z=x \\
\eta(z) & \textrm{if}\ z\neq x\
\end{array}
\right.,
\end{equation}
\begin{equation}\label{afterj2}
\eta^{x,y}(z):=\left\{
\begin{array}{ll}
\eta(x)-1 & \textrm{if}\ z=x \\
\eta(y)+1 & \textrm{if}\ z=y\\
\eta(z) & \textrm{if}\ z\neq x,y
\end{array}
\right.,
\end{equation}
\begin{equation}\label{afterj3}
\eta^{\mf e}(z):=\left\{
\begin{array}{ll}
\eta(e^-) & \textrm{if}\ z=e^+ \\
\eta(e^+) & \textrm{if}\ z=e^-\\
\eta(z) & \textrm{if}\ z\neq e^\pm\,.
\end{array}
\right.
\end{equation}
We call respectively $c_{x}(\eta)$, $c_{x,y}(\eta)$ and $c_{\mf e}(\eta)$ the corresponding rate of transitions. These dynamics are defined from definition \ref{d:locrate} similarly to $ c_x(\eta) $ in example \ref{e:eta^x}, to get a dynamic $ c_{x,y}(\eta) $ we  orient the domains $ \Lambda_x= \{x,y\}\in \mc E_N $ and associate with each $ x $  the oriented edges $ \mc A_x=\{(x,y),(y,x):y=x+\veps\,e_i,i\in 1,\dots,n,{x,y}\in \mc E_N\} $, expression \eqref{eq:partgen} becomes
\begin{equation}\label{eq:x,y gen}
\mc{L}_Nf(\eta)=\sum_{(x,y)\in E_N}c_{x,y}(\eta)(f(\eta^{x,y})-f(\eta)).
\end{equation}
Note in the first case \eqref{afterj1} the dynamics is non-conservative because a particle is created or annihilated in $ x $, while the other two \eqref{afterj2} and \eqref{afterj3} are conservative because particles just move between vertices. We can also conceive dynamics that are a superimposition of a conservative part and non-conservative one, these are called \emph{reaction-diffusion} dynamics, for example superimposing \eqref{eq:x,y gen} to \eqref{eq:xgen} we have
\begin{equation*}\label{key}
\mc{L}_Nf(\eta)=\sum_{(x,y)\in E_N}c_{x,y}(\eta)(f(\eta^{x,y})-f(\eta))+\sum_{x\in  V_N}c_x(\eta)(f(\eta^x)-f(\eta)).
\end{equation*}

The models we presented so far describe the dynamics inside the domain $  V_N $. The dynamics at its boundary $ \p V_N $ has to be introduced superimposing a boundary generator as in next subsection \ref{ss:boundary}.
When the boundary conditions are periodic, that is $  V_N=\mathbb{T}^n_N $, we have just the  bulk part $ \mc L_Nf(\eta) $.

\subsection{Driven interacting systems}\label{ss:boundary} We can represent non-isolated systems too. Physically, these are systems which are open with respect to the number of particles, in other words matter can be exchanged with the surroundings.  We model this effect of external reservoirs superimposing at the boundaries  a creation/annihilation of particles or masses at exponential times. In general, the boundary term forces a flux of current in the steady state and the dynamics is not reversible, usually we have reversibility\footnote{This concept is introduced later in this chapter in section \ref{sec:inv measure}} only for very special boundary condition. Since nonzero steady current is a peculiarity of nonequilibrium stationarity states,  driven interacting systems with an invariant measure can be considered a good representation for these physical states at the microscopical scale.

So to get driven interacting systems we add to the generators $ \mc{L}_N $ of \eqref{eq:partgen} and \eqref{eq:massgen} a boundary generator $ \mc{L}_b $ defined respectively
\begin{equation}\label{eq:bgenpart}
(\mc{L}_bf)(\eta)=\sum_{x\in\partial V_N}c_{x,b}(\eta)(f(\eta^{x})-f(\eta))
\end{equation}
and
\begin{equation}\label{eq:bgenmass}
(\mc{L}_bf)(\eta)=\sum_{x\in\partial V_N}\int\Gamma^{\eta}_{x,b}(dj)(f(\eta^{x,j})-f(\eta)),
\end{equation}
where $ \Gamma^{\eta,b}_{x}(dj) $ is a distribution on the real line defined in a similar way to $ \Gamma^\eta_{\Lambda_x}(dj) $ in \eqref{eq:massgen}. The indexes $ b  $ stay for " boundary ".

\section{Reversibility and invariant measure}\label{sec:inv measure}
First, we remember we deal with finite space state. A stationary state of the system corresponds to a probability distribution $ \mu_N $ (ensemble) on the configuration space $ \Sigma_N $. A state is \emph{invariant} (stationary) under the dynamics if
\begin{equation}\label{eq:invcond1}
\bb E_{\mu_N}(\mc{L}_Nf)=0 \textrm{ for all observables } f.
\end{equation} 
where $ \bb E_{\mu_N} $ denotes the expectation with respect to $ \mu_N $.
In the particle case (for the energies-masses case see remark \ref{r:sum}), a necessary and sufficient condition for a state $ \mu_N $ to be invariant is
\begin{equation}\label{eq:inv}
\sum_{\eta\in\Sigma_N}\mu_N(\eta)e^{t\mc L_N}(\eta,\eta')=\mu_N(\eta'),
\end{equation}
This means that distributing the initial condition $ \eta $ according to $ \mu_N $, then the distribution of $ \eta_t $, at any later time $ t\geq 0 $, is again $ \mu_N $.
Condition \eqref{eq:invcond1} is also equivalent to \begin{equation}\label{invcond2}
\mu_N\mc{L}_N=0.
\end{equation} 
We consider only \emph{irreducibile} models, i.e. there is a strictly positive probability to go from any configuration to any other. In this case, according to general results on Markov chains, the invariant state is unique, strictly positive and it coincides with the limiting distribution of the systems when $ t\to\infty $.

If the generator $ \mc{L}_N $ satisfies the detailed balance condition with respect to some probability distribution $ \mu_N $, namely
\begin{equation}\label{eq:dbc2}
\bb E_{\mu_N}(g\mc{L}_N f)=\bb E_{\mu_N}(f\mc{L}_N g) \textrm{ for all observables $ f,g $},
\end{equation}      
or equivalently
\begin{equation}\label{eq:dbc1} 
{\mu_N}(\eta)c(\eta,\eta')= {\mu_N}(\eta')c(\eta',\eta) \textrm{ for any pair of states $ \eta,\eta'\in\Sigma_N $} 
\end{equation}                             
then $ \mu_N $ is necessarily an invariant state. In this case the process is said to be \emph{time reversal invariant} or simply \emph{reversible} (with respect to $ {\mu_N} $). We talk about reversibility because of the following fact. Let $ \mathbb{P_{\mu_N}} $ be the probability distribution on the space of paths $ \{\eta_t\}_{t\geq 0} $ with initial configuration $ \eta_0$ distributed according to the invariant state $ \mu_N $. Since ${\mu}_N  $ is invariant, the probability $ \bb P_{\mu_N}  $ is invariant with respect to time shifts.
We can thus regard $ \bb P_{\mu_N} $ as a distribution on paths defined
also for $ t \leq 0 $. This probability distribution is invariant
under time reversal if and only if the detailed balance
condition \eqref{eq:dbc2} holds. Indeed, if $ \theta $ is the time reversal,
i.e., $ (\theta\eta)_t := \eta_{-t} $ , we have that $ \bb P_{\mu_N}\circ\theta $ is the stationary
Markov process with generator $  \mc{L}_N^*  $, the adjoint of $ \mc{L}_N  $ with
respect to $ \mu_N $ , and condition \eqref{eq:dbc2} is precisely the condition that $ \mc{L}_N^* = \mc{L}_N $. When the unique invariant state does not satisfy the
detailed balance condition \eqref{eq:dbc2}, the corresponding process is not time reversal invariant. The time reversed Markov process having the same invariant measure $ \mu_N $ and generator $ \mc L_N^* $ can be constructed defining on the same state space of $ \mc L_N $ the Markov chain with transition rates
\begin{equation}\label{eq:trevMC}
 c^*(\eta,\eta'):=\frac{\mu_N(\eta')c(\eta',\eta)}{\mu_N(\eta)} ,
\end{equation}
in the present stationary case it gives to a set of trajectories exactly the same probability that the original process gives to the set of time reversed trajectories and has the same invariant measure $ \mu_N  $ of the process $ c(\eta,\eta') $, moreover when \eqref{eq:dbc1} holds $ c^*(\eta,\eta')=c(\eta,\eta') $ otherwise they are different. 
Time reversal invariant processes correspond to equilibrium thermodynamic states: the probability to see a path or its time reversal is the same, i.e we don't "distinguish" the future from the past. The converse is not necessarily true:
there can be microscopic models not invariant under
time reversal corresponding to equilibrium macroscopic
states \cite{BJ-L04},\cite{GJ-LL96,GJ-LL99}, and \cite{GJLV97}. This is not surprising: going from
the microscopic to the macroscopic description there is
loss of information.

\subsection{Transition Graph and Kolmogorov criterion}\label{ssec:Tgrap+Kol} We recall the Kolmogorov criterion for reversibility of a Markov Chain. This is
particularly useful when the invariant measure is not known. Consider a continuous
time Markov chain with finite state space $ V $ and having transition rate $ c(\eta, \eta') $ for
the jump from $ \eta\in V $ to $ \eta'\in V $. We define its \emph{transition graph}. This is the graph
having as vertices the set  $ V $ and  edges $ E  $ all the pairs $ (\eta,\eta' ) $ such that
$ c(\eta, \eta' ) > 0 $. Here to denote the invariant measure we are not using the subscript $ N $ as in last section \ref{sec:inv measure} because the discussion is for  general Markov processes, including the cases of the dynamics with $ V=\Sigma_N $.
We consider finite ($ |V|<\infty $) and irreducible dynamics hence there exists a unique invariant measure $ \mu $   characterized by the stationary \eqref{invcond2}, i.e.
\begin{equation*}\label{eq:invgraph}
\sum_{\eta':(\eta,\eta')\in E}\mu(\eta)c(\eta,\eta')=\sum_{\eta':(\eta',\eta)\in E}\mu(\eta')c(\eta',\eta).
\end{equation*}
Reversibility will be satisfied when
\begin{equation}\label{eq:revgraph}
\mu(\eta)c(\eta,\eta')=\mu(\eta')c(\eta',\eta), \hspace{1cm} \forall (\eta,\eta')\in E.
\end{equation}
 A set $ C = (\eta_0 , \dots , \eta_n ) $ is  \e{cycle} in the transition graph if  $\eta_i\neq\eta_{i+1}$ for all  $i\in\{0,n-1\}$, $ \eta_{n+1}=\eta_0$ and $ (\eta_i, \eta_{i+1} ) \in E  $.  Reversibility  has an equivalent formulation, that is condition \eqref{eq:revgraph} holds if and only if for any cycle in the transition graph we have
\begin{equation*}\label{eq:kol1}
\prod_{i=0}^{n}c(\eta_i,\eta_{i+1})=\prod_{i=0}^{n}c(\eta_{n+1-i},\eta_{n-i}),
\end{equation*}
this condition in turn is equivalent to ask 
\begin{equation}\label{eq:kol2}
\sum_{i=0}^{n}\log \frac{c(\eta_i,\eta_{i+1})}{c(\eta_{i+1},\eta_{i})}=0
\end{equation}
for any cycle $ C = (\eta_0, \dots, \eta_n ) $.
\begin{definition}
A discrete vector field is a function $ j: E \to\mathbb{R} $ that is \emph{antisymmetric}, i.e.
$ j(\eta,\eta') = - j(\eta',\eta) $ for any $ (\eta,\eta') \in E  $. A discrete vector field $ j $ is called a \emph{gradient}
vector field if there exists a function $ h:V\to\mathbb{R}  $ such that $ j(\eta, \eta') = f (\eta') - f (\eta) $.
\end{definition}
Condition \eqref{eq:kol2} has a simple and direct geometric
interpretation. First of all observe that both conditions \eqref{eq:dbc1} and \eqref{eq:kol2} can hold
only if when $ (\eta,\eta') \in E $ then also $ (\eta' , \eta) \in E $. We can then define on the transition
graph $ (V, E) $ the discrete vector field \begin{equation}\label{eq:kolvec}
j(\eta, \eta' ) := \log\frac{c(\eta,\eta')}{c(\eta',\eta)}
\end{equation}
\begin{proposition}
Condition \eqref{eq:kol2} is equivalent to require that the discrete vector field \eqref{eq:kolvec} is of gradient type.
\end{proposition}
\begin{proof}
The proof is identical to the one of proposition \ref{prop:prop1form} in chapter \ref{ch:DEC} where discrete vector field are treated extensively, therefore we refer there for the proof.
\end{proof}

\section{Particles models}\label{sec:part models}
Here we describe the microscopic particles models we are using in later sections. We consider periodic boundary conditions, i.e the generator consists only of the bulk dynamics. If the dynamics is conservative then irreducibility will be restricted to the configurations with the same number of particles. We name \emph{ nearest neighbours dynamics}  the conservative dynamics such that particles are exchanged between nearest neighbours sites, for example $ c_{\Lambda_x}(\eta) $ with $ \Lambda_x=(x,y)=e $ or $ \Lambda_x=\{x,y\} $, remember that   $ (x,y)\in E_N $ as in section \ref{sec:CTMC}. The conservative models we are going to study in this work belong to this subclass.  

\subsection{Exclusion process}\label{ss:expr} In exclusion process particles move according to independent random walks with the exclusion rule that there cannot be more than one particle in a single lattice site (hard core interaction). This is a kind of classical Pauli principle. Hence this is a conservative dynamics preserving the number of particles and for
which the elementary random jumps are constituted by the exchange of the occupation numbers between nearest neighbors vertices.
We  consider $  V_N=\mathbb{T}^n_N $, i.e periodic boundary conditions on a torus of edges one with $ N $ sites.  If $ (x,y)\in E_N $ then $ (x,y)=(x,x\pm e_i) $ for some $ i\in{1,\dots,n} $, with $ e_i $ canonical versors. The space of state is $ \Sigma_N=\{0,1\}^{\mathbb{T}^n_N} $.
The stochastic dynamics can be summarized writing down the generator \eqref{eq:partgen}, it takes the form
\begin{equation}\label{eq:EPgen}
\mathcal L_N^\textrm{E}f(\eta)=\sum_{\mf e \in \mc{E}_N}c_{\mf e}(\eta)\Big(f(\eta^{\mf e})-f(\eta)\Big).
\end{equation}
Rates $ c_{\{x,y\}}(\eta) $ can be written in the general form
\begin{equation}\label{eq:EPrate}
c_{\{x,y\}}(\eta)=\eta(x)(1-\eta(y))\tilde c_{x,y}(\eta)+\eta(y)(1-\eta(x))\tilde c_{y,x}(\eta),
\end{equation}
We define these jump rates $\tilde c_{e}(\eta)$ so that, taking $ e=(x,y) $, if the configuration $\eta$ has a particle at $x$ and an empty site at $y$ then the particle at $x$ jumps to $y$ with rate $\tilde c_{(x,y)}(\eta)$ while if the configuration $\eta$ has a particle at $y$ and an empty site at $x$ then the particle at $y$ jumps to $x$ with rate $\tilde c_{(y,x)}(\eta)$. When $ \tilde c_{x,y}(\eta)  $ and $ \tilde c_{y,x}(\eta) $ don't depend on $ \eta $ the dynamics is called \e{simple exclusion process} (SEP), if they are equal the dynamics is a \e{symmetric simple exclusion process} (SSEP) while  if they are different it is an \e{asymmetric simple exclusion process} (ASEP). If we want to write the generator with ordered sets, we proceed like we did at the end of section \ref{sec:CTMC} to get \eqref{eq:x,y gen}. Since we are on a torus we can take  $ \mc A_x=\{(x,y),(y,x) \} $ for all $ x $ in \eqref{eq:partgen}, then the transition rates have the general form
\begin{equation}\label{eq:EPjump}
c_{x,y}(\eta):=c(\eta,\eta^{x,y})=\eta(x)(1-\eta(y))\tilde c_{x,y}(\eta)
\end{equation}
and the generator becomes
\begin{equation}\label{eq:EPjumpgen}
\mathcal L_N^\textrm{E}f(\eta)=\sum_{e\in E_N}\eta(x)(1-\eta(y))\tilde c_{e}(\eta)\Big(f(\eta^{e})-f(\eta)\Big).
\end{equation}
Jump rates  $\tilde c_{x,y}(\eta)$ and $\tilde c_{y,x}(\eta)$ are local function non-depending on the configuration in $x$ and in $y$.

\subsection{The 2-SEP}The model we are considering is the 2-SEP (2-simple exclusion process) for which on each lattice site there can be at most 2 particles. The generalization to k-SEP can be easily done likewise. Namely the interaction is simply hardcore but with respect to SEP, where in each site can be at most one particle, in every site there can be at most $ k $ particles. The state space is $ \Sigma_N=\{0,1,2\}^{\bb T^n_N} $. We  consider $  V_N=\mathbb{T}^n_N $ and $\Sigma=\{0,1,2\}$.
The dynamics is defined by 
\begin{equation}\label{eq:2SEPgen}
\mathcal L_N^\textrm{2-SEP}f(\eta)=\sum_{(x,y)\in E_N}c_{x,y}(\eta)\left(f(\eta^{x,y})-f(\eta)\right).
\end{equation}
with rates
$$
c_{x,y}(\eta)=\chi^+(\eta(x))\chi^-(\eta(y))
$$
where $\chi^+(\alpha)=1$ if $\alpha >0$ and zero otherwise while $\chi^-(\alpha)=1$ if $\alpha<2$ and zero otherwise.

\subsection{Glauber dynamics}\label{ss:Gdyn} 
Glauber dynamics  is an irreducible non-conservative dynamics where the configuration changes in one site at a time through the annihilation of a particle when is  present and  the creation of a particle in the opposite case. We consider $ V_N=\bb T^n_N $ and $ \Sigma_N=\{0,1\}^{\bb T^n_N} $. At each event of the Markov chain the configuration increases or decreases by one the number of particles in one site. 
This class of models has  generator of the form
\begin{equation}\label{eq:Glgen}
\mathcal L_N^\textrm{G}f(\eta)=\sum_{x\in  V_N}c_x(\eta)\Big(f(\eta^x)-f(\eta)\Big), \mathrm{ where}\ c_x(\eta):=c(\eta,\eta^x).
\end{equation}
In this case the general form of the rates is
\begin{equation*}\label{Glauber rate form}
c_{x}(\eta)=\eta(x)c_{x}^{-}(\eta)+(1-\eta(x))c_{x}^{+}(\eta),
\end{equation*}
where $c_{x}^{-}$ and $c_{x}^{+}$ are local function non-depending on the configuration in $x$.  We can consider a dynamics where particles can both jump and be created-annihilated too, the  generator is of the form
\begin{equation}\label{Gl+Kaw}
\mathcal{L}_Nf(\eta)=\left[ \mathcal{L}_N^\textrm{E}+\mathcal{L}_N^\textrm{G} \right]f(\eta).
\end{equation}
Let the Glauber dynamics be defined by the rates
\begin{equation}\label{eq:specialG}
\left\{ 
\begin{array}{ll}
c_x^+=\tau_x c^+(\eta)\\
c_x^-=\tau_x c^-(\eta)\\
\end{array}.
\right.
\end{equation}
If the exclusion process is the simple exclusion process (SEP), i.e. $ c_e(\eta)=1  $ in \eqref{eq:EPgen},  the full dynamics \eqref{Gl+Kaw} is reversible with respect to the unique Bernoulli invariant measure
\begin{equation}
\mu_p(\eta)=\underset{x\in V_N}{\prod}p^{\eta(x)}(1-p)^{(1-\eta(x))},
\end{equation}
provided that for all $ \eta\in\Sigma_N $ the rates \eqref{eq:specialG} satisfies 
\begin{equation}\label{eq:Glspecial2}
\frac{c^-_x(\eta)}{c^+_x(\eta)}=\frac{1-p}{p}.
\end{equation}

\subsection{Zero range model} The zero range model   is a conservative dynamics but without bound on the number of particles which can occupy a single site. We consider $ V_N=\bb T^n_N $ and $ \Sigma_N=\bb N^{\bb T^n_N} $. In this dynamics a particle interact only with the other particles occupying the same site. The interaction can be either attractive or repulsive and the dynamics is nearest neighbours.  This class of models has  generator of the form
\begin{equation}\label{eq:Zgen}
\mathcal L_N^\textrm{Z}f(\eta)=\sum_{(x,y)\in E_N}c_{x,y}(\eta)\Big(f(\eta^{x,y})-f(\eta)\Big), 
\end{equation}
where the general form of the rates is defined trough a function $ g:\bb N\to \bb R^+ $  such that $ g(0)=0 $ and $ g(k)>0 ,k\geq 1$ and that describes the type of interaction, i.e.
\begin{equation}\label{eq:Zrate}
\left\{ 
\begin{array}{ll}
c_{x,y}(\eta)=g(\eta(x)) \textrm{  for $ (x,y)\in E_N $},\\
c_{x,y}(\eta)=0 \textrm{  otherwise}.
\end{array}
\right.
\end{equation}

\subsection{Weakly asymmetric particles models} \label{ss:wapmod}The effect of an external field is modeled by perturbing the rates and giving a net drift toward a specified direction.  

Let $ F:\Lambda\to\mathbb{R}^n $ be a smooth vector field with components $ F(x)=(F_1,\dots,F_n) $, describing the force acting on the particles of the systems. We associate to $ F $ a discrete vector field on the lattice that corresponds to a discrete version of the continuous vector field defined by 
\begin{equation}\label{eq:disF}
\mathbb{F}(x,y)=\int_{(x,y)}F(z)\cdot dz,
\end{equation}
note that $ \bb F(x,y)=-\bb F(y,x) $. In \eqref{eq:disF} $(x,y)$ is an oriented edge of the lattice (which has length of order $\frac{1}{N}  $), the integral is a line integral
that corresponds to the work done by the vector field $F$ when a particle moves from $x$ to
$y$. So we think about $ \mathbb F(x,y) $ as work done per particle. The value of   $\mathbb F\left(y,x\right)$, by antisymmetry,  corresponds to minus the value in \eqref{eq:disF}.
Let $ c_{\Lambda_x}(\eta) $ a nearest neighbours conservative dynamics, that is $ \Lambda_x=(x,y)\in E_N $.  For example  the exclusion process \eqref{eq:EPjump}  and the 2-SEP \eqref{eq:2SEPgen}, respectively the perturbed rates are defined by 
\begin{equation}\label{eq:pertrate}
c_{\Lambda_x}^{\mathbb F}(\eta)=c_{\Lambda_x}(\eta)e^{\mathbb F(x,y)j},
\end{equation}
where $ j $ is the number of particle moving from $ x $ to $ y $ during the jump $ c_{\Lambda_x}(\eta) $, for models \eqref{eq:EPjump} and \eqref{eq:2SEPgen} the current $ j=1 $. 
When the size $ |y-x| $ is of order $ 1/N $, then the discrete vector field \eqref{eq:disF} is of order $ 1/N $ too and we have $ c^{\mathbb F}_{\Lambda_x}(\eta)=c_{\Lambda_x}(\eta)\left(1+\mathbb F(x,y)\right)+o(1/N) $. If $ F=-\nabla H $ is a gradient vector field, then $ \mathbb F(x,y)=H(x)-H(y) $, and equation \eqref{eq:pertrate} becomes $ c_{\Lambda_x}^{\mathbb F}(\eta)=c_{\Lambda_x}(\eta)e^{(H(x)-H(y))j} $. We already underlined that $ \mathbb{F}(x,y) $ in \eqref{eq:disF} is of order $ 1/N $, this means that on the microscopic scale the external field is small with the scaling parameter and  this is the  reason for the name weakly asymmetric when $ N $ is large.

\subsection{Instantaneous particle current }\label{subsec:Ipc} Given $ (x,y)\in E_N $ now we indicate with $ c_{x,y}(\eta) $ a general nearest neighbours local  dynamics where  a particle in $ x $ jumps in $ y=x\pm e_i $ according to a  conservative dynamics and correspondingly $ \eta $ change into $ \eta^{xy} $. Therefore the generator is \begin{equation*}\label{eq:jumpgen}
\mc{L}_N f(\eta)=\sum_{(x,y)\in E_N}c_{x,y}(\eta)\big(f(\eta^{x,y})-f(\eta)\big).
\end{equation*}
The \e{instantaneous current} for these stochastic lattices  is defined by
\begin{equation}\label{eq:istcur}
j_\eta(x,y):=c_{x,y}(\eta)-c_{y,x}(\eta)\,.
\end{equation}
For each fixed configuration $\eta$ this is a discrete vector field. The intuitive interpretation of the instantaneous current is the rate at which particles cross the bond $(x,y)$. Let $\mathcal N_t (x,y)$ be the number of particles that jumped from site $x$ to site $y$ up to time $t$. The \e{current flow} across the bond $(x,y)$ up to time $t$ is defined as
\begin{equation}\label{eq:currver}
J_t (x,y):=\mathcal N_t (x,y)-\mathcal N_t (y,x)\,.
\end{equation}
This is a discrete vector field depending on the trajectory $ \{\eta_t\} $ of the system of particles,  by definition $ J_t (x,y)=-J_t (y,x) $. Note that while $ J_t(x,y) $ is a function on path space, i.e. it depends on the particular history, the instantaneous current $ j_{\eta}(x,y) $ is a function on configuration space, i.e. it depends only on the present configuration.
Between the instantaneous current $ j_\eta(x,y) $ and the current flow $ J_t (x,y) $ there is a strict  connection given by the key observation (see for example \cite{Spo91} section 2.3 in part II) that
\begin{equation}\label{eq:mart}
M_t (x,y)=J_ t (x,y)-\int_0^tj_{\eta(s)}(x,y)ds
\end{equation}
is a martingale. This allows to treat the difference between $ J_ t (x,y)  $ and the integral $ \int_0^t ds\, j_{\eta(s)}(x,y) $ as  microscopic fluctuation term and computing $ j_\eta(x,y) $ as follows. Consider an initial configuration $ \eta_0=\eta $, the explicit expression of the instantaneous current can be naturally obtained as the average flow per unit of time integrated over an infinitesimal time interval, i.e.
\begin{equation}\label{eq:flow/t}
 j_\eta(x,y):=\us{ t\to0}{\lim}\frac{{\mathbb{E}^{\eta}(J_{t}(x,y))}}{t}.
\end{equation}
The expectation is $ \mathbb{E}^{\eta}(J_t(x,y))=\int\mathbb{P}^\eta(d\{\eta_t\}_t)J_t(x,y) $, where the integration is over all trajectories $ \{\eta_t\} $ starting from $ \eta $ at time 0.  For a trajectory $\{\eta_t\}_t $ the probability to observe more than one jump goes like $ O(t^2) $, then it is negligible since we are interested in an infinitesimal time interval. While  from \eqref{eq:Mrates} we have that  
$ \us{t\to 0}{\lim}\frac{\mathbb{P}^\eta(\eta_t=\eta')}{ t}=\us{t \to 0}{\lim}\frac{p_{ t}(\eta,\eta')}{t} =c(\eta,\eta')$ when $ \eta'=\eta^{xy} $ or $ \eta'=\eta^{yx} $. Therefore  when the limit for $ t $ that goes to zero is considered, $ J_t(x,y) $  takes value $ +1 $ if a jump from $ x $ to $ y $ happens,  $ -1 $ it the opposite case happens and  $ 0 $ in the other cases. So the current defined in  \eqref{eq:flow/t}  becomes 
\begin{equation*}\label{eq:jfromJ}
j_\eta(x,y)= c_{x,y}(\eta)-c_{y,x}(\eta).
\end{equation*}

Defining the discrete divergence for a discrete vector field $ \phi $ on $ E_N $, i.e. an antisymmetric function $ \phi:E_N\to\mathbb{R} $, as $ \div\phi(x):=\us{y\sim x}{\sum}\phi(x,y) $, where the sum is on the nearest neighbours $ y\sim x $ of $ x $,  the local conservation of the number of particles is expressed by
\begin{equation}\label{eq:numcon1}
\eta_t(x)-\eta_0(x)+\div J_t(x)= 0.
\end{equation}
 Using \eqref{eq:mart} in \eqref{eq:numcon1} we get
\begin{equation}\label{eq:numcon2}
\eta_t(x)-\eta_0(x)+\int_0^t ds\, \div j_s(x)+\div M_t(x)=0.
\end{equation}

We can deduce that  at the equilibrium, that is when the detailed balance condition is true, the average flow $ \mathbb{E}^\eta_{\mu_N} (J_t(x,y)) $ is constantly zero, where  the subscript $ \mu_N $   indicates  the average respect to the equilibrium measure $ \mu_N $.   For a small time interval $ \Delta t $  from \eqref{eq:istcur}, \eqref{eq:flow/t} and  the detailed balance condition $ \mathbb{E}^\eta_{\mu_N}(J_{\Delta t}(x,y)) \sim \bb E_{\mu_N} (j_\eta(x,y)) {\Delta t}=0 $, since this is true for any time interval $ \Delta t $ and the current flow $ J_t(x,y) $ is additive we conclude $ \mathbb{E}^\eta_{\mu_N}(J_{t}(x,y))=0 $.  We use the fact that the right side of \eqref{eq:mart} is a martingale and we have that 
\begin{equation}\label{eq:asJ=asj}
 \bb E^\eta_{\mu_N}  (J_t(x,y))=\bb E_{\mu_N} (j_\eta(x,y)) t.
\end{equation}

A relevant notion in the derivation of the hydrodynamic behavior (see chapter \ref{ch:HSL} later) for diffusive particle systems is the definition of gradient particle systems. The basic definition is the following. A particle system is called of \emph{gradient type} if there exists a local function $h$ such that
\begin{equation}\label{eq:gradj}
j_\eta(x,y)=\tau_yh(\eta)-\tau_x h(\eta)\,.
\end{equation}
The relevance of this notion is on the fact that the proof of the hydrodynamic limit for gradient systems is extremely simplified.
Moreover for gradient models it is possible to obtain explicit expressions of the transport coefficients.
The instantaneous current of any translational covariant stochastic lattice gas, i.e. in the present framework \eqref{eq:tcov} becomes $ c_{x,y}(\eta)=c_{x+z,y+z}(\tau_z\eta) $, is translational covariant, this means that it satisfies the symmetry relation
\begin{equation}\label{eq:jinv}
j_\eta(x,y)=j_{\tau_z\eta}(x+z,y+z).
\end{equation}

\section{Energies-masses models}\label{sec:en models}
These models as we have seen in section \ref{sec:CTMC} have the peculiarity to exchange  continuous quantity between interacting lattice sites. We interpret lattice variables of these models as energy or mass along the context. 
For these models here and in what follows we indicate the configuration with $ \xi=\{\xi(x)\}_{x\in V_N} $.
We start describing the most famous model of this class the Kipnis-Marchioro-Presutti (KMP) model \cite{KMP82}, later on we describe others nearest neighbours model.

\subsection{KMP model} This is a generalized stochastic lattice gas on which energies associated to oscillators located at the vertices of a lattice randomly evolve. More precisely let $\Lambda=(0,1]^n$  and let $ V_N= \Lambda\cap \frac 1N \mathbb Z^n $. The oscillators are located at the vertices of the lattice. The value $\xi(x)\in \mathbb R^+$ is the energy associated to the oscillator at site $x$. The interpretation of the configuration $\xi$ as a configuration of energy follows by the original definition of the model \cite{KMP82}. Since we are discussing also generalizations and different models we prefer to switch the interpretation to a configuration of mass.  We will mainly consider the one dimensional case for which $\Lambda$ is an interval and $ V_N$ is a linear lattice. We call $x\in  V_N$ an internal vertex if all its nearest neighbors $y\in \frac 1N\mathbb Z^n$ belong also to $ V_N$. A vertex $x\in  V_N$ that is not internal is instead a boundary vertex in  $\partial  V_N$. The stochastic evolution is encoded in the generator that is of the type
\begin{equation}\label{eq:genKMP}
\mathcal L_N f(\xi) =\sum_{\{x,y\}\in\mc{E}_N} \mc{L}_{\{x,y\}}f(\xi) +\sum_{x\in \partial  V_N}\mc{L}_{b,x} f(\xi)\,.
\end{equation}
The first term in \eqref{eq:genKMP}
is the bulk contribution to the stochastic evolution while the second term is the boundary part of the dynamics that modelizes the interaction of the system with external reservoirs.

We introduce the model using the language of section \ref{sec:CTMC}, slightly different language from the classic one. This approach simplifies notation and is suitable for generalizations. Let $\varepsilon^x=\left\{\varepsilon^x(y)\right\}_{y\in  V_N}$ be the configuration of mass with all the sites different from $x$ empty and having unitary mass at site $x$. This means that $\varepsilon^x(y)=\delta_{x,y}$ where $\delta$ is the Kronecker symbol. The bulk contribution to the stochastic dynamics is defined by
\begin{equation}\label{eq:KMPbulkgen}
\mc{L}_{\{x,y\}}f(\xi):=\int_{-\xi(y)}^{\xi(x)}\frac{dj}{\xi(x)+\xi(y)}\big[f(\xi-j\left(\varepsilon^x-\varepsilon^y\right))-f(\xi)\big]\,.
\end{equation}
It is immediate to check that definition \eqref{eq:KMPbulkgen} is symmetric in $x$ and $ y $ so that we can consider without ambiguity a sum over unordered pairs in \eqref{eq:genKMP}. Formula \eqref{eq:KMPbulkgen} define the model as a random current model. The intuition behind the formula is the following. On each bond of the system there is a random flow of mass that happens accordingly to a random exponential clock with rate $1$. When this random clock rings there is a flow between the two endpoint sites that is uniformly distributed among all the possible currents that keep the masses at the two sites positive. The new configuration $\xi'=\xi-j\left(\varepsilon^x-\varepsilon^y\right)$ is the starting configuration minus the divergence of a current on the lattice different from zero on the single edge $(x,y)$ where it assumes the value $j$. The choice of a uniformly random current in \eqref{eq:KMPbulkgen} corresponds to the usual KMP dynamics.

The boundary part of the generator can be defined in several ways. Let us fix a possible definition that is good for symmetric and weakly asymmetric models. Consider the left boundary of a one dimensional system on $\Lambda=(0,1)$. The system is in contact with an external reservoir with chemical potential $\lambda<0$. The effect of the interaction with the source is that at rate one the value of the variable $\xi(1/N)$ is substituted by a random value exponentially distributed with parameter $-\lambda$. When the value $\xi(1/N)$ is substituted by the value $z$ we imagine that there is a current $j=z-\xi(1/N)$ across the edge
$\left(0,1/N\right)$. With a change of variables we can then write the boundary term of the generator with a random current representation as
\begin{equation}\label{eq:bKMPgen}
\begin{split}
 & \mc{L}_{b,1/N}\,f(\xi)=\int_0^{+\infty}|\lambda| e^{\lambda z}\left[f(\xi+\varepsilon^{1/N}(z-\xi(1/N)))-f(\xi)\right]\,dz \\
& =\int_{-\xi(1/N)}^{+\infty}|\lambda| e^{\lambda (\xi(1/N)+j)}\left[f(\xi+\varepsilon^{1/N}j)-f(\xi)\right]\,dj\,.
\end{split}
\end{equation}
At the  other boundary $ x=1 $ we have an identical dynamics but possibly with a different chemical potential. The dynamics \eqref{eq:KMPbulkgen}  can be generalized, as in \eqref{eq:massgen}, substituting the uniform distribution on $[-\xi(y),\xi(x)]$ with
a different probability measure (or just positive measure) $\Gamma_{x,y}^\xi(dj)$ , i.e.
\begin{equation}\label{eq:KMPgeneralized}
\mc{L}_{x,y}f(\xi):=\int\Gamma_{x,y}^\xi(dj)[f(\xi-j\left(\varepsilon^x-\varepsilon^y\right))-f(\xi)\big]\,.
\end{equation}

\subsection{Dual KMP model} A natural choice for $ \Gamma_{x,y}^\xi(dj) $ in \eqref{eq:KMPgeneralized} is the discrete uniform distribution on the integer points in $[-\xi(y),\xi(x)]$. This means that if $\xi$ is a configuration of mass assuming only integer values then
\begin{equation}\label{eq:KMPdrate}
\Gamma_{x,y}^\xi(dj)=\frac{1}{\xi(x)+\xi(y)+1}\sum_{i\in [-\xi(y),\xi(x)]}\delta_i(dj)
\end{equation}
where $\delta_i(dj)$ is the delta measure concentrated at $i$ and the sum is over the integer values belonging to the interval.
If the initial configuration is such that the values of the variables $\xi$ are all integers then this fact is preserved by the dynamics and we obtain a model that can be interpreted as a model of evolving particles. This is exactly the dual model of KMP \cite{KMP82}. We call the stochastic dynamics associated to the choice \eqref{eq:KMPdrate} KMPd where the last letter means \emph{dual}. The boundary dynamics can be fixed similarly to \eqref{eq:bKMPgen}. In this case it is natural to substitute the exponential distribution by a geometric one.

\subsection{Gaussian model} Another interesting model could be related to Gaussian distributions. In this case the interpretation in terms of mass is missing since the variables can assume also negative values. The bulk dynamics is defined by a distribution of current having support on all the real line and defined by
\begin{equation}\label{eq:gaussrate}
\Gamma_{x,y}^\xi(dj)=\frac{1}{\sqrt{2\pi \gamma^2}}e^{-\frac{\left(j-\frac{(\xi(x)-\xi(y))}{2}\right)^2}{2\gamma^2}}dj\,.
\end{equation}
In general we use the same notation both for a  measure and the corresponding density.
Note that all these models \eqref{eq:KMPbulkgen},\eqref{eq:KMPdrate} and \eqref{eq:gaussrate} share the symmetry $\Gamma^\xi_{x,y}(j)=\Gamma^\xi_{y,x}(-j)$ and for this reason we can write
in \eqref{eq:genKMP} a sum over unordered nearest neighbours sites. Also in this case it is possible to introduce a boundary part of the dynamics.

\subsection{Weakly asymmetric energies-masses models}\label{subsec:WAm} As in the particle  case we consider dynamics perturbed by a space and time dependent external field, so the random distribution of the current is changed. In particular suppose that on the lattice it is defined a discrete vector field $\mathbb F$. 
We can think about $ \mathbb{F} $ as discrete version of a continuous field $ F $ as in \eqref{eq:disF}.
Since $ \mathbb{F} $ is a discrete vector field, it is a collection of numbers $\mathbb F(x,y)$ for any ordered pair of nearest neighbors lattice points satisfying the antisymmetry relationships $\mathbb F(x,y)=-\mathbb F(y,x)$. If the vector field is time dependent these numbers are time dependent. The motion of the mass is influenced by the presence of the field and we have a model with a random current across each bond distributed according to a measure $\Gamma^{\mathbb F}$, since we want to perturb the symmetric nearest neighbours models we are considering in these subsections, we modify \eqref{eq:KMPgeneralized} as follows
\begin{equation}\label{eq:bulkxyE}
\mc{L}_{\{x,y\}}^{\mathbb F}f(\xi):=\int\Gamma^{\xi,\mathbb F}_{x,y}(dj)\big[f(\xi-j\left(\varepsilon^x-\varepsilon^y\right))-f(\xi)\big]\,,
\end{equation}
where in analogy with \eqref{eq:pertrate} the natural choice of the measure $\Gamma^{\mathbb F}$ is
\begin{equation}\label{eq:pertEnrate}
\Gamma^{\xi,\mathbb F}_{x,y}(dj)=\Gamma^{\xi}_{x,y}(dj)e^{\frac{\mathbb F(x,y)}{2}j}\,.
\end{equation}
The factor $\frac 12$ in the exponent appears just for convenience of notation in the following.
A perturbation of this type is for example the one used in \cite{BGL05} to compute dynamic large deviations for the KMP model and  corresponds therefore to the choice
\begin{equation}\label{eq:KMPpertrate}
\Gamma^{\xi,\mathbb F}_{x,y}(dj)=\frac{e^{\frac{\mathbb F(x,y)}{2}j}}{\xi(x)+\xi(y)}\chi_{[-\xi(y),\xi(x)]}(j)\,dj\,.
\end{equation}

By the symmetry of the measure $\Gamma$ and the antisymmetry of the discrete vector field $\mathbb F$ we have that $\Gamma^{\xi,\mathbb F}_{x,y}(j)=\Gamma^{\xi,\mathbb F}_{y,x}(-j)$ and we can define the generator considering sums over unordered bonds, hence the full dynamics is
\begin{equation}\label{eq:fullKMPgen}
\mc{L_N}f(\xi)=\sum_{\{x,y\}\in\mc{E}_N}\mc{L}_{\{x,y\}}^{\mathbb F}f(\xi)+\sum_{x\in\partial V_N}\mc{L}_{b,x}f(\xi).
\end{equation}

In principle the perturbation $ \bb{F} $ should affect also the boundary but we don't consider this because we are studying fast boundary dynamics like \eqref{eq:bKMPgen}. Macroscopically this fact won't be relevant, see the discussion in sections \ref{sec:SL} and \ref{s:per2}.

\subsection{Stationarity}\label{subsec:sta} For  models \eqref{eq:KMPbulkgen},\eqref{eq:KMPdrate} and their perturbed versions \eqref{eq:pertEnrate}  we want to study the detailed balance condition  with respect to   a full dynamics like in  \eqref{eq:fullKMPgen}. 
We have for the symmetric KMP model that the detailed balance condition
\begin{equation}\label{eq:KMPdbc}
\mu_N^\lambda(\xi)\Gamma_{x,y}^\xi(j)=\mu_N^\lambda\big(\xi-j\left(\varepsilon^x-\varepsilon^y\right)\big)
\Gamma_{x,y}^{\xi-j\left(\varepsilon^x-\varepsilon^y\right)}(-j)
\end{equation}
is satisfied when
\begin{equation}\label{eq:KMPrevmea}
\mu_N^\lambda(\xi)=\prod_z|\lambda| e^{\lambda\xi(z)}
\end{equation}
is the density of the product of exponentials of the same parameter $-\lambda >0$. Considering perturbed KMP models like in \eqref{eq:KMPpertrate}, for any ${x,y}$ the detailed balance condition
\begin{equation}\label{detbal2}
\mu_N^{\lambda(\cdot)}(\xi)\Gamma_{x,y}^{\xi,\mathbb F}(j)=\mu_N^{\lambda(\cdot)}\big(\xi-j\left(\varepsilon^x-\varepsilon^y\right)\big)
\Gamma_{x,y}^{\xi-j\left(\varepsilon^x-\varepsilon^y\right),\mathbb F}(-j)
\end{equation}
is satisfied provided that $\mathbb F(x,y)=\lambda(y)-\lambda(x)$ and
\begin{equation}\label{prolampo}
\mu_N^{\lambda(\cdot)}(\xi)=\prod_z|\lambda(z)| e^{\lambda(z)\xi(z)}
\end{equation}
is the density of an inhomogeneous product of exponentials.

The boundary dynamics \eqref{eq:bKMPgen} satisfies the detailed balance condition with respect to an exponential measure with parameter $-\lambda>0$ coinciding with the one of the external source. This means that we have for a boundary site $x$
\begin{equation}\label{detbal3}
\mu_N^{\lambda(\cdot)}(\xi)\Gamma_{x,b}^\xi(j)=\mu_N^{\lambda(\cdot)}(\xi+j\varepsilon^x)\Gamma_{x,b}^{\xi+j\varepsilon^x}(-j)
\end{equation}
where $\mu_N^{\lambda(\cdot)}$ is like \eqref{prolampo}  and the value $\lambda(x)$ coincides with $\lambda$ in \eqref{eq:bKMPgen}. In \eqref{detbal3} we called
$$
\Gamma_{x,b}^\xi(dj)=|\lambda| e^{\lambda(\xi(x)+j)}\chi_{[-\xi(x),+\infty)}(j)\,dj\,.
$$
In agreement with general results in \cite{BDeSGJ-LL07}, by the above computations we obtain that a KMP model is reversible if the external field is of gradient type $\mathbb F(x,y)=\psi(y)-\psi(x)$ for a function $\psi$ such that $\psi(x)=\lambda(x)$ if $x\in \partial  V_N$, where $\lambda(x)$ is the parameter of the external source at $x$. In this case the invariant measure is product like in
\eqref{prolampo} with $\lambda(\cdot)$ replaced by $\psi(\cdot)$.
When the model is not reversible the invariant measure is not known. A similar result can be obtained also for KMPd.
For the Gaussian model \eqref{eq:gaussrate}  the detailed balance condition is satisfied with respect to a product of Gaussian distributions having the same arbitrary mean value and variance equal to $2\gamma^2$ where $\gamma^2$ is the variance of the stochastic current across an edge.

A special situation is when the system is in contact with sources having all the same chemical potential and there is no  external field. In this case we have that the KMP and the KMPd models are equilibrium models reversible with respect to homogeneous product measures. Given a reference measure $\mu$ on $\mathbb R$, we denote by $\mu^\lambda$ the probability measure obtained inserting a chemical potential term of the form
\begin{equation}\label{eq:mu_lambda}
\mu^{\lambda}(dx)=\frac{\mu(dx)e^{\lambda x}}{Z(\lambda)}\,,
\end{equation}
where $Z(\lambda)$ is the normalization constant.
The corresponding average density given by $\rho[\lambda]:=\int \mu^{\lambda}(dx)x=(\log Z(\lambda))'$ is increasing in $\lambda$ and we call $\lambda[\rho]$ the inverse function. In the case of the KMP model it is natural to fix $\mu(dx)=dx$ and restrict to negative values of $\lambda$. In this case $Z(\lambda)=-\lambda^{-1}=\rho[\lambda]$ and $\lambda[\rho]=-\rho^{-1}$. For the KMPd model we fix $\mu(dx)=\sum_{k=0}^{+\infty}\delta_k(dx)$ and again we restrict to negative values of $\lambda$. In this case we have $Z(\lambda)=(1-e^\lambda)^{-1}$, $\rho[\lambda]=\frac{e^\lambda}{1-e^\lambda}$ and $\lambda[\rho]=\log\frac{\rho}{1+\rho}$.

\subsection{Instantaneous energy-mass current }\label{subsec:Ie-mc} Here we have to adapt the definition of instantaneous current to the formalism of the interacting nearest neighbours energies-masses models of this section. The generator is  \eqref{eq:bulkxyE}, the case $ \mathbb{F}=0 $ is treated as a subcase  and we omit the index when the external field is zero. 
The \e{instantaneous current} for the bulk dynamics is defined as
\begin{equation}\label{eq:istce-m}
j^{\mathbb{F}}_\xi(x,y):=\int \Gamma^{\xi,\mathbb{F}}_{x,y}(dj)j\,.
\end{equation}
Its interpretation is the rate at which masses-energies cross the bond $ (x,y) $.
The \e{current flow} now is indicated with $\mathcal J_t(x,y)$ and it is the net total amount of mass-energy that has flown from $x$ to $y$ in the time window $[0,t]$. It can be defined as sum of all the differences between  the mass-energy measured in $ x  $ before of every jump on the bond $ \{x,y\}=\{x,x+e_i\} $ and the mass-energy measured in $ x $ after every jump on the bond $ \{x,y\}=\{x,x+e_i\} $ . Let $ \tau_i $ be the time of the $ i-th $ jump on the bond $ \{x,y\}=\{x,x+ e_i\} $ for some $ i $, we write the current flow as follows
\begin{equation}\label{eq:e-mflow1}
\mc{J}_t(x,y):=\sum_{\tau_i:\tau_i\in[0,t]} J_{\tau_i}(x,y) \,,
\end{equation}
where  $ J_{\tau}(x,y) $ is the \emph{present flow} defined as the current flowing from $ x $ to $ y $ at the   jump time  $ \tau $
\begin{equation}\label{eq:eq_e-mflow2}
J_{\tau}(x,y):=\lim_{h\downarrow 0}\xi_{\tau-h}(x)-\lim_{h\downarrow 0}\xi_{\tau+h}(x).
\end{equation}
The flow $ \mc{J}_{t}(x,y) $ has to be a discrete vector depending on the trajectory $ \{\xi_t\} $, so we set $ J_{\tau}(y,x):=-J_{\tau}(x,y) $.
As in the particles case the flow $ \mc{J}_t(x,y) $ is a function on the path space, while the instantaneous current $ j^{\mathbb{F}}_\xi(x,y) $ is a function on the configuration space. Mutatis mutandis, see \cite{Spo91} section 2.3 part II, the difference 
\begin{equation}\label{eq:massmart}
M_t(x,y)=\mc{J}_t(x,y)-\int^t_0 ds\,j_{\xi(s)}(x,y).
\end{equation}
is a martingale.  
Repeating what we did  subsection \ref{subsec:Ipc} to get \eqref{eq:istcur} with a formalism suitable to energies-masses models,   the instantaneous current \eqref{eq:istce-m} can be naturally obtained as the average flow per unit of time integrated over an infinitesimal time interval, i.e.  
\begin{equation}\label{eq:massflow/t}
 j^{\mathbb{F}}_\xi(x,y):=\us{ t\to0}{\lim}\frac{{\mathbb{E}^{\xi}(\mc{J}_{t}(x,y))}}{t}.
\end{equation}

Moreover, analogous considerations to the ones we did subsection \ref{subsec:Ipc} starting from the conservation law \eqref{eq:numcon1} can be done also in this case with the use of the martingale \eqref{eq:massmart}.

\begin{example}
For example the instantaneous current across the edge $(x,y)$ for the KMP process is given by
\begin{equation}\label{eq:KMPic}
\int_{-\xi(y)}^{\xi(x)}\frac{jdj}{\xi(x)+\xi(y)}=\frac 12\left(\xi(x)-\xi(y)\right)\,.
\end{equation}
\end{example}
This computation shows that the KMP model is of gradient type. In general a model of stochastic masses-energies on a lattice is called of \emph{gradient} type  if the instantaneous current can be written as
\begin{equation}\label{eq:massgra}
j^{\mathbb{F}}_{\xi}(x,y)=\tau_y h(\xi)-\tau_x h(\xi),
\end{equation}
where $h$ is a local function and $\tau_z$ is the shift operator by the vector $z$ defined in section \ref{sec:CTMC}. Formula \eqref{eq:KMPic} shows that for KMP formula \eqref{eq:massgra} holds with $h(\xi)=-\frac{\xi(0)}{2}$. Also KMPd is gradient with respect to the same function $h$.
The local conservation of masses-energies analogous of \ref{eq:numcon1} and \ref{eq:numcon2} is going to be discussed later in section \ref{sec:TC}, when to study their scaling limits  we will need the respective discrete continuity equations. 

\begin{example}
For the weakly asymmetric KMP model in the case of a constant external field $E$ in the direction from $x$ to $y$ the density is
\begin{equation*}\label{eq:KMPEgamma}
\Gamma^{\xi,E}_{x,y}(j)=\frac{1+Ej}{\xi(x)+\xi(y)}+o(E)
\end{equation*}
and the instantaneous current is
\begin{eqnarray}\label{eq:KMPEic}
& & j_{\xi}^E(x,y)=\int_{-\xi(y)}^{\xi(x)}\Gamma^{\xi,E}_{x,y}(j)jdj\nonumber \\
& &=\frac{2}{E(\xi(x)+\xi(y))}\left[e^{\frac E2\xi(x)}\xi(x)
+e^{-\frac E2\xi(y)}\xi(y)-2\frac{e^{\frac E2\xi(x)}-e^{-\frac E2\xi(y)}}{E}\right]\nonumber  \\
& &=\frac 12\big(\xi(x)-\xi(y)\big)+\frac E6\big[\xi(x)^2+\xi(y)^2-\xi(x)\xi(y)\big]+o(E)\,.
\end{eqnarray}
The hydrodynamic behavior of the model under the action of an external field in the weakly asymmetric regime, i.e. when the external field $E$ is of order $1/N$, is determined by the first two orders in the expansion \eqref{eq:KMPEic}. In particular any perturbed KMP model having the same expansion as in \eqref{eq:KMPEic} will have the same hydrodynamic behavior of the model \eqref{eq:KMPpertrate} in the weakly asymmetric regime.
\end{example}
\begin{example}
For the KMPd model with a discrete version of the above computation we get
\begin{equation}\label{eq:KMPdEic}
j_{\xi}^E(x,y)
=\frac 12\big(\xi(x)-\xi(y)\big)+\frac{E}{12}\big[2\xi(x)^2+2\xi(y)^2-2\xi(x)\xi(y)+\xi(x)+\xi(y)\big]+o(E)\,.
\end{equation}
\end{example}

With the same definitions of \ref{subsec:Ipc}, the local conservation of the mass-energy  is expressed by $ \xi_t(x)-\xi_0(x)+\div \mc{J}_t(x)= 0 $, i.e.
\begin{equation}\label{eq:masscon}
\xi_t(x)-\xi_0(x)+\int_0^t ds\, \div j^{\mathbb{F}}_{\xi(s)}(x)+\div M_t(x)=0.
\end{equation}
In particular the microscopic fluctuation \eqref{eq:massmart} has mean zero and the expected value of the current through a bond
can be obtained from the expected value of an additive functional involving the instantaneous current trough the same bond and the analogous considerations of subsection \ref{subsec:Ipc} to conclude that the average currents are zero in the equilibrium case  are true.  

The natural scaling limit for this class of processes is the diffusive one, where the rates have to be multiplied by $N^2$ to get a non trivial scaling limit, this is discussed in the introduction to part \ref{p:mact} and chapter \ref{ch:HSL}.  Instead of \eqref{eq:massmart} we will consider, in the the second part regarding the macroscopic theory, the speeded up martingale   
\begin{equation*}\label{eq:macromart}
M_t(x,y)=\mc{J}_t(x,y)-N^2\int^t_0 ds\,j^\mathbb{F}_{\xi(s)}(x,y).
\end{equation*}

\section{Asymmetric energies-masses models}\label{sec:Amodels}

For their own interest and for comparison with the results we will get later when we are studying the strong limit of a constant external field for the large deviation functional in the case of the perturbed KMP \eqref{eq:KMPpertrate}, i.e. $ \bb F(x,y)=E $ for all $ (x,y)=(x,x+e_i) $. We consider now  some possible one dimensional totally asymmetric models for which the mass can move only in one preferred direction and we study their invariant measure.
If on a bond $(x,y)$ the asymmetry is from $x$ to $y$ then we have that the measure $\Gamma_{x,y}^\xi$ determining the distribution
of the current has a support contained on the interval
$[0,\xi(x)]$. We consider only the KMP case and assume that the density $\Gamma_{x,y}^\xi$ depends only on $\xi(x)$ and not on $\xi(y)$.
The distribution $\Gamma_{x,y}^\xi$ of a totally asymmetric model should be obtained as a limit for large values of a constant external field of the distribution $\Gamma_{x,y}^{\xi,E}$ of a weakly asymmetric model having an instantaneous current with an expansion like \eqref{eq:KMPEic}. This is because macroscopically we will discuss the limit for large values of the field of weakly asymmetric models  having an hydrodynamic behavior deduced from the expansion \eqref{eq:KMPEic}. Since the expansion \eqref{eq:KMPEic} is for small values of the field while we discuss here microscopically the behavior for large values of the field it is reasonable to have some freedom in the determination of the limiting models. We discuss indeed two different cases. One has a product invariant measure while for the other one we discuss a duality representation of the invariant measure using a convex analytic approach.

The macroscopic domain is $\Lambda=(0,1]$ and the asymmetry is in the positive direction. For simplicity of notation we consider the models defined on the lattice $\{1,2,\dots ,N\}$ instead of $ V_N$. Since the computations in this section are only microscopic the lattice size is not relevant.

\subsection{Totally asymmetric KMP model version 1}\label{subsec:tkmp1dis}

In this section we discuss a model with  distribution of the current flowing across a bond in the bulk given by
\begin{equation}\label{eq:TAKMP1}
\Gamma_{x,x+1}^\xi= \chi_{[0,\xi(x)]}(j)\, dj\,.
\end{equation}
We fix the interaction with the boundary left source like in \cite{BDeSGJ-LL07}. We imagine to have a ghost site at $0$ where an exponential random energy of parameter $-\lambda>0$ is available and this random energy available at site $0$ is transported into site
$1$ with the same mechanism \eqref{eq:TAKMP1} of the bulk. We have then the boundary part of the dynamics at the boundary site $1$
given by
\begin{eqnarray}\label{eq:pera}
\mc{L}_{b,1}f(\xi) &=& \int_{0}^{+\infty}|\lambda| e^{\lambda z}\left(\int_0^{z}dj\left[f(\xi+j\varepsilon^{1})-f(\xi)\right]\right)\,dz\\
& = & \int_0^{+\infty} e^{\lambda j} \left[f(\xi+j\varepsilon^{1})-f(\xi)\right]\,dj\,.\nonumber
\end{eqnarray}
At the right boundary the dynamics is like on the bulk but the mass that is moving to the right exits
from the system and disappears. More precisely at site $N$ with rate $\xi(N)$ the amount of mass present is transformed into
$\xi'(N)$ that is uniformly distributed on $[0,\xi(N)]$. We could imagine  a different mechanism that allows also a creation of mass connected with a reservoir with a given chemical potential. This mechanism however changes the distribution of the mass just at site $N$ and in particular is not  observable macroscopically.
Consider the product measure (recall $\lambda<0$)
\begin{equation}
\label{eq:scat}
\mu_N^\lambda(d\xi)=\prod_{x\in  V_N}|\lambda| e^{\lambda\xi(x)}d\xi(x)\,.
\end{equation}
With a change of variables we get
 \begin{eqnarray}\label{eq:1}
\mathbb E_{\mu_N^\lambda}\left[\mc{L}_{\{x,x+1\}}f\right] &=&\int_{ (\mathbb R^+)^N}  \mu_N^\lambda(d\xi)\int_0^{\xi(x)}dj\left[f\left(\xi-j\left(\varepsilon^x-\varepsilon^{x+1}\right)\right)-f(\xi)\right]\nonumber\\
&=&\int_{(\mathbb R^+)^N}  \mu_N^\lambda(d\xi)f(\xi)\left[\xi(x+1)-\xi(x)\right]\,.
 \end{eqnarray}
Still with a changes of variables at the boundaries we get
 \begin{eqnarray}\label{eq:2}
\mathbb E_{\mu_N^\lambda}\left[\mc{L}_{b,1}f\right] &=&\int_{ (\mathbb R^+)^N} \mu_N^\lambda(d\xi)\int_0^{+\infty}e^{\lambda j}\left[f(\xi+j\varepsilon^{1})-f(\xi)\right]\, dj\nonumber \\
&=&\int_{(\mathbb R^+)^N} \mu_N^\lambda(d\xi)f(\xi)\left[\xi(1)+\lambda^{-1}\right]\,,
 \end{eqnarray}
 and
 \begin{eqnarray}\label{eq:3}
\mathbb E_{\mu_N^\lambda}\left[\mc{L}_{b,N}f\right] &=&\int_{(\mathbb R^+)^N} \mu_N^\lambda(d\xi)\int_0^{\xi(N)}\left[f\left(\xi-j\varepsilon^{N}\right)-f(\xi)\right]\, dj\nonumber\\
&=&-\int_{(\mathbb R^+)^N} \mu_N^\lambda(d\xi)f(\xi)\left[\lambda^{-1}+\xi(N)\right]\,.
 \end{eqnarray}
Summing up \eqref{eq:1}, \eqref{eq:2} and \eqref{eq:3} we obtain that \eqref{eq:scat} is invariant for the dynamics.

\subsection{Totally asymmetric KMP version 2}\label{subsec:tkmp2dis}

In this section we discuss a second possible totally asymmetric limit dynamics. This is the model that is obtained considering a constant external field in \eqref{eq:KMPpertrate} and taking the limit suitably normalizing the rates. The dynamics on the bulk is defined by a distribution of the current flowing across a bond given by
\begin{equation}\label{TAKMP2}
\Gamma_{x,x+1}^\xi=\delta_{\xi(x)}\,.
\end{equation}
This means that at rate one all the mass present on a site jumps to the nearest neighbor site on the right.
On the torus this dynamics is not irreducible since eventually all the mass will concentrate on a single lattice site moving randomly like an asymmetric random walk. The boundary driven case has not this problem and the dynamics is irreducible.

\smallskip

Given two probability measures $\mu$ and $\nu$ on $\mathbb R$ we define their
convolution $\mu\ast\nu=\nu\ast\mu$ as the measure on $\mathbb R$
defined by
$$
[\mu\ast\nu](A)=\int_\mathbb R \mu(A-x) d\nu(x)\,,
$$
for any measurable subset $A$. Let us also define the family of
Gamma measures $\left\{\gamma_n\right\}_{n\geq 0}$ of parameter $|\lambda|$  as follows. We
set $\gamma_0:=\delta_0$, then we define $\gamma_1$ as the absolute
continuous probability measure on $\mathbb R^+$ having density $|\lambda|
e^{\lambda x}$. Finally we define $\gamma_n:=\gamma_1^{\ast n}$
where the right hand side symbol means a n-times convolution of
$\gamma_1$. Note that $\gamma_j\ast\gamma_i=\gamma_{i+j}$.
We fix the dynamics at the boundaries like
\begin{equation}\label{casanuova}
\mc{L}_{b,1}f(\xi)=\int_0^{+\infty}\gamma_1(j)\left[f(\xi+j\varepsilon^{1})-f(\xi)\right]\,dj
\end{equation}
and
\begin{equation}\label{casanuovaN}
\mc{L}_{b,N}f(\xi)=\left[f\left(\xi-\xi(N)\varepsilon^{N}\right)-f(\xi)\right]\,.
\end{equation}

The invariant measure for this second version of the totally asymmetric KMP model is not of product type and it seems not to have a simple expression. We give a representation of this measure as a convex combination of products of Gamma distributions. This is done developing a kind of duality between this process and a totally asymmetric version of KMPd. It is not clear if there is a usual duality relationship between the two processes that we are
intertwining.  It is interesting to analyze this duality within the general approach to duality in \cite{CGRS16} where the case of an asymmetric KMP model is also discussed.

Consider a product measure $\nu_N$ having marginals $\nu^{(x)}$, i.e.
\begin{equation}\label{prodguess}
\nu_N(d\xi):=\prod_{x=1}^N\nu^{(x)}(d\xi(x))=:\otimes_{x=1}^N\nu^{(x)}\,.
\end{equation}
For a measure of this type we have
\begin{eqnarray}
& &\int_{\left(\mathbb R^+\right)^n}\nu_N(d\xi)\int_0^{+\infty}\gamma_k(j)f(\xi+j\varepsilon^{1})\,dj=\nonumber \\
& &\int_{\left(\mathbb R^+\right)^n}(\nu^{(1)}\ast \gamma_k)(d\xi(1))\nu^{(2)}(d\xi(2))\dots \nu^{(N)}(d\xi(N))\, f(\xi)\,.\nonumber
\end{eqnarray}
Likewise we have
\begin{eqnarray}
& &\int_{\left(\mathbb R^+\right)^n}\nu_N(d\xi)f\big(\xi +\xi(x)(\varepsilon^{x+1}-\varepsilon^x)\big)=\nonumber\\
& &\int_{\left(\mathbb R^+\right)^n}\nu^{(1)}(d\xi(1))\dots \gamma_0(d\xi(x))(\nu^{(x)}\ast\nu^{(x+1)})(d\xi(x+1))\dots \nu^{(N)}(d\xi(N))\,f(\xi)\,.\nonumber
\end{eqnarray}
The right hand side in the above formula is the expected value of the function $f$ with respect
to a product measure having $x$-marginal equal to $\gamma_0$,
$(x+1)-$marginal equal to $\nu^{(x)}\ast\nu^{(x+1)}$ and all the remaining
marginal equal to the one of $\nu_N$. Finally we have also that
\begin{eqnarray}
& &\int_{\left(\mathbb R^+\right)^n}\nu_N(d\xi)f\left(\xi-\xi(N)\varepsilon^{N}\right)=\nonumber \\
& &\int_{\left(\mathbb R^+\right)^n}\nu^{(1)}(d\xi(1))\dots
\gamma_0(d\xi(N))\,f(\xi)\,.\nonumber
\end{eqnarray}
The above computations give that the action of the asymmetric generator $\mathcal L_{N,a}$ on a measure $\nu_N$ like in \eqref{prodguess}
is given by
\begin{eqnarray}
\nu_N\mathcal L_{N,a}&=& \Big\{\left[\left(\nu^{(1)}\ast \gamma_1\right)\otimes \nu^{(2)}\otimes \dots \otimes \nu^{(N)}\right]-\left[\nu^{(1)}\otimes \dots \otimes \nu^{(N)}\right]\Big\}\nonumber \\
&+& \sum_x\Big\{\left[\nu^{(1)}\otimes \dots \gamma_0\otimes \left(\nu^{(x)}\ast\nu^{(x+1)}\right)\otimes \dots \otimes \nu^{(N)}\right]-
\left[\nu^{(1)}\otimes \dots \otimes \nu^{(N)}\right]\Big\}\nonumber \\
& & +\Big\{\left[\nu^{(1)}\otimes \dots \otimes \nu^{(N-1)}\otimes \gamma_0\right] -\left[
\nu^{(1)}\otimes \dots \otimes \nu^{(N)}\right]\Big\}\,.\label{eq:gammamarg}
\end{eqnarray}
We can now show that there is a solution of the Kolmogorov equation $\partial_t\nu_N(t)=\nu_N(t)\mathcal L_{N,a}$ that can be written in the form
\begin{equation}
\nu_N(t)=\sum_{\eta\in \mathbb N^N}c_t(\eta)\gamma_{\eta(1)}\otimes\gamma_{\eta(2)}\otimes \dots
\otimes\gamma_{\eta(N)}\,.\label{eq:gammaprod}
\end{equation}
In the formula \eqref{eq:gammaprod} $\eta=(\eta(1),\dots \eta(N))\in \mathbb N^N$ can be interpreted as a configuration
of particles on the lattice and
$c_t(\eta)\geq 0$ for any fixed $t$ is a suitable probability measure on $\mathbb N^N$
to be determined. Formula \eqref{eq:gammaprod} says that we are searching for a
solution that can be written as a mixture of products
of Gamma measures for any time. Defining
$$
\gamma^\eta:=\gamma_{\eta(1)}\otimes\gamma_{\eta(2)}\otimes \dots
\otimes\gamma_{\eta(N)}
$$
we write compactly \eqref{eq:gammaprod} as
$\sum_\eta c_t(\eta)\gamma^{\eta}$. For a measure of the type \eqref{eq:gammaprod} we have
\begin{equation}\label{eq:gammacomb}
\nu_N(t)\mathcal L_{N,a}=\sum_\eta c_t(\eta)\left(\gamma^\eta\mathcal L_{N,a}\right)\,,
\end{equation}
and we can now use formula \eqref{eq:gammamarg}.
Reorganizing the terms, the right hand side of \eqref{eq:gammacomb} becomes
\begin{eqnarray}\label{eq:gammacomb1}
& &\sum_\eta \gamma^\eta\Big\{\left[c_t(\eta-\varepsilon^1)\chi(\eta(1)>0)-c_t(\eta)\right]\nonumber \\
& &+\sum_{x=1}^{N-1}\sum_{k=0}^{\eta(x+1)}
\left[\chi(\eta(x)=0)c_t(\eta+k(\varepsilon^{x}-\varepsilon^{x+1}))-c_t(\eta)\right]\nonumber \\
& &\left.+\left[\chi(\eta(N)=0)\sum_{k=0}^{+\infty}c_t(\eta+k\varepsilon^N)-c_t(\eta)\right]\right\}\,,
\end{eqnarray}
where $\chi$ denotes the characteristic function.
Using \eqref{eq:gammacomb1} we can write formula \eqref{eq:gammacomb} compactly as
\begin{equation}\label{eq:gammacomb2}
\sum_\eta\gamma^\eta\partial_t c_t(\eta)=\nu_N(t)\mathcal L_{N,a}=\sum_\eta\gamma^\eta\left(c_t(\eta)\mathcal L_{N,a}^d\right)
\end{equation}
where $\mathcal L_{N,a}^d$ is a Markov generator of a stochastic dynamics on the variables $\eta$. We interpret
\eqref{eq:gammacomb2} as a duality relationship between the two stochastic dynamics $\mathcal L_{N,a}$ and $\mathcal L_{N,a}^d$.
The upper index $d$ is the shorthand of \emph{dual}. The variables $\eta$ represent configurations of particles on the lattice
and $\eta(x)$ that is always an integer number is the number of particles at site $x$. By formula \eqref{eq:gammacomb1} the stochastic dynamics associated to $\mathcal L_{N,a}^d$ can be described as follows. In the bulk the dynamics has a distribution of current
given by $\Gamma_{x,x+1}^\eta=\delta_{\eta(x)}$. At the left boundary one particle is created with rate $1$ while all the particle at the right boundary are erased at rate $1$. This is a totally asymmetric version of the model KMPd.

We proved that the model with generator $\mathcal L_{N,a}$ starting at time zero with a distribution of the type $\sum_\eta c_0(\eta)\gamma^\eta(d\xi)$ will have a distribution of energies at time $t$ that is $\sum_\eta c_t(\eta)\gamma^\eta(d\xi)$ where $c_t(\eta)$ is the distribution of particles at time $t$ for the model with generator $\mathcal L_{N,a}^d$ starting at time $0$ with the distribution of particles given by $c_0(\eta)$. In particular, considering the limit for $t\to +\infty$, this relationship between the two processes will hold also for the corresponding invariant measures for which we get
\begin{equation}\label{eq:dual-inv}
\mu_N(d\xi)=\sum_\eta \mu_{N,d}(\eta)\gamma^\eta(d\xi)\,.
\end{equation}
In \eqref{eq:dual-inv} $\mu_N$ is the invariant measure for the process $\mathcal L_{N,a}$ while $\mu_{N,d}$ is the invariant measure for the process $\mathcal L^d_{N,a}$.

\bigskip We conjecture
that the large deviations rate functional for the empirical measure when particles are distributed according to the invariant measure of the original  model is the same of the corresponding one associated to a product of exponentials, as it is for the asymmetric model of   subsection \ref{subsec:tkmp1dis}.  A direct microscopic computation of this rate functional would be very interesting.

\chapter{Discrete  calculus in IPS}\label{ch:DEC}
 
 In this chapter we introduce the so called "discrete exterior calculus" (DEC) in a suitable language to treat the discrete operators we will need for our lattice gases.
 
Discrete exterior calculus (DEC) is  motivated by
potential applications in computational methods for field theories (elasticity, fluids, electromagnetism) and in
areas of computer vision/graphics,  for these applications see for example   \cite{DeGDT15}, \cite{DMTS14} or   \cite{BSSZ08}.  DEC is developing as an alternative approach for computational science to the usual discretizing process from the continuous theory. It considers the discrete mesh as the only  thing given and develops an entire calculus using only discrete combinatorial and geometric
operations. The derivations may require that the objects on the discrete mesh, but not the mesh itself, are
interpolated as if they come from a continuous model. Therefore DEC could have interesting applications in fields where there isn't any continuous underlying structure as Graph theory \cite{GP10} or problems that are inherently discrete since  they are defined on a lattice \cite{AO05} and it can stand in its own right as theory that parallels the continuous one.  The language of DEC is founded on the concept of discrete differential form, this characteristic allows to preserve in the discrete context some of the usual geometric and topological structures of continuous models, in particular the Stokes'theorem 
\begin{equation}\label{eq:Stokes}
\int_\mc{M} d\omega=\int_{\partial\mc{M}}\omega.
\end{equation}
Stating that a differential form $\omega$ over the boundary of some orientable manifold  $\mc{M}$ is equal to the integral of its exterior derivative $d\omega$ over the whole of $\mc{M}$. Equation (\ref{eq:Stokes}) can be considered the milestone to define the discrete exterior calculus since it contains the main objects to set a discrete exterior calculus, namely the concepts of discrete differentials form, boundary operator and  discrete exterior derivative, moreover it is very natural in the discrete setting.
A   qualitative review for DEC   is  \cite{DKT08}, while to have a deeper view we suggest \cite{Cra15} and \cite{Hir03}. 

\vspace{0.5cm}
First we are establishing the necessary objects to introduce the discrete analogous of differential forms, i.e. the discrete exterior derivative $d$ and its adjoint operator $\delta$ in the context of a cubic cellular complex for a  lattice on a discrete manifold. In this way we can  state the Hodge decomposition. This is done in section \ref{sec:DEC}, here the ideas for the construction of DEC  come from \cite{DKT08} and \cite{DHLM05}, but respect to this works we  present more precise definitions for a manifold setting and   cubic cells instead of simplexes, the concepts of comparability, consistency  and local orientation are introduced to define rigorously what is a  discrete manifold. Here we are not exposing  the continuous theory, good references for a treatment of Geometry with differential forms (and its application) to compare with DEC are \cite{AMR88,Fla89} and \cite{Fra12}.  
Then, in section \ref{sec:disop}, we are illustrating how these operators work up to dimension three using as manifold the discrete torus.

\section[DEC on cubic mesh]{Discrete exterior calculus on cubic mesh}\label{sec:DEC}

Intuitively, \e k-differential forms are objects that can be integrated on a \e k dimensional region of the space. For example 1-forms are like $dF=f(x)dx$ or $dG=\frac{\partial G}{\partial x}dx+\frac{\partial G}{\partial y}dy+\frac{\partial G}{\partial z}dz$,  which can be integrated respectively over a interval in $\mathbb{R}$ or  over a curve in $ \mathbb{R}^3 $. With this idea in mind discrete differential forms are going to be defined. As we said we are working  with an abstract cubic complex, instead of a simplicial one. This abstract complex can be tought as a collection of discrete sets of maximal  dimension  $ n $. This collection could be previously derived  from  a continuous structure on a manifold $\mc{M}$ of dimension $n$.

\subsection{Primal cubic complex and dual cell complex}

The next definitions fix in an abstract way  the objects on which DEC operates, the language is the typical one in  algebraic topology \cite{Mun84}.
\begin{definition}\label{d:simdef}
	A \emph{k-simplex} is the convex span $\mf s_k=\{v_0,v_1,\dots,v_k\}$ of $k+1$ geometrically independent points of $ \mathbb{R}^N $ with $ N\geq k $, they are called \emph{vertices of the k-simplex} and  $ k $ its  \emph{dimension}.  A simplex  $ s_k=(v_0,v_1,\dots,v_{k+1}) $ is \e{oriented}   assigning one of the two possible   equivalence classes of   ordering of its vertices $ v_i $. Two orderings are in the same  class if they differ for an even permutation, while they are not for an odd permutation.  The \e k-simplex with same vertices but different ordering from $  s_k $ is said to have \e{opposite orientation} and   denoted with $ -s_k $.  
\end{definition}

One orientation of a simplex can be called conventionally \e{positive} and the opposite one  \e{negative}. 
\begin{definition}[Orientation convention\footnote{This convention tell us the following. We consider $ \mathbb{R}^k $  with a right handed orthonormal basis $ e_1,\dots, e_k $. A simplex $s_1=(v_0,v_1)$ embedded in $ \bb R^1 $  can take orientation from $v_0$ to $v_1$, let's call it positive assuming $ v_1>v_0 $, otherwise from $v_1$ to $v_0$, that is negative. A simplex $s_2=(v_0,v_1,v_2)$ embedded in $ \bb R^2 $  can take anticlockwise orientation with the normal pointing outside the plane along the "right hand rule", let's call it positive,  otherwise clockwise with the normal pointing outside the plane along the "left-hand rule", that is negative. A simplex $s_3=(v_0,v_1,v_2,v_3)$ embedded in $ \bb R^3 $ can take  orientation along the "screw-sense" about the simplex embodied in the familiar "right-hand rule", let's call it positive, otherwise  orientation along the "left-hand rule", that is negative.} for simplexes]\label{d:orsim}
	A way to define the sign of a \e k-simplex $ s_k=(v_0,\dots,v_{k+1}) $ is that of embedding it  in $ \mathbb{R}^k $ equipped with a  right handed orthonormal basis and saying it is positive oriented  if $ \det(v_1-v_0,v_2-v_0,\dots,v_{k+1}-v_0)>0 $ and negative in the opposite case.
\end{definition}

"Inside" a \e k-simplex we can individuate  some proper simplexes, we need to define this and how they relate with the "original" one.
\begin{definition}
	A \emph{ j-face} of a \e k-simplex is any \emph{j}-simplex ($j<k$) spanned by  a proper  subset of  vertices of $ \mf s_k $, this gives a strict  partial order relation $\mf s_j\prec \mf s_k$ and if $ \mf s_j $ is a face of $ \mf s_k $ we denote it $ \mf s_j (\mf s_k)$. A \e j-face is \e{shared} by two \e k-simplexes  $ (j<k) $ if it is a face of both.
\end{definition}

We will need to say    when it is possible and how to compare two \e k-simplexes (of the same dimension), i.e. their reciprocal orientations. The idea of next definition is that this is possible when there exists an hyperplane where  both the \e k-simplexes lie.

\begin{definition}\label{def:consistency}
	We say that two   \e k-simplexes in $ \bb R^N $ with $ N\geq k $ are \e{comparable} if they belong to the same \e k-dimensional hyperplane. Moreover, we say that two comparable oriented simplexes are \e{consistent} if they have same  orientation sign.
\end{definition}

The orientation sign refers to definition \ref{d:orsim}. The condition for two \e k-simplexes to be in the same hyperplane is equivalent   to ask that   any convex span of $ k+2 $ points  chosen from the union of their vertices is not a \e {(k+1)}-simplex. Another relevant  concept in DEC is that one of  induced orientation, that is the orientation that  a face of a simplex inherited when the last one is oriented.  

\begin{definition}\label{def:indor}
	We call \emph{induced orientation} by an oriented \emph k-simplex $ s_k $ on a \e j-face $ s_j(s_k) $  the corresponding  ordering of its vertices in the sequence ordering the vertices   of $ s_k $. If  two \e k-simplexes $ s_k $ and $ s'_k $ induce opposite orientations on a shared \e j-face we say that the \e j-face \e{cancels}.
\end{definition}

Now we introduce the definition of \e k-cube.

\begin{definition}\label{def:s-dec}
	A \e k-cube $\mf c_k=\{v_0,v_1,\dots,v_{2^k-1}\}$  is the convex span of $ 2^k $ points of $ \bb R^N $ with $ N\geq k $ such that there exist $ k! $ different \e k-simplexes having their $ k+1 $ vertices chosen between the \e{vertices} $ v_i $ of $ \mf c_k $ and sharing two by two only one \e {(k-1)}-face. Moreover, vertexes are extremal\footnote{Extremal means that a vertex  can not be written as convex combination of the other vertexes.} points of the convex combination. Each one of these $ k! $ simplexes $ \mf s^i_k $ is said \e{a proper  k-simplex of $ \mf c_k $} and we denote it $ \mf s^i_k(\mf c_k) $ where $ i\in\{0,1,\dots,k!\} $. The dimension of $ \mf c_k $ is $ k $. Note that there is more than one  way to choose these $ k! $ proper \e k-simplex and we call each of them a \e{simplicial decomposition of $ \mf c_k $} denoted $  \varDelta \mf c_k= \os{k!}{\us{i=1}{\cup}}\mf s^i_k $.
\end{definition}
When we don't need to specify the index $ i $ of the these internal simplexes we  omit it. 
\begin{remark}
	A shared \e{(k-1)}-face $ \mf s_{k-1}(\mf s^i_k(\mf c_k))=\mf s_{k-1}(\mf s^{i'}_k(\mf c_k)) $( with $ i\neq i' $) is inside $ \mf c_k $ and not on its boundary. A precise definition of boundary for simplexes and cubes will be given later.
\end{remark}

\begin{definition}
	A \emph{ j-face} of a \e k-cube is any \emph{j}-cube ($j<k$) spanned by  a proper  subset of  vertices of $ \mf c_k $ and not intersecting its interior, this gives a strict  partial order relation $\mf c_j\prec \mf c_k$ and if $ \mf c_j $ is a face of $ \mf c_k $ we denote it $ \mf c_j (\mf c_k)$. A \e j-face is \e{shared} by two \e k-cubes  $ (j<k) $ if it is face of both.
\end{definition}

The concept of orientation for \e k-cubes follows  from that one for \e k-simplexes.

\begin{definition}\label{d:cubor}
	A \e k-cube $ c_k=(v_0,v_1,\dots, v_{2^k-1}) $ is \e{oriented} assigning to each  \e k-simplex  $ \mf s_k^i $ in a simplicial decomposition of $ \mf c_k $ an orientation such that  the  \e{(k-1)}-faces $ s_{k-1}(s^i_k(c_k)) $ that they share cancel\footnote{Namely,  on them, it is induced an opposite orientation,  see definition \ref{def:indor}.}. Two  \e{oriented simplicial decompositions} are in the same equivalence class of orientation if any two comparable  not shared  \e {(k-1)}-simplexes $ s_{k-1}(s^i_k(c_k)) $ (i.e. lying on the same \e{(k-1)}-face on the boundary of $ \mf c_k $) are consistent. We denote $ \varDelta  c_k=\os{k!}{\us{i=1}{\cup}}s^i_k $ an oriented simplicial decomposition.
\end{definition}

\begin{definition}
	Analogously to definition \ref{def:consistency} for simplexes, two \e k-cubes are \e{comparable} if they lie in the same \e k-dimensional hyperplane, while we say that they are \e{consistent} if the \e k-simplexes of the two simplicial decompositions are consistent.
\end{definition}
It is enough to check the consistency between  any \e k-simplex in the simplicial decomposition of one of the two \e k-cube and any \e k-simplex in the decomposition of the other \e k-cube because of definition \ref{d:cubor}.

\begin{example}
	Consider the  2-cube $\mf c_2=\{v_0,v_1,v_2,v_3\} $. A simplicial decomposition is given by  $ s^A_2=(v_0,v_1,v_3) $ and $ s^B_2=(v_1,v_2,v_3) $. Indeed let $ \mf s_1=\{v_1,v_3\} $ be the 1-simplex shared by $ s^A_2 $ and $ s^B_2$, then they cancel on $ \mf s_1 $ because $ s_1(s^A_2)=(v_1,v_3)$ and $s_1(s^B_2)=(v_3,v_1)=-(v_3,v_1) $. Another decomposition is that one given by $  s^C_2=(v_0,v_2,v_3) $ and $  s^D_2=(v_0,v_1,v_2) $. These two decompositions  are in the same equivalence class  because they induce on the 1-simplexes  $ \{v_0,v_1\} $, $\{v_1,v_2\} $, $ \{v_2,v_3 \} $ and $\{v_3,v_0\} $ consistent orientations.  
\end{example}

\begin{proposition}\label{p:cubor}
	There are only two possible equivalence classes of orientation\footnote{ [Orientation convention for cubes]The result of this simplicial decomposition is that also a \e k-cube has only two possible orientations.
		In one dimensions a 1-cube is also a 1-simplex;
		in two dimensions a 2-cube $(v_0,v_1,v_2,v_3)$ can be anticlockwise oriented (positive) with the normal pointing outside the plane along the "right hand rule" or clockwise oriented (negative) otherwise clockwise with the normal pointing outside the plane along the "left-hand rule";
		in three dimensions a 3-cube $(v_0, \dots,v_7)$ can have, looking at it from outside,  all the faces  anticlockwise oriented (positive) or viceversa all clockwise (negative). These two possibilities corresponds respectively  to have all the  normals to its faces pointing outside or inside the volume.} for a k-cube $ \mf c_k $, when an orientation $ c_k $ is assigned the other one is denoted $ -c_k $. One orientation can be  conventionally defined to be positive and the other negative. 
\end{proposition}

\begin{proof} Once a single simplex $ \mf s_k(\mf c_k) $ is oriented, the proposition follows from the facts that  a simplex can be oriented in only two ways because the cancelling condition on  the shared \e{(k-1)}-faces in a simplicial decomposition force all the others to assume an orientation propagating to the entire cubic cell. 
\end{proof}

Now the concept of induced orientation for simplexes can be transferred to cubes.

\begin{definition}\label{d:indorcub}
	We call \emph{induced orientation} by an oriented \emph k-cube\footnote{We saw in definition \ref{d:cubor} and proposition \ref{p:cubor} that  a \emph k-cube $ \mf c_k $ is oriented trough the orientation of  the \emph k-simplexes of its decomposition, hence we talk equivalently of  orientation  induced by the oriented \emph k-simplexes of the simplicial decomposition of $ c_k $.}  $ c_k $ on a \e j-face   $ c_j(c_k) $ the orientation assigned on it by the inductively oriented  \emph{j}-simplexes $ s_{j}(s_k(c_k)) $  of its simplicial decomposition. If  two \e k-cubes induce opposite orientations on a shared \e j-face we say that the face \e{cancels}.
\end{definition} 

We do an example considering the  oriented 2-cube   $c_2=(v_0,v_1,v_2,v_3)$. Let be  $ s^A_2=(v_0,v_1,v_3) $ and $ s^B_2=(v_1,v_2,v_3) $ the 2-simplexes  orienting a simplicial decomposition of  $ c_2 $. The inductively oriented 1-faces  $ c_1(c_2) $ are $ (v_0,v_1) $,$ (v_1,v_2) $, $ (v_2,v_3) $ and $(v_3,v_0) $.

Our intention is to operate on object made by many cubes, like lattice. So we introduce collections of cubes suitable to define a discrete calculus.  Later we will restrict ourselves to the case of discrete manifold. 

\begin{definition}
	A \emph{cubic complex} $\mc{C}$  of \e{dimension n} is a finite collections of  \emph{elementary} cubes $ \mf c_k $, called also \emph{cells},  such that $ 0 \leq k\leq n $, every face of an elementary cube  is in $\mc{C}$ and the intersection of any two cubes of $\mc{C}$ is either empty or a face of both. To each cubes is  assigned  an orientation $ c_k $.  The \e{local orientation} of $ \mc C $ is the orientation of the cubes of dimension $ n $. We denote with $ |\mc C| $ the topological set of $ \mathbb{R}^N $ ($ n\leq N $) given by the union of all \e n-cubes $ \mf c_n $ in $ \mc C $. 
\end{definition}

Fixing an  orientation of the \emph n-cube is like orienting the  \e n-volume of the \e n-hyperplane containing it or equivalently the \e n-volume of the space  $ \mathbb{R}^n $ where it can be embedded. For a  discrete manifold is like to orient  the tangent space for a continuous one.
The meaning of having finite collections of cubes is that  to deal with compact sets in the continuous case.

In particular we are interested in cubic complexes which are  discrete version of an orientable compact  boundaryless manifold $ \mc M $ of dimension $ n $. The idea to define a discrete manifold of dimension $n$ is that to have a cubic complex (in local sense)   topologically equivalent  to a \e n-ball in $\mathbb{R}^n$. Moreover we want to orient a discrete manifold, this is possible using the "cancelling" notion of definition \ref{d:indorcub}.

\begin{definition}\label{def:dMan}
	A \emph{cubic complex} $\mc{C}$  of dimension \e n is a  \emph{discrete manifold} (boundaryless) if every  \e{(n-1)}-cube is shared exactly by two \e n-cubes. A manifold is \e{orientable} if the orientations of all \e n-cubes can be chosen such that  every shared \e{(n-1)}-face cancels.
\end{definition}

The meaning of this definition is that  if we consider a \e{(n-1)}-cube  ${\mf c}_{n-1}\in\mc{C}$ the set $\us{\mf c_n:\,\mf c_{n-1}\in \,\mf c_n}{\cup}\mf c_n$   is  simply connected and homeomorphic to a unit \e n-dimensional ball.

In DEC  important concepts are the ones of dual cell and dual complex. To define a centre of a cube we use barycentric coordinates,  while in \cite{DHLM05} they use the concept of circumcentre for a simplicial complex. This last concept is very simple in the simplicial case but for our case we should introduce the concept of Voronoi diagram \cite{JosTheo13} and we prefer to avoid this.

\begin{definition}\label{def:centre}
	The \emph{centre} of a \emph k-cube $\mf c_k$ is the barycentre of its vertices, denoted  with $\mathtt{b}(\mf c_k)$ $  $. 
\end{definition}

\begin{remark}
	The union $ c_k\cup c_k'/s_k\cup s_k' $ of two comparable consistent \e k-cubes/\e k-simplexes sharing a \e{(k-1)}-face inherits an orientation that is consistent with that one of $ c_k/s_k $ and $ c_k'/s'_k $, in the sense that any \e k-simplex that  can be generated by the vertices in the union and  properly contained in it is defined to be consistent with  $ c_k/s_k $ and $ c'_k/s'_k $. 
\end{remark}

\begin{definition}
	We define the operation +  as  $  \mf c_k+\mf c'_k:=\mf c_k\cup \mf c'_k $,   when $ \mf c_k=\mf c'_k $ we set  $ \mf c_k+\mf c_k=2 \mf c_k:={\mf c_k\cup \mf c_k} $ as multiset\footnote{ Multisets are sets that can differentiate for  multiple instances of the same element, for example $ {a,b} $ and $ {a,a,b} $ are different multisets.} $ \{\mf c_k,\mf c_k\} $. For oriented cubes $ c_k $ we define an analogous + operation and in addition we define the inverse element, that is  $ c_k+(-c_k)=\emptyset $. The same operation is defined for \e k-simplexes substituting the occurrences of $ \mf c_k $ and $ c_k $ respectively with $ \mf s_k $ and  $ s_k  $.
\end{definition}

The dual of a \emph k-cube, called \emph{dual cell}, is derived from the  duality operator (see next definition \ref{def:dualop}) $\ast:c_k\rightarrow \ast(c_k)$ and   the set of dual cells of a cubic complex will be the \emph{dual complex} $\ast \mc{C}$. Remember that when two \e k-cubes induce two opposite orientation on a shared  \e{(k-1)}-face this last one cancels.

{\begin{definition}\label{def:dualop}
		For a discrete manifold $ \mc C $ of dimension $ n $  the \emph{ duality map} acts on a \emph k-cube $ \mf c_k $ giving  the   \e{(n-k)}-\e{dual cell} $ \ast(\mf c_k) $ obtained with the following union of  (\e{n-k})-simplexes 
		\begin{equation}\label{eq:*op}
		\ast(\mf c_k)=\us{\mf c_n:\mf c_k\prec \mf c_n}{\bigcup}\,\,\us{ \mf c_k\prec \mf c_{k+1}\prec\dots\prec\ \mf c_n}{\bigcup}\{\mathtt{b}(\mf c_k),\mathtt{b}(\mf c_{k+1}),\dots,\mathtt{b}(\mf c_n)\},
		\end{equation}
		where in both union $ \mf c_k $ is fixed, $ \mf c_n $ varies on the first union and fixed in the second one, while the \e{(n-k-2)}-tuple $ \mf c_{k+1},\dots,\mf c_{n-1} $   vary on the second union according to the rule specified in subscript. 
		For an oriented \e k-cube $ c_k=(v_0,\dots, v_{2^{k-1}}) $ the oriented dual cell $ \ast (c_k) $ is obtained assigning to each \e{(n-k)}-simplex \footnote{These simplexes give a simplicial decomposition for the \e{(n-k)}-cube $ \us{ \mf c_k\prec \mf c_{k+1}\prec\dots\prec\ \mf c_n}{\bigcup}\{\mathtt{b}(\mf c_k),\dots,\mathtt{b}(\mf c_n)\} $.} $\{\mathtt{b}(\mf c_k),\dots,\mathtt{b}(\mf c_n)\}  $  an orientation $(\mathtt{b}(\mf c_k),\dots,\mathtt{b}(\mf c_n)) $ if  the oriented \e n-cube $ (v_0,\dots,v_{k-1},\mathtt{b}(\mf c_k),\dots,\mathtt{b}(\mf c_n)) $ is consistent with the local orientation of $ c_n $ and $- (\mathtt{b}(\mf c_k),\mathtt{b}(\mf c_{k+1}),\dots,\mathtt{b}(\mf c_n))  $ otherwise.  The \emph{dual complex} $*\mc{C}$, or \e{dual discrete manifold} ,is the collection $\{{\ast(c_k)}\}_{c_k\in\mc C} $. 
	\end{definition}

	\begin{remark}\label{r:asy} 
		The oriented (\e{n-k})-simplexes in the union \eqref{eq:*op} share two by two exactly one \e{(n-k-1)}-face that cancels. 
		In definition \ref{def:dualop} with  the oriented \e n-cubes $ c_n $ and $ -c_n $ we always associate a non-oriented object $ \mtt c(\mf c_n)=c_0  $, while with a 0-cube $ c_0 $ we associate an oriented \e n-cube always consistent with the local orientation. This "asymmetry"  is due to the fact that a vertex  doesn't have an intrinsic orientation. Formally we could fix this considering a vertex $ c_0 $ with two orientations $ \pm c_0 $ such that $*(\pm c_n)=\pm c_0  $ and $ *(\pm c_0)=\pm c_n  $. 
	\end{remark}

	\subsection{Discrete differential forms and exterior derivative}
	
	At this point we are ready to define  discrete versions of differential forms and Stokes Theorem (\ref{eq:Stokes}). Remind the idea of continuous \emph k-form as  object that can be integrated only on a \e k-submanifold and  defined as a linear map from  \emph k-dimensional sets to $\mathbb{R}$. When  \emph k-dimensional sets are defined on a mesh of a discrete manifold  we call them chains, a linear mapping from  chains to a real numbers is  quite a  natural discrete counterpart of a differential form.
	
	\begin{definition}
		Let $\mc{C}$ be a cubic complex and  $\{c^i_k\}_{i\in I_k}$ the collection of all elementary oriented \e k-cubes in $ \mc C $ indexed by $ I_k $. The space of   \emph{k-chains } $C_k(\mc{C})$ is the space with basis $\{c^i_k\}_{i\in I_k}$ of  the finite formal sums $\gamma_k=\us{i\in I_k}{\sum} \gamma_k^i c_k^i$  where the coefficient $ \gamma_k^i $ is an integer.
	\end{definition}

	In defining the discrete \emph k-forms we are not technical as usual in algebraic topology, we want just to stress that a discrete \emph k-form is a map from the space of \e k-chains to $\mathbb{R}$.

	\begin{definition}
		A discrete \emph{k-form} $\omega^k$ is a linear mapping from $C_k(\mc{C})$ to $\mathbb{R}$, i.e.
		\begin{equation}
		\omega^k(\gamma_k)=\omega^k\left(\us{i\in I_k}{\sum}\gamma_k^i c_k^i\right)=\us{i\in I_k}{\sum}\gamma^i_k\omega^k \left(c_k^i\right) .
		\end{equation}
		We add two forms $\omega_1$ and $\omega_2$ adding their values in $\mathbb{R}$, i.e. $
		(\omega_1+\omega_2)\left(\gamma\right)=\omega_1\left(\gamma\right)+\omega_2\left(\gamma\right)$. The vector space of \emph k-forms is denoted $\Omega^k(\mc{C})$.
	\end{definition}
	Any discrete \e k-form can be written as finite linear combination respect to a basis $\{\alpha^k_i\}_i$ with   same cardinality of $\{c_k^i\}_i$ and determined by the relaton $\alpha_i^k(c_k^j)=\delta_{ij}$.
	So we have a \emph{natural pairing} between chains and discrete forms, that is the bilinear pairing 
	\begin{equation}\label{eq:pairing}
	[\omega^k,\gamma_k]:=\omega^k(\gamma_k).
	\end{equation}
	Writing $\omega^k=\underset{i}{\sum}\omega^k_i\alpha^k_i$, here $\omega^k_i$ are real coefficient, the pairing \eqref{eq:pairing} becomes $\underset{i\in I_k}{\sum}\omega^k_ic_k^i$.  So the natural pairing \eqref{eq:pairing} leads to a natural notion of \emph{duality} between chains and discrete forms.
	In DEC natural pairing plays the role that integration of forms plays in differential exterior calculus. The two can be related by a discretization procedure, for example in the manifold case, thinking to have a piecewise linear \footnote{In case of non-piecewise linear manifold the discretization process present some technicalities, but it is still possible to give meaning to \eqref{eq:omega^k_d}} manifold that can be subdivided in cubes $\{\sigma_k^i\}_i$ and a differential \emph k-form $\omega^k$, the integration  of $\omega^k$ on each \emph k-cube gives its discrete counterpart $\omega^k_d$, where the subscript \emph d is for discrete, defined as
	\begin{equation}\label{eq:omega^k_d}
	\omega^k_d(\sigma_k):=\int_{\sigma_k}\omega^k.
	\end{equation}
	In this way a discrete \emph k-form is  a natural representation of a continuous \emph k-form.
	\begin{remark}
		A \emph discrete k-form can be viewed as a \emph{k-field} taking different values on different \emph k-cubes of an oriented cubic complex, e.g. a \emph k-form $\omega^k$ on  a \emph k-cube $c_k$ is such that 
		\begin{equation}\label{eq:omegafield}
		\omega^k(-c_k)=-\omega^k(c_k).
		\end{equation}
		With this in mind, for  a discrete \emph 1-form  we use also the name  \e{discrete vector field}.
	\end{remark}
	
	To define a discrete exterior derivative, that will give us a  discrete version of \eqref{eq:Stokes}, we have to introduce a discrete boundary operator. As we did so far the definition for cubes goes trough the one for simplexes.
	\begin{definition}
		The \emph{boundary operator} $\partial_k:C_k(\mc{C})\rightarrow C_{k-1}(\mc{C})$ is the linear operator that acts  on an oriented \e k-simplex $s_k=(v_0,\dots,v_{k})$ as
		\begin{equation}
		\partial_k s_k=\partial(v_0,\dots,v_{k})=\us{i=0}{\os{k}\sum}(-1)^i(v_0,\dots,\hat{v}_i,\dots, v_{k}),
		\end{equation}
		where $(v_0,\dots,\hat{v}_i,\dots, v_{k})$ is the oriented \emph{(k-1)}-simplex obtained omitting the vertex $v_i$.  Let $\varDelta c_k$ be a simplical decomposition (see definition \ref{def:s-dec}) of $c_k$, the boundary operator on \e k-cubes acts as   
		\begin{equation}
		\partial_k c_k=\us{s_k\in\varDelta c_k}{\sum}\partial_k s_k.
		\end{equation}
	\end{definition}
	
	The boundary of  non-oriented objects is obtained doing the boundary of the correspondent oriented objects  and then considering the resulting sets without orientation.
	
	\begin{example} Consider $c_2=(v_0,v_1,v_2,v_3)$ and $\varDelta c_k=(v_0,v_1,v_3)\cup(v_1,v_2,v_3)$, then $\partial_2 c_2=\partial_2(v_0,v_1,v_3)+\partial_2(v_1,v_2,v_3)=(v_1,v_3)-(v_0,v_3)+(v_0,v_1)+(v_2,v_3)-(v_1,v_3)+(v_1,v_2)=(v_0,v_1)+(v_1,v_2)+(v_2,v_3)+(v_3,v_0)$.
	\end{example}
	In practice $\partial_k$, applied to $c_k$, gives back the faces of $c_k$ with the orientation induced by $c_k$. 
	In other terms this $\partial_k$ extracts the oriented border of  an oriented \emph k-cube. A remarkable property of this operator is that  the border of a border is the void set, therefore
	\begin{equation}\label{eq:bofb}
	\partial_k\circ\partial_{k+1}=0.
	\end{equation}
	
	Now with the duality defined by the  natural paring \eqref{eq:pairing} the time is ripe to introduce the discrete exterior derivative (or coboundary operator) $d^k:\Omega^k(\mc{C})\rightarrow\Omega^{k+1}(\mc{C})$  defined by  duality \eqref{eq:pairing} and  the boundary operator.
	\begin{definition}\label{d:dder}
		For  a cubic complex $ \mc{C} $ the \emph{ discrete exterior derivative (or coboundary  operator)}    is the linear operator $d^k:\Omega^k(\mc{C})\rightarrow\Omega^{k+1}(\mc{C})$ such that 
		\begin{equation}  \label{eq:def d}
		[ d^k\omega^k,c_{k+1}]  = [ \omega^k,\partial_{k+1}c_{k+1}] 
		\end{equation}
		where  $\omega^k\in\Omega^k(\mc{C})$ and $c_{k+1}\in C_{k+1}(\mc{C})$. Moreover we set $ d\Omega^{n}(\mc C)=0 $. Definition \eqref{eq:def d}  is equivalent to
		$ d^k(\omega^k):=\omega^k\circ\partial_{k+1} $.
	\end{definition}
	From definition \ref{d:dder} and  \eqref{eq:bofb} it is straightforward the property
	\begin{equation}\label{dd=0}
	d^{k+1}\circ d^k=0.
	\end{equation}
	\begin{assumption}
		Unless otherwise specified, we are omitting if  operators are referred to the primal complex or its dual. We assume  the right one at the right moment. Moreover for discrete manifolds in definition \ref{def:dMan} the  operators of this section don't change in the dual.
	\end{assumption}
	
	From  definition \eqref{eq:def d} and natural pairing \eqref{eq:omega^k_d}  we have  a discrete  Stokes Theorem: consider a chain $\gamma_k$  and a discrete form $\omega^k$ then
	\begin{equation}\label{eq:distokes}
	\int_{\gamma_k} d^k\omega^k \equiv [ d^k\omega^k,\gamma_k] =  [\omega^k,\partial_k \gamma_k] \equiv \int_{\partial_k \gamma_k}\omega^k.
	\end{equation}
	
	\subsection{Hodge star and codifferential}
	The  counterpart of $d^k$, denoted with $ \delta^{k+1} $, mapping  a \e {(k+1)}-form into a \e{k}-form is the tool still missing to have all what we need from DEC. Given two \e k-forms $\omega^k_1$ and $\omega^k_2$, this operator is defined as the adjoint of $d$ with respect to the scalar product 
	\begin{equation}\label{eq:*scalar2}
	\langle \omega_1^k, \omega_2^k \rangle= \underset{i\in I_k}{\sum}\omega^k_{1,i} \omega^k_{2,i}
	\end{equation}
	This scalar product is the discrete version of formula \eqref{eq:smoothsp} in footnote \ref{fo:star} below.
	\begin{definition}
		The discrete codifferential operator\footnote{\label{fo:star}In the smooth case, for a manifold $ \mc{M} $ of dimension \e n,   the Hodge star is the map  $\star: \Omega_k(\mc{M}) \rightarrow \Omega_{n-k} (\mc{M})$, defined by its local metric and the local scalar product of \e k-forms $\langle\langle \omega^k_1,\omega^k_2 \rangle\rangle=(\omega^k_1)^{i_1,\dots,i_k}(\omega^k_2)_{i_1,\dots,_k}$, such that 
			\begin{equation*}\label{eq:*smooth}
			\omega^k_1\wedge\star\omega^k_2:=\langle\langle \omega^k_1,\omega^k_2\rangle\rangle vol^n,
			\end{equation*}
			where $vol^n$ is the volume form on the manifold. Denoting with $ d $ the exterior derivative for differential forms, this operator can be computed through its action on the the basis of \e k-forms $\mathbf{dx}^k=dx^{i_1}\wedge\dots \wedge dx^{i_k}$ ($i_1<\dots<i_k$), that gives back the forms $\star \mathbf{dx}^k=C\mathbf{dx}^{n-k}=C dx^{i_{k+1}}\wedge\dots \wedge dx^{i_n}$ ($i_{k+1}<\dots<i_{n}$) with $C$ is  such that $\mathbf{dx}^k\wedge C \mathbf{dx}^{n-k}=\langle\langle \mathbf{dx}^k,\mathbf{dx}^{k} \rangle\rangle vol^n=vol^n$.
			On a smooth manifold $ \mc{M} $, for details see chapter 14 in \cite{Fra12} or section 6.2 in \cite{AMR88}, the adjoint operator  $ \delta $ of $ d $ is defined respect to the scalar product 
			\begin{equation}\label{eq:smoothsp}
			\langle \omega^k_1,\omega^k_2\rangle:=\int_{\mc{M}}\omega^k_1\wedge\star\omega^k_2=\int_{\mc{M}}\langle\langle \omega^k_1,\omega^k_2\rangle\rangle vol^n.
			\end{equation}
			Thinking about,  \eqref{eq:*scalar2} is the discrete version of the most right term of the last formula \eqref{eq:smoothsp} and it can be introduced in a sophisticated way that "emulates" the continuous case of this footnote  \ref{fo:star}. Let's try to sketch  this parallelism.
			Since the duality $\ast$ maps a primal cell into an only one dual cell and vice versa,  the most spontaneous thing to set a \e{discrete Hodge star} $ \star $  from \e k-forms into  \e {(n-k)}-forms is doing it from  $ \Omega^{k}(\mc{C}) $ into $ \Omega^{n-k}(\ast\mc{C}) $, i.e. $\star:\Omega^k(\mc{C})\rightarrow\Omega^{n-k}(\ast\mc{C})$, this can be defined with the relation $ (\omega^k,c_k)=(\star\omega^k,\ast c_k) $. With this definition \eqref{eq:*scalar2} can be written as $ \langle\omega^k_1,\omega^k_2\rangle:=\sum_{i\in I_k}(\omega^k_1,c_k^i) (\star\omega^k_2,\ast c^i_k)= \underset{i\in I_k}{\sum}\omega^k_{1,i} \star\omega^k_{2,i} $ that is the equivalent of the middle term in \eqref{eq:smoothsp}. For details how to define a discrete wedge product and the counterpart of the \e n-volume form see section 12 of \cite{DHLM05} and \cite{DMTS14}} $ \delta^{k}:\Omega^k(\mc C)\to\Omega^{k-1}(\mc C) $ is defined by $ \delta^0\Omega^0(\mc C)=0 $ and the equation 
		\begin{equation}\label{eq:codef}
		\langle d^k\omega_1^{k}, \omega_2^{k+1} \rangle= \langle \omega_1^{k}, \delta^{k+1}\omega_2^{k+1} \rangle,
		\end{equation}
		where $ \omega^k_1\in\Omega^k(\mc C) $ and $ \omega_2^{k+1}\in \Omega^{k+1}(\mc C) $.
	\end{definition}                                    
	Also for this operator we have a property analogous to \eqref{dd=0}, that is 
	\begin{equation}\label{eq:deltaodelta}
	\delta^k\circ\delta^{k+1}=0.
	\end{equation}                             
	For a   matrix description of the operators of this chapter see \cite{DKT08}. The operators $ \ast,\partial,d,\star $ and $ \delta $ we introduced on $ \mc C $ can be defined also on the dual complex $ \ast\mc C $. In particular when a discrete manifold is considered and the inverse map of duality map $ \ast $ can be written as  
	\begin{equation}\label{e:ast}
	\ast^{-1}=(-1)^{k(n-k)}\ast
	\end{equation}  where $ \ast $ acts on $ \ast \mc C $ as on $ \mc C $ , i.e. when $ \ast\mc C $ is still  made of \e k-cubes, all definitions applies in the same way. In this case\footnote{\label{f:delta}In this case we have also the formula $ \delta^{k+1}\omega^{k+1}=(-1)^{nk+1}\star d^{n-(k+1)}\star\omega^{k+1} $, which is analogous to the one for $ \delta $ in smooth boundaryless maniflods. This is proved using $ \star^{-1}=(-1)^{k(n-k)}\star $, which follows from \eqref{e:ast} and the definition of $ \star  $ in footnote \ref{fo:star}, see also subsection 5.5 in \cite{DKT08}.} the same notation for the discrete operators is used both on $ \mc C $ and $ \ast \mc C $.

	\subsection{Hodge decomposition}
	The Hodge decomposition in our context is as follows.
	
	\begin{theorem}\label{th:HT}
		Let $ \mc{C} $ be a discrete manifold of dimension n and let $ \Omega^k(\mc{C}) $ be the space of   k-forms on $\mc{C}$. The following orthogonal decomposition holds for all k:
		\begin{equation}\label{eq:dhodd}
		\Omega^k(\mc{C})=d^{k-1}\Omega^{k-1}(\mc{C})\oplus\delta^{k+1}\Omega^{k+1}(\mc{C})\oplus\Omega^k_H(\mc{C}),
		\end{equation}
		where $\oplus$ means direct sum and  $ \Omega^k_H(\mc{C})= \{\omega^k|d^{k}\omega^k=\delta^k\omega^k=0\}$ is the space of harmonic forms.
	\end{theorem}
	
	\begin{proof}
		Consider  the discrete forms  $ \omega^{k-1}_1\in\Omega^{k-1}(\mc C) $, $ \omega^{k+1}_2\in\Omega^{k+1}(\mc C) $ and $ h^k\in \Omega_H^k(\mc C) $ chosen arbitrary form their respective spaces.
		Since $ \delta^{k+1}  $ is the adjoint of $ d^k $ using  \eqref{dd=0} and \eqref{eq:codef} we have $ \langle d^{k-1}\omega_1,\delta^{k+1} \omega_2 \rangle=\langle d^kd^{k-1}\omega_1, \omega_2 \rangle =0$.  By the definition of $ \Omega^k_H(\mc C) $ and \eqref{eq:codef} we have $ \langle h^k,d^{k-1} \omega^{k-1}_1 \rangle=\langle \delta^k h^k, \omega^{k-1}_1 \rangle=0$, likewise $\langle h^k, \delta^{k+1}\omega^{k+1}_2 \rangle=\langle  d^kh^k, \omega^{k+1}_2 \rangle=0 $. Therefore the spaces $ d^{k-1}\Omega^{k-1}(\mc C)$, $ \delta^{k+1}\Omega^{k+1}(\mc C)$ and $ \Omega^k_H(\mc C) $ are each other orthogonal. A general \e k-form $ \omega^k\in \Omega^k(\mc C)$  belongs to $ \left(d\Omega^{k-1}(\mc{C})\oplus\delta\Omega^{k+1}(\mc{C})\right)^\perp$ if and only if  $ \langle \omega^k,d^{k-1} \omega^{k-1}_1+ \delta^{k+1}\omega^{k+1}_2 \rangle=\langle \delta^k \omega^k, \omega^{k-1}_1 \rangle+\langle  d^k \omega^k, \omega^{k+1}_2 \rangle=0 $ for all $ d^{k-1} \omega^{k-1}_1+ \delta^{k+1}\omega^{k+1}_2\in d\Omega^{k-1}(\mc{C})\oplus\delta\Omega^{k+1}(\mc{C}) $, namely $ d^k\omega^k=\delta^k\omega=0 $. Then  we  showed  $ \left(d\Omega^{k-1}(\mc{C})\oplus\delta\Omega^{k+1}(\mc{C})\right)^\perp=\Omega^k_H$, so the decomposition is complete and generates all $ \Omega^k(\mc C) $.
	\end{proof}
	
	\section[Discrete operators  on   $ \mathbb{T}^3 $]{Discrete operators  on  the discrete torus}\label{sec:disop}
	
	In this section we show how the discrete exterior derivative and its adjoint work on a cubic complex of dimension three, namely we work with \e 0-forms, \e 1-forms, \e 2-forms and \e 3-forms, setting discrete equivalent of gradient, curl and divergence operators.
	\emph{ We consider a discrete mesh of edge 1 on the discrete torus $ \mathbb{T}^3_N=\mathbb{Z}^3/N\mathbb{Z}^3 $ of side N and we refer explicitly to other dimensions whenever appropiate}.  
	Thinking about the bases $dx_i$, $dx_idx_j $ and $ dx_1dx_2dx_3 $ respectively for \e 1-forms, \e 2-forms and \e 3-forms, the parallel of the left   table \ref{tab:forme} between  the smooth case and the discrete case is really natural. For smooth form the action of $ d $ and $ \delta $ can be summarized  as in the right table  \ref{tab:forme}.

	\begin{table}
		\centering
		\caption[]{Left table: conceiving continuous  and discrete forms. Right table: $ d $ and $ \delta $ applied to smooth forms.}
		\label{tab:forme}
		\begin{tabular}{|cc|cc|}
			\hline
			\multicolumn{2}{|c|}{smoooth case}&\multicolumn{2}{c|}{discrete case}\\ 
			\hline
			$ \omega^0 $: & scalar field &  $ \omega^0 $: & vertex field \\
			$ \omega^1 $: & vector field &  $ \omega^1 $: & edge field\\
			$ \omega^2 $: & vector  field & $ \omega^2 $: & face field\\
			$ \omega^3 $: & scalar field &  $ \omega^3 $: & cell field\\
			\hline
		\end{tabular}
		\hspace{0.08mm}
		\begin{tabular}{|c|c|c|}
			\hline
			\multicolumn{3}{|c|}{smoooth case}\\
			\hline
			form & $ d $ &  $ \delta $ \\
			\hline
			$ \omega^0 $ &  grad$ \,\omega^0 $ &  0 \\
			$ \omega^1 $ &  curl$ \,\omega^1 $ &  -div$ \,\omega^1 $ \\
			$ \omega^2 $ &  div$ \,\omega^2 $  &  curl$ \,\omega^2 $ \\
			$ \omega^3 $ &   0 & -grad$ \,\omega^3 $ \\
			\hline
		\end{tabular}{ }
	\end{table}
	
	Now we are indicating also the discrete operators of last section \ref{sec:DEC} without the index $ k $,  unless otherwise specified, it will be implicit to use the right one according to the \e k-form on which they act. We want to show how to compute $ d $ and $ \delta $. Calculations will be performed in some cases, the others follow mutatis mutandis. We recover also the usual divergence, Gauss and Stokes theorems in our discrete setting. A part proposition \ref{prop:Omega^1_H}, all what  we do in this section for the discrete torus can be done with some extra work and notation for general discrete manifold with non regular mesh (i.e. that can not be defined using the canonical basis $ \{e_1,e_2,e_3\} $).

	\subsection{Notation}\label{ss:not}Let  $\{{e_1,e_2,e_3}\} $ be a canonical right handed orthonormal basis.  We define the sets of all the \e{vertices}  
	\begin{equation*}\label{key}
	V_N:=\{ x=(x_1,x_2,x_3):x_i=0,e_i,\dots,(N-1)e_i\}.
	\end{equation*}
	all the oriented \e{edges} $ E_N=\left(\os{3}{\us{i=1}{\bigcup}}E_N^{i,+}\right)\cup \left(\os{3}{\us{i=1}{\bigcup}}E_N^{i,-}\right) $ where 
	\begin{equation*}\label{key}
	E_N^{i,\pm}=\{\text{positive/negative oriented \e 1-cubes } e=\pm(x,x+e_i),i\in{1,2,3}\},
	\end{equation*} all the oriented \e{faces} $ F_N=\left(\os{3}{\us{i=1}{\bigcup}}F_N^{i,+}\right)\cup \left(\os{3}{\us{i=1}{\bigcup}}F_N^{i,-}\right) $ where 
	\begin{equation*}\label{key}
	F_N^{1,\pm}=\{\text{positive/negative oriented  \e 2-cubes } f_1=\pm(x,x+e_2,x+e_2+e_3,x+e_3) \},
	\end{equation*} 
	\begin{equation*}\label{key}
	F_N^{2,\pm}=\left\{\text{positive/negative oriented \e 2-cubes } f_2=\pm(x,x+e_3,x+e_3+e_1,x+e_1) \right\},
	\end{equation*} 
	\begin{equation*}\label{key}
	F_N^{3,\pm}=\left\{\text{positive/negative oriented \e 2-cubes } f_3=\pm(x,x+e_1,x+e_1+e_2,x+e_2)\right\}
	\end{equation*}
	and finally all the oriented \e{cells} $C_N=\left(C_N^{+}\right)\cup \left(C_N^{-}\right)$  where 
	\begin{equation*}\label{key}
	C^{\pm}_N:=\left\{\textrm{positive/negative oriented \e 3-cubes } \pm c:\texttt{b}(c)=x+\frac{e_1}{2}+\frac{e_2}{2}+\frac{e_3}{2}\right\}.
	\end{equation*}
	
	\begin{remark}\label{r:face}
		Observe that for a face the orientation can be defined as for a \e 2-simplex, see definition \ref{d:simdef} and \ref{d:orsim}, i.e. with the ordering given to the vertices in its sequence.
	\end{remark}

	We indicate a general  oriented edge,  oriented face and oriented cell respectively with $ e=(x,y) $, $ f $ and $ c $, where $ y=x\pm e_i $ for some $ i $, while for the   non-oriented  ones we use a  calligraphic writing, that is $ \mf e=\{x,y\} $,  $ \mf f $ and $ \mf c $. We indicate in calligraphic also the non-oriented collections just defined above, that is   $ \mc E_N $, $ \mc F_N $ and $ \mc C_N $. 
	The collections of \e k-cube defining the discrete Manifolds $ \mc C $  on $ \bb T^3_N $, $ \bb T^2_N $ and $ \bb T_N $ are respectively $ \left\{V_N,E^+_N,F^{+}_N,C_N^+\right\} $, $ \left\{V_N,E^{1,+}_N,E^{2,+}_N,F^{+}_N\right\} $ where $ F^{+}_N $ is defined as $ F^{3,+}_N $ and   $ \left\{V_N,E^{+}_N\right\} $ where $ E_N^+ $ is defined as $ E_N^{+,1} $. 
	
	A general vertex, edge, face and cell field is going to be denoted respectively with $h(x)$, $j(x,y) $,  $\psi(f) $ and  $ \rho(c)$. 
	
	The dual complex $ \ast\mc C $ is constructed on a mesh with the same geometrical structure of  the original one and  obtained translating each vertex of $ (e_1/2,e_2/2,e_3/2) $. We denote with $ \bb T^{n,\ast}_N $ the dual torus  obtained translating the vertexes. So we use an analogous notation and indicate with an index $ * $ the collections and the elements of the dual complex, that is respectively $ V_N^* $, $ E_N^* $, $ F_N^* $ and $ C_N^* $  and $ x^* $, $ e^*=(x^*,y^*) $, $ f^* $ and $ c^* $. 
	
	\begin{remark}\label{r:orob}
		With a vertex $ x\in V_N $ are associated the three edges  $ \{x,x+e_1\}$,$\{x,x+e_2\} $ and $ \{x,x+e_3\} $, consequently we associate with $ x $  also the oriented edges $ \pm (x,x+e_1)$,$\pm(x,x+e_2) $ and $ \pm(x,x+e_3)$. While in two dimension with $ x\in V_N $ are associated $ \{x,x+e_1\}$ and $\{x,x+e_2\} $ and in one dimension only $ \{x,x+e_1\}$.
		
		With a vertex $ x\in V_N $ are associated the three  faces $ \mf f_1,\mf f_2 $ and $ \mf f_3 $ such that    $\mathtt b(\mf f_k)=x+\frac{e_i}{2}+\frac{e_j}{2}  $  where $ i,j,k\in\{1,2,3\} $ and  $ i\neq j \neq k $, consequently we associate with $ x $  also the oriented faces $ \pm f_1$,$\pm f_2 $ and $ \pm f_3$. In two dimensions we don't have a subscript on the faces and only one  face $ \mf f $, defined as $ \mf f_3 $, is associated with  $ x $. 
		
		Finally with a vertex $ x\in V_N $ is associated only one  cell $ \mf c $ such that $ \texttt{b}(\mf c)=x+\frac{e_1}{2}+\frac{e_2}{2}+\frac{e_3}{2} $, consequently the oriented cells $ \pm c $ are associated with $ x $.
	\end{remark}

	\subsection{\label{ss:d 0to1}$ \mathbf{d:\Omega^0\rightarrow\Omega^1} $} 
	Consider $ (x,y)\in E_N $ and compute
	$
	dh(x,y)=h\circ\partial(x,y)=h(y-x)=h(y)-h(x).
	$
	Choosing $ (x,y)=(x,x+e_i) $  we define the discrete gradient
	\begin{equation}\label{eq:dgrad}
	dh(x,x+e_i)=h(x+e_i)-h(x)=:\nabla_i h(x).
	\end{equation}
	For a \e 1-chain $ \gamma=\os{m}{\us{i=1}{\sum}}(x_k,y_k) $, where $ y_k=x_{k+1} $, setting $ \int_\gamma \nabla h\cdot dl:=\us{k=1}{\os{m}{\sum}}\nabla_{i_k} h(x_{k}) $ we get
	\begin{equation}\label{eq:dlineint}
	dh(\gamma)=\int_{\gamma}\nabla h\cdot dl=h\circ\partial(\gamma)=\int_{\partial\gamma}h = h(y_m)-h(x_1),
	\end{equation}
	where $ \p\gamma $ is the boundary of $ \gamma $. This gives us  the discrete versions of  line integral. For completeness we translates in our language the well known result relating gradient fields to  zero integrations on closed paths.
	\begin{proposition}\label{prop:prop1form}
		A 1-form $ j(x,y)\in d\Omega^0$ if and only if $ j(\gamma)=\oint_\gamma j=0 $ for all closed path(chain) $ \gamma $ on $ V_N $. Moreover, the vertex function  \begin{equation}\label{eq:h^x(y)}
		h^x(y):=\us{(w,z)\in\gamma_{x\rightarrow y}}{\sum}j(w,z) 
		\end{equation} is such that $ j(y,y')=h^x(y')-h^x(y) $ for every $ (y,y')\in E_N $ and it doesn't depend on the particular path  $ \gamma_{x\rightarrow y} $   from $ x $ to $ y $. Any two functions $ h^{x'}(\cdot) $  and  $ h^x(\cdot) $  differ for  an additive constant.
	\end{proposition}
	
	\begin{proof}
		If $ j(x,y)=dh(x,y) $ from \eqref{eq:dlineint} we get $ \oint_\gamma j =0 $ since $ y_m=x_1 $. Now let's prove the opposite implication. Let $ \gamma $ be a closed path and $ x,y $ any two points on it. Since $ \gamma=\gamma_{x \rightarrow y}+\gamma_{y\rightarrow x} $, from $ \us{(w,z)\in\gamma}{\sum}j(w,z)=\us{(w,z)\in\gamma_{x\rightarrow y}}{\sum}j(w,z)+\us{(w,z)\in\gamma_{y\rightarrow x}}{\sum}j(w,z)=0 $, we have $ \us{(w,z)\in\gamma_{x\rightarrow y}}{\sum}j(w,z)=\us{(w,z)\in\gamma'_{x\rightarrow y}}{\sum}j(w,z) $, where $ \gamma'_{x\rightarrow y}=-\gamma_{y\rightarrow x} $. Hence the function \eqref{eq:h^x(y)} doesn't depend on the particular path from $ x  $ to $ y $ and for any $ (y,y')\in E_N $ its gradient is $ h^x(y')-h^x(y)=\us{(w,z)\in\gamma_{x\rightarrow y}\cup(y,y')}{\sum}j(w,z)-\us{(w,z)\in\gamma_{x\rightarrow y}}{\sum}j(w,z)=j(y,y')$. Taking a closed path $ \gamma=\gamma_{x'\rightarrow y}+\gamma_{y\rightarrow x}+\gamma_{x\rightarrow x'} $ we obtain $ h^{x'}(y)=h^x(y)+\us{(w,z)\in\gamma_{x'\rightarrow x}}{\sum}j(w,z) $  because of $ \oint_\gamma j=0 $.
	\end{proof}
	
	\subsection{$ \mathbf{d:\Omega^1\rightarrow\Omega^2} $} 
	Let   $ f_k\in F^{k,+}_N$ be the face with index $ k $ that is  associated with  $ x $ as in previous notation in subsection \ref{ss:not}, see remark \ref{r:orob}. We have
	$ dj(f_k)=j\circ\partial (f_k)= \us{(x,y)\in \p f_k}{\sum}j(x,y)$, that we write
	\begin{equation}\label{eq:dj}
	dj(f_k)=\us{(x,y)\in \p f_k}{\sum}j(x,y)=:\int_{\p f_k}j\cdot dl.
	\end{equation}
	We compute $ dj(f_3) $ as
	$
	dj(f_3)=j(x,x+e_1)+j(x+e_1,x+e_1+e_2)-j(x+e_2,x+e_1+e_2)-j(x,x+e_2),
	$
	calling $ \nabla_l j_m(x):= j(x+e_l,x+e_l+e_m)-j(x,x+e_m) $, where $ j_m(x):=j(x,x+e_m) $, the last equation becomes
	$ dj(f_3)=\nabla_1 j_2(x)-\nabla_2 j_1(x) $, so  $ dj(f_k) $ is also the curl defined as 
	\begin{equation*}\label{eq:dcurl}
	dj(f_k)=\varepsilon^{klm}\nabla_lj_m(x)=:(\mathrm{curl\,}j)(f_k),
	\end{equation*}
	where $ \veps^{klm} $ is the Levi-Civita symbol summed on repeated indexes with the Einstein convention.
	Let $ S $ be the \e 2-chain $ S=\os{m}{\us{i=1}{\sum}} f_i $,  where $ \{f_i\}_{i=1}^{m} $ can be any collection  of  oriented faces in $ F_N $. Setting $ \int_S \textrm{curl}j\cdot d\varSigma:=\us{i=1}{\os{m}{\sum}}(\textrm{curl\, }j)(f_i) $,  we have
	\begin{equation}\label{eq:dstokes3d}
	\int_S\textrm{curl\,}j\cdot d\varSigma=\int_\Gamma j\cdot dl.
	\end{equation}
	where $  \Gamma=\os{m}{\us{i=1}{\sum}}\p f_i  $. If all the faces in the collection have same orientation, i.e. $ f_i\in F^\pm_N $ for all $ i\in 1,\dots,m $, the chain $  \Gamma=\os{m}{\us{i=1}{\sum}}\p f_i  $ is the closed path  corresponding to  the inductively oriented boundary of the oriented surface $ S $, indeed   when a edge is  shared by two faces $ f_i $ and $ f_j $ it cancels and  gives two contributions with same modulo but opposite sign in $ \int_{\p f_i}j\cdot dl $ and $ \int_{\p f_j}j\cdot dl $. We obtained in this last case  a  discrete version of the usual Stokes theorem.
	\subsubsection*{\quad Two dimensions.}In two dimensions  $ dj(f)=\nabla_1 j_2(x)-\nabla_2 j_1(x) $,  where $ f $ is the face associated with $ x $ as $ f_3 $ in notation and   $ \textrm{curl}j $ is a \e 2-form defined  on the faces  lying on the plane of  $ e_1 $ and $ e_2 $. Proceeding like we did in three dimensions we have the discrete Green-Gauss formula.
	
	\subsection{$ \mathbf{d:\Omega^2\rightarrow\Omega^3} $} Consider the $ c\in C_N^+ $,  we compute and define the divergence $ d $ as follows 
	\begin{equation}\label{eq:dpsi}
	d\psi(c)=\psi\circ\partial(c)=\us{f\in F_N:f\in \p c}{\sum}\psi(f)=:\div\psi(c).
	\end{equation}
	As in notation the cell $ c $ is associated with the vertex $ x\in V_N $.  Let $ f_i $ with $ i\in \{1,2,3\} $ the three positive faces  associated with $ x $. We have then
	$\us{f:f\in \p c}{\sum}\psi(f)=\us{i=1}{\os{3}{\sum}}\psi(f_i+e_i)-\psi(f_i) $. We define the flow across $ c $ as
	\begin{equation}\label{eq:ddiv}
	\int_{\partial c}\psi\cdot d\varSigma:=\us{i=1}{\os{3}{\sum}}\Phi(f_i),
	\end{equation}
	where $ \Phi(f_i):=\psi(f_i+e_i)-\psi(f_i) $. Let $ V $ be the 3-chain $ V=\os{m}{\us{i=1}{\sum}}c_i $, where $ \{c_i\}_{i=1}^{m} $ can be any collection  of  oriented cubes in $ C_N $. Setting $ \int_V \textrm{div\,}\psi \,dx:=\us{i=1}{\os{m}{\sum}}(\textrm{div}\psi)(c_i) $, we have
	\begin{equation}\label{eq:ddivth3d}
	\int_V\textrm{div\,}\psi\, dx=\int_{S} \psi\cdot d\varSigma,
	\end{equation}
	where $ S=\os{m}{\us{i=1}{\sum}}\p c_i $. If all the cubes in the collection have same sign orientation, i.e. $ c_i\in C_N^\pm $ for all $ i\in 1,\dots,m $, the chain  $ S=\os{m}{\us{i=1}{\sum}}\p c_i $ is the closed surface corresponding to the  inductively oriented boundary of the oriented volume $ V $, indeed  when a face is  shared by two cells $ c_i $ and $ c_j $ it cancels and  gives two contributions with same modulo but opposite sign in $ \int_{\partial c_i}\psi\cdot d\varSigma $ and $ \int_{\partial c_j}\psi\cdot d\varSigma $. 
	We obtained in this last case a discrete version of the usual   divergence theorem.
	
	\bigskip
	Now we compute the codifferential operator $ \delta $, this can be done using the fact that it is the adjoint of $ d $ with respect to the scalar product \eqref{eq:*scalar2} or using the formula for $ \delta $ in footnote \ref{f:delta} since \eqref{e:ast} is true. We don't write the computation but we give just the explicit form for $ \delta $.
	
	\subsection{$ \mathbf{\delta:\Omega^3\rightarrow\Omega^2} $}
	Let   $ f_i\in F^{i,+}_N$ be the face with index $ i $ that is  associated with  $ x $ as in  notation of subsection \ref{ss:not}, see remark \ref{r:orob}, we have
	\begin{equation}\label{eq:co3-2}
	\delta\rho(f_i)=\rho(c-e_i)-\rho(c)=\us{c':f_i\in \p c'}{\sum}\rho(c')
	\end{equation}
	With the help of the dual complex we can recover line integrals analogously to subsection \ref{ss:d 0to1}. To do it we should interpret $ \rho $ as 0-form on the dual vertex $ x^*=\ast c $ of a cell $ c $  (see the operator $ \star $ defined in  footnote  \ref{f:delta})   and $ \delta\rho $ as discrete vector field of gradient type ($ \nabla_i\rho(c):=\rho(c)-\rho(c-e_i) $) on the dual edges $ (x^*,y^*)=*(f_i) $ of the faces  $ f_i\in F^+_N $.  The collection of cells  $ c $ of $ C_N $ that would define the dual chain $ \gamma^* $ equivalent to  $ \gamma $ in \eqref{eq:dlineint} should be such that   any two consecutive  cells $ c_i$,   $c_{i+1} $ in the collection  share a common face that cancels.

	\subsection{$ \mathbf{\delta:\Omega^2\rightarrow\Omega^1} $}\label{ss:delta3to2}
	On an edge $ (x,y)\in 	E_N $ the codifferential operator is 
	\begin{equation}\label{eq:co2-1}
	\delta\psi(x,y)=\sum_{f:(x,y)\in \p f}\psi(f).
	\end{equation}
	Let $\{f_j\}_{j=1}^3 $ be the faces associated with $ x $ and let's set $ \nabla_i\psi(f_j):=\psi(f_j)-\psi(f_j-e_i) $, the  codifferential  can be rewritten $ \delta\psi(x,x+e_k)=\epsilon^{klm}\nabla_l\psi(f_m)=: \textrm{curl}\,\psi(x,x+e_k)$ getting the relation
	$
	\textrm{curl}\,\psi(x,x+e_k)=\sum_{f:(x,x+e_k)\in f}\psi(f)
	$.
	The definition of curl as integral on the border of a surface as \eqref{eq:dj} and a the discrete version of the Stokes theorem like \eqref{eq:dstokes3d} can be recovered with the dual complex. We should interpret (see the operator $ \star $ defined in  footnote  \ref{f:delta})   $ \mathrm{curl}\,\psi(x,x+e_k) $ as a 2-form on the dual  face $f^*= *(x,x+e_k) $ and $ \psi(f)  $ as discrete vector fields on the dual edges $ (x^*,y^*)=*f $. With a suitable collection of edges $ \{e_i\}_{i=1}^m $ we can define any desired surface $ S^* $.
	
	\subsubsection*{\quad Two dimensions.} 
	Formula \eqref{eq:co2-1} is still true, calling $ f $ the face associated to $ x $,  the codifferential is interpreted as an  orthogonal gradient
	\begin{equation}\label{eq:dortgra}
	\left(\begin{array}{ll}
	\delta\psi(x,x+e_1)\\
	\delta\psi(x,x+e_2)
	\end{array}\right)
	=\left(\begin{array}{ll}
	\;\;\;\nabla_2\psi(f)\\
	-\nabla_1\psi(f)
	\end{array}\right)=:\nabla^\perp \psi(x).
	\end{equation}
	
	\subsection{$ \mathbf{\delta:\Omega^1\rightarrow\Omega^0} $}\label{ss:delta1to0}
	For a vertex $ x\in V_N $
	the codifferential results to be 
	\begin{equation}\label{eq:cod1-0}
	\delta j(x)=\sum_{y:(y,x)\in E_N}j(
	y,x).
	\end{equation}
	To highlight  the analogous nature of this operator to that of \eqref{eq:dpsi}, we define $ \Phi_i(x):=j(x,x+e_i)-j(x-e_i,x) $. Here $ \Phi_i  $ can be thought as the net flow passing trough the "source/sink" $ x  $.
	We define a divergence operator as $ \div j(x):=-\delta j(x) $, i.e.
	\begin{equation}\label{eq:cod1-0}
	\textrm{div\,}j(x)=\us{y:(x,y)\in E_N}{\sum}j(x,y)=\sum_i\Phi_i(x)
	\end{equation}
	To recover the usual definition of divergence like \eqref{eq:dpsi} with \eqref{eq:ddiv} and a discrete  divergence theorem like \eqref{eq:ddivth3d} we should interpret the divergence in \eqref{eq:cod1-0} with the help of the dual complex. We have to think about (see the operator $ \star $ defined in  footnote  \ref{f:delta}) $ \mathrm{div} j(x) $ as a 3-form on the cube $c^*= *(x) $ and $ j(x,y)  $ as  2-forms on the faces $ f^*=*(x,y) $. With a suitable collection of vertexes $ \{x_i\}_{i=1}^m $ we can define any desired volume $ V^* $.
	\subsubsection*{\quad Other dimensions.} In other dimensions definition \eqref{eq:cod1-0} doesn't change.
	
	\bigskip
	\bigskip
	
	We conclude determining the  space of harmonic discrete vector fields on the discrete manifold of $ \mathbb{T}^n_N $. 
	
	\begin{proposition}\label{prop:Omega^1_H}
		Consider    on $ \mathbb{T}^n_N $ the collections $ V_N $ and $ E_N^+ $  in $ \mc C $ defined in notation \ref{ss:not}, then in any dimension  d the harmonic space $  \Omega^1_H(\mc {C}) $ has dimension d and it is generated by the discrete vector fields
		\begin{equation}\label{eq:Omega^1_Hbasis}
		\vphi^{(i)}(x,x+e_j):=\delta_{ij},\qquad  i,j=1,\dots, n. 
		\end{equation}
	\end{proposition}
	
	\begin{proof}
		For each $ i\in\{1,\dots,n\} $ we have both $ d^1\vphi^{(i)}(f)=0 $ for all $ f\in F_N $ and $ \delta^1 \vphi^{(i)}(x)=0 $ for all $ x\in V_N $, then $ \vphi^{(i)} \in \Omega^1_H $ and the set of discrete vector fields $ \{\vphi^{(i)}(x)\}_i $ generate a \e d-dimensional subspace of $ \Omega_H^1 $.  Since $ \textrm{Ker} \,d^1=d^0\Omega^0 \oplus \Omega^1_H $ we have $ \dim \Omega^1_H=\dim(\textrm{Ker}\, d^1)-\dim(\Ima d^0) $, this is also the dimension of the quotient space $ \Ker d^1/\Ima d^{0}  $  of closed discrete vector fiels (\e 1-forms) differing for  an exact \e 1-form. By duality (see natural pairing \eqref{eq:pairing}) between $ d $ and $ \p $ the quotient space $ \Ker d^1/\Ima d^{0}$ is isomorphic to quotient space of closed \e 1-chains that differ for the  boundary of a \e 2-chain $\Ker \p^1/\Ima \p^{2} $, which dimensionality is  given by the number of independent circle $ S^1 $ in the cartesian product $ S_1\times\dots\times S_1 $ homeomorphic to the continuous torus  $  \mathbb{T}^n $, that is $ n $. Therefore $ \{j_i(x)\}_{i=1}^n $ generates all $ \Omega_H^1 $. \end{proof}

\section{Functional Hodge decomposition}\label{sec: HD} 

The instantaneous current $ j_{\eta}(x,y) $ defined in subsections \ref{subsec:Ipc} and \ref{subsec:Ie-mc} for every fixed configuration $ \eta\in\Sigma_N $ is a discrete vector field, i.e. $ j_{\eta}(x,y)=-j_{\eta}(y,x) $. The configurations space was defined at the beginning of \ref{sec:CTMC}. For $ \eta $ fixed, the Hodge decomposition in theorem \ref{th:HT} tell us  it can be decomposed into a gradient part $ j^d_\eta $, a curl part $ j^\delta_\eta  $ and  an harmonic part $ j^H_\eta  $. For practical purpose, this theorem is resumed  immediately after. When a  discrete vector fields is \e{translational covariant}, that is for any $ z\in V_N $ 
\begin{equation}\label{eq:covcond}
j_{\eta}(x,y)=j_{\tau_z\eta}(x+z,y+z) \text{ for all } ((x,y),\eta)\in E_N\times \Sigma_N,
\end{equation}
where the translation configuration $ \tau_{z}\eta $ is going to be defined here after, we develop a generalization of the Hodge decomposition proving a functional Hodge decomposition for translational covariant discrete vector fields,  the proof is constructive providing a method
to compute explicitly these three components, which will be associated with some functions on
the configurations space. This decomposition is unique.
Since we deal translational covariant stochastic lattice gas, see definition \eqref{eq:tcov}, their instantaneous current is translational covariant too. A relevant notion in the derivation of the hydrodynamic behaviour for diffusive particle systems is the definition of gradient particle system, see \eqref{eq:grpar}.
The relevance of this notion is on the fact that the proof of the hydrodynamic limit for gradient systems is extremely simplified.
In addition, when also the invariant measure is known it is possible to obtain explicit expressions of the transport coefficients. The knowledge of the invariant measure typically happens in the reversible case, but the hydrodynamics limit and transport coefficients derived in this case are still the same when a  boundary driven out of equilibrium version of the system is considered, see \cite{ELS96}. The hydrodynamic limits are the macroscopic behaviours of the microscopic dynamics presented in chapter \ref{ch:IPSmodels} and they are going to be discussed later in chapter \ref{ch:HSL}.   So our goal is to introduce the idea that this  functional decomposition could be useful  in the study of  the hydrodynamics of nongradient systems. Moreover  in chapter \ref{ch:SNS} we will apply the results of this section to study some nonequilibrium stationary states.

\bigskip

The spaces of \e k-forms are the ones of the last section \ref{sec:disop}, where  \e k-forms $ h\in\Omega^0,j\in\Omega^1,\psi\in\Omega^2 $ and $ \rho\in\Omega^3 $  are defined on the suitable (according to the dimension , see notation in subsection \ref{ss:not}) cubic complex $ \mc C=\{V_N,E^+_N,F^+_N,C^+_N\}$ on the discrete torus $ \mathbb{T}^n_N:=\mathbb{Z}^n/N\mathbb{Z}^n $ of mesh one. The configuration space on the discrete torus is $ \Sigma_N=\Sigma^{V_N} $.
The space of 1-forms $ \Omega^1 $ is the vector space of discrete vector fields, where a discrete vector field $ \vphi $ on $ E_N $ is a map $ \vphi:E_N\to \bb R $ such that for $ (x,y)\in E_N $ we have $ \vphi(x,y)=-\vphi(y,x) $, and it is  endowed with the scalar product  \eqref{eq:*scalar2}, that we write
\begin{equation}\label{sc}
\langle \vphi,\chi\rangle:=\frac 12\sum_{(x,y)\in E_N}\vphi(x,y)\chi(x,y)\,, \qquad \vphi,\chi\in \Omega^1.
\end{equation}
In section \ref{sec:disop} we defined on $ \bb T^n_N $ the discrete exterior derivative $ n $ and its adjoint $ \delta $ respect to this scalar product. As in section \ref{sec:disop} we don't write the index on $ d $ and $ \delta $. In this subsection we will need of $ d $ and $ \delta $ respectively as in subsections \ref{ss:d 0to1} and \ref{ss:delta3to2}, anyway their definitions are reported here after in a proper form for this section. For details about the Hodge decomposition  and the related objects see sections \ref{sec:DEC} and \ref{sec:disop}. Here we resume  the definitions for the case of discrete vector fields(1-forms) defining the  elements of the  subspaces generating $ \Omega^1 $ in the Hodge decomposition.  This decomposition tell us that
\begin{equation}\label{eq:HD1}
\Omega^1=d\Omega^0\oplus\delta\Omega^2\oplus\Omega^1_H,
\end{equation}
where $ \Omega^0$ is the space of real function (0-forms) on the vertices $ V_N $ of the torus $g:V_N\to\mathbb{R} $, $ \Omega^2$ is the space of antisymmetric real function (2-forms) on $ \psi:F_N\to\mathbb{R} $  on the oriented faces $ F_N $ of the torus (see notation in subsection \ref{ss:not}), i.e. such that $ \psi(f)=-\psi(-f) $ where $ -f $ is the face $ f $ with opposite orientation, and $ \Omega^1_H $ is the space of harmonic discrete vector fields, that is the \e d-dimensional space (see proposition \ref{prop:Omega^1_H}) generated by the discrete vector fields
 \begin{equation}\label{eq:Hbasis}
 \vphi^{(i)}\left(x,x+e_{j}\right):=\delta_{i,j}\,, \qquad i,j=1,2,3\,.
 \end{equation}
We remind that with $ e_i $ we indicate the vectors of a right-handed canonical basis, i.e. $ (e_i)_j=\delta_ {ij} $, and the edges are along these directions. 
 Before to write explicitly the terms of  \eqref{eq:HD1} for a discrete vector field $ j_\eta(x,y) $ when $ \eta $ is fixed, we introduce  the group of translations
on the discrete torus, that will allow us  to write the generalized Hodge decomposition we will obtain. We denote by $\tau_x$ the translation by the element $x\in \mathbb T^d_N$. The translations act on configurations by $(\tau_x\eta)(z):=\eta(z-x)$ and on functions by $(\tau_x g)(\eta):=g(\tau_{-x}\eta)$. For notational convenience it is useful to define $\tau_{\mathfrak f}$ for an nonoriented face $\mathfrak f$ as $\tau_{\mathfrak f}:=\tau_x$, where $ x  $ is the vertex to which the   face $ \mf f $ is  associated as in notation \ref{ss:not}. Likewise for an nonoriented edge $\mathfrak e=\{x,x+e_{i}\}$  we define $\tau_{\mathfrak e}:=\tau_x$. Given an integer vector $ z=(z_1,\dots,z_n) $,  at fixed $ \eta $ for  a 2-form $ \psi(\eta,f) $ we define the following translation group on faces
\begin{equation}\label{eq:facetras}
\psi(\eta,\tau_z f)=\psi(\eta,f+z),
\end{equation}
where $ f+z $ is the face such that $ \mathtt{b}(f+z)=\mathtt{b}(f)+z $, where $ \mathtt{b}(f) $ is the centre of a face in definition \ref{def:centre}. When for any $ z\in V_N $  $ \psi(\eta,\cdot) $ satisfies  
\begin{equation}\label{eq.covpsi}
\psi(\eta,f)=\psi(\tau_z\eta,\tau_z f)\quad \text{ for all } (\eta,f)\in F_N\times\Sigma_N,
\end{equation}
we say that it is \e{translational covariant}. For pratical reasons we are using the name discrete vectors field and discrete form also when there is a dependence on $ \eta $ as in $ j_\eta(x,y) $ or $ \psi(\eta,f) $, in these cases it is understood that $ \eta $ is fixed.

We name  a discrete vector field $ j_\eta(x,y) $  of \emph{gradient type} if there exists a  function $h:V_N\times \Sigma_N\to\bb R$ such that
\begin{equation}\label{grgr}
j_{\eta}(x,y)=dh(\eta,x,y)=h(\eta,y)-h(\eta,x),\textrm{ for all }\big((x,y),\eta\big)\in E_N\times\Sigma_N,
\end{equation}
where for a fixed $ \eta $ the action of $ d $ on a vertex function  is defined by the most hand right side. We say that interacting particles systems is of \e{gradient type} if \eqref{grgr} turns out to be
\begin{equation}\label{eq:grpar}
j_{\eta}(x,y)=dh(\eta,x,y)=\tau_y h(\eta)-\tau_xh(\eta),\textrm{ for all }\big((x,y),\eta\big)\in E_N\times\Sigma_N.
\end{equation}\label{eq:grpar}
Similarly we name respectively a discrete vector $ j_\eta(x,y) $ of \emph{curl type} if there exist a  2-form $ \psi:F_N\times \Sigma_N $   such that
\begin{equation}\label{eq:curltype}
j_{\eta}(x,y)=\delta\psi(\eta,x,y)=\sum_{f:(x,y)\in \p f}\psi(\eta,f), \textrm{ for all }\big((x,y),\eta\big)\in E_N\times\Sigma_N.
\end{equation}
Where for a fixed $ \eta $ the action of $ \delta $ on 2-form is defined by the most right hand side. For an edge $ e=(x,y)\in E_N $, we define with $ f^{+}_i(e)/f^{-}_i(e)$ the unique positive  face $ f_i^+/f_i^-\in F^{i,+}_N $ such that $ e $ is on the border\footnote{ To say this in term of   definitions \ref{def:indor} and \ref{d:indorcub}  there is an inductively  oriented edge $ e(f^+_i)/e(f^-_i) $ of $ f^+_i/f^-_i $ such that $ e(f^+_i)/-e(f^-_i) =e $.}  of $f^+_i/-f^-_i$ and we write $ e\in\p f^+_i/-\p f^-_i $, see figure \ref{fig: due facce}. Consider $ e=(x,x+e_i) $, we say that an interacting particle system is  of \e{curl type} if  \eqref{eq:curltype} turns out to be
{\small \begin{equation}\label{eq:curlpar}
j_{\eta}(x,x+e_i)=\sum_{\us{j\neq i}{ j:\,e\in\pm\partial f^{\pm}_j(e),}} \left(\tau_{\mf f_j^+(e)}g_j(\eta)-\tau_{\mf f_j^-(e)}g_j(\eta) \right),  \textrm{ for all }((x,x+e_i),\eta)\in E_N\times\Sigma_N.
\end{equation}}
Finally we call both a discrete vector field $j_\eta(x,y) $  and an interacting  particle systems of \e{harmonic type} if $ j_{\eta}(x,y) $ is a linear combination of the harmonic vectors $ \vphi^{(i)} $  \eqref{eq:Hbasis}, i.e.
\begin{equation}\label{eq:hartype}
j_\eta(x,y)=\sum_{i=1}^{n}c_i\vphi^{(i)}, \textrm{ for all }\big((x,y),\eta\big)\in E_N\times\Sigma_N.
\end{equation}

Given a discrete vector field $j \in \Omega^1$ , here we don't specify the dependence on $ \eta $ but the same can be done once the configuration $ \eta $ is fixed, we write
\begin{equation}\label{eq:Hsplit}
j=j^d+j^\delta+j^H
\end{equation}
to denote the unique splitting following from \eqref{eq:HD1} in the three orthogonal components, namely $ j^d $ is of gradient type, $ j^\delta $ is of curl type and $ j^H $ of harmonic type. This decomposition can be computed as follows.  The coefficients $c_i$ of the harmonic part $ j^H $ as in \eqref{eq:hartype} are determined by
\begin{equation}\label{cco}
c_i=\frac{1}{N^n}\sum_{x\in \mathbb T^n_N}j\left(x,x+e_{i}\right)\,.
\end{equation}
We remember that the definition of divergence for a discrete vector field $ j  $ is   \begin{equation}\label{eq:divj}
\div j(x)=\us{y:(x,y)\in E_N}{\sum}j(x,y).
\end{equation} 
To determine the gradient component $j^d$ we need to determine a function $h$ for which $j^d(x,y)=dh(x,y)=h(y)-h(x)$. This is done just taking the divergence on both side of \eqref{eq:Hsplit} obtaining that $h$ solves the discrete Poisson equation $\div dh=\div j$ because  by construction \footnote{Considering $ \delta:\Omega^1\to\Omega^0 $ we write the divergence as $ \div j(x)=-\delta j(x) $ and then $  \div j^\delta=\div \delta\psi=-\delta\delta \psi=0$.} $ \div j^\delta=\div \delta\psi=0 $. The remaining component $j^\delta$ is computed just by difference $j^\delta=j-j^d-j^H$.

We remind that we indicate with upper indices $ ^* $ the elements of the dual complex $ \ast\mc{C} $ on the dual torus $ \mathbb{T}^{d,*}_N $ with mesh of size one and connecting the dual vertices $ x^* $, namely the centres of the cells $ \mf c\in \mc C_N $. Moreover, we remind that both an edge $ e $, a face $ f $ and a cell $ c $ have two orientations as in notation of subsection \ref{ss:not} and with $ -e,-f,-c $ we indicate the same geometrical object with opposite orientation.  How to construct in general context the dual elements  is explained in definition \ref{def:dualop} 
\begin{table}[htp]
\caption[]{Duality between primal complex and dual complex}
\label{tab:duale}
\begin{tabular}{|cc|cc|cc|}
\hline
\multicolumn{2}{|c}{1d}&\multicolumn{2}{|c|}{2d}&\multicolumn{2}{c|}{3d}\\
\hline
Primal cube & Dual Cell & Primal cube & Dual Cell&Primal cube & Dual Cell\\ 
\hline
vertex&edge&vertex&face&vertex&cell\\
edge&vertex&edge&edge&edge&face\\
&&face&vertex&face&edge\\
&&&&cell&vertex\\
\hline
\end{tabular}
\end{table}
Now we recap how  duality works in the present case. There is a natural duality $ \ast $ between the nonoriented elements(\e k-cubes) of $\mathcal C$ and the ones  of  $\ast \mathcal C$ which is listed in table \ref{tab:duale}. Also between oriented elements(\e k-cubes) of $ \mc C $ and $ \ast C $ there is a duality. In one dimension we associate  with a vertex $ x\in V_N $ the positive oriented edge $ e^* \in E_N^*$ of the dual  with $ x $ at its centre ($ {\tt b}(e^*)=x $) and consistently oriented with $e_1 $, i.e. along $ e_1 $, while the dual of an edge $ e\in E_N $  is the vertex $ x^*\in V_N^* $ given by the centre $ {\tt b } (e) $ of $ e $.
 In two dimensions we associate with a vertex $ x\in V_N $   the anticlockwise(positive) oriented face $ f^*\in F_N^* $ with the vertex $ x $ at its centre ($ {\tt b}(f^*)=x $), while  we associate with an oriented edge $ e\in E_N $ the oriented dual edge $ e^*\in E^*_N $  obtained simply rotating the original edge anticlockwise of $\frac \pi 2$ around its centre and finally the dual of a face $ f\in F_N $ is the vertex $ x^*\in V_N^* $ given by its centre $ {\tt b}(f) $. 
  In three dimensions we associate with a vertex $ x\in V_N $ the positive oriented cell $ c^*\in C^*_N $ with normal pointing outside and  the vertex $ x $ at its centre ($ {\tt b}(c^*)=x $). While  we associate with an oriented edge $ e\in E_N $  the oriented dual face $ f^*\in F^*_N $  crossed at its centre by the edge $ e $ and oriented in agreement with the edge according to the right hand rule, analogously we obtain the dual of an oriented face $ f\in F_N $ as the edge $ e^*\in E_N^* $ crossing the centre of  face $ f $ and oriented in agreement with the face according to the right hand rule. Finally, the dual of a cell $ c\in C_N $ is the vertex $ x^*\in V^*_N $ given by its centre ($ {\tt b}(c)=x^* $). 
The same duality $ * $ is defined also on the dual complex $ \ast \mc C $ because of the regularity of the mesh of $ \mc C $, in this case it is valid\footnote{Therefore the adjoint operator $ \delta $ can be computed via the formula in footnote $ \ref{f:delta}$.} \eqref{e:ast}. So applying in two dimensions twice this duality we obtain the original edge $ e  $ with reversed orientation, i.e. $ \ast\ast e=-e $ and applying it in three dimensions twice this duality we obtain the original  edge $ e $/face $ f $  with original orientation, i.e. $ \ast\ast e/f=e/f $.

\subsection{One dimensional case}
We consider the one dimensional torus with $ N $ sites $ \mathbb{T}_N:=\mathbb{Z}/N\mathbb{Z} $ that we represent as a ring such that we move anticlockwise going from $ x $ to $ x+1 $. The cubic complex is $ \mc{C}=\{V_N,E^+_N\} $ with $ E^+_N:=E_N^{1,+} $, see notation in subsection \ref{ss:not}. We have the following theorem.

\begin{theorem}\label{belteo}
Let $j_{\eta}$  be a translational covariant discrete vector field. Then there exists a function $h(\eta)$ and a translational invariant function $C(\eta)$ such that
\begin{equation}\label{imbr}
j_{\eta} (x,x+1)=\tau_{x+1}h(\eta)-\tau_{x}h(\eta)+C(\eta)\,.
\end{equation}
The function $C$ is uniquely identified and coincides with
\begin{equation}\label{eq:C(eta)}
C(\eta)=\frac1N\sum_{x\in V_N} j_{\eta}(x,x+1).
\end{equation}
The function $h$ is uniquely identified up to an arbitrary additive translational invariant function and coincides with 
\begin{equation}\label{eq:h(eta)}
h(\eta)=\sum_{x=1}^{N-1}\frac{x}{N}j_{\eta}(x,x+1).
\end{equation}
\end{theorem}

\begin{proof}
A short proof can be obtained simply inserting \eqref{eq:C(eta)} and \eqref{eq:h(eta)} into \eqref{imbr}. We give a proof that
illustrates the intuition behind. First of all
summing over $x\in V_N$ both sides of \eqref{imbr} we obtain that $ C $ is uniquely determined by \eqref{eq:C(eta)}. By \eqref{eq:covcond} we have that $ C $ is invariant by translations
$$
C(\tau_z\eta)=\frac 1N\sum_{x=1}^N j_{\tau_z\eta}(x,x+1)=\frac 1N\sum_{x=1}^N j_{\eta}(x-z,x-z+1)=C(\eta)\,.
$$
We define a function $U:V_N\times \Sigma_N\times V_N\to \mathbb R^+$ that associates with the triple $(x,\eta,y)$  the value $U_x(\eta,y)$ defined setting $U_x(\eta,x):=0$ and for $y\neq x$ by
\begin{equation}\label{sete}
U_x(\eta,y):=\sum_{(w,z)\in \gamma_{x\to y}}\left[j_\eta(w,z)-C(\eta)\right]=\sum_{z\in [x,y-1]}\left[j_\eta(z,z+1)-C(\eta)\right]\,,
\end{equation}
where $ \gamma_{x\to y} $ is a path (\e 1-chain) made of edges in $ E_N^+ $  from $ x $ to $ y-1 $  as in proposition \ref{prop:prop1form}, we rewrite this path summing on    the collection of sites $[x,y-1]$ encountered going form $x$ to $y-1$ moving anticlockwise.
Since
\begin{equation}\label{egrad}
\sum_{z=1}^{N}\left[j_{\eta}(z,z+1)-C(\eta)\right]=0
\end{equation}
the function $ U $ is well defined. Condition \eqref{egrad} is equivalent to the condition that $j_{\eta}-C(\eta)$ is a gradient vector field for any configuration $\eta$ as \eqref{grgr} and $U_x$ is one of the associated potentials fixed in such a way that it is zero at $x$, see proposition  \ref{prop:prop1form}. We have indeed
\begin{equation}\label{indeed}
U_x(\eta,y+1)-U_x(\eta,y)=j_{\eta} (y,y+1)-C(\eta)\,.
\end{equation}
Consider the averaged function
\begin{equation}\label{indeedav}
\overline U(\eta,y):=\frac 1N \sum_{x=1}^N U_x(\eta,y)\,,
\end{equation}
for which still we have the equality
\begin{equation}\label{indeedproprio}
\overline U(\eta,y+1)-\overline U(\eta,y)=j_{\eta}(y,y+1)-C(\eta).
\end{equation}
We will show  there exists  a function $h:\Sigma_N\to \mathbb R$ such that $\overline U(\eta,y)=\tau_y h(\eta)$ and \eqref{indeedproprio} becomes \eqref{imbr}. First of all we prove the invariance property
\begin{equation}\label{invariance}
\overline U(\tau_z \eta, y+z)=\overline U(\eta,y)\,.
\end{equation}
We have
\begin{equation*}
\begin{split}
\overline U(\tau_z \eta, y+z)=&\frac 1N\sum_{x=1}^N U_x (\tau_z\eta,y+z)\\
=&\frac 1N \sum_{x=1}^N\sum_{v\in [x,y+z-1]}  \left[j_{\tau_z\eta}(v,v+1)-C(\tau_z\eta)\right]\\
=&\frac 1N \sum_{x=1}^{N}\sum_{v\in [x,y+z-1]}\left[j_\eta(v-z,v-z+1)-C(\eta)\right]\\
=&\frac 1N \sum_{x=1}^{N}U_{x-z}(\eta,y)=\overline U(\eta,y)\,.
\end{split}
\end{equation*}
Between the second and the third we used the invariance \eqref{eq:covcond} and the translation invariance of $ C $.
If we define $h(\eta):=\overline U(\eta,0)$, by the invariance \eqref{invariance} we have
\begin{equation*}
\overline U(\eta,y)=\overline U(\tau_{-y}\eta,0)=
h(\tau_{-y}\eta)=\tau_yh(\eta)\,.
\end{equation*}
Formula \eqref{eq:h(eta)} can be easily deduced from \eqref{indeedav} and the definition of $h$ or directly checking the validity of \eqref{imbr} using \eqref{eq:C(eta)}.
To finish the proof suppose now that there are two functions $h$ and $h'$ satisfying \eqref{imbr}. Since $C$ is univocally determined we deduce $\tau_{x+1}(h-h')=\tau_x(h-h')$ that means that $h-h'$ is translational invariant. This finishes the proof of the theorem.
\end{proof}
The basic idea of the above theorem is the usual strategy to construct the potential of a gradient discrete vector field,  see  proposition \ref{prop:prop1form}, plus a subtle use of the translational covariance of the model.
It is interesting to observe that a one dimensional system of particles is of gradient type (with a possibly not local $h$)
if and only if $C(\eta)=0$. This corresponds to say that for any fixed configuration $\eta$ then $j_\eta(x,y)$
is a gradient vector field. This was already observed in \cite{BDeSGJ-LL07,Na98}. In next subsections, we will generalize the statement of theorem \ref{belteo}  to higher dimensions.

\subsection{One dimensional interacting particles systems:  examples}
Now we are applying   our one dimensional functional decomposition to some interacting particle systems to  find the function $ h(\cdot) $. So some example are produced. A fact to underline is  the possible presence of non local $ h(\cdot) $ even though the dynamics is local.
\begin{example}[\e{The gradient case}]
If the vector field is of gradient type, i.e. $j_\eta(x,x+1)=\tau_{x+1}h(\eta)-\tau_xh(\eta)$ then we have
$$
\sum_{x=1}^{N-1}\frac{x}{N}j_\eta(x,x+1)=h(\eta)-\frac 1N\sum_{x=1}^N\tau_xh(\eta)\,,
$$
and since the second term on the right hand side is translational invariant we reobtain in formula \eqref{eq:h(eta)} the original function $h$ up to a translational invariant function.
\end{example}

\begin{example}[\e{The 2-SEP}]
The model we are considering is the 2-SEP for which on each lattice site there can be at most 2 particles (\cite{KipLan99} section 2.4). The generalization to k-SEP can be easily done likewise. More precisely we have $\Sigma=\{0,1,2\}$.
The dynamics is defined by the rates
$$
c_{x,x\pm 1}(\eta)=\chi^+(\eta(x))\chi^-(\eta(x\pm 1))
$$
where $\chi^+(\alpha)=1$ if $\alpha >0$ and zero otherwise while $\chi^-(\alpha)=1$ if $\alpha<2$ and zero otherwise.
We denote also $D^\pm_\eta(x,x+1)$ local functions associated with the presence on the bond $(x,x+1)$ of what we call respectively a positive or negative discrepancy. More precisely $D^+_\eta(x,x+1)=1$ if $\eta(x)=2$ and $\eta(x+1)=1$ and zero otherwise. We have instead $D^-_\eta(x,x+1)=1$ if $\eta(x+1)=2$ and $\eta(x)=1$ and zero otherwise. We define also $D_\eta:=D^+_\eta-D^-_\eta$.

The instantaneous current across the edge $(x,x+1)$ associated with the configuration $\eta$ is
$$
j_\eta(x,x+1):=\chi^+(\eta(x))-\chi^+(\eta(x+1))+D_\eta(x,x+1)
$$
For this specific model formulas \eqref{eq:C(eta)} and \eqref{eq:h(eta)} become
$$
C(\eta)=\frac 1N\sum_{x\in  V_N}D_\eta(x,x+1)
$$
$$
h(\eta)=-\chi^+(\eta(0))+\sum_{x=1}^{N-1}\frac{x}{N}D_\eta(x,x+1)\,.
$$
\end{example}

\begin{example}[\e{ASEP}]
The asymmetric simple exclusion process is characterized by the rates $c_{x,x+1}(\eta)=p\eta(x)(1-\eta(x+1))$ and $c_{x,x-1}(\eta)=q\eta(x)(1-\eta(x-1))$. Let  $\mathfrak {C}(\eta)$ be the collection of clusters of particles in the configuration $\eta$ defined at the end of section \ref{sec:CTMC} in definition \ref{def:cluster}.
Given a cluster $c\in \mathfrak C$ we call $\partial^lc,\partial^rc\in\{0,1,\dots ,N-1\}$ the first element of the lattice on the left of the leftmost site of the cluster and the rightmost site of the cluster respectively. The decomposition \eqref{imbr} holds with
\begin{equation}\label{caldaia}
C(\eta)=\frac{(p-q)\left|\mathfrak C(\eta)\right|}{N}\,,
\end{equation}
where $\left|\mathfrak C(\eta)\right|$ denotes the number of clusters and
\begin{equation}\label{bagnetto}
h(\eta)=\frac 1N\sum_{c\in \mathfrak C(\eta)}\left[p\partial^rc-q\partial^lc\right]\,.
\end{equation}
\end{example}

\subsection{Two dimensional case}

We consider now the two dimensional torus   $ \mathbb{T}_N^2:=\mathbb{Z}^2/N\mathbb{Z}^2 $ with $ N^2 $ sites that we represent on the plane and  separated by edges of length one which are oriented consistently with the canonical basis $ e_1,e_2 $.  Namely we consider  the two dimensional torus with the cubic complex $ \mc{C}=\left\{V_N,E^{+}_N,E^{+}_N,F^{+}_N\right\}$ where $E^+_N:=E_N^{1,+}\cup E_N^{2,+}$ and $ F^+_N=F^{3,+}_N $, see notation in subsection \ref{ss:not}. The functional Hodge decomposition in the two dimensional case is as follows. 

\begin{theorem}\label{belteo2}
Let $j_\eta$  be a covariant discrete vector field. Then there exist four functions $h,g,C^{(1)}, C^{(2)}$ on configurations of particles such that for an edge of the type $e=(x,x\pm e_i)$ we have
\begin{equation}\label{hodgefun2}
j_\eta(e)=\big[\tau_{e^+}h(\eta)-\tau_{e^-}h(\eta)\big]+\big[\tau_{\mathfrak f^+(e)}g(\eta)-\tau_{\mathfrak f^-(e)}g(\eta)\big]\pm C^{(i)}(\eta)\,.
\end{equation}
The functions $C^{(i)}$ are translational invariant and uniquely identified. The functions $h$ and $g$ are uniquely identified up to additive arbitrary translational invariant functions.
\end{theorem}

\begin{proof}
Since for any fixed $\eta$ we have that $j_\eta(\cdot)$ is an element of $\Omega^1$ we can perform for any fixed $\eta$ the Hodge decomposition writing $j_\eta=j_\eta^d+j_\eta^\delta+j_\eta^H$. The harmonic component $j_\eta^H$ is related to the functions
$C^{(i)}$ that are equivalent to the constants $c_i$ computed by \eqref{cco}. Since we have now discrete vector fields depending on the configuration $\eta$ we will obtain not constants but functions of the configurations. In particular we have a formula completely analogous of \eqref{cco} that is
\begin{equation}\label{defci}
C^{(i)}(\eta):=\frac{1}{N^2}\sum_{x\in V_N}j_\eta(x,x+e_{i}).
\end{equation}
The functions in \eqref{defci} are clearly translational invariant (this is done similarly to what we did in on dimension) and uniquely determined  as can be seen taking the scalar product with $ \vphi^{(i)} $ on both sides of \eqref{hodgefun2}. We have
$$
j^H_\eta=C^{(1)}(\eta)\vphi^{(1)}+C^{(2)}(\eta)\vphi^{(2)}\,.
$$
As a second step we show that the translational covariance of $j_\eta$ is inherited also by the other components $j^\delta_\eta, j^d_\eta$.  Fix arbitrarily $z\in \mathbb T_N^2$ and consider two new discrete vector fields depending on configurations
\begin{equation*}\label{invdgH}
\left\{
\begin{array}{l}
\tilde j_\eta^d(x,y) := j_\eta^d(x+z,y+z) \\
\tilde j_\eta^\delta(x,y) := j_\eta^\delta(x+z,y+z). \\
\end{array}
\right.
\end{equation*}
For any fixed $\eta$ we have $\tilde j^d_{\eta}\in d\Omega^0$ and $\tilde j_{\eta}^\delta\in \delta\Omega^2$ since translations preserve these properties. We can write
\begin{eqnarray*}
& &\tilde j_\eta^d(x,y) +\tilde j_\eta^\delta(x,y)+j^H(x+z,y+z)=j_\eta(x+z,y+z)=\\
& &j_{\tau_{-z}\eta}(x,y)=j_{\tau_{-z}\eta}^d(x,y)+j_{\tau_{-z}\eta}^\delta(x,y)+j_\eta^H(x,y)\,.
\end{eqnarray*}
We used the translational covariance of $j_\eta$ and the translational invariance of $j^H_\eta$.
Since for any fixed $ \eta $ the Hodge decomposition is unique we obtain $j_{\tau_{z}\eta}^d(x+z,y+z)=j_{\eta}^d(x,y)$ and
$j_{\tau_{z}\eta}^\delta(x+z,y+z)=j_{\eta}^\delta(x,y)$ that is the translational covariance of $ j_\eta^d $ and $ j_\eta^\delta $.

Since for any fixed $\eta$ we have that $j_\eta^d$ is translational covariant and of gradient type, we show that there exists a function $ h $ such that  
\begin{equation}\label{eq:j^d=grad}
j^d_{\eta}(x,x+e_i)=\tau_{x+e_i}h(\eta)-\tau_x h(\eta) \qquad i=1,2\,.
\end{equation}
The proof is very similar to the one dimensional case. We define a function $U_x(\eta, y)$ where $x,y\in \mathbb T_N^2$. This is defined by
\begin{equation}\label{grd2}
U_x(\eta,y):=\sum_{i=0}^{k-1}j_\eta^d(z^{(i)},z^{(i+1)})
\end{equation}
where $z^{(0)}, \dots z^{(k)}$ is any path  of vertexes on the grid going from $z^{(0)}=x$ to $z^{(k)}=y$. Since $j^d_\eta$ is of gradient type  the value is independent of path, see proposition \ref{prop:prop1form}.
As in the one dimensional case we consider the averaged function
\begin{equation}\label{indeedav2}
\overline U(\eta,y):=\frac{1}{N^2} \sum_{x\in  V_N} U_x(\eta,y)\,,
\end{equation}
for which we have the equality
\begin{equation}\label{indeedproprio2}
\overline U(\eta,y)-\overline U(\eta,x)=j^d_{\eta}(x,y)\,,
\end{equation}
for any $(x,y)\in E_N$.
To finish the proof we need to show that there exists a function $h:\Sigma_N\to \mathbb R$ such that $\overline U(\eta,y)=\tau_y h(\eta)$ and \eqref{indeedproprio2} becomes \eqref{eq:j^d=grad}. First of all we prove the invariance property
\begin{equation}\label{invariance2}
\overline U(\tau_z \eta, y+z)=\overline U(\eta,y)\,.
\end{equation}
This follows by the symmetry
\begin{equation}\label{kgl}
U_x(\eta,y)=U_{x+z}(\tau_z\eta, y+z)
\end{equation}
that is obtained directly by the definition and the translational covariance of $j^d$.
Likewise in the one dimensional case we have then
\begin{eqnarray*}
& &\overline U(\tau_z \eta, y+z)=\frac{1}{N^2}\sum_{x\in  V_N}U_x (\tau_z\eta,y+z)\\
& &=\frac{1}{N^2} \sum_{x'\in V_N}U_{x'+z} (\tau_z\eta,y+z)=\frac{1}{N^2} \sum_{x'\in  V_N}U_{x'} (\eta,y)=\overline U(\eta,y)
\end{eqnarray*}
We used the change of variables $x'=x-z$ and the invariance \eqref{kgl}.
If we define $h(\eta):=\overline U(\eta,0)$, by the invariance \eqref{invariance2} we have
\begin{equation}\label{lem4}
\overline U(\eta,y)=\overline U(\tau_{-y}\eta,0)=h(\tau_{-y}\eta)=\tau_yh(\eta)\,. \end{equation}
The uniqueness of $ h $ up to an additive translational invariant function is obtained like in the one dimensional case.

Consider now $j_\eta^\delta\in \delta\Omega^2$ for any fixed $\eta$ and translational covariant.  This means that for any fixed $\eta$ there exists a 2-form $\psi(\eta,f)$ such that $ j^\delta_\eta $ is of curl type, i.e.
\begin{equation}\label{craxi}
j_\eta^\delta(e)=\delta\psi(\eta, e)=\sum_{f\,:\, e\in \p f}\psi(\eta,f)=\psi(\eta, f^+(e))-\psi(\eta,f^-(e))\,, \qquad e\in E_N,
\end{equation}
where we recall that $f^+(e)$ is the unique anticlockwise oriented face to which $e$ belongs and $f^-(e)$ is the unique anticlockwise oriented face to which $-e$ belongs, see figure \ref{fig: due facce}. A two form satisfying \eqref{craxi} is identified up to an arbitrary constant.  Every anticlockwise oriented face $ f $ is associated with the vertex $x\in V_N$ such that the centre of $ f$ is $\mathtt{b}(f)=x+(1/2,1/2) $. For any $ x\in V_N $, we introduce $\psi_{x}(\eta,\cdot)$ that is a 2-form satisfying \eqref{craxi} and
such that $\psi_{x}(\eta,f^x)=0$, where $f^x\in F_N^{+}$ is the anticlockwise oriented face associated with $x$. This two form can be defined as follows, for a face $ f\in F^+_N $
\begin{equation}\label{eq:psi_x}
\left\{\begin{array}{ll}
\psi_{x}(\eta,f):=\psi(\eta,f)-\psi(\eta,f^x), \\
\psi_{x}(\eta,-f):=-\psi(\eta,f)+\psi(\eta,f^x) .
\end{array}\right.
\end{equation}
We define
\begin{equation}\label{betto}
\overline\psi(\eta,f):=\frac{1}{N^2}\sum_{x\in V_N}\psi_{x}(\eta,f)
\end{equation}
that satisfies
\begin{equation}\label{craxi2}
j_\eta^\delta(e)=\overline \psi(\eta, f^+(e))-\overline\psi(\eta,f^-(e))\,, \qquad e\in E_N\,.
\end{equation}
For a face $ f\in F_N^{+} $ we call $ S_{f'\to f} $ any  path of faces (2-chain) from $ f' $ to $ f $ defined as $ S_{f'\to f}=\os{m}{\us{i=1}{\sum}}f_i $, where $ \{f_1=f',f_2\dots,f_{m-1},f_m=f\} $  is  any collection of anticlockwise faces  connecting $ f' $ to $ f $ such that for each $ i\in\{1,\dots,m-1\} $ the faces $ f_{i} $ and $ -f_{i+1} $ share an oriented edge $ (x_i,y_i) $. The two-form in \eqref{eq:psi_x} can be written as 
\begin{equation*}
\psi_{x}(\eta,f)=\psi(f)+\underset{{f'\in\, \mc{S}_{f^x\to f}}}{\sum}\big(   \psi(f')+\psi(-f')\big)-\psi(f^x)=
\os{n-1}{\us{i=1}{\sum}}\big(\psi(f_{i+1})-\psi(f_i)\big),
\end{equation*}
since  
$j_\eta^\delta(x_i,y_i)=\delta\psi(\eta,x_i,y_i)=\psi(\eta,f_{i+i})-\psi(\eta,f_i) $ we get 
\begin{equation}\label{eq:psi_x int}
\psi_{x}(\eta,f)=\os{n-1}{\us{i=1}{\sum}}j^\delta(x_i,y_i),
\end{equation}  
by the translational covariance of $j^\delta$ and \eqref{eq:psi_x int} we have  that $ \psi(\eta,\cdot) $ is translational covariant too
\begin{equation}\label{eq:psi_xinv}
\psi_{x}(\eta,f)=\psi_{x+z}(\tau_z\eta, f+z),
\end{equation}
therefore, exploiting the invariance \eqref{eq:psi_xinv},
with essentially the same computation as in the gradient case we can show that
\begin{equation}\label{cius}
\overline\psi(\tau_z \eta, \tau_z f)=\overline\psi(\eta,f)\,\qquad\forall z.
\end{equation}
We define the function $g(\eta):=\overline\psi(\eta,f^0)$ where $f^0$ is the anticlockwise oriented face associated with the vertex
$(0,0)\in V_N$. From \eqref{cius}   we obtain
\begin{equation}\label{gugarios}
\overline\psi(\eta,f)=\overline\psi(\tau_{-x}\eta,f^0)=g(\tau_{-x}\eta)=\tau_{\mf f}g(\eta),
\end{equation}
where the positive face $ f $ is associated with $x\in V_N$.
Using \eqref{gugarios} the relation \eqref{craxi2} becomes
$$
\left\{
\begin{array}{l}
j_\eta^\delta(x,x+e_1)=\tau_{x}g(\eta)-\tau_{x-e_2} g(\eta)\\
j_\eta^\delta(x,x+e_2)=\tau_{x-e_1}g(\eta)-\tau_x g(\eta)
\end{array}
\right., 
$$ that  is resumed by 
\begin{equation}\label{eq:j^delta=deltag}
j_\eta^\delta=\tau_{\mf f^+(e)}g(\eta)-\tau_{\mf f^-(e)}g(\eta).
\end{equation}

To conclude the proof suppose now that there are two couple of functions $(h,g)$ and $(h',g')$ satisfying \eqref{hodgefun2}. Then we deduce $\tau_{x+e_i}(h'-h)=\tau_x(h'-h) $ and $\tau_{x-e_i}(g'-g)=\tau_x(g'-g) $ from the uniqueness of our splitting, this means that both $(h'-h)(\cdot)$ and $(g'-g)(\cdot)$ are translational invariant. This completes the proof of the proposition.
\end{proof}

\begin{remark}
We underline that the proof is constructive, indeed the configurations functions $ h $ and $ g $ in \eqref{hodgefun2} can be explicitly computed using respectively \eqref{indeedav2} in the preferential point $ y=0 $ and \eqref{betto} in  the preferential face $ f^0 $.
\end{remark}

\begin{remark}\label{re:gradpart}

In a general dimension   the gradient part of the current $ j^g $ will have always the same form $ j^g_{\eta}(x,x+e_i)=\tau_{x+e_i} h(\eta)-\tau_x h(\eta)  $ for each $ i\in\{1,\dots,d\} $ and the gradient function $ h $ will have always the same construction because $ U_x(\eta,y) $ in \eqref{grd2} is just a linear integral along a path from $ x $ to $ y $. Hence in dimension $ d $ the path $ z^{(0)}, \dots, z^{(k)} $ will be free to move in $ d $ dimension and the averaged function equivalent  to \eqref{indeedav2} is $ \overline U(\eta,y) :=\frac{1}{N^d} \us{x\in V_N}{\sum} U_x(\eta,y)$ with $ V_N $ the set of vertexes of the discrete torus $ \bb T^d_N $, so the proof leading to $ h(\eta) $ doesn't change.
\end{remark}

\subsection{Two dimensional interacting particles systems:  examples}\label{ss:vortmodels}
In this subsection we are applying   our two dimensional functional decomposition to perturb some gradient dynamics in such a way to obtain a non-gradient dynamics with a non trivial and local decomposition \eqref{hodgefun2}.

\begin{example}[\e{A non gradient lattice gas with local decomposition}]
We construct a model of particles satisfying an exclusion rule, with jumps only trough nearest neighbours sites and having a non trivial decomposition of the instantaneous current \eqref{hodgefun2} with $C^{(i)}=0$ and $h$ and $g$ local functions. The functions $h$ and $g$ have to be chosen suitably in such a way that the instantaneous current is always zero inside clusters (definition \ref{def:cluster}) of particles and empty clusters and has to be always such that $j_\eta(x,y)\geq 0$ when $\eta(x)=1$ and $\eta(y)=0$. A possible choice is the following perturbation of the SEP. We fix $h(\eta)=-\eta(0)$ and $g(\eta)$ with $D(g)=\{0, e^{(1)},e^{(2)}, e^{(1)}+e^{(2)} \}$ (we denote by $0$ the vertex $(0,0)$) defined as follows. We have $g(\eta)=\alpha$ if $\eta(0)=\eta(e^{(1)}+e^{(2)})=1$ and
$\eta(e^{(1)})=\eta(e^{(2)})=0$. We have also  $g(\eta)=\beta$ if $\eta(0)=\eta(e^{(1)}+e^{(2)})=0$ and
$\eta(e^{(1)})=\eta(e^{(2)})=1$. The real numbers $\alpha,\beta$ are such that $|\alpha|+|\beta|<1$. For all the remaining configurations we have $g(\eta)=0$.  Since $\Sigma=\{0,1\}$ the rates of jump are uniquely determined by $c_{x,y}(\eta)=\left[j_\eta(x,y)\right]_+$.
\end{example}

\begin{example}[\e{A perturbed zero range dynamics}] 
We consider a perturbation of a zero range dynamics, see \eqref{eq:Zgen},  having a local non trivial decomposition. We say that the face
$\{0, e^{(1)}, e^{(2)}, e^{(1)}+e^{(2)}\}$ is full in the configuration
$\eta\in \mathbb N^{V_N}$ if $\min\{\eta(0), \eta(e^{(1)}), \eta(e^{(2)}), \eta(e^{(1)}+e^{(2)})\}>0$.
Consider two non negative functions $w^\pm$ that are identically zero when the face $\{0, e^{(1)}, e^{(2)}, e^{(1)}+e^{(2)}\}$
is not full. Given a positive function $\tilde h:\mathbb N\to \mathbb R^+$, we define the rates of jump as
\begin{equation}\label{chiu}
c_{e^-,e^+}(\eta)=\tilde h(\eta(e^-))+\tau_{\mathfrak f^+(e)}w^++\tau_{\mathfrak f^-(e)}w^-\,.
\end{equation}
This corresponds to a perturbation of a zero range dynamics such that one particle jumps from one site with $k$ particles with a rate $\tilde h(k)$. The perturbation increases the rates of jump if the jump is on the edge of a full face. The gain depends on the orientation and the effect of different faces is additive. For such a model the instantaneous current
has a local decomposition \eqref{hodgefun2} with $h(\eta)=-\tilde h(\eta(0))$ and $g(\eta)=w^+(\eta)-w^-(\eta)$.
\end{example}

\subsection{Three dimensional case}

Here we want to  consider  the three dimensional torus   $ \mathbb{T}_N^3:=\mathbb{Z}^3/N\mathbb{Z}^3 $ with $ N^3 $ sites that we represent  on a cubic lattice  and  separated by edges of length one which are oriented consistently with the canonical basis $e_1,e_2,e_3$.  The cellular complex is $ \mc{C}=\left\{V_N,E^{+}_N,,F^{+}_N,C_N^+\right\}$ as in notation in subsection \ref{ss:not}. We are going to discuss  the general strategy of the statement we would like to prove.

\begin{theorem}\label{belteo3}
Let $j_\eta$  be a covariant discrete vector field. Then there exist six functions $h,g_1,g_2,g_3,C^{(1)}, C^{(2)}, C^{(3)}$ on configurations of particles such that for an edge of the type $e=(x,x\pm e_i)$ we have
\begin{equation}\label{hodgefun3}
j_\eta(e)=\big[\tau_{e^+}h(\eta)-\tau_{e^-}h(\eta)\big]+\us{j\neq i}{\sum_{ j:\,e\in\pm\partial f^{\pm}_j(e),}} \left[\tau_{\mf f_j^+(e)}g_j(\eta)-\tau_{\mf f_j^-(e)}g_j(\eta) \right]\pm C^{(i)}(\eta)\,.
\end{equation}
The functions $C^{(i)}$ are translational invariant and uniquely identified. The functions $h$ and $g_i$ are uniquely identified up to additive arbitrary translational invariant functions.
\end{theorem}

The proof related to gradient part $ \big[\tau_{e^+}h(\eta)-\tau_{e^-}h(\eta)\big] $ is exactly like for the two dimensional case for the reasons explained in remark \ref{re:gradpart}. Again because of Hodge splitting $ j_\eta=j_\eta^d+j^\delta_\eta+j^H_\eta $ the harmonic part is given by $ j^H_\eta=C^{(1)}(\eta)\vphi^{(1)}+C^{(2)}(\eta)\vphi^{(2)}+C^{(3)}(\eta)\vphi^{(3)}$ where 
\begin{equation}\label{eq:C^3}
C^{(i)}(\eta)=\frac 1{N^3} \us{x\in V_N}{\sum}j_\eta(x,x+e_i).
\end{equation}

From $ \Omega^1=d^0\Omega^0\oplus\delta^2\Omega^2\oplus\Omega^1_H $  doing a dimensional counting we have that the dimension of $ \delta^2\Omega^2 $ is $ 2N^3-2 $.
So  we can deduce that the dimension of $ \Ker \delta^2 $ is $N^3+2  $ because of $ \Omega^2=\Ima\delta^2\oplus\Ker\delta^2 $. Since the kernel of $ \delta^2 $ is $ \delta^3\Omega^3\oplus\Omega^2_H $ and  the dimension of $ \delta^3\Omega^3 $ is $ N^3-1 $,   that one of $ \Omega^2_H $ is $ 3 $. One can show that $ \Omega^2_H $ is generated by the 2-forms $ \psi^{(i)}(f_j)=\delta_{ij} $  with $ i,j\in 1,2,3 $ and $ f_j\in F^{j,+}_N $. Consider now $j_\eta^\delta\in \delta\Omega^2$ for any fixed $\eta$ and translational covariant.  This means that for any fixed $\eta$ there exists a 2-form $\psi(\eta,f)$ such that $ j^\delta_\eta $ is of curl type, i.e.
\begin{equation}\label{eq:j=delta}
j_\eta^\delta(e)=\delta\psi(\eta,e)=\sum_{ i:\,e\in\pm\partial f^{\pm}_i(e)} \left(\psi(\eta,f^+_i(e))-\psi(\eta,f^-_i(e)\right),\quad e\in E_N 
\end{equation}
To prove the theorem \ref{belteo3} we should show there exist a 2-form $ \overline\psi(\eta,f) $ such that $ j^\delta_\eta(e)=\delta\overline{\psi}(\eta,e) $ and translational covariant in the sense \eqref{eq.covpsi}. With this we could define the configuration functions $ g_i(\eta):=\overline{\psi}(\eta,f^0_i) $ where $ f^0_i $ are the three positive faces associated with the vertex $ (0,0)$, see remark \ref{r:orob}. With the same argument of \eqref{gugarios} we would have that $ \overline{\psi}(\eta,f_i)=\tau_{\mf{f}_i}g_i(\eta) $ where $ f_i $ is a positive face in $ F^{i,+}_N $ associated to a vertex $ x\in V_N $ and consequently the statement \ref{belteo3}. So the difficult part of the proof is to construct a covariant 2-form $ \overline{\psi}(\eta,f) $ from $ \psi(\eta,f) $. Note that for a 2-form defined as $ \psi'(\eta):=\psi(\eta)+\delta\rho(\eta)+\psi^2_H(\eta) $ where $ \delta \rho(\eta)\in\delta\Omega^3 $ and $ \psi^2_H(\eta)\in\Omega^2_H $ we have that $ \delta\psi'(\eta)=\delta\psi(\eta) $. For each fixed configuration, the degrees of freedom in defining $ \psi'(\eta) $  are $ N^3+2 $, exploiting this structure   we should be able to find a 2-form $ \overline \psi(\eta) $ as desired. A constructive proof should give an explicit covariant solution $ \psi_x(\eta) $ canonically associated to a vertex $ x\in V_N $ and then $ \overline\psi(\eta) $   would be defined like in \eqref{betto}.
\bigskip

Our functional Hodge decomposition tell us that in dimensions higher than one an interacting particles systems can be seen as an overlapping of a gradient type systems \eqref{eq:grpar}, a curl type systems \eqref{eq:curltype} and a harmonic systems\eqref{eq:hartype}. The proof in three dimensions could provide a good idea how to generalize the proof to any dimension.

\chapter[Stationary Noneq. States]{Gibbsian Stationary Nonequilibrium States}\label{ch:SNS}

In this chapter we study the structure of stationary non equilibrium states for interacting particle systems from a microscopic viewpoint.
In particular we discuss two different discrete geometric constructions. We apply
both of them to determine non reversible transition rates corresponding to a fixed invariant measure.
In section \eqref{sec:Dfreefl} we discuss the first one, this uses the equivalence of this problem with the construction of divergence
free flows on the transition graph (defined in subsection \ref{ssec:Tgrap+Kol}). Since divergence free flows are characterized by cyclic decompositions we
can generate families of models from elementary cycles on the configurations space.  We consider only local cycles obtained by local  perturbations of configurations of particles but the case of non local cycles is also interesting. 
The second construction, discussed in section \ref{sec:sta&ort}, is based on the   functional discrete Hodge decomposition for
translational covariant discrete vector fields develop in section \ref{sec: HD}.  
With this second construction the stationary condition can be interpreted
as an orthogonality condition with respect to an harmonic discrete vector field and we use this decomposition to construct models having a fixed invariant measure.

We study the problem in dimension one and two in the case of interacting particles systems on stochastic lattice gases but the  two approaches can be extended to more general and/or abstract models (e.g. higher dimension and/or  more general states  space and transitions). For any model the stationary condition can be interpreted with the two geometric construction. We concentrate on finite range translational invariant measures but different situations with long range interactions and/or non translational invariant measure are also interesting. We focus only on local dynamics but analogous constructions can be done in the   non-local case too.
\medskip

First,  to present the material, we need to introduce some notation and basic tools in graph theory.
We refer to section \ref{sec:CTMC} for notation in stochastic lattice gases and to sections \ref{sec:disop} and \ref{sec: HD} for notation and definitions about discrete vector fields, edges and faces of a lattice and its dual lattice.

\section{Graphs}\label{s:graphs}
Let $(V, \mathcal E)$ be a finite un-oriented graph without loops. This means that $|V|<+\infty$ and a generic element
of $\mathcal E$, called un-oriented edge or simply an edge, is $\{x,y\}$ a subset of cardinality $2$ of $V$. To every un-oriented graph $(V, \mathcal E)$ we associate canonically an oriented graph $(V,E)$ such that the set of oriented edges $E$ contains all the ordered pairs $(x,y)$ such that $\{x,y\}\in\mathcal E$. Note that if $(x,y)\in E$ then also $(y,x)\in E$. If $e=(x,y)\in E$ we denote
$e^-:=x$ and $e^+:=y$ and we call $\mathfrak e:=\{x,y\}$ the corresponding un-oriented edge. A sequence $(z_0,z_1,\dots ,z_k)$ of elements of $V$ such that $(z_i,z_{i+1})\in E$, $i=0,\dots k-1$, is called an oriented  path, or simply a path. A path with distinct vertices except $z_0=z_k$ is called a \e{cycle}. If $C=(z_0,z_1,\dots ,z_k)$ is a cycle and there exists an $i$ such that $(x,y)=(z_i,z_{i+1})$ we write $(x,y)\in C$. Likewise if there exists an $i$ such that $x=z_i$ we write $x\in C$. Two cycles that contain the same collection of oriented edges and differ by the starting point will be identified, namely we introduced a relation and equivalence classes. We call $\mathcal C$ the collection of all the equivalence classes of finite length cycles and call $C$ a generic element.

A discrete vector field $\vphi$ on $(V,\mathcal E)$ is a map $\vphi:E\to \mathbb R$ such that $\vphi(x,y)=-\vphi(y,x)$.
A discrete vector field is of gradient type if there exists a function $g:V\to \mathbb R$ such that
$\vphi(x,y)=d g(x,y):= g(y)-g(x)$.
The divergence of a discrete vector field $\vphi$ at $x\in V$ is defined by
$$
\div\vphi(x):=\sum_{y\,:\, \{x,y\}\in \mathcal E}\vphi(x,y)\,.
$$
We call $\Omega^1$ the $|\mathcal E|$ dimensional vector space of discrete vector fields. We endow $\Omega^1$ with the scalar product
\begin{equation}\label{sc}
\langle \vphi,\psi\rangle:=\frac 12\sum_{(x,y)\in E}\vphi(x,y)\psi(x,y)\,, \qquad \vphi,\psi\in \Omega^1\,.
\end{equation}

A flow on an oriented graph is a map $Q:E\to \mathbb R^+$ that associates the amount of mass flowed $Q(x,y)$ to any edge $(x,y)\in E$. The divergence of a flow $Q$ at site $x\in V$ is defined as
\begin{equation}\label{eq:divfree}
\div Q(x):=\sum_{y:(x,y)\in E}Q(x,y)-\sum_{y:(y,x)\in E}Q(y,x)\,.
\end{equation}
Given a cycle $C$ we introduce an \e{elementary flow} associated with the cycle by
\begin{equation}\label{lavas}
Q_C(x,y):=\left\{
\begin{array}{ll}
1   & \textrm{if}\ (x,y)\in C\,,  \\
0 & \textrm{if} \ (x,y)\not\in C\,.
\end{array}
\right.
\end{equation}
A similar definition can be given also for a path. Since for each $x\in C$ the outgoing flux is equal to ingoing flux and it is equal to $1$ we have $\div Q_C=0$.

\medskip
Later we will need to use the concept of cluster, that we introduce here with next definition referring  for the graph $ (V_N,E_N) $ to the end of section \ref{sec:CTMC}.
\begin{definition}\label{def:cluster}
Given a configuration of particles $\eta\in\Sigma_N=\{0,1\}^{V_N}$, we call $\mathfrak C(\eta)$ the collection of clusters of particles
that is induced on $V_N$.
A \e{cluster} $c\in \mathfrak C(\eta)$ is a subgraph of $(V_N,\mathcal E_N)$. Two sites $x,y\in V_N$ belong to the same cluster $c$ if
$\eta(x)=\eta(y)=1$ and there exists an non-oriented path $(z_0,z_1, \dots ,z_k)$ such that $\eta(z_i)=1$ and
$(z_i,z_{i+1})\in \mathcal E_N$.
\end{definition}

\section[Divergence free flows]{Invariant measures and divergence free flows}\label{sec:Dfreefl}

In this section we repeat few notions about Markov chains already introduced in chapter \ref{ch:IPSmodels}, in particular the notions of invariant measure and transition graph. We do this for practical convenience and because we want these notions in a more abstract context as in section \ref{s:graphs} where we treated a general abstract graph. When the graph will be referred to interacting models of chapter \ref{ch:IPSmodels} a vertex $ x\in V $ is a configuration $ \eta\in\Sigma_N $ and an edge $ \{x,y\} $ is a subset $ \{\eta,\eta'\} $ of cardinality two of $ V=\Sigma_N $. Usually for interacting models with Markov rates $ c(\eta,\eta') $ we are interested in the transition graph, that is $ (\eta,\eta')\in E  $ if and only $ c(\eta,\eta')>0 $. For example for the exclusion process \eqref{eq:EPjumpgen} if $ (\eta,\eta')=(\eta,\eta^{xy}) $  the flow on this edge is $ Q(\eta,\eta^{xy}) $. 

\subsection{Markov chains: transition graph and stationarity} 
We are considering finite ($ |V|<\infty $) and irreducible continuous time Markov chains with transitions rates $c(x,y)$ for a jump from $ x\in V $ to $ y\in V $.  The \e{transition graph} of the Markov chain  is the graph $ (V,E) $
having as  edges $ E  $ all the pairs $ (x,y) $ such that
$ c(x,y ) > 0 $. There exists an unique invariant measure $\mu$ that is strictly positive and is characterized by the stationarity condition
\begin{equation}\label{stsu}
\sum_{y:(x,y)\in E}\mu(x)c(x,y)=\sum_{y:(y,x)\in E}\mu(y)c(y,x)\,.
\end{equation}
A Markov chain is reversible when it is satisfied detailed balance condition
$$\mu(x)c(x,y)=\mu(y)c(y,x)\,, \qquad \forall (x,y)\in E\,. $$
When the Markov chain is not reversible it is possible to define a time reversed Markov chain with transition rates defined by
$$
c^*(x,y):=\frac{\mu(y)c(y,x)}{\mu(x)}\,.
$$
This process is characterized by the following feature. In the stationary case it gives to a set of trajectories exactly the same probability that the original process gives to the set of time reversed trajectories. In particular the time reversed chain has the same invariant measure $\mu$ of the original process.

\subsection{Divergence free flow} 
The problem of determining the invariant measure of a given Markov chain and the problem of construct a divergence free flow on its transition graph on $ V $ are strictly related.
Suppose that we have an irreducible Markov chain with transition rates $c$ and invariant measure $\mu$. Then on the transition graph we can define the \e{flow}
\begin{equation}\label{defbas}
Q(x,y):=\mu(x)c(x,y)\,.
\end{equation}
The stationary condition \eqref{stsu} coincides with the divergence free condition for the flow $Q$, that is 
\begin{equation}\label{eq:divfree}
\div Q(x)=0 \text{ for all } x\in V.
\end{equation}
 Conversely suppose that we have a divergence free flow $Q$ and a fixed strictly positive target distribution $\mu$ . We obtain a Markov chain having invariant measure $\mu$ defining the rates as
\begin{equation}\label{qpi}
c(x,y):=\frac{Q(x,y)}{\mu(x)}\,.
\end{equation}
Indeed all the Markov chains having $\mu$ invariant are obtained in this way for a suitable divergence free flow.
Once again the proof follows inserting the rates \eqref{qpi} into \eqref{stsu} and using the divergence free condition of $Q$. The natural interpretation of the flow $Q$ is that it represents the typical flow observed in the stationary state of the chain.

In terms of flows, the connection between a Markov chain and its time reversed is even more transparent. Consider a Markov chain having transition rates $c$ and invariant measure $\mu$ and let $Q$ be defined by \eqref{defbas}. The time reversed flow $Q^*$ is defined simply reversing the original flow $Q$, i.e. we have $Q^*(x,y):=Q(y,x)$. It is clear that $Q^*$ is still a divergence free flow. The rates of the time reversed Markov chain are obtained by \eqref{qpi} but using the reversed flow and the same invariant measure $\mu$: $c^*(x.y)=\frac{Q^*(x,y)}{\mu(x)}$. In particular reversibility corresponds to the symmetry by inversion $Q(x,y)=Q(y,x)$.

This connection between the two problems is important since the geometric structure of divergence free flows is quite well understood and there are simple representations theorems \cite{BFG15,GV12,MacQ81}. We have indeed that on a finite oriented graph any divergence free flow can be written as a superposition of elementary flows associated to cycles
\begin{equation}\label{Qr}
Q=\sum_{C\in \mathcal C}\rho(C)Q_C\,,
\end{equation}
for suitable positive weights $\rho$. We call \eqref{Qr} \e{cyclic decomposition} of a flow $ Q $. This decomposition is in general not unique and under suitable assumptions on the flow (for example summability) the result can be extended also to infinite graphs \cite{BFG15}. We can summarise the concepts of this section with the following remark and proposition.
\begin{remark}
From our discussion we understand that the problem of finding dynamics with a given invariant measure is equivalent to the problem of finding a divergence free flow on the transition graph.
\end{remark}
\begin{proposition}\label{prop:decom cycl}
If $(V,E)$ is a finite graph, a flow $Q$ has a cyclic decomposition if and only if $ \div Q(x)=0$ for all $x\in V$.
\end{proposition}
\begin{proof}
For a proof of proposition \ref{prop:decom cycl} see theorem 2.8 in \cite{GV12}. From  its proof one can deduce that the decomposition is not unique.
\end{proof}

Using this simple construction it is possible to generate with \eqref{qpi} non trivial and interesting non-reversible Markov dynamics. An non-reversible dynamics is generated setting  a cyclic decomposition where there is at least one cycle of length greater than two (see example \ref{ex:symflow}) 
such that  $ Q(x,y)\neq Q(y,x)$ for some $(x,y)\in E$. Namely, when we set a family of cycle $ \mc C $, to avoid reversibility we must verify for at least one edge $(x,y)\in E$ such that  $(x,y)\in C_1,\dots,C_k$  and $(y,x)\in C'_1,\dots,C'_l$  it holds $Q(x,y)=\underset{i=1}{\overset{k}{\sum}}\rho(C_i) \neq \underset{i=1}{\overset{l}{\sum}}\rho(C'_i)=Q(y,x)$. 

To ends the section we do two very simple examples to show how the construction works.

\begin{example}\label{ex:symflow}
For a symmetric flow for which $Q(x,y)=Q(y,x)$ the decomposition \eqref{Qr} can be done using just cycles of length 2 of the form $C=(x,y,x)$.
\end{example}
\begin{example}
For practical purposes, we give a graphical example how to construct a cyclic decomposition for a graph with five vertices where it's defined a zero divergence flow, see fig  \ref{fig: grafo a div 0}. In figure \ref{fig:algoritmo dec cicl} we set a collection $ \mc C $ of cycles  for this graph. 
We  proceed in an iterative way: we individuate the edges $(x,y)$ where the flow takes the value $\underset{(x,y)\in E}{\rm{min}} Q(x,y)$ and associate the elementary cycle $C_1$ with $\rho(C_1)=\underset{(x,y)\in E}{\rm{min}} Q(x,y)$ to some of these. Then we considers a new graph $(V',E')$  erasing the edges, belonging to  $C_1$, where the flow took the minimum value at the previous step and subtracting $\rho(C_1)$ to the value of the flow on the other edges in $C_1$, later we repeat the procedure until it is possible to get the family of elementary cycles ${C_1,\dots,C_{k-1}=C_3}$.  When the initial graph is reduced to  just one cycle, selecting also  this last cycle $C_k=C_4$, with the proper weight, we complete the family $\mathcal{C}$ and the real numbers $\rho(C_i)$  give the weight function $ \rho(C) $ in \eqref{Qr} for  a cyclic decomposition.
\end{example}
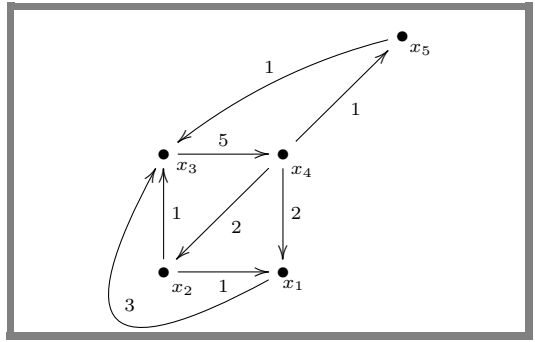
\begin{figure}[]
\begin{mdframed}[style=fig,userdefinedwidth=7cm]
\centering
\hspace{0.1cm}
\xymatrix@C=1.2cm@R=1.2cm{
 &  & \bullet\ar@/_/[lld]_{1}^(0.0){\hspace{0.3cm}x_{5}} \\
\bullet  \ar[r]_(.2){x_{3}}^{5} & \bullet \ar[ld]^{2} \ar[d] ^(0.15){x_{4}}^{2} \ar[ru]_{1}&  \\
\bullet \ar[r]_{1} \ar[u]_{1} _(-0.15){x_{2}} & \bullet\ar@/^2cm/[lu]^(0.0){x_{1}}_{3} &\\
} 
\end{mdframed}
\caption{
\small {In figure it is defined, on a graph with five vertices and eight edges, a zero-divergence flow. The values $Q(x_i,x_j)$  with $(i,j)\in\{1,2,3,4,5\}$ are the ones in the middle of the arrows.} }
\label{fig: grafo a div 0}-
\end{figure}
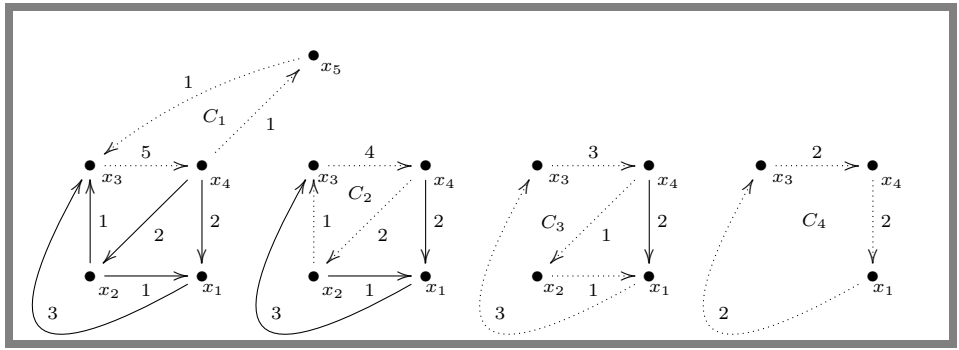
\begin{figure}[]
\begin{mdframed}[style=fig]
\centering
\hspace{0.1cm}
\xymatrix@C=1.1cm@R=1.1cm{
 &  & \bullet\ar@{.>}@/_/[lld]_{1}^(0.0){\hspace{0.3cm}x_{5}} &  &   &  &  &  &  \\
\bullet  \ar@{.>}[r]_(.2){x_{3}}^{5} & \bullet \ar[ld]^{2} \ar[d] ^(0.15){x_{4}}^{2} \ar@{.>}[ru]_{1}^(0.3){C_1}& \bullet  \ar@{.>}[r]_(.2){\hspace{0.1cm}x_{3}\hspace{0.3cm}}^{4} & \bullet \ar@{.>}[ld]^{2}_(0.4){\quad C_2} \ar[d] ^(0.15){x_{4}}^{2}& \bullet  \ar@{.>}[r]_(.2){x_{3}}^{3} & \bullet \ar@{.>}[ld]^{1} \ar[d] ^(0.15){x_{4}}^{2} & \bullet  \ar@{.>}[r]_(.2){x_{3}}^{2} & \bullet  \ar@{.>}[d] ^(0.15){x_{4}}^{2}_{C_4\hspace{0.5cm}} &  \\
\bullet \ar[r]_{1} \ar[u]_{1} _(-0.15){x_{2}} & \bullet\ar@/^2cm/[lu]^(0.0){x_{1}}_{3} & \bullet \ar[r]_{1} \ar@{.>}[u]_{1} _(-0.15){x_{2}} & \bullet\ar@/^2cm/[lu]^(0.0){x_{1}}_{3} & \bullet \ar@{.>}[r]_{1} _(0.15){x_{2}} & \bullet\ar@{.>}@/^2cm/[lu]^(0.0){x_{1}}_{3}_(0.9){\hspace{0.3cm}C_3} &  & \bullet\ar@{.>}@/^2cm/[lu]^(0.0){x_{1}}_{2} & \\
} 
\end{mdframed}
\caption{
\small {In figure the dotted cycles are the ones identified for an elementary cyclic decomposition. The sequence of graphs is the sequence obtained erasing the previous cycle $C_i$ and updating the values of $Q(\cdot,\cdot)$ along the rules of the algorithm.} }
\label{fig:algoritmo dec cicl}
\end{figure}

\section{Cyclic decomposition: applications }\label{sec:applications}

We use the method of the previous section to set  non-reversible continuous time Markov chains in particular we focus on the interacting particles systems presented in chapter \eqref{ch:IPSmodels}. Where we denoted with $ (V_N,E_N) $ the graph having as vertexes in $ V_N $ the lattice sites of $ \mathbb{T}^n_N $ and as edges $ E_N $ the set of oriented edges connecting them, see also \ref{sec:part models}. The non-oriented graph was $ (V_N,\mc E_N) $. At the end of the section we show a general structure to generate of non-reversible irreducible local dynamics in interacting models, which is suitable to generate also non local dynamics. 

\subsection{Metropolis algorithm} The first example is the classic Metropolis algorithm \cite{LB15} that corresponds to the generation
of reversible dynamics.
As we observed the reversible chains have a symmetric typical flow Q. This means that fixing the rates as in \eqref{qpi}
with a symmetric flow $Q$ we obtain a reversible chain with invariant measure $\mu$. For example choosing
$Q(x,y)=\max\{\mu(x),\mu(y)\}$, that is clearly symmetric, gives the classic rates
$$
r(x,y)=\max\left\{1,\frac{\mu(y)}{\mu(x)}\right\}\,.
$$

\subsection{Random walks} Before consider the case of several particles let us shortly recall a construction in \cite{GV12} for one single particle. We consider the bidimensional case although the arguments presented below can be generalized to any dimension. Consider
for example  two  nonnegative (i.e. with nonnegative coordinates) vector fields $q^+(x)$ and $q^-(x)$ on the continuous torus $[0,N]\times [0,N]$ with periodic boundary conditions and such that
\begin{equation}
j(x):=q^+(x)-q^-(x)\,,
\label{diffj}
\end{equation}
is a  divergence free vector field.

We define a flow $Q$ on the lattice fixing the values  $Q\left(y,y  \pm e^{(i)}\right)$  as the flux of $q^\pm_i$ across  the edges dual to $\left\{y,y\pm e^{(i)}\right\}$, more precisely
\begin{align}
Q\left(y,y+e^{(1)}\right):= &\int_{0}^{1} q^+_1\left(y+ \frac{e^{(1)}-e^{(2)}}{2} + \alpha e^{(2)}
\right)d\alpha\,,\\
Q(y,y+e^{(2)}):= &\int_{0}^{1}   q^+_2\left(y+ \frac{-e^{(1)}+e^{(2)}}{2} + \alpha e^{(1)}  \right)d\alpha\,,\\
Q(y,y-e^{(1)}):=& \int_{0}^{1}    q^-_1\left(y- \frac{e^{(1)}+e^{(2)}}{2} + \alpha e^{(2)}
\right)d\alpha\,,\\
Q(y,y-e^{(2)}):=& \int_{0}^{1} q^-_2\left(y  - \frac{e^{(1)}+e^{(2)}}{2} + \alpha e^{(1)}
\right)d\alpha\,.
\end{align}
Given $y\in \mathbb T_N^2$ we have that $\div Q (y)$ coincides with the flow (from inside to outside) of $j$ through the boundary of the box $B_y:=\{z \in \bb Z^2 \,:\, |y-z|_\infty \leq 1/2\}$
$$
\div Q(y)=\int_{\partial B_y}j\cdot \hat n\, d\Sigma=0\,.
$$
The last equality above follows by $\div j=0$ and the  Gauss-Green Theorem.

If we define the rate of transitions of a random walk by \eqref{qpi} we obtain a random walk with invariant measure $\mu$.

\subsection{Two dimensional Exclusion process }\label{basile}
We consider a system of particles satisfying an exclusion rule, i.e. the configuration on one single site is $\Sigma=\{0,1\}$, and evolving on $V_N=\mathbb T^2_N$. The construction can be naturally generalized to any planar graph. Generalizations to higher dimensions are also natural but require more notation.

Let $\mu$ be a Gibbsian probability measure on the configuration space $\Sigma^{V_N}$
given by
\begin{equation}\label{gibbs}
\mu(\eta)=\frac 1Z e^{-H(\eta)}\,,
\end{equation}
where $H$ is the Hamiltonian.
To explain notation we consider the case of an Hamiltonian with one body  and two body interactions
\begin{equation}\label{H}
H(\eta)=- \sum_{\{x,y\}\in \mathcal E_N}J_{\{x,y\}}\eta(x)\eta(y)-\sum_{x\in V_N}\lambda_x\eta(x)\,.
\end{equation}
We consider such an Hamiltonian for simplicity but more general interactions can be handled similarly. Indeed the specific form of the Hamiltonian does not play any role in the following. It is only important that it has interactions of bounded range.
In \eqref{H} the parameters $\left(J_{\{x,y\}}\right)_{\{x,y\}\in \mathcal E_N}$ and $\left(\lambda_x\right)_{x\in V_N}$ are arbitrary real numbers describing respectively the interactions associate to the bonds and the chemical potentials of the sites. In the following we will always restrict to translational invariant Hamiltonian.
Given $W \subseteq \mathbb T^2_N$ we define the energy
restricted to $W$ as
\begin{equation}\label{en-ris}
H_{W}(\eta):=-\sum_{\{x,y\}\cap W \neq \emptyset}J_{\{x,y\}}\eta(x)\eta(y)-\sum_{x\in W}\lambda_x\eta(x)\,.
\end{equation}
We define also
\begin{equation}\label{en-ris*}
H_{W}^*(\eta):=-\sum_{\{x,y\}\subset W }J_{\{x,y\}}\eta(x)\eta(y)-\sum_{x\in W}\lambda_x\eta(x)\,.
\end{equation}
For any $W$ we have $H=H_W+H^*_{W^c}$.
Given an oriented edge $e\in  E_N$ of the lattice there is only one anticlockwise oriented such that 
$e\in \p f^+(e)$. There is also an unique anticlockwise face, that we call $f^-(e)$, such that $e\in -\p f^-(e)$
(see figure \ref{fig: due facce}).

\begin{figure}[]
\begin{mdframed}[style=fig,userdefinedwidth=8.7cm]
\vspace{0.5cm}
\centering
\entrymodifiers={+<0.5ex>[o][F*:black]}
\xymatrix@C=1.5cm@R=1.5cm{
 {}\ar@{.}[rrrr]\ar@{.}[dddd] & \ar@{.}[dddd] & {}\ar@{.}[dddd] & {}\ar@{.}[dddd] & {}\ar@{.}[dddd]\\
 {}\ar@{.}[rrrr] & {} & {} & {} & {}\\
 {}\ar@{.}[rrrr]_(.53){\textit{\normalsize y}} & {} & {} & {} & {}\\
 {}\ar@{.}[rrrr]_(.53){\textit{\normalsize x}}^(0.38){\textit{\tiny{$f^{+}(e)$}}} ^(0.63){\textit{\tiny{$f^{-}(e)$}}} & {} \ar@{}[ur]|{\textit{\Huge{$\circlearrowleft$}}} & {} \ar@{->}[u]|{\textit{}} & {} \ar@{}[ul]|{\textit{\Huge{$\circlearrowleft$}}} & {}\\
  {} \ar@{->}[u]^(.5){\textit{\normalsize $e^{(2)}$}}  \ar@{->}[r]_(.5){\textit{\normalsize $e^{(1)}$}} \ar@{.}_(-0.05){\text{\normalsize }}[rrrr] & {} & {} & {} & {}
}
\end{mdframed}
\caption{\small{ The discrete two dimensional torus of side 5; the opposite sides of the square are identified. For the oriented edge $(x,y)=e$ we draw the two anticlockwise oriented faces $f^{-}(e)$ and $f^{+}(e)$. }}
\label{fig: due facce}
\end{figure}
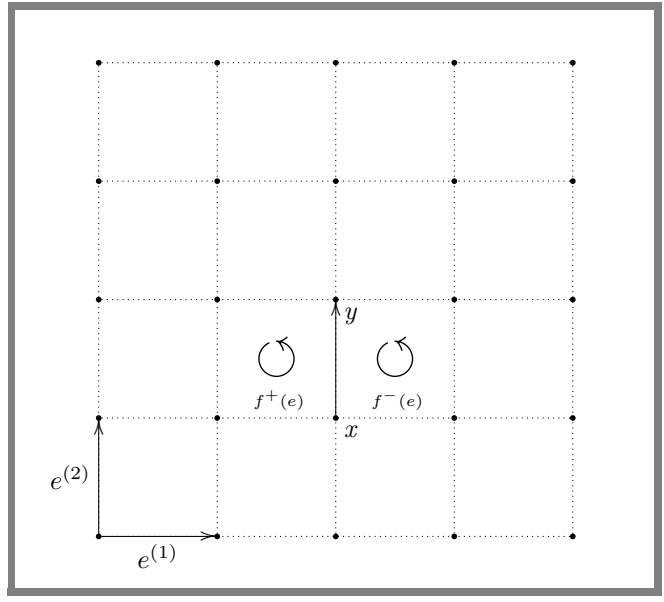
We introduce before a simplified version of the class of models we will consider.
Let $w^+,w^-$ two positive numbers.
We consider the following transition rates for the jump of one particle from $e^-$ to $e^+$
\begin{equation}\label{irates}
c_{e^-,e^+}(\eta):=\eta(e^-)(1-\eta(e^+))\left(w^+e^{H_{\mathfrak f^+(e)}(\eta)}+w^-e^{H_{\mathfrak f^-(e)}(\eta)}\right)\,.
\end{equation}
When we write $H_{\mathfrak f}$ for a suitable face $\mathfrak f$ we consider the face as a set of vertices.
We claim that the generator \eqref{eq:EPjumpgen} with the rates of jump \eqref{irates}  has \eqref{gibbs}
as invariant measure, moreover if $w^-\neq w^+$ then the dynamics is not reversible. Since the Hamiltonian \eqref{H} has only finite range interactions the dynamics induced by \eqref{irates} is local.

First we give a direct proof of this claim and then we show how a generalized version of the rates \eqref{irates} is naturally constructed by divergence free flows on the configuration space.
We need to check the validity of the stationary equations
\begin{equation}\label{staz}
\sum_{e\in E_N}\Big[\mu(\eta)c_{e^-,e^+}(\eta)-\eta(e^-)(1-\eta(e^+))\mu\left(\eta^{e^-,e^+}\right)c_{e^+,e^-}\left(\eta^{e^-,e^+}\right)\Big]=0\,,\qquad \forall \eta\,.
\end{equation}
Using \eqref{gibbs} and \eqref{irates} we will obtain
in \eqref{staz} a sum of several terms having as a factor
$e^{-H^*_{\mathfrak f^c}(\xi)}$ for some face $\mathfrak f$ and some configuration $\xi$.
Observe that if $\eta$ and $\xi$ are obtained one from the other just changing
the occupation numbers on sites belonging to  $\mathfrak f$ then we have $H^*_{\mathfrak f^c}(\eta)=H^*_{\mathfrak f^c}(\xi).$ This means that all the factors can be written as $e^{-H^*_{\mathfrak f^c}(\eta)}$ for different $\mathfrak f$.
We will use the relationship $f^+(e)=f^-(-e)$ (with $-e$ we denote the edge oriented oppositely with respect to $e$).

Let $\mathcal F_N$ be the collection of  un-oriented faces. We group together all the terms in \eqref{staz}
that have the energetic factor equal to $e^{-H^*_{\mathfrak f^c}}$  for a given $\mf f\in \mathcal F_N$.
The sum of all these terms is equal to
\begin{equation}\label{arpa}
\frac{e^{-H^*_{\mathfrak f^c}(\eta)}}{Z}\left\{\sum_{e \in \p f}\Big[\eta(e^-)(1-\eta(e^+))(w^+-w^-)+\eta(e^+)(1-\eta(e^-))(w^--w^+)\Big]\right\}\,.
\end{equation}
In \eqref{arpa} $f$ is the unique anticlockwise oriented face corresponding to $\mathfrak f$.
The sum appearing in
\eqref{arpa} is zero  since coincides with the telescopic sum
\begin{equation}\label{sz}
(w^+-w^-)\sum_{e\in \p f}\left[\eta(e^-)-\eta(e^+)\right]=0\,.
\end{equation}
Considering all the faces in $\mathcal F_N$ we can write \eqref{staz} as
\begin{equation}\label{fdz}
\sum_{\mathfrak f\in \mathcal F_N}\left\{\frac{e^{-H_{\mathfrak f^c}^*(\eta)}}{Z}(w^+-w^-)\sum_{e\in \p f}\left[\eta(e^-)-\eta(e^+)\right]\right\}=0\,.
\end{equation}
By \eqref{sz} we have that \eqref{fdz} is  satisfied.
The reversibility condition becomes
\begin{equation}
(w^+-w^-)e^{-H_{\mathfrak f^{-}(e)^c}^*(\eta)}=(w^+-w^-)e^{-H_{\mathfrak f^{+}(e)^c}^*(\eta)},
\end{equation}
that is not satisfied when $w^+\neq w^-$ apart special cases.

\smallskip

We show now how it is possible to conceive a generalized version of \eqref{irates} constructing a divergence free flow like in \eqref{Qr} on the configuration space
and using then \eqref{qpi}. We use cycles for which the configuration of particles is frozen outside a given face $\mathfrak f$.  The cycles that we use are shown in figure \ref{fig: EP 2d} and correspond to letting the particles rotate around a fixed face according to specific rules and following the two orientations. Let us consider two non-negative functions $w^\pm$ such that
$D(w^\pm)\cap \{0,e_{1},e_{2}, e_{1}+e_{2}\}=\emptyset$. If $w^\pm$ are local then the corresponding rates will be also local.

\begin{remark}\label{zeb}
More generally we can assume that $w^\pm$ depends on the configuration of particles restricted to  $\{0,e_{1},e_{2}, e_{1}+e_{2}\}$ only through the number of particles
$\eta(0)+\eta(e_{1})+\eta(e_{2})+\eta(e_{1}+e_{2})$.
\end{remark}
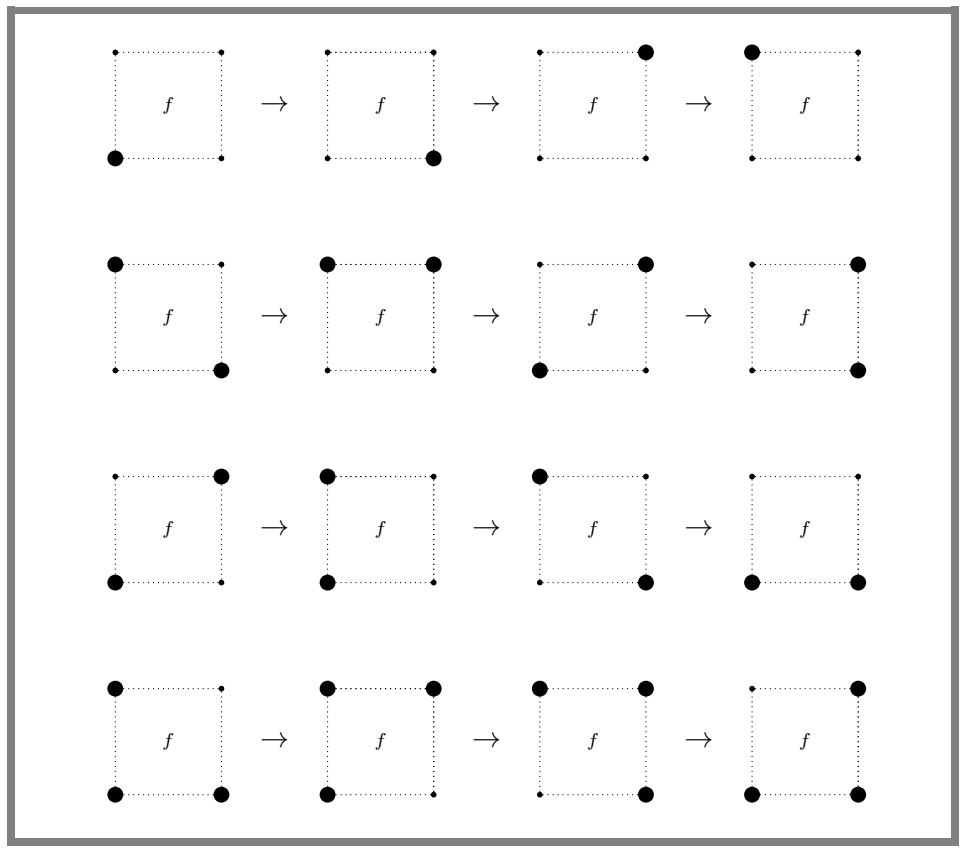
\begin{figure}[]
\begin{mdframed}[style=fig]

\[
\hspace{0.1cm}
\entrymodifiers={+<0.5ex>[o][F*:black]}
\xymatrix@C=1.2cm@R=1.2cm{
  \ar@{.}[r] \ar@{}[dr]|{f}\ar@{}[rrrd]|{\textit{\large{$\rightarrow$}}} &   \ar@{.}[d] &   \ar@{.}[d] \ar@{.}[r]\ar@{}[rrrd]|{\textit{\large{$\rightarrow$}}} \ar@{}[dr]|{f} &  \ar@{.}[l]  &  \ar@{.}[r] \ar@{}[rd]|{f}\ar@{}[rrrd]|{\textit{\large{$\rightarrow$}}} &  *+[o][F*:black]{}\ar@{.}[d] &  *+[o][F*:black]{}\ar@{.}[r] \ar@{.}[d] &  \ar@{.}[d] \ar@{}[ld]|{f}  \\
  *+[o][F*:black]{} \ar@{.}[u]  &  \ar@{.}[l]   &   \ar@{.}[r]  & *+[o][F*:black]{}  \ar@{.}[u] \ar@{.} [u] &  \ar@{.}[u] \ar@{.}[r]  &  &  \ar@{.}[r] &   \ar@{.}[u]\\
  *+[o][F*:black]{}\ar@{.}[r] \ar@{}[dr]|{f} \ar@{}[rrrd]|{\textit{\large{$\rightarrow$}}}&   \ar@{.}[d] &   *+[o][F*:black]{}\ar@{.}[d] \ar@{.}[r] \ar@{}[dr]|{f}\ar@{}[rrrd]|{\textit{\large{$\rightarrow$}}} &  *+[o][F*:black]{}\ar@{.}[l]  &  \ar@{.}[r] \ar@{}[rd]|{f}\ar@{}[rrrd]|{\textit{\large{$\rightarrow$}}} &  *+[o][F*:black]{}\ar@{.}[d] &  \ar@{.}[r] \ar@{.}[d] & *+[o][F*:black]{} \ar@{.}[d] \ar@{}[ld]|{f}  \\
   \ar@{.}[u]  &  *+[o][F*:black]{}\ar@{.}[l] &   \ar@{.}[r]  &  \ar@{.}[u] \ar@{.} [u] &  *+[o][F*:black]{}\ar@{.}[u] \ar@{.}[r]  &  &  \ar@{.}[r] &  *+[o][F*:black]{} \ar@{.}[u]\\
  \ar@{.}[r] \ar@{}[dr]|{f}\ar@{}[rrrd]|{\textit{\large{$\rightarrow$}}} &  *+[o][F*:black]{} \ar@{.}[d] &   *+[o][F*:black]{}\ar@{.}[d] \ar@{.}[r] \ar@{}[dr]|{f}\ar@{}[rrrd]|{\textit{\large{$\rightarrow$}}} &  \ar@{.}[l]  &  *+[o][F*:black]{}\ar@{.}[r] \ar@{}[rd]|{f}\ar@{}[rrrd]|{\textit{\large{$\rightarrow$}}} &  \ar@{.}[d] &  \ar@{.}[r] \ar@{.}[d] &  \ar@{.}[d] \ar@{}[ld]|{f}  \\
   *+[o][F*:black]{}\ar@{.}[u]  &  \ar@{.}[l] & *+[o][F*:black]{}  \ar@{.}[r]  &   \ar@{.}[u] \ar@{.} [u] &  \ar@{.}[u] \ar@{.}[r]  & *+[o][F*:black]{} &  *+[o][F*:black]{}\ar@{.}[r] &   *+[o][F*:black]{}\ar@{.}[u]\\
  *+[o][F*:black]{}\ar@{.}[r] \ar@{}[dr]|{f} \ar@{}[rrrd]|{\textit{\large{$\rightarrow$}}}&   \ar@{.}[d] &  *+[o][F*:black]{} \ar@{.}[d] \ar@{.}[r] \ar@{}[dr]|{f}\ar@{}[rrrd]|{\textit{\large{$\rightarrow$}}} & *+[o][F*:black]{} \ar@{.}[l]  &  *+[o][F*:black]{}\ar@{.}[r] \ar@{}[rd]|{f}\ar@{}[rrrd]|{\textit{\large{$\rightarrow$}}} & *+[o][F*:black]{}  *+[o][F*:black]{}\ar@{.}[d] &  \ar@{.}[r] \ar@{.}[d] & *+[o][F*:black]{} \ar@{.}[d] \ar@{}[ld]|{f}  \\
  *+[o][F*:black]{} \ar@{.}[u]  &  *+[o][F*:black]{}\ar@{.}[l] &  *+[o][F*:black]{} \ar@{.}[r]  &   \ar@{.}[u] \ar@{.} [u] &  \ar@{.}[u] \ar@{.}[r]  & *+[o][F*:black]{} &  *+[o][F*:black]{}\ar@{.}[r] &   *+[o][F*:black]{}\ar@{.}[u]\\
}
\vspace{0.3cm}
\]

\end{mdframed}
\caption{\small{Anticlockwise cycles in $\mathcal C_f$. Each line corresponds to a cycle and the evolution is from left to right. The configuration of particles is frozen outside the face. In the first cycle there is 1 particle rotating in the face, in the second and third 2 and in the last one 3. The clockwise cycles are obtained reading from right to left the figure. }}
\label{fig: EP 2d}
\end{figure}
We call $\mathcal C_f$ the collection of cycles in the configuration space associated to the oriented face $f$. In figure \ref{fig: EP 2d} we draw the structure of these cycles associated to an $f\in F_N^+$ (we call $F_N^+$ the collection of anticlockwise oriented faces). For different configurations of the particles outside $\mathfrak f$ we obtain different elements of $\mathcal C_f$.
The weights in \eqref{Qr} are fixed by
\begin{equation}\label{weights cycles kawa 2-d}
\rho(C):=\left\{
\begin{array}{ll}
e^{-H^*_{\mathfrak f^c}}\tau_{\mathfrak f}w^+ & \textrm{if}\  C\in \mathcal C_f\,, f\in F_N^+\,,\\
e^{-H^*_{\mathfrak f^c}}\tau_{\mathfrak f}w^- & \textrm{if}\  C\in \mathcal C_f\,, f\in F_N^-\,.
\end{array}
\right.
\end{equation}
A cycle $C\in \mathcal C_f$ is individuated by the face $f$, the type of rotation around the face among the possible ones
in figure \ref{fig: EP 2d} and the configuration $\eta_{\mathfrak f^c}$ outside the face.
Definition \eqref{weights cycles kawa 2-d} is then well posed since all the functions appearing depend only on $\eta_{\mathfrak f^c}$
and are therefore constant on the cycle $C$.

The value of $Q\left(\eta,\eta^{x,y}\right)$ associated to a jump of one particle from $x$ to $y$ ($e=(x,y)$ is an oriented
edge of the lattice) in the configuration $\eta$ is determined as follows. This value is zero unless the site $x$ is occupied and the site $y$ is empty. This gives a factor $\eta(x)(1-\eta(y))$. If this constraint is satisfied then there is a contribution corresponding to $e^{-H^*_{\mathfrak f^+(e)^c}(\eta)}\tau_{\mathfrak f^+(e)}w^+(\eta)$ from a cycle with anticlockwise rotations and a contribution corresponding to
$e^{-H^*_{\mathfrak f^-(e)^c}(\eta)}\tau_{\mathfrak f^-(e)}w^-(\eta)$ from a cycle with clockwise rotations. Applying formula \eqref{qpi}
we obtain the following generalized version of the rates \eqref{irates}
\begin{equation}\label{iratesgen}
c_{e^-,e^+}(\eta)=\eta(e^-)(1-\eta(e^+))
\left(e^{H_{\mathfrak f^+(e)}(\eta)}\tau_{\mathfrak f^+(e)}w^+(\eta)+e^{H_{\mathfrak f^-(e)}(\eta)}\tau_{\mathfrak f^-(e)}w^-(\eta)\right)\,.
\end{equation}
The exponential factors in \eqref{weights cycles kawa 2-d} have been chosen in such a way that applying \eqref{qpi}
the non locality of the measure $\mu$ is erased and we obtain local rates.
To these rates it is always possible to add some reversible rates coming from cycles of length 2 of the
form $C=(\eta,\eta^{x,y},\eta)$. Since this happens also in one dimension we discuss this issue in the next section.

If $\mu$ is a Bernoulli measure then $e^{H_{\mathfrak f^\pm(e)}}$ depends only on the number of particles in the face $\mathfrak f^\pm(e)$. This dependence can be compensated by $w^\pm$ by Remark \ref{zeb}.

\subsection{One dimensional Exclusion process}\label{ss:odex}
We consider now a conservative dynamics on a ring with $N$ sites. It will result an exclusion process as \eqref{eq:EPjumpgen}
We search for non reversible local rates having (\ref{gibbs}) as  invariant  measure.

The class of cycles that we consider contains all the cycles in which one single particle jumps across an edge and then come back. The length of these cycles is 2 and any superposition of them is generating a symmetric flow. The family of such cycles, when one single particle jumps across the edge $\mathfrak e$, is denoted by $\mathcal C_{\mathfrak e}^r$.

The most simple way to  introduce irreversibility  is to consider  cycles associated to two particles evolving.
The structure of the cycles associated to the movement of two particles that we consider is
illustrated in figure \ref{fig: due v- due p Kaw1d}. Again the configuration outside the window drawn is frozen. Moving from the top to the bottom of the figure we have a cycle in the collection $\mathcal C^+_{\mathfrak e}$. Moving instead from the bottom to the top we have a cycle in the collection $\mathcal C_{\mathfrak e}^-$. The lower index $\mathfrak e$ denotes the edge around which the evolution takes place while instead the upper index denotes the direction of movement. Observe that on cycles $\mathcal C_{\mathfrak e}^\pm$ particles are not jumping across the edge $\mathfrak e$ but
instead across $\tau_{\pm 1}\mathfrak e$.
\begin{figure}[]
\begin{mdframed}[style=fig,userdefinedwidth=8cm]
\[
\entrymodifiers={+<0.5ex>[o][F*:black]}
\xymatrix@C=1.5cm@R=1.5cm{
*+[o][F*:black]{}
\ar@{.}[rrr]^{\textit{\normalsize $\mathfrak e$}} &\ar@{}[dr]|{\textit{\large{$\downarrow$}}}&*+[o][F*:black]{}&\\
\ar@{.}[rrr]^{\textit{\normalsize $\mathfrak e$}}&*+[o][F*:black]{}\ar@{}[dr]|{\textit{\large{$\downarrow$}}}&*+[o][F*:black]{}&\\
\ar@{.}[rrr]^{\textit{\normalsize $\mathfrak e$}}&*+[o][F*:black]{}\ar@{}[dr]|{\textit{\large{$\downarrow$}}}&&*+[o][F*:black]{}\\
*+[o][F*:black]{}\ar@{.}[rrr]^{\textit{\normalsize $\mathfrak e$}}&&&*+[o][F*:black]{}\\
} \]
\vspace{0.3cm}
\end{mdframed}
\caption{\small{Following the arrows from top to bottom we obtain a cycle in the collection $\mathcal C_{\mathfrak e}^+$. Recall that the configuration is frozen outside the portion of the lattice drawn. Going in the opposite direction from bottom to top
we obtain a cycle in $\mathcal C_{\mathfrak e}^-$.}}
\label{fig: due v- due p Kaw1d}
\end{figure}
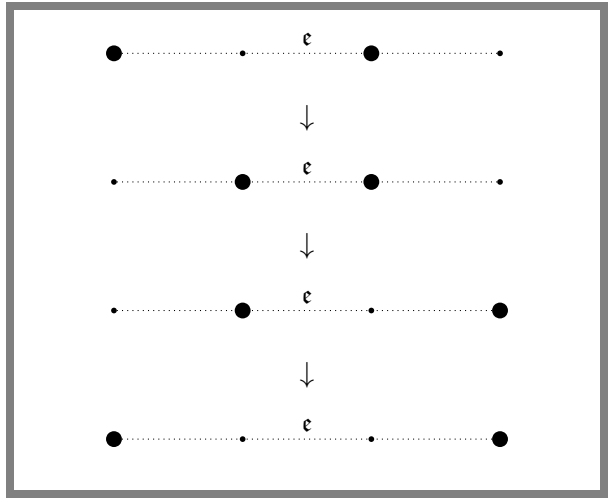
We call $W(\mathfrak e):=\tau_{-1}\mathfrak e\cup \mathfrak e\cup \tau_1\mathfrak  e$, i.e. the set of vertices belonging to one of the edges considered. We consider non negative functions $w, w_\pm$ such that $D(w)\cap \{0,1\}=\emptyset$ and $D(w_\pm)\cap W(\{0,1\})=\emptyset$. We construct a divergence free flow with a decomposition \eqref{Qr} with weights defined by
\begin{equation}\label{weights cycles kawa 1-d}
\rho(C):=\left\{
\begin{array}{ll}
e^{-H^*_{\mathfrak e^c}}\tau_{\mathfrak e} w & \textrm{if}\  C\in \mathcal C_{\mathfrak e}^r \,, \\
e^{-H^*_{W^c(\mathfrak e)}}\tau_{\mathfrak e} w^+ & \textrm{if}\ C\in \mathcal C_{\mathfrak e}^+\,,\\
e^{-H^*_{W^c(\mathfrak e)}}\tau_{\mathfrak e} w^- & \textrm{if}\ C\in \mathcal C_{\mathfrak e}^-\,.
\end{array}
\right.
\end{equation}
Definition \eqref{weights cycles kawa 1-d} is not ambiguous since the functions used are constant on the cycles and can be interpreted as functions of the cycle itself.
The flow is decomposed into a reversible part coming from the superposition of the reversible cycles and an irreversible one
$Q=Q^r+Q^i$.
The reversible part in the case of a jump of a particle from $x$ to $x+1$ is for example
\begin{equation}\label{qr}
Q^r(\eta,\eta^{x,x+1})=\tau_xw e^{-H^*_{\{x,x+1\}^c}}\eta(x)(1-\eta(x+1))\,,
\end{equation}
and similarly for a jump from $x+1$ to $x$.
Since any symmetric flow can be decomposed using these elementary cycles we obtain the most general form of the reversible rates having \eqref{gibbs} has an invariant measure
\begin{equation}\label{genrev}
c^r_{e^-,e^+}(\eta)=\tau_{\mathfrak e}w e^{H_{\mathfrak e}}\eta\left(e^-\right)\left(1-\eta\left(e^+\right)\right)\,.
\end{equation}

For the irreversible part $Q^i(\eta,\eta^{x.x+1})$  we may have a contribution from one cycle in $\mathcal C^\pm_{\{x-1.x\}}$ and one cycle in $\mathcal C^\pm_{\{x+1.x+2\}}$. Let us call $\chi_i\,, i=1,\dots , 4$ the characteristic functions associated to the local distribution of particles like in figure \ref{fig: due v- due p Kaw1d} when $\mathfrak e=\{0,1\}$ and numbering from the top of the figure toward the bottom. For example we have
$$
\chi_1(\eta)=\eta(-1)(1-\eta(0))\eta(1)(1-\eta(2))\,,
$$
for the first configuration from the top in figure \ref{fig: due v- due p Kaw1d} and similar expressions for the other cases. We have
\begin{eqnarray}\label{qi0q}
Q^i(\eta,\eta^{x,x+1})&=&\tau_{x-1}\left(w^+\chi_2e^{-H^*_{W^c(\{0,1\})}}\right)+\tau_{x+1}\left(w^+\chi_1
e^{-H^*_{W^c(\{0,1\})}}\right)
\nonumber \\
&+&\tau_{x-1}\left(w^-\chi_1e^{-H^*_{W^c(\{0,1\})}}\right)+\tau_{x+1}\left(w^-\chi_4
e^{-H^*_{W^c(\{0,1\})}}\right)\,.
\end{eqnarray}
Likewise for jumps in the opposite direction we have
\begin{eqnarray}\label{qi0-q}
Q^i(\eta,\eta^{x+1,x})&=&\tau_{x-1}\left(w^+\chi_4e^{-H^*_{W^c(\{0,1\})}}\right)+\tau_{x+1}\left(w^+\chi_3
e^{-H^*_{W^c(\{0,1\})}}\right)
\nonumber \\
&+&\tau_{x-1}\left(w^-\chi_3e^{-H^*_{W^c(\{0,1\})}}\right)+\tau_{x+1}\left(w^-\chi_2
e^{-H^*_{W^c(\{0,1\})}}\right)\,.
\end{eqnarray}
Using the general rule \eqref{qpi} we obtain the rates of transition for the irreversible part
\begin{equation}\label{qi0}
c^i_{x,x+1}(\eta)=
\tau_{x-1}\Big[\left(w^+\chi_2+w^-\chi_1\right)e^{H_{W(\{0,1\})}}\Big]+\tau_{x+1}\Big[\left(w^+\chi_1+w^-\chi_4\right)
e^{H_{W(\{0,1\})}}\Big]\,.
\end{equation}
Likewise for jumps in the opposite direction we have
\begin{equation}\label{qi0-}
c^i_{x+1,x}(\eta)=\tau_{x-1}\Big[\left(w^+\chi_4+w^-\chi_3\right)e^{H_{W(\{0,1\})}}\Big]+\tau_{x+1}\Big[\left(w^+\chi_3+w^-\chi_2\right)
e^{H_{W(\{0,1\})}}\Big]\,.
\end{equation}
All the terms in \eqref{qi0} have a factor $\eta(x)(1-\eta(x+1))$ (contained in the characteristic functions) so that this rate of jump is zero if this function is zero
and likewise in \eqref{qi0-} there is a factor $\eta(x+1)(1-\eta(x))$.

With some algebra putting all together we  obtain the following rates
\begin{align}\label{euna}
&c_{x,x+1}(\eta)=\eta(x)(1-\eta(x+1))\Big\{e^{H_{W(\{x-1,x\})}}\Big[\eta(x-1)(1-\eta(x-2))\tau_{x-1}w^+ \nonumber\\
&+\eta(x-2)(1-\eta(x-1))\tau_{x-1}w^-\Big]+e^{H_{W(\{x+1,x+2\})}}\Big[\eta(x+2)(1-\eta(x+3))\tau_{x+1}w^+\\
&+\eta(x+3)(1-\eta(x+2))\tau_{x+1}w^-\Big]+e^{H_{\{x,x+1\}}}\tau_x w\Big\}\,,\nonumber
\end{align}
and
\begin{align}\label{edue}
& c_{x+1,x}(\eta)=\eta(x+1)(1-\eta(x))\Big\{e^{H_{W(\{x-1,x\})}}\Big[\eta(x-2)(1-\eta(x-1))\tau_{x-1}w^+ \nonumber\\
& +\eta(x-1)(1-\eta(x-2))\tau_{x-1}w^-\Big]+e^{H_{W(\{x+1,x+2\})}}\Big[\eta(x+3)(1-\eta(x+2))\tau_{x+1}w^+\\
&+\eta(x+2)(1-\eta(x+3))\tau_{x+1}w^-\Big]+e^{H_{\{x,x+1\}}}\tau_x w\Big\}\,,\nonumber
\end{align}
where $w, w^\pm$ are arbitrary non-negative functions such that $D(w)\cap \{0,1\}=\emptyset$ and $D(w^\pm)\cap W(\{0,1\})=\emptyset$.

\subsection{Reversible and gradient models in one dimension}
A problem of interest is to find models that are at the same time of gradient type and reversible. This problem is discussed in Section II.2.4 of \cite{Spo91}. In dimension $d>2$ it is difficult to find models satisfying the two conditions at the same time. In $d=1$ it has been proved in \cite{Na98} that this is always possible for any invariant finite range  Gibbs measure and nearest neighbours exchange dynamics with $\Sigma=\{0,1\}$. Let us show how it is possible to prove this fact with a simple argument for a special class of interactions. The interactions that we consider are of the form $J_A\prod_{x\in A}\eta(x)$ where $A$ are intervals of the lattice. The numbers $J_A$ are translational invariant and satisfy the relation $J_A=J_{\tau_xA}$. Consider the most general form of the reversible rates \eqref{genrev} and fix the arbitrary function as $w=1$. The instantaneous current is given by $j_\eta(x,x+1)=e^{H_{\{x,x+1\}}}\left[\eta(x)-\eta(x+1)\right]$. The instantaneous current is different from zero only on the left and right boundary of each cluster of particles, see definition \ref{def:cluster}. More precisely it is positively directed on the right boundary and negatively directed on the left one. Since the interactions are associated only to intervals and are translational invariant then $e^{H_{\mathfrak e}}$ when $\mathfrak e$ is a boundary edge of a cluster depends only on the size of the cluster. This means that for each cluster the sum of the instantaneous currents on the two boundary edges is identically zero and this implies that
$\sum_{x\in V_N}j_\eta(x,x+1)=0$ for any configuration $\eta$. By Theorem \ref{belteo}, this coincides with the gradient condition. We generated in this way
a gradient reversible dynamics for each Gibbs measure of this type.

\subsection{Glauber dynamics} \label{ss:Gladyn}
We discuss in detail the one dimensional case of a Glauber dynamics \eqref{eq:Glgen}. The higher dimensional cases can be discussed very similarly.
It is convenient to write the rates as
\begin{equation}\label{GG}
c_x(\eta)=\eta(x)c^-_x(\eta)+(1-\eta(x))c^+_x(\eta)\,.
\end{equation}
Decomposition \eqref{GG} identifies uniquely $c_x^\pm$ only if we require $D(c^\pm_x)\cap \{x\}=\emptyset$.

The reversible cycles of length 2 are of the type $C=(\eta,\eta^x,\eta)$.
We consider also minimal irreversible cycles involving two neighbouring sites $x$ and $x+1$ (see figure \ref{fig: cicli Gl1-d}).
\begin{figure}[]
\begin{mdframed}[style=fig]
\[
\entrymodifiers={+<0.5ex>[o][F*:black]}
\xymatrix@C=1.5cm@R=1.cm{
*+[o][F*:black]{}\ar@{.}[r]^{\textit{\normalsize $\mathfrak e$}}& \ar@{}[r]|{\textit{\large{$\rightarrow$}}} &*+[o][F*:black]{}\ar@{.}[r]^{\textit{\normalsize $\mathfrak e$}}&*+[o][F*:black]{} \ar@{}[r]|{\textit{\large{$\rightarrow$}}} &
  \ar@{.}[r]^{\textit{\normalsize $\mathfrak e$}}&*+[o][F*:black]{} \ar@{}[r]|{\textit{\large{$\rightarrow$}}} &
\ar@{.}[r]^{\textit{\normalsize $\mathfrak e$}}&}
 \]
 \vspace{0.01cm}
 \end{mdframed}
\caption{\small{An elementary irreversible cycle associated to the bond $\mathfrak e=\{x,x+1\}$. The configuration outside the window is frozen. The sequence following the arrows from the left to the right define a cycle in $\mathcal C^+_{\mathfrak e}$ and a cycle in $\mathcal C_{\mathfrak e}^-$ is generated going in the opposite direction.}}
\label{fig: cicli Gl1-d}
\end{figure}
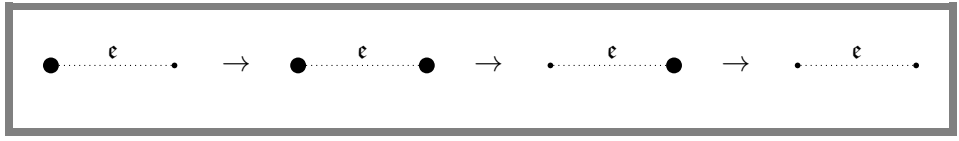
We fix the weights associated to the reversible cycles as $\rho(C):=e^{-H^*_{x^c}}\tau_xw$, where $w$ is any non-negative function
such that $D(w)\cap \{0\}=\emptyset$.
To fix the weights associated to the irreversible cycles consider $w^\pm$ non-negative functions
such that $D(w^\pm)\cap \{0,1\}=\emptyset$. We set
\begin{equation}\label{weights cycles Glauber 1-d}
\rho(C):=\left\{
\begin{array}{ll}
e^{-H^*_{\mathfrak e^c}}\tau_{\mathfrak e}w^+ & \textrm{if} \ C\in \mathcal C^+_{\mathfrak e}\,,\\
e^{-H^*_{\mathfrak e^c}}\tau_{\mathfrak e}w^- & \textrm{if} \ C\in \mathcal C^-_{\mathfrak e}\,.
\end{array}
\right.
\end{equation}
To the value of $Q(\eta,\eta^x)$ may contribute 4 irreversible cycles and a reversible one.
Two irreversible cycles are associated to the bond $\{x-1,x\}$ and two to the bond $\{x,x+1\}$.
When $\eta(x)=1$ we obtain
\begin{eqnarray}\label{g10}
Q^i(\eta,\eta^x)&=&\tau_x\left[e^{-H^*_{\{0,1\}^c}}\left(w^+\eta(1)+w^-(1-\eta(1)\right)\right]\nonumber \\
&+& \tau_{x-1}\left[e^{-H^*_{\{0,1\}^c}}\left(w^-\eta(0)+w^+(1-\eta(0)\right)\right]\,.
\end{eqnarray}
When $\eta(x)=0$ we obtain
\begin{eqnarray}\label{g01}
Q^i(\eta,\eta^x)&=&\tau_x\left[e^{-H^*_{\{0,1\}^c}}\left(w^-\eta(1)+w^+(1-\eta(1)\right)\right]\nonumber \\
&+& \tau_{x-1}\left[e^{-H^*_{\{0,1\}^c}}\left(w^+\eta(0)+w^-(1-\eta(0)\right)\right]\,.
\end{eqnarray}
For the reversible part we have in any case
\begin{equation}\label{gr}
Q^r(\eta,\eta^x)=\tau_x\left(we^{-H^*_{0^c}}\right)\,.
\end{equation}
We obtain that $c_x^\pm$ can be written as a sum of an irreversible part $c_x^{\pm,i}$ and of a reversible one $c_x^{\pm, r}$.
We have for the reversible part
\begin{equation}\label{arrit}
c_x^{r}=\tau_x\left(w e^{H_0}\right)\,.
\end{equation}
This is the form of the reversible rates
written in full generality. The Bernoulli case \cite{GJLV97} is recovered as a special case.
For the irreversible part we have to distinguish when $\eta(x)=1$ and when $\eta(x)=0$ using respectively
\eqref{g10} and \eqref{g01}. We obtain
\begin{eqnarray}\label{g01r}
c_x^{\pm,i}=&=&\tau_x\Big[e^{H_{\{0,1\}}}\left(w_\mp\eta(1)+w_\pm(1-\eta(1))\right)\Big]\nonumber \\
&+& \tau_{x-1}\Big[e^{H_{\{0,1\}}}\left(w_\pm\eta(0)+w_\mp(1-\eta(0))\right)\Big]\,.
\end{eqnarray}
A similar construction can be done also in dimension $d>1$.

\begin{remark}\label{r:MacRev}
If we fix the invariant measure as a Bernoulli measure of parameter $p$  the rates always satisfy the condition of macroscopic reversibility in \cite{GJLV97}
\begin{equation} \label{macroscopic rev cond}
\frac{\bb{E}_\lambda \left( c^-_x \right)}{\bb{E}_\lambda\left( c^+_x\right)}=\frac{1-p}{p}\,, \qquad \forall \lambda\in[0,1]\,,
\end{equation}
for the hydrodynamic scaling limit of a Glauber dynamics + exclusion process. In \eqref{macroscopic rev cond} $\mathbb E_\lambda$ denotes the expected value with respect to a Bernoulli measure of parameter $\lambda$ and $c^\pm_x$ are the one uniquely identified by \eqref{GG}
and the condition on the domain.
\end{remark}

\subsection{The general structure}
The constructions developed for some specific models can be summarized in a general form as follows.
Let $W\subseteq \mathbb T_N^2$ be a finite region of the lattice. Consider $C$ a finite cycle on the configurations space restricted
to the finite region. This means a graph having vertices $\Sigma^W$ and edges between pairs of configurations that can be obtained one from the other according to a local modification depending on the type of dynamics we are considering, for example by a jump of one particle from one site to a nearest neighbor. We have $C=(\eta^{(1)}_W,\dots,\eta^{(k)}_W)$ where we write explicitly that the configurations are restricted to $W$. We can associate to the single cycle $C$ several cycles on the full configuration space $\Sigma^{V_N}$. In particular for any configuration of particles $\xi_{W^c}$ outside of the region $W$, we have the cycle
$C[\xi]=\big(\xi_{W^c}\eta^{(1)}_W,\dots , \xi_{W^c}\eta^{(k)}_W\big)$. Since the cycle $C[\xi]$ is labeled by the configuration $\xi_{W^c}$ we can associate a weight $\rho(C)=g(\xi)$ where $D(g)\cap W=\emptyset$.
Starting from these cycles we can generate also other cycles by translations
$\tau_xC[\xi]=\big(\tau_x\left[\xi_{W^c}\eta^{(1)}_W\right],\dots , \tau_x\left[\xi_{W^c}\eta^{(k)}_W\right]\big)$ and to have translational
covariant models we give the weights in such a way that $g(\xi)=\rho(C[\xi])=\rho(\tau_xC[\xi])$. We obtain a divergence free flow as a superposition of all these cycles
\begin{equation}\label{acor}
Q=\sum_{x\in V_N}\sum_{\xi_{W^c}}g(\xi)Q_{\tau_xC[\xi]}\,.
\end{equation}
More generally we can consider a collection $C_1,\dots ,C_l$ of finite cycles on configurations restricted to finite regions $W_1,\dots ,W_l$ with weights associated to functions $g_1,\dots,g_l$ with suitable domains. Accordingly in \eqref{acor} there will be also a sum over the possible types of cycles. We consider for simplicity the case of one single cycle like in \eqref{acor}. Suppose that we need to determine the value of $Q(\eta,\eta^{x,x+e_{1}})$ corresponding to a jump of one particle from $x$ to $x+e_{1}$ in the configuration $\eta$. Let $\chi_1,\dots ,\chi_k$ be the characteristic functions associated to the configurations restricted to $W$ along the cycle, i.e.
$$
\chi_i(\eta):=\left\{
\begin{array}{ll}
1 & \textrm{if}\ \eta_W=\eta^{(i)}_W\,, \\
0 &  \textrm{otherwise}\,.
\end{array}
\right.
$$
Suppose that along the cycle $C$ there are jumps of particles along the direction $e_{1}$ at times $i_1,\dots ,i_n$ in the positions
$x_{1},\dots,x_{n}\in W$. This means that $\eta_W^{\left(i_{m}+1\right)}=\left(\eta_W^{\left(i_m\right)}\right)^{x_m,x_m+e_{1}}$, $m=1,\dots ,n$. We then have
\begin{equation}\label{prto}
Q(\eta,\eta^{x,x+e_{1}})=\sum_{m=1}^n\tau_{x-x_{m}}\left[\chi_{i_m}(\eta)g(\eta)\right]\,.
\end{equation}
Having a fixed invariant measure $\mu$, the rates can be determined using \eqref{qpi}. The functions $g$ can be chosen in such a way that the rates are local.

\subsection{Examples with global cycles}\label{ewgc}

There are simple and natural models that cannot be constructed using local cycles. This is immediately clear for example for the totally asymmetric simple exclusion process (TASEP) on a ring, since particles can move only in one direction. Consider for example the case when particles can move only anticlockwise. In this case an elementary cycle is obtained as follows. Starting from each configuration move the particles one by one anticlockwise in such a way that each particle will occupy at the end the position
of the first particle in front of it in the anticlockwise direction. Note that there are many possible cycles of this type depending on the order on which particles are moved. Across each edge of the ring there will be in any case only one particle jumping. It is not difficult to see by \eqref{qpi} that using these cycles it is possible to construct TASEP with Bernoulli invariant measures.

Another example of this type is  a continuous time version of the irreversible Glauber
dynamics in \cite{PSS17} for spins taking values $\pm 1$ on a ring. In this case elementary cycles can be naturally constructed in this way. We consider a given configuration and we flip the spins one by one starting from one single site and moving anticlockwise.
Going twice around the ring  we come back to the original configuration.

Both models can be understood also by the following symmetry argument.
Suppose that on a graph $(V,E)$ we have a flow $Q$ for which it is possible to find a bijection such that to any $(x,y)\in E$ we can associate a $(y',x)\in E$ and $Q(x,y)=Q(y',x)$. In this case the flow $Q$ is automatically divergence free since on each vertex the outgoing flow coincides with the incoming one. In the case of the TASEP from each configuration $\eta$, in the configuration space, the number of arrows exiting is equal to the number of arrows entering and coincides with the number of clusters on $\eta$. Giving the same weight to all the arrows corresponding to transitions for configurations having the same number of particles the bijection is automatically constructed. A similar more tricky construction can be done also for the model in \cite{PSS17}.

\section{Stationarity and orthogonality}\label{sec:sta&ort}

We discuss a different approach to the problem of constructing non reversible stationary non equilibrium states. This is a geometric construction based on the functional Hodge decomposition discussed in Section \ref{sec: HD}. The main idea is that it is possible to interpret the stationary equations as an orthogonality condition with respect to a suitable harmonic discrete vector field. Translational covariant discrete vector fields can be generated using the functional discrete Hodge decomposition. This interpretation can be given in general, however, to have a clearer view of the geometric construction, we discuss the simplified case
of a Bernoulli invariant measure. More general measures  (see Remark \ref{re:remarco}) and higher dimensions can be discussed as well.
We remember that $ (V_N, \mc E_N) $ and $ (V_N,E_N) $ are as at the begging of section \ref{sec:applications}, see also section \ref{sec:part models}.

\subsection{One-dimensional Exclusion process }\label{1dkb}
We start discussing possibly the simplest case, a one dimensional  Exclusion process (subsection \ref{ss:expr}) with a translational covariant rate of exchange given by $c_{\mathfrak e}(\eta)$. Of course we can write
the exchange rate as the sum of two jumps rates
\begin{equation}\label{e=s}
c_{\mathfrak e}(\eta)=c_{e^-,e^+}(\eta)+c_{e^+,e^-}(\eta)\,.
\end{equation}
We consider the case of the one dimensional ring. We can write $c_{x,x+1}(\eta)=\eta(x)(1-\eta(x+1))\tilde{c}_{x,x+1}(\eta)$ and $c_{x+1,x}(\eta)=\eta(x+1)(1-\eta(x))\tilde{c}_{x+1,x}(\eta)$ where  $D\left(\tilde{c}_{x,y}\right)\cap\{x,y\}=\emptyset$. The stationary condition can be written as
\begin{equation}\label{mauf}
\mu(\eta)\sum_{\mathfrak e\in \mathcal E_N}\left[c_{\mathfrak e}(\eta)-\frac{\mu(\eta^{\mathfrak e})}{\mu(\eta)}c_{\mathfrak e}(\eta^{\mathfrak e})\right]=0\,.
\end{equation}
Using \eqref{e=s} and the fact that $\mu$ is a Bernoulli measure so that $\frac{\mu(\eta^{\mathfrak e})}{\mu(\eta)}=1$, with some algebra we can show that \eqref{mauf} is equivalent to
\begin{equation}\label{equivuf}
\sum_{x\in V_N}\Big[\left(\eta(x+1)-\eta(x)\right)\left(\tilde{c}_{x+1,x}(\eta)-\tilde{c}_{x,x+1}(\eta)\right)\Big]=0\,.
\end{equation}
The expression inside squared parenthesis in \eqref{equivuf} is symmetric for the exchange $x\leftrightarrow x+1$. We can naturally interpret \eqref{equivuf} as the orthogonality condition
\begin{equation}\label{orti}
\langle \gamma_\eta , \mathbb I\rangle=0\,,
\end{equation}
where $\mathbb I$ is the harmonic discrete vector field defined by $\mathbb I (x,x+1)=1$ (and consequently $\mathbb I(x+1,x)=-1$) while $\gamma_\eta$ is the discrete vector field defined by
$$\gamma_\eta(x,x+1)=\Big(\eta(x+1)-\eta(x)\Big)\Big(\tilde{c}_{x+1,x}(\eta)-\tilde{c}_{x,x+1}(\eta)\Big)$$
(and consequently $\gamma_\eta(x+1,x)=-\gamma_\eta(x,x+1)$).
\begin{remark}\label{re:remarco}
In the case of a different invariant measure the stationary condition can be written again as \eqref{orti} where the discrete vector field $\gamma_\eta$ is given by
\begin{equation}\label{ggibbs}
\gamma_\eta(x,x+1):=\Big(\eta(x)- r_x(\eta)\eta(x+1)\Big)\left(\tilde c_{x,x+1}(\eta)-\frac{ \tilde c_{x+1,x}(\eta)}{r_x(\eta)}\right)\,,
\end{equation}
with $r_x(\eta):=\frac{\pi\left(\eta_{\{x,x+1\}^c}1_x0_{x+1}\right)}{\pi\left(\eta_{\{x,x+1\}^c}0_x1_{x+1}\right)}$. Also in this case $\gamma_\eta$ is translational covariant and assuming that $r_x(\eta)$ is local (as it is for a finite range Gibbs measure) we can proceed
similarly to the Bernoulli case.
\end{remark}
The vector field $\gamma_\eta$ is translational covariant and can be decomposed like \eqref{imbr}. The orthogonality condition \eqref{orti} can be satisfied if and only if the harmonic part $C$ in the decomposition \eqref{imbr} is identically zero and we get that the stationary condition is satisfied if and only if there exists a function $h$ such that
\begin{equation}\label{non}
\gamma_\eta(x,x+1)=(\eta(x+1)-\eta(x))\left(\tilde{c}_{x+1,x}(\eta)-\tilde{c}_{x,x+1}(\eta)\right)=\tau_{x+1}h(\eta)-\tau_xh(\eta)\,.
\end{equation}
To solve the above equation we have to consider a function $h$ such that the right hand side of \eqref{non} is zero when
$\eta(x)=\eta(x+1)$ since the left hand side is clearly zero in this case. Moreover when $\eta(x)\neq \eta(x+1)$, we have to impose that
\begin{equation}\label{non+}
\frac{\tau_{x+1}h(\eta)-\tau_xh(\eta)}{\eta(x+1)-\eta(x)}
\end{equation}
is a function invariant under the exchange of the values of $\eta(x)$ and $\eta(x+1)$. This is because by \eqref{non} we have that \eqref{non+}
has to coincide with $\tilde{c}_{x+1,x}(\eta)-\tilde{c}_{x,x+1}(\eta)$ that does not depend on $\eta(x),\eta(x+1)$. If we fix an $h$
that satisfies these constraints, we obtain by \eqref{non} the rates $\tilde c$.

A first possibility is to fix $h(\eta)=\eta(0)C(\eta)$ where $C(\eta)$ is a translational invariant function.
In this way we get $\tau_{x+1}h(\eta)-\tau_xh(\eta)=(\eta(x+1)-\eta(x))C(\eta)$ and we have
$$
\tilde{c}_{x+1,x}(\eta)-\tilde{c}_{x,x+1}(\eta)=C(\eta)\,.
$$
Since the left hand side does not depend on $\eta(x)$ and $\eta(x+1)$ and $C$ is translational invariant, the only possibility is that
$C$ is a constant function.
All the non negative solutions $X,Y$ of the equation $X-Y=A$ are given by
\begin{equation}\label{posol}
\left\{
\begin{array}{l}
X=[A]_++S\,,\\
Y=[-A]_++S\,,
\end{array}
\right.
\end{equation}
where $[\cdot]_+$ denotes the positive part and $S\geq 0$ is arbitrary.
We obtain that the general solution in this case is
\begin{equation}\label{elnc}
\left\{
\begin{array}{l}
\tilde c_{x+1,x}(\eta)=\Big[C\Big]_++\tau_xs(\eta)\,, \\
\tilde c_{x,x+1}(\eta)=\Big[-C\Big]_++\tau_xs(\eta)\,,
\end{array}
\right.
\end{equation}
where $s$ is an arbitrary non negative function such that $D(s)\cap\{0,1\}=\emptyset$. In formula \eqref{elnc} the additive part involving the function $s$ corresponds to the reversible part while the remaining part corresponds to the irreversible part. The asymmetric exclusion process is obtained as a special case.
The corresponding decomposition into cycles contains necessarily
global cycles as discussed in subsection \ref{ewgc}.

Another possibility is to consider a function $h$ of the form $h=\tau_1 g+g$. Then we have
\begin{equation}\label{grgh}
\tau_{x+1}h-\tau_xh=\tau_{x+2}g-\tau_xg\,.
\end{equation}
A very general class  of functions $h$ that can be used in \eqref{non} is then obtained by $h=g+\tau_1g$ where
\begin{equation}\label{trd}
g(\eta)=\big(\eta(-1)-\eta(-2)\big)\big(\eta(1)-\eta(0)\big)\tilde g(\eta)
\end{equation}
and $\tilde g$ is any function such that $D(\tilde g)\cap W(\{-1,0\})=\emptyset$, where we recall that the symbol $W(\mathfrak e)$ has been defined just above \eqref{weights cycles kawa 1-d}.

We obtain for our class of functions $h$ a general form of the rates given by
\begin{equation}\label{eln}
\left\{
\begin{array}{l}
\tilde c_{x+1,x}(\eta)=\Big[(\eta(x+3)-\eta(x+2))\tau_{x+2}\tilde g-(\eta(x-1)-\eta(x-2))\tau_x\tilde g\Big]_++\tau_xs(\eta)\,, \\
\tilde c_{x,x+1}(\eta)=\Big[(\eta(x-1)-\eta(x-2))\tau_x\tilde g-(\eta(x+3)-\eta(x+2))\tau_{x+2}\tilde g\Big]_++\tau_xs(\eta)\,,
\end{array}
\right.
\end{equation}
where $s$ is an arbitrary non negative function such that $D(s)\cap\{0,1\}=\emptyset$. In formula \eqref{eln} the additive part involving the function $s$ corresponds to the reversible part while the part involving the function $\tilde g$ corresponds to the irreversible part.

To compare the models obtained in \eqref{eln} with the models obtained in subsection \ref{ss:odex} we have first of all
to recall that here we are considering product invariant measures so that in formulas \eqref{euna} and \eqref{edue}
the value of the energy is constant and can be incorporated into the arbitrary functions. If we fix $\tau_1 \tilde g=w^+-w^-$ we obtain that the rates obtained in subsection \ref{ss:odex} are a subset of the models defined by \eqref{eln}. This is obtained observing that
in both cases we will have
\begin{equation}\label{pes}
\tilde{c}_{x+1,x}-\tilde{c}_{x,x+1}=\big(\eta(x+3)-\eta(x+2)\big)\tau_{x+2}\tilde g+
\big(\eta(x-2)-\eta(x-1)\big)\tau_{x}\tilde g\,.
\end{equation}
Formula \eqref{eln} gives the most general positive solution to \eqref{pes} while instead this is not the case for the rates of
subsection \ref{ss:odex} (note for example that selecting $s=0$ in \eqref{eln} we obtain rates such that $\min\left\{\tilde{c}_{x+1,x},\tilde{c}_{x,x+1}\right\}=0$ while this is not always possible for the rates in subsection \ref{ss:odex}).
This means that any model constructed with cycles like in ection \ref{ss:odex} can be obtained by \eqref{eln} for some $\tilde g$ and $s$;
however among the rates defined by \eqref{eln} there are some models for which the typical current cannot be decomposed
by cycles of length two and by cycles like in Figure 3.

We stress again that in this case, as well as in the followings, we obtain very general families
of models parametrized by arbitrary functions. Considering simple cylindric functions $\tilde g$ and $s$ it is
possible to obtain simple and completely explicit models.

\subsection{One-dimensional Glauber dynamics}\label{ss:Glort1}

The stationary condition for a Glauber dynamics (subsection \ref{ss:Gdyn}) on the one dimensional torus is
\begin{equation}\label{stazb}
\mu(\eta)\sum_{x\in V_N}\left[c_x(\eta)-\frac{\mu(\eta^x)}{\mu(\eta)}c_x(\eta^x)\right]=0\,.
\end{equation}
We consider translational covariant rates that can be written as
\begin{equation}\label{rate}
c_x(\eta)=\eta(x)\tau_xc^-(\eta)+(1-\eta(x))\tau_xc^+(\eta)\,,
\end{equation}
where $D\left(c^\pm\right)\cap\{0\}=\emptyset$. We consider the case of
Bernoulli invariant measures of parameter $p$. The stationary condition \eqref{stazb} is equivalent to
\begin{eqnarray}\label{ste1}
& &\sum_{x\in V_N}\left(\eta(x)\tau_xc^-(\eta)+(1-\eta(x))\tau_xc^+(\eta)\right)\nonumber \\
& &=\sum_{x\in V_N}\left(\eta(x)\frac{(1-p)}{p}\tau_xc^+(\eta)+(1-\eta(x))\frac{p}{(1-p)}\tau_xc^-(\eta)\right)=0\,,
\end{eqnarray}
that with some algebra becomes
\begin{equation}\label{step3}
\sum_{x\in V_N}\left[\left(1-\frac{\eta(x)}{p}\right)\left(\tau_xc^+(\eta)-\frac{p}{(1-p)}\tau_xc^-(\eta)\right)\right]=0\,.
\end{equation}
As before we can naturally interpret \eqref{step3} as the orthogonality condition $\langle\vphi_\eta,\mathbb I\rangle=0$ where the translational covariant vector field $\vphi_\eta$ is defined by
\begin{equation}\label{deffi}
\vphi_\eta(x,x+1):=\left(1-\frac{\eta(x)}{p}\right)\left(\tau_xc^+(\eta)-\frac{p}{(1-p)}\tau_xc^-(\eta)\right)\,,
\end{equation}
setting $\vphi_\eta(x+1,x)=-\vphi_\eta(x,x+1)$ by antisymmetry. Since we are in 1 dimension we have that \eqref{step3} holds
if and only if there exists a function $h$ such that
$$
\vphi_\eta(x,x+1)=\tau_{x+1}h(\eta)-\tau_x h(\eta)\,.
$$
By translational covariance this relation is equivalent to
\begin{equation}\label{unasola}
c^+(\eta)-\frac{p}{(1-p)}c^-(\eta)=\frac{\tau_{1}h(\eta)- h(\eta)}{\left(1-\frac{\eta(0)}{p}\right)}\,.
\end{equation}
An important fact to observe is that the left hand side of \eqref{unasola} does not depend on $\eta(0)$ and then this must be true also for the right hand side. We have a general family of functions satisfying this constraint that is given by
\begin{equation}\label{genht}
h(\eta)= \left(1-\frac{\eta(-1)}{p}\right)\left(1-\frac{\eta(0)}{p}\right)\tilde{h}(\eta)\,,
\end{equation}
where $\tilde h$ is an arbitrary function such that $D(\tilde h)\cap \{-1,0\}=\emptyset$.
The other constraint that has to be satisfied is that the rates are nonnegative functions and this is obtained considering positive solutions of \eqref{unasola} by using \eqref{posol}. We obtain
\begin{equation}\label{formratesg1}
\left\{
\begin{array}{l}
c^+(\eta)=\left[\left(1-\frac{\eta(1)}{p}\right)\tau_1\tilde h-\left(1-\frac{\eta(-1)}{p}\right)\tilde h\right]_++s(\eta)\,,\\
c^-(\eta)=\frac{1-p}{p}\left(\left[\left(1-\frac{\eta(-1)}{p}\right)\tilde h-\left(1-\frac{\eta(1)}{p}\right)\tau_1\tilde h\right]_++s(\eta)\right)\,,
\end{array}
\right.
\end{equation}
where $s$ is an arbitrary non negative function such that $D(s)\cap \{0\}=\emptyset$.
Again the part with $s$ is the reversible contribution.

\subsection{Two dimensional exclusion process }

In two dimensions the stationary condition under the hypothesis of Bernoulli invariant measure is analogous to \eqref{mauf}. We have indeed
\begin{equation}\label{dentino}
\sum_{x\in V_N}\sum_{i=1,2}\left[\left(\eta\left(x+e_{i}\right)-\eta(x)\right)\left(\tilde{c}_{x+e_{i},x}(\eta)-\tilde{c}_{x,x+e_{i}}
(\eta)\right)\right]=0\,,
\end{equation}
that can be interpreted as $\langle\gamma_\eta,\mathbb I\rangle=0$ where the translational covariant vector field $\gamma_\eta$ is given by \begin{equation}\label{kkk}
\gamma_\eta\left(x,x+e_{i}\right)=\left(\eta\left(x+e_{i}\right)-\eta(x)\right)\left(\tilde{c}_{x+e_{i},x}(\eta)
-\tilde{c}_{x,x+e_{i}}(\eta)\right)\,, \qquad i=1,2\,.
\end{equation}
The values of $\gamma_\eta\left(x+e_{i},x\right)$ are fixed by antisymmetry. The orthogonality condition is satisfied if and only if
for any fixed $\eta$ the vector $\gamma_\eta$ is obtained as a linear combination of a vector field in $d \Omega^0$,
a vector in $\delta\Omega^2$ and the vector $\vphi^{(1)}-\vphi^{(2)}$. Since the vector filed $\gamma_\eta$ is translational covariant
there exist functions $h,g$ and a real number $\omega$ such that
\begin{equation}\label{hodgefun22}
\gamma_\eta(e)=\big[\tau_{e^+}h(\eta)-\tau_{e^-}h(\eta)\big]+\big[\tau_{\mathfrak f^+(e)}g(\eta)-\tau_{\mathfrak f^-(e)}g(\eta)\big]\pm \omega\,,
\end{equation}
where the sign $\pm$ has to be fixed as $+$ if $e=(x,x+e_{1})$  or $e=(x,x-e_{2})$ for some $x$ and has to be fixed as $-$
in the remaining cases.
We discuss two cases. In both cases the orthogonality with the vector $\mathbb I$ is verified but the splitting \eqref{hodgefun22} is not trivial.

\smallskip

The first case is as follows. Let $k(\eta)$ and $v(\eta)$ be two function having a structure like in Section \ref{1dkb} but respectively along the two directions of the plane. More precisely let $\tilde k$ be a function
such that $D\left(\tilde k\right)\subseteq \{je_{1}\,,\,j=1,\dots ,N\}$ and $D\left(\tilde k\right)\cap \{-2e_{1}, -e_{1}, 0, e_{1}\}=\emptyset$. Define then
$$
k(\eta):=\left[\eta\left(-e_{1}\right)-\eta\left(-2e_{1}\right)\right]
\left[\eta\left(e_{1}\right)-\eta\left(0\right)\right]\tilde k(\eta)\,.
$$
Likewise let $\tilde v$ be a function
such that $D\left(\tilde v\right)\subseteq \{je_{2}\,,\,j=1,\dots ,N\}$ and $D\left(\tilde v\right)\cap \{-2e_{2}, e_{2}, 0, e_{2}\}=\emptyset$. Define then
$$
v(\eta):=\left[\eta\left(-e_{2}\right)-\eta\left(-2e_{2}\right)\right]
\left[\eta\left(e_{2}\right)-\eta\left(0\right)\right]\tilde v(\eta)\,.
$$
We define $h^1(\eta):=k(\eta)+\tau_{e_{1}}k(\eta)$ and  $h^2(\eta):=v(\eta)+\tau_{e_{2}}v(\eta)$.
Finally we define
\begin{equation}\label{terr}
\left\{
\begin{array}{l}
\gamma_\eta(x,x+e_{1})=\tau_{x+e_{1}}h^1-\tau_x h^1\,,\\
\gamma_\eta(x,x+e_{2})=\tau_{x+e_{2}}h^2-\tau_x h^2\,.
\end{array}
\right.
\end{equation}
The orthogonality $\langle\gamma_\eta,\mathbb I\rangle=0$ follows by
$$
\langle\gamma_\eta,\vphi^{(1)}\rangle=\langle\gamma_\eta,\vphi^{(2)}\rangle=0\,,
$$
obtained like in the one dimensional case. Inserting \eqref{terr} in the left hand side of \eqref{kkk} we obtain
the rates $\tilde c$ like in Section \ref{1dkb}.

\smallskip

The second case that it is possible to describe is the following. Let $b$ be a function such that $D(b)\cap \left\{0,e_{1},e_{2},e_{1}+e_{2}\right\}=\emptyset$. Consider also $c^{(i)}$, $i=1,2$ two arbitrary constants. We define
\begin{equation}\label{clml}
\tilde{c}_{x,y}(\eta)
-\tilde{c}_{y,x}(\eta)=\left[\tau_{\mathfrak f^+(x,y)}b(\eta)-\tau_{\mathfrak f^-(x,y)}b(\eta)\right]\pm c^{(i)}\,,
\end{equation}
when $y=x\pm e_{i}$.
By the property of $D(b)$ it is possible to obtain by \eqref{clml} the rates $\tilde c$ using again \eqref{posol}.
According to \eqref{clml} if we define the vector field $\psi_\eta(x,y):=\tilde{c}_{x,y}(\eta)
-\tilde{c}_{y,x}(\eta)$, we have $\psi_\eta\in\delta\Omega^2\oplus\Omega^1_H$ for any $\eta$. We define also the vector field
$\vphi_\eta(x,y):=\eta(y)-\eta(x)$ and for any $\eta$ we have $\vphi_\eta\in d \Omega^0$. Since the stationary conditions
\eqref{dentino} can also naturally be interpreted as $\langle\vphi_\eta,\psi_\eta\rangle=0$ and the two discrete vector fields belong to orthogonal subspaces the stationarity conditions are automatically satisfied. 

Another possible construction similar to \eqref{clml} is the following.  Let $b$ be a function such that $D(b)\cap \left\{0,e_{1},e_{2},e_{1}+e_{2}\right\}=\emptyset$. We define
\begin{equation}\label{clml2}
\tilde{c}_{x,y}(\eta)
-\tilde{c}_{y,x}(\eta)=\tau_{\mathfrak f^+(x,y)}b(\eta)-\tau_{\mathfrak f^-(x,y)}b(\eta)\,.
\end{equation}
By the property of $D(b)$ it is possible to obtain by \eqref{clml2} the rates $\tilde c$ using again \eqref{posol}.
It remains to show that $\gamma_\eta$ in \eqref{kkk} is orthogonal to $\mathbb I$. Once again $\gamma$ is not purely a gradient or a curl and the splitting \eqref{hodgefun2} is not trivial. The proof of the orthogonality is obtained by the following argument. Given any configuration of particles $\eta$ the 2 dimensional torus is subdivided into the disjoint clusters $\mathfrak C(\eta)$ of particles (see definition \ref{def:cluster}). To any configuration of clusters it is associated in the dual graph a collection of contours \cite{Ga99}. We consider a contour as a closed oriented cycle. The oriented contours are defined as follows. If $\eta(x)=1$ and $\eta(y)=0$ we add to the contours the dual oriented edge obtained rotating $(x,y)$ anticlockwise of $\frac \pi 2$ around its middle point.  Denoting by $\mathcal C^*_\eta$ the oriented contours associated to the configuration $\eta$, we have
\begin{equation}\label{4geko}
\langle\gamma_\eta,\mathbb I\rangle=\sum_{C^*\in \mathcal C^*_\eta}\sum_{ e^*\in C^*}\left(\tau_{e_*^+}b-\tau_{e_*^-}b\right)=0\,,
\end{equation}
where $\tau_{x^*}=\tau_{\mathfrak f}$ if $\mathfrak f$ is the face associated to the dual vertex $x^*$.
The last equality in \eqref{4geko} follows by the fact that the sum on every single oriented contour of the telescopic sum is identically zero.

\subsection{Two dimensional Glauber dynamics}\label{ss:Glort2}
We consider now the two dimensional Glauber dynamics but the same construction can be done also in higher dimensions. The stationary condition for a Bernoulli invariant measure of parameter $p$ can be written again like \eqref{step3}
\begin{equation}\label{step32}
\sum_{x\in V_N}\left[\left(1-\frac{\eta(x)}{p}\right)\left(\tau_xc^+(\eta)-\frac{p}{(1-p)}\tau_xc^-(\eta)\right)\right]=0\,.
\end{equation}
This expression is of the form $\sum_{x\in V_N}\tau_x g(\eta)=0$ where the function $g$ is given by
\begin{equation}\label{carsoli}
g(\eta)=\left(1-\frac{\eta(0)}{p}\right)\left(c^+(\eta)-\frac{p}{(1-p)}c^-(\eta)\right)\,.
\end{equation}
Given a $\tilde g\in \Omega^0$ such that $\sum_{x\in V_N}\tilde g(x)=0$ then there exists a $\vphi\in \Omega^1$ such that
$\tilde g(x)=\div \vphi(x)$. This fact can be proved in several ways, for example considering a $\vphi$ of gradient type
$\vphi(x,y)=h(y)-h(x)$ we have that $h$ has to satisfy the discrete Poisson equation $\Delta h:=\div dh=\tilde g$ that has always a solution.
In our case we have that $\tilde g(x)=\tilde g(\eta,x)$ depends on $\eta$ so that
there exists an $\eta$ dependent discrete vector field $\vphi_\eta$ such that $\div \vphi_\eta(\cdot)=\tilde g(\eta,\cdot)$. Since $\tilde g(\eta, x)=\tau_xg(\eta)$ is translational covariant, with the same arguments as in Theorem \ref{belteo2} we can prove that
the $\eta$ dependent discrete vector field $\vphi_\eta$ can be chosen translational covariant also. The most general translational covariant discrete vector field has the following structure. Let $u(\eta), v(\eta)$ be two function then we define
\begin{equation}\label{dvfc}
\left\{
\begin{array}{l}
\vphi_\eta(x,x+e_{1})=\tau_x u(\eta)\,,  \\
\vphi_\eta(x,x+e_{2})=\tau_x v(\eta)\,,
\end{array}
\right.
\end{equation}
and this is the most general translational covariant $\eta$ dependent discrete vector field.
We obtain that the stationary condition is satisfied if and only if there exists two functions $u,v$ such that
\begin{equation}\label{quest}
g(\eta)= \div \vphi_\eta=u(\eta)-\tau_{-e_{1}}u(\eta)+v(\eta)-\tau_{-e_{2}}v(\eta)\,,
\end{equation}
where $g$ is given by \eqref{carsoli}.
As in the one dimensional case to solve the equation \eqref{quest} we have to fix the functions $u,v$ in such a way that
$$
\frac{u(\eta)-\tau_{-e_{1}}u(\eta)+v(\eta)-\tau_{-e_{2}}v(\eta)}{\left(1-\frac{\eta(0)}{p}\right)}
$$
has a domain $D$ such that $D\cap \{0\}=\emptyset$. A general family that satisfies this constraint is given by
$$
\left\{
\begin{array}{l}
u(\eta)=\left(1-\frac{\eta(0)}{p}\right)\left(1-\frac{\eta\left(e_{1}\right)}{p}\right)\tilde u(\eta)\,,\\
v(\eta)=\left(1-\frac{\eta(0)}{p}\right)\left(1-\frac{\eta\left(e_{2}\right)}{p}\right)\tilde v(\eta)\,,
\end{array}
\right.
$$
where $D\left(\tilde u\right)\cap \{0,e_{1}\}=\emptyset$ and $D\left(\tilde v\right)\cap \{0,e_{2}\}=\emptyset$.
Under these assumptions we can find the positive solutions $c^\pm$ using again \eqref{posol}. Writing only the non reversible contribution we have that $c^+(\eta)$ is equal to
\begin{align}\label{uuu}
&\left[\left(1-\frac{\eta\left(e_{1}\right)}{p}\right)\tilde u-\left(1-\frac{\eta\left(-e_{1}\right)}{p}\right)\tau_{-e_{1}}\tilde u\right.\nonumber \\
&\left.+\left(1-\frac{\eta\left(e_{2}\right)}{p}\right)\tilde v-\left(1-\frac{\eta\left(-e_{2}\right)}{p}\right)\tau_{-e_{2}}\tilde v\right]_+
\end{align}
while $c^-(\eta)$ is given by $c^-(\eta)=\frac{1-p}{p}\left[-i(\eta)\right]_+$, where $i(\eta)$ is the function appearing inside the positive part in \eqref{uuu}.

\bigskip

In appendix \ref{a:mic} the construction of invariant Glauber dynamics is done using a functional Hodge decomposition for vertex functions.

\facciatabianca

\cleardoublepage
\pagestyle{fancy}
\fancyhf{}
\rhead{}
\lhead{}
\cfoot{ \footnotesize \thepage}
\renewcommand{\headrulewidth}{0pt}

{\sffamily\part{Macroscopic theory}\label{p:mact}}

\section*{Why a macroscopic theory?}

At the very begging of Part I we anticipated that the stochastic models of chapter \ref{ch:IPSmodels} manifest on large scale deterministic behaviours when considered at a macroscopic space scale, that is when the separation lattice scale $ 1/N $ goes to zero. Remember that, at the beginning of section \ref{sec:CTMC}, we distinguished for the discrete domain $ V_N $  a microscopic space scale $ 1 $ and a macroscopic space scale $ 1/N $ . These deterministic patterns are called \e{scaling limits} and are the main motivation of interest for the probabilistic models we presented. These patterns are for example some autonomous equations (PDEs) governing the dynamics of a  macroscopic observable (e.g. mass density, energy density, electric charge, currents) or thermodynamics functionals (e.g. functionals describing fluctuations).  In this second part we will explain how in these probabilistic models from a random motion on the small scale a deterministic motion on the large scale emerges and we will study their stationary states from a large deviations point of view. A remarkable property is the fact that on macroscopic level the microscopic interactions appear only indirectly through the thermodynamics (e.g. equation state, mathematical property of the thermodynamic functionals) and the transport coefficients of the  PDE obtained in the scaling limit.

\vspace{0.5cm}

A macroscopic theory of nonequilibrium systems has of course as goal that of establishing a nonequilibrium thermodynamics, namely a description of the laws governing the evolution of observable quantities and understanding which are the thermodynamic functionals (read also potentials) describing the  nonequilibrium thermodynamic transformations.
In physical systems out of equilibrium there is neither a  rigorous derivation of the Fourier law (and then of the heat equation) \cite{BLR00}
nor a mathematical description of the thermodynamic potentials. This is due to a lack of a theory for the ensembles in nonequilibrium: because of the presence of non-conservative terms (i.e. viscous terms) in the microscopic Newton's equations of the matter it is no longer possible to make the same hypothesis leading to the ensembles of (equilibrium) statistical mechanics, as discussed in the introduction of part \ref{p:mt}.

Even though stochastic interacting systems  like the ones of chapter \ref{ch:IPSmodels} are very simple compared to real physical systems, their  interactions are enough interesting  to lead to rigorous results for complex phenomena in both these directions, namely we can prove  hydrodynamic equations describing physical phenomena (for example diffusive equations) and computing explicitly their large deviations functionals.  So the scale at which we are interested is a macroscopic scale  where the local randomness can still manifest in fluctuations from the  macroscopic behaviour of the hydrodynamics equation but  microscopic details remain hidden. At this scale an important property to explain is the \e{local equilibrium} property. Let's suppose for example to have a gas ( order of $ 10^{23} $ degree of freedom)  confined in a volume $ \Lambda $ with an initial configuration far away from a stationary state. If we consider a microscopically large but macroscopically small volume\footnote{ For each point $ x $ in $ V_N $, we denote by $ B_x $ a small neighbourhood of $ x $ such that to be small if
compared to the total volume $ \Lambda$, but large enough if compared to the intermolecular
distance. Hence we can  assume that each neighbourhood contains an infinite number
of particles. } $ B_x $ around the point $ x\in V_N $,
due to the "strong" interaction among the molecules, it is natural to
believe that the system reaches very rapidly a local equilibrium state, i.e. in
each  $ B_x $ the system is close to an equilibrium state  characterized  by a small number of macroscopic quantities
$ \boldsymbol{\rho}(x) =(\rho_1(x),\dots,\rho_m(x)) $, called thermodynamics characteristics. This is a \e{local equilibrium}. The parameters $ \rho_i  $ can be the temperature,
the density, the pressure, etc. . For our purpose we consider only one index $ i $, so $ \boldsymbol{\rho}(x)=\rho(x) $. At a later time $ t $, we expect   to observe  in a
small neighbourhood $ B_x $ a state close to a new local equilibrium state characterized by
 a space-time parameter $  \rho(t, x) $. This  parameter $   \rho(t, x) $ is expected to evolve  in time according to a partial
differential equation. For our stochastic interacting models this  corresponds to the \e{hydrodynamic equation}, namely a deterministic behaviour of the empirical measure \eqref{eq:empmisxi} in the sense \eqref{eq:wconKMP}. From this discussion about the local equilibrium, it seems quite natural  to define (as it was for space)  a macroscopic time scale $ t $ at which  matter/energy is transported ( i.e. the parameter $  \rho(t,x) $ evolves) and a microscopic scale $ \theta_N t $ (with $ \theta_N>0 $ diverging with $ N $) at which particles interact many times before to give an observable effect, namely when at the scale time $ t $ the local equilibrium reaches another state for $ \rho(x,t) $.
So we could call this phenomenological macroscopic scale  also the \e{local equilibrium scale}. Its mathematical meaning is that  also out of equilibrium we can compute integrals of local functions as integrals of expectations   with respect to a stationary measure parametrized by a local equilibrium density in $ B_x $. Usually  the stationary measure characterizing the local equilibrium is  the  (product)  reversible measure of the systems at equilibrium, so before investigating the evolution of the systems out of equilibrium the state  at equilibrium has to be known. 

For the interacting models of chapter \ref{ch:IPSmodels} with  a diffuse hydrodynamics equation, at  this macroscopic scale of the local equilibrium is possible to formulate\cite{BDeSGJ-LL15} a fluctuation theory called \e{macroscopic fluctuation theory} (MFT) around the deterministic large scale motion as we will explain in section \ref{sec:Stldt}. In general the computation of the invariant measure is a difficult problem, but  this theory describes indirectly the invariant measure  of our microscopic models with a variational method to compute the large deviations functional for the probability that   the empirical measure \eqref{eq:empmisxi} (having the meaning of the macroscopic quantity  $ \rho$) stays close to an atypical profile, namely a profile different from the typical one given by the stationary solution of the hydrodynamics equation for $ \rho$.  The macroscopic fluctuation theory generalise to nonequilibrium the near equilibrium Onsager's fluctuation  theory  in the context of diffusive systems in contact with boundary reservoirs and possibly in presence of an external field. The boundary reservoirs represent the environment and are characterized   by their chemical potentials, in addition they are assumed to be much larger than the systems so that  their state is constant in time and they fix  the observable $ \rho $ at the boundaries of $ V_N $ in such a way that no fluctuations can be observed at the boundary. This kind of boundaries are idealized with fast dynamics like in  \eqref{eq:bKMPgen} with a chemical potential $ \lambda(x) $ for $ x\in \p V_N $ that fixes the value of $ \rho $ at this point as $ \rho(x,t)=|\lambda(x)| $. With respect to the Onsager's, theory MFT is  a far away  equilibrium theory: it allows to consider nonlinear hydrodynamics equations with fixed arbitrary different boundary conditions  and the presence of an external driving force (like potential difference) that doesn't have to be assumed small. Fluctuations described by the MFT in general are not Gaussian and the main source of new phenomena is the non-linearity of the underlying evolution equations. 

Nonequilibrium is characterised by the presence of currents flowing trough the systems in addition to the usual thermodynamics variables.  The MFT is based on the following formula for the probability of joint space-time fluctuations $ \rho $  and \rt{ $j $ } (in dimension \e n)
\begin{equation}\label{eq:MFTformula}
\bb{P}(\pi_N(\xi_t)\sim\rho)\simeq\exp\left\{-N^n\frac{1}{\k T_e}\frac 14\int dt\int dx\,(j-J(\rho))\cdot (\sigma(\rho))^{-1}(j-J(\rho))\right\},
\end{equation}  
where $ T_e $ is the environment temperature, $ \rho $  the thermodynamic variable, $ \sigma(\rho) $ the mobility, $ j $ the actual value of the current such that $ \p_t\rho+\div  j=0 $ and $ J(\rho) $ the hydrodynamic current for the given value $ \rho $, that is $ J(\rho) $ is the current describing the hydrodynamic limit $ \p_t\rho+ \div J(\rho)=0$ of the stochastic microscopic model. The probability $ \bb P  $  represents the ensemble over the space-time trajectories, for our stochastic interacting systems is over the trajectories of the empirical measure. Similarly for the autonomous equations describing the macroscopic evolution law of observables, this formula can not be derived from the real interactions between molecules and their Newton's equations because of the limits  of nowadays mathematical theory. The exponent in \eqref{eq:MFTformula} is interpreted as a functional proportional to the energy dissipated by the extra current $ j-J(\rho) $ as  follows. Given a fluctuation $ \rho $ we perturb the original  system by adding an external field $ F $ (like $ \bb F $ in subsection \ref{subsec:WAm} for our probabilistic microscopical models) for which the prescribed fluctuation becomes the trajectory followed by the system, namely we modify the hydrodynamics in a suitable way. For a fluctuation $ (\rho,j) $ it is  $ F= (\sigma(\rho))^{-1} (j-J(\rho)) $ and the exponent becomes proportional to $ \int dt\int dx F\cdot \sigma(\rho)F $. For example in the case of an electric circuit  $ \sigma^{-1} $ is the resistance, $ F $  the eletric field and the integral in \eqref{eq:MFTformula} is the energy dissipated by the extra current $ j-J(\rho) $ according to the Ohm's law. Moreover observing only the fluctuations of the density $ \rho $, the optimal perturbation (in the sense of probabilistic cost of the fluctuation) will be of the form $ F=-\nabla H $ for a smooth field $ H $, see section IV in \cite{BDeSGJ-LL15}. From \eqref{eq:MFTformula} is derived the  variational method already mentioned in the last paragraph, which  express the probability $ \bb P_{\mu_N} $ of the fluctuations of the stationary ensemble, that is when the distribution of the lattice variables are characterized by their invariant measure $  \mu_N$ and the fluctuation $ \rho $ is not the stationary solution $ \bar \rho $ of the hydrodynamics equation.  This  method is a dynamical variational principle that leads to the definition of a  functional $ V(\rho) $ estimating the cost of a spatial  fluctuation $ \rho(x) $ for the empirical measure and called \e{quasi-potential}, i.e. 
\begin{equation}\label{eq:qpformula}
\bb P_{\mu_N}(\pi_N(\xi)\sim\rho)\simeq \exp\{-N^n V(\rho) \}
\end{equation}  
where $ V(\rho) $ is defined in \eqref{eq:QP-MFT}. This formula \eqref{eq:qpformula} generalise (see section I of \cite{BDeSGJ-LL15}) the Boltzmann-Einstein formula and $ V(\rho) $ is an extension to non-equilibrium of the availability of classic thermodynamics which coincides with Gibbs free energy when the external pressure $ p_e $ and temperature $ T_e $ are in equilibrium with the ones of the system in contact with the environment, see  chapter 7 in \cite{Pi66}. 
In non-equilibrium the quasi-potential $ V(\rho) $ shows properties of diffusive systems that are not possible in equilibrium. For example, in models with external fields and inhomogeneous boundary conditions, there are singularities with the meaning of  nonequilibrium phase transitions,  see \cite{BDeSGJ-LL10} and later section \ref{sec:opt}. Furthermore \eqref{eq:qpformula} predicts in general non-local rate functional, namely long range correlations in stationary nonequilibrium states.
An overview  of the MFT is given here in sections \ref{sec:dLD} and \ref{sec:Stldt}.

\section*{Macroscopic results: overview}
In the case of diffusive models of stochastic lattice gases, the continuum limit of this class of models, under a diffusive space-time scaling, is given by a nonlinear diffusion equation of the form
\begin{equation}\label{idr}
\partial_t\rho=\nabla\cdot\left(D(\rho)\nabla \rho\right)\,,
\end{equation}
where the symmetric positive definite matrix $D$ is the diffusion matrix.
When the stochastic models satisfy a suitable constraint, called \emph{gradient condition} (see \ref{eq:gradj}), it is possible to derive an explicit form for the diffusion matrix.
If this condition is violated\footnote{THIS SECTION WAS NOT IN THE VERSION SUBMITTED TO THE REFEREE BUT IT WAS ALREADY KNOWN BEFORE THE END OF THE PhD AND DISCUSSED DURING THE DEFENSE OF THE THESIS (SLIDES OF THAT DAY ARE AVAILABLE).} then, in general, the diffusion coefficient has just a variational characterization \cite{Spo91}.

For reversible and gradient models
\cite{BDeSGJ-LL06} the typical current is given by
\begin{equation}\label{tygr}
J(\rho)=-D(\rho)\nabla \rho\,,
\end{equation}
which  is indeed the classic form of the Fick's law.

In chapter \ref{ch:HSL}, we show that for a class of generalized gradient models the relation \eqref{tygr} is typically violated and conjecture that has to be replaced by
\begin{equation}\label{dstorto}
J(\rho)=-\mathcal D(\rho)\nabla\rho\,,
\end{equation}
where the diffusion matrix $\mathcal D$ is positive definite but, not necessarily, symmetric. We consider, for simplicity, the two dimensional setting,  but a similar result can be obtained in any dimension. 

\vspace{0.5cm}

In the case of out of equilibrium diffusive systems, the first computation of large deviations rate functionals was obtained for the exclusion model starting from an exact representation of the invariant measure \cite{Der07,DLS01,DLS02a}. The same result was then obtained with the dynamic variational approach of the macroscopic fluctuation theory \cite{BDeSGJ-LL01,BDeSGJ-LL02,BDeSGJ-LL15}. This macroscopic approach was then generalized to a wider class of models characterized by a constant diffusion matrix and a quadratic mobility \cite{BDeSGJ-LL15,BGL05}.

For the weakly asymmetric  exclusion process (see subsection \ref{ss:wapmod}) the computation of the rate functional \eqref{eq:qpformula} has been done in \cite{ED04} starting from an exact representation of the invariant measure and then in \cite{BGL09,BDeSGJ-LL10,BDeSGJ-LL11} using the macroscopic fluctuation theory.

For asymmetric models we have the following.
The large deviations rate functional for the invariant measure of the boundary driven asymmetric  exclusion process has been computed using an exact representation \cite{DLS02b,DLS03}. The corresponding rate functional is not local. The dynamic large deviations for asymmetric models is less understood. The rate functional for the exclusion process with periodic boundary conditions is discussed in \cite{Var04}.  The case with boundary sources has been discussed in \cite{BD06}. Once the dynamic rate functional \eqref{eq:MFTformula} is identified then it is possible to define and in some cases to compute the corresponding quasi-potential \cite{Bah10,BDeSGJ-LL15}. In particular in \cite{Bah10} the functionals in \cite{DLS02b,DLS03} and the quasi-potentials for other conservation laws with convex mobilities have been computed.

Other examples of exact computations of large deviations rate functionals for stationary non equilibrium states are for example \cite{Ber08} for a two component diffusive system and  \cite{Gab08} for the totally asymmetric exclusion process with particles of different classes on a ring.

A different way of obtaining the functionals in \cite{DLS02b,DLS03} for the asymmetric exclusion process is to compute the functional for the weakly asymmetric case and then to consider the limit for large values of the external field $ E $. This has been done in \cite{BDeSGJ-LL10,BDeSGJ-LL11,BGL09}. As reference for weakly asymmetric models see  subsections \ref{ss:wapmod} and \ref{subsec:WAm}.
With this approach the existence of Lagrangian phase transitions for finite but large external fields \cite{BDeSGJ-LL10,BDeSGJ-LL11} has also been proved. In chapters \ref{ch:LDT} and \ref{ch:TALD} we are using exactly this approach generalizing it to the case of models having constant diffusion matrix and convex quadratic mobility $ \sigma(\rho)=c_2\rho^2+ c_1\rho+c_0 $, where $ \rho  $ is the macroscopic thermodynamic parameter of the model.

In sections \ref{sec:QPweas} and \ref{sec:opt}  we will compute the large deviations rate functional for the empirical measure when the particles are distributed according to the invariant measure of weakly asymmetric models of heat conduction. Instead of discussing the general case we concentrate on three prototypical models corresponding respectively to quadratic and convex mobilities having two coinciding roots, two distinct roots and no roots. The first case corresponds microscopically to the KMP model with $ \sigma(\rho)=\rho^2 $. The second case corresponds microscopically to the dual of the KMP model (KMPd) and  with $ \sigma(\rho)= \rho(\rho+1)  $. We study the third case only macroscopically, indeed  it is difficult to build up a corresponding explicit microscopic model (see the comments after \eqref{compiti} for the technical reasons). We call KMPx this unspecified model with $ \sigma(\rho)=\rho^2+1 $. In this case the density can assume also negative values. For all these models, in the one dimensional case, we will compute the quasi-potential using the dynamic variational approach of the macroscopic fluctuation theory \cite{BDeSGJ-LL15}. The corresponding rate functionals are generically not local and depends on the value of the constant external field $E$ generating the asymmetry.

In chapter \ref{ch:TALD} we are studying  the asymptotic behavior when $E\to\pm\infty$ of the quasi-potentials. In the case of the exclusion model these limits allow to recover the large deviations rate functionals of the totally asymmetric exclusion process \cite{BGL09,BDeSGJ-LL10,BDeSGJ-LL11}. In the case of KMP and KMPd we have that the functionals become local in the limit. Moreover the dependence on one of the boundary sources disappears.  More precisely the limit functionals correspond to the large deviations rate functionals for empirical measures when the variables are distributed according to product measures: in the case of KMP the single marginal is exponential and in the KMPd case is geometric. We recover in this way the large deviations rate functionals for the  version  of the totally asymmetric KMP dynamics in subsection \ref{subsec:tkmp1dis}. The model KMPx has instead a very different behaviour. The limiting functional is again non local and has a structure very similar to the one of the simple exclusion. In this case the possibility of Lagrangian phase transitions appears for large values of the field acting in the same direction with respect to the external reservoirs.

As we already discussed in the overview of the microscopic results of part \ref{p:mt} and section \ref{sec:Amodels}, several totally asymmetric dynamics having the same weakly asymmetric diffusive scaling limit can be defined. The results we obtained in this part \ref{p:mact} jointly with that ones of section \ref{sec:Amodels} in part \ref{p:mt} suggest the natural conjecture that a large class of totally asymmetric models of heat conduction should have a local rate functional for the invariant measure.

\subsection*{Structure of the part \ref{p:mact}}
The structure of part \ref{p:mact} is as follows.
In chapter \ref{ch:HSL} we give a short overview of the scaling limit for interacting systems introduced in chapter \ref{ch:IPSmodels}, that is the main bridge between the microscopic and macroscopic descriptions. We discuss the transport coefficients and introduce macroscopically the three prototypical models we are going to study. Moreover, in section of \ref{sec:slv} of chapter \ref{ch:HSL}, we show that for a class of generalized gradient models the typical Fick's law  \eqref{tygr} in diffusive models is  violated\footnote{THIS SECTION WAS NOT IN THE VERSION SUBMITTED TO THE REFEREE BUT IT WAS ALREADY KNOWN BEFORE THE END OF THE PhD AND DISCUSSED DURING THE DEFENSE OF THE THESIS (SLIDES OF THAT DAY ARE AVAILABLE).} and conjecture that has to be replaced by
a more general relation.

In this part \ref{p:mact}, we mainly treat the energies-masses models of section \ref{sec:en models}, for which there are not totally rigorous results apart \cite{KMP82}, neither for scaling limits nor for large deviations.

Chapter \ref{ch:LDT} in the first two sections is dedicated to a general exposition of the problem of large deviations with the aim to explain what they are and how to compute them. Next section \ref{sec:dLD} enters into the problem of large deviations for interacting particles/energies-masses systems showing how to compute the dynamical large deviation rate functional for these systems. While in section \ref{sec:Stldt} is exposed the macroscopic fluctuation theory, i.e. the main theory we are using to get the results we described above in the abstract. The presentation of the MFT is  tought with a parallel with the Freidlin-Wentzell theory.
Particular emphasis is on how  the dynamic large deviations for our class of models relates   the quasi-potential and the large deviations for the invariant measure and on the Hamiltonian structure underlying the MFT.

In chapter \ref{ch:TALD}  we will compute as solution of an infinite dimensional Hamilton-Jacobi equation  the quasi-potential for any diffusive, one dimensional, weakly asymmetric, boundary driven model having a constant diffusion coefficient and a quadratic mobility. This is done in  sections \ref{sec:QPweas} and \ref{sec:opt}, where we will see that the quasi-potential in general is    a non-local functional. 
In section \ref{sec:adjhy} it is obtained  a generalized Onsager-Machlup principle  trough the problem of the adjoint dynamics for weakly asymmetric models. 

The last part of chapter \ref{ch:TALD} is dedicated to the study  of the macroscopic behaviour of  totally asymmetric models of heat conduction, i.e. when for the weak asymmetric models $E\to \pm \infty$. 
At a later time we  will compute in sections \ref{sec:TAL-} and \ref{sec:TAL+}    the limit of the quasi-potentials for large values of the external field. We are considering only the cases of convex mobilities.
To do this is we will need to use the results of  section \ref{sec:stasol}, where we study the stationary solutions of the hydrodynamic equations and the corresponding associated currents. In particular we discuss the asymptotic behaviours for large fields $E$.

The totally asymmetric limit is a local functional except in the case of single site variables that can assume all the possible real values (no roots case for the mobility). In this case the limiting quasi-potential has a non local structure similar to that of the exclusion process.

Some technical questions related to chapter  \ref{ch:TALD} are presented in sections  \ref{sec:exsist} and \ref{sec:KMPsta} of appendix \ref{a:appA}. 
In the final perspectives chapter \eqref{ch:pers} we will propose some possible developments of our results   and in addiction we will show a perturbative application of the mathematical structure of section \ref{sec:Dfreefl}.

\bigskip

\noindent \textbf{\textit{Keywords:}} Hydrodynamics scaling limits, large deviations, Freidlin-Wentzell theory, macroscopic fluctuation theory, diffusive models.


\clearpage
\pagestyle{fancy}
\fancyhf{}
\fancyhead[RE]{\sffamily \fontsize{10}{12} \selectfont \nouppercase{\leftmark}}
\fancyhead[RO]{\sffamily \fontsize{10}{12} \selectfont \nouppercase{\rightmark}}
\fancyhead[LE]{\sffamily \fontsize{10}{12} \selectfont Part \thepart: Macroscopic theory}
\fancyhead[LO]{\sffamily \fontsize{10}{12} \selectfont Part \thepart: Macroscopic theory}
\cfoot{\footnotesize \thepage}
\renewcommand{\headrulewidth}{0.01cm}

\chapter{Hydrodynamic scaling limits}\label{ch:HSL}

For deterministic model described by Newton's equation the rigorous derivation of PDEs describing the evolution of thermodynamics quantities is often a too optimistic programme, mainly  because of the lack of good ergodic property of the system. To overcome the problem two assumptions are traditionally made: or modelling the problem with a  stochastic microscopic evolution or assuming a low density of particles.  In the present framework we are interested in the first assumption and we are not having a complete rigorous point of view. For a rigorous and didactic treatment the main reference is \cite{KipLan99}. The microscopic dynamics consist of random walks of particles on a lattice $ V_N $ (see definition in section \ref{sec:CTMC}), that are  constrained to some rule expressing the local interaction, these are the so called interacting particle systems introduced by Spitzer \cite{Spit70}.

Examples of these interacting particles systems are the models presented in chapter \ref{ch:IPSmodels}, where   we introduced also the energies-masses models of section \ref{sec:en models}. 
Both for particles models and energies-masses models of chapter \ref{ch:IPSmodels}, the hydrodynamics equation is derived in a weak formulation studying the evolution of  a quantity called empirical measure (see next section \ref{sec:SL})  using the generator form of the respective process and discrete operations like in chapter \ref{ch:DEC}.
When the mesh of $ V_N $ goes to zero  for large $ N $, the proof  of an hydrodynamics scaling limit consists in proving that  the  distribution $ \bb P_N $ induced by a quantity called empirical measure (see next section \ref{sec:SL}) converges weakly to a distribution concentrated on the unique weak solution $ \rho(x,t) $ of the derived hydrodynamics equation. We write this with
$
\bb P_N\os{d}{\us{N}{\longrightarrow}}\delta_\rho.
$

The main difficulty in the proof of scaling limits is existence and conservation of local equilibrium (see the main introduction to part \ref{p:mact}).  In section \ref{sec:TC} we discuss briefly a mathematical formulation of it.
For our microscopical models, the time scale $ \theta_N $ to find  a non trivial behaviour in the derivation of the hydrodynamics PDEs  is $ \theta_N=N^2  $. 
Calculations in the proof of an hydrodynamics scaling limit are easier when  the stochastic system is gradient, see subsections \ref{subsec:Ipc} and \ref{subsec:Ie-mc}.

\section{Scaling limits}\label{sec:SL}

We can keep  in mind the KMP model, this is gradient (formula \eqref{eq:KMPEic}) and the hydrodynamic behaviour is relatively well understood, in  \cite{BGL05} they give a non-rigours derivation of its scaling limit assuming the local equilibrium.  We point out that KMP model is just a prototype for the discussion here, but the same scheme applies both to others particles and energies-masses models where the  diffusion coefficient and the mobility in \eqref{eq:drc} and \eqref{eq:drcE} are different. We consider the one dimensional case with $\Lambda=(0,1]$ and $ V_N:=\Lambda\cap\frac{1}{N}\mathbb{Z}^n $. In the symmetric case \eqref{eq:KMPgeneralized} the model with rates $ \Gamma^\xi_{x,y}(dj) $ is diffusive and the natural scaling of the system is obtained considering a lattice of mesh $\frac 1N$ and rescaling time by a factor $N^2$. This is done simply multiplying by $N^2$ the rates of jump (for notational convenience we will multiply by a factor $2N^2$). Often the accelerated process is denoted with $ \xi_{\,N^2t} $, here we don't adopt this notation but we use $ \xi_{\,t} $ remembering we are speaking about the accelerated process. The observable that describe macroscopically the evolution of the mass of the system is the so called empirical measure. This is a positive measure on $\Lambda$ associated to any fixed microscopic configuration $\xi$. It is defined as a convex combination of delta measures as
\begin{equation}\label{eq:empmisxi}
\pi_N(\xi):=\frac 1N\sum_{x\in V_N}\xi(x)\delta_x\,.
\end{equation}
Physically this represents a mass density or an energy density along the interpretation of the model. Integrating a continuous function $f:\Lambda\to \mathbb R$ with respect to $\pi_N(\xi)$ we get
\begin{equation*}\label{eq:defpif}
\int_\Lambda f\,d\pi_N(\xi)=\frac 1N\sum_{x\in V_N}f(x)\xi(x)\,.
\end{equation*}
In the hydrodynamic scaling limit the empirical measure, that for any finite $N$ is atomic and random, becomes deterministic and absolutely continuous for suitable initial conditions $\xi_0$ that are associated to a given density profile $\gamma(x)dx$ in the sense that in probability
\begin{equation}\label{eq:asc}
\lim_{N\to +\infty}\int_\Lambda f\, d\pi_N(\xi_0)=\int_\Lambda f(x)\gamma(x)dx.
\end{equation}
Let $ D([0,T];\Sigma_N) $ be  the space of right continuous with left limits paths from $  [0,T] $ to $ \Sigma_N $  with the Skorokhod topology and $  P_N^{\gamma} $ the distribution of the Markov chain of the the considered energies-masses/particles interacting models  with initial condition associated to $ \gamma $ as in \ref{eq:asc}, this is a measure on $ D([0,T];\Sigma_N) $ . While let $ D([0,T];\mc M^1(\Lambda)) $ be  the space of trajectories from $ [0,T] $ to the space of positive measure $ \mc M^1(\Lambda) $ with the Skorokhod topology and $ \bb P_N^\gamma:= P^{\gamma}_N\circ\pi_N^{-1} $ the measure induced on  $ D([0,T];\mc M^1(\Lambda)) $ by the empirical measure.
We have that $\pi_N(\xi_t)$ is associated to the density profile $\rho(x,t)dx$ where $\rho$ is the weak solution to the heat equation
with initial condition $\gamma$, i.e.
\begin{equation}\label{eq:wconKMP}
\bb P^\gamma_N\os{d}{\us{N}{\longrightarrow}}\delta_\rho,
\end{equation} 
where $ \bb P^\gamma_N $ is the distribution of the empirical measure and  $ \delta_{\rho} $ is the distribution concentrated on the unique weak solution of the heat equation with initial condition $ \gamma $. The boundary conditions are fixed by the interactions with the external sources. We consider the case when the boundary dynamics is also accelerated by $N^2$ and the final effect of this fast interaction is that the values of the densities at the boundaries are fixed by the external sources \cite{ELS96}.  We have then that $\rho$ is the unique weak solution associated to  a Cauchy problem with Dirichelet boundary condition like
\begin{equation}\label{eq:drc}
\left\{
\begin{array}{l}
\partial_t\rho=\Delta \rho\\
\rho(x,0)=\gamma(x)\\
\rho(0,t)=\rho_-\\
\rho(1,t)=\rho_+\,.
\end{array}
\right.
\end{equation}
Without loss of generality we will always consider the case $\rho_-\leq \rho_+$.
This is a space time law of large numbers and the corresponding fluctuations can be described by a large deviations principle \cite{BGL05}. Given a space and time dependent density profile $\rho(x,t)dx$, the probability that the empirical measure will be in a suitable neighborhood of it, it is exponential unlikely with a corresponding rate functional like \eqref{eq:tilde(rho)} that we call
dynamic large deviations rate functional. This is the main ingredient for the dynamic variational study of stationary non equilibrium states of the macroscopic fluctuation theory \cite{BDeSGJ-LL15}.
This problem is going to be discussed later in chapters \ref{ch:LDT} and in \ref{ch:TALD}, before we need to discuss the hydrodynamic behaviour of weakly asymmetric models as the ones introduced in subsection \ref{subsec:WAm}. So the bulk part of the generator is obtained as a sum of possibly time dependent contributions like \eqref{eq:bulkxyE} multiplied by $N^2$. In particular we consider a model with rates  $ \Gamma^{\xi,\bb F}_{x,y}(dj) $ determined by \eqref{eq:KMPpertrate} or more generally
having an expansion like \eqref{eq:KMPEic} where the external field $E$ is however substituted by a space and time dependent vector field $\mathbb F$  obtained by a discretization of a smooth vector field on $\Lambda$ like \eqref{eq:disF}.
The hydrodynamic behavior of this model is similar to the symmetric one and the external field appears macroscopically with a new term. The hydrodynamic equation is
\begin{equation}\label{eq:drcE}
\left\{
\begin{array}{l}
\partial_t\rho(x,t)=\Delta \rho(x,t)-\div \left(\rho^2(x,t) F(x,t)\right)\,,\\
\rho(x,0)=\gamma(x)\,,\\
\rho(0,t)=\rho_-\,,\\
\rho(1,t)=\rho_+\,
\end{array}
\right.
\end{equation}
and mutatis mutandis the probability distribution of the empirical measure $ \bb P^\gamma_N $ will have an equivalent formulation of \eqref{eq:wconKMP} for \eqref{eq:drcE}. Of course for others energies-masses model the diffusion and transport coefficients will be in general different, for example for \eqref{eq:KMPdrate} the mobility is  $ \rho^2+\rho $ instead of $ \rho^2 $. More generally, in chapter \ref{ch:TALD} we are going to study the problem of large deviations for all diffusive models with quadratic mobility, namely we consider \eqref{eq:drcE} where the first equation becomes 
\begin{equation}\label{eq:mob2^}
\partial_t\rho(x,t)=\Delta \rho(x,t)-\div \left(\sigma(\rho(x,t)) F(x,t)\right),
\end{equation}
with $ \sigma(\rho)=c_2\rho^2+c_1\rho+c_0 $  general real polynomial of second order. The possible physical situations along the sign of the coefficient $ c_0,c_1,c_2 $
can be resumed into four prototypes containing all the phenomenology, that is
\begin{equation}\label{eq:4prot}
\sigma(\rho)=\left\{
\begin{array}{ll}
\rho(1-\rho) & \textrm{SEP}\\
\rho^2 & \textrm{KMP}\\
\rho(\rho+1) & \textrm{KMPd}\\
\rho^2+1 & \textrm{KMPx}\,,
\end{array}
\right.
\end{equation}
with  KMPx we denote  an unknown hypothetical microscopical model with this hydrodynamic equation, see the discussion after \eqref{compiti} about the technical difficulty to find such model. Since for physical reason the mobility $ \chi(\rho)>0 $, in this model  space profiles that change sign are possible solutions because in $ \rho=0 $ the mobility is  strictly positive. The physical is interpretation is that we have positive and negative charges that are moving.   We can forget about the simple exclusion process in what follows since it is widely treated in  \cite{BDeSGJ-LL03,BGL09}. 

\section{Transport coefficients}\label{sec:TC}

In this section we give a derivation scheme of the hydrodynamic equation \eqref{eq:mob2^} assuming local equilibrium.
The general form of the hydrodynamic equation associated to weakly asymmetric diffusive stochastic particle systems is
\begin{equation}\label{hydrog}
\partial_t\rho=\div\left(D(\rho)\nabla \rho-\sigma(\rho)F\right)\,.
\end{equation}
The symmetric and positive definite matrix $D$ is called the \e{diffusion matrix} while the symmetric and positive definite matrix $\sigma$ is called the \e{mobility}.
For all the models that we are discussing the diffusion matrix coincides with the identity matrix while the mobility
is a multiple of the identity matrix $\sigma(\rho)\mathbb I$ (we are calling $\sigma$ both the matrix and the scalar value on the diagonal). The transport coefficients can be explicitly computed for gradient models for which at  equilibrium the invariant measure is known.

It is convenient to write the hydrodynamic equation \eqref{hydrog} as a conservation law $\partial_t\rho+\nabla\cdot J_F(\rho)=0$ where
\begin{equation}\label{jeff}
J_F(\rho):=-\nabla \rho+\sigma(\rho)F
\end{equation}
is the typical current observed in correspondence to the density profile $\rho$.

We briefly illustrate the general structure of the proof of the hydrodynamic limit for gradient reversible models that
allows to identify and compute exactly the transport coefficients \cite{KipLan99,Spo91}. This argument is the bridge between the microscopic and the macroscopic description of a system and allows to identify the hydrodynamic equations associated to the microscopic models.

The starting point for the hydrodynamic description of the system is the discrete continuity equation that is
\begin{equation}\label{disc-cont}
\xi_t(x)-\xi_0(x)=-\div \mathcal J_t(x)\,,
\end{equation}
where $\mathcal J_t$ has been defined in subsection \ref{subsec:Ie-mc} and $\div$ denotes the discrete divergence for a discrete vector field $ \phi $ defined in \eqref{eq:cod1-0}, that we write  as $ \div\phi(x):=\us{y\sim x}{\sum}\phi(x,y) $   summing  on the nearest neighbours $ y\sim x $ of $ x $.
Using \eqref{eq:massmart} we can rewrite \eqref{disc-cont} as
\begin{equation}\label{rev}
\xi_t(x)-\xi_0(x)= -N^2\int_0^t\div j_{\xi_s}(x)\,ds+ M_t(x)\,,
\end{equation}
where $M_t(x)$ is a martingale term obtained by summing some martingales of the type \eqref{eq:massmart}, as we did for particles systems in subsection \ref{subsec:Ipc}.

Since we are interested only on the transport coefficients we consider the model defined on the 1-dimensional continuous torus $\Lambda=(0,1]$ with periodic boundary conditions.
Multiplying equalities \eqref{rev} by a test function $\psi$, dividing by $N$ and summing over $x$  we obtain
\begin{equation}\label{gate}
\int_{[0,1]} \psi \,d\pi_N(\xi_t)-\int_{[0,1]} \psi\, d\pi_N(\xi_0)=-N\int_0^t\sum_{x\in V_N}\div j_{\xi_s}(x)\,\psi\left(x\right)\,ds+o(1)\,.
\end{equation}
The infinitesimal term comes from the martingale terms and can be shown
to be negligible in the limit of large $N$ \cite{KipLan99, Spo91}, the intuitive reason rely on the probabilistic nature of the martingale giving an analogous of \eqref{eq:asJ=asj}. Now we prefer to rewrite  the gradient condition \eqref{eq:massgra} respect to $ h(\xi)=\frac {\xi(0)} 2 $, that is \begin{equation}\label{eq:newgraKMP}
j_{\xi}(x,y)=\tau_x h(\xi)-\tau_y h(\xi),
\end{equation}
using this gradient condition \eqref{eq:newgraKMP} and  performing a double discrete integration by part, up to the infinitesimal term one has that the right hand side of \eqref{gate} is
\begin{equation}\label{aibp}
\frac{1}{N}\sum_{x\in V_N}\int_0^t \tau_xh(\xi_{{s}})\left[N^2\left(\psi\left(x+\frac{1}{N}\right)+\psi\left(x-\frac{1}{N}\right)-
2\psi\Big(x\Big)\right)\right]\,ds\,.\nonumber \\
\end{equation}
Considering a $C^2$ test function $\psi$ the term inside squared parenthesis in the last term of \eqref{aibp}
coincides with $\Delta\psi\left(x\right)$ up to an uniformly infinitesimal term.

At this point the main issue in proving hydrodynamic behavior is to prove the validity of a local equilibrium property. Let us define
\begin{equation}\label{local}
A(\rho)=\mathbb E_{\mu_{N}^{\lambda[\rho]}}\left(h (\xi)\right)\,,
\end{equation}
where  $\mu_N^\lambda$ is the product of measures \eqref{eq:mu_lambda} associated to the chemical potential $\lambda$ and we recall that
$\lambda[\rho]$ is the chemical potential associated to the density $\rho$ (see the discussion after \eqref{eq:mu_lambda}).
The local equilibrium property is explicitly stated through a replacement lemma that states that (in probability)
\begin{equation}\label{eq:repla}
\frac{1}{N}\sum_{x\in V_N}\int_0^t \tau_xh(\xi_{{s}})\Delta\psi\left(x\right)\,ds\simeq
\frac{1}{N}\sum_{x\in V_N}\int_0^t A\left(\frac{\int_{B_{x}}d\pi_N(\xi_s)}{|B_{x}|}\right)\Delta\psi\left(x\right)\,ds
\end{equation}
where $B_{x}$ is a microscopically large but macroscopically small volume around
the point $x\in V_N$. For a precise formulation of \eqref{eq:repla} see lemma 1.10 and corollary 1.3 respectively  in chapter 5 and  in chapter 6 of \cite{KipLan99}.
This allows to write, up to infinitesimal corrections, equation \eqref{gate} in terms only of the empirical measure.
Substituting the r.h.s. of \eqref{eq:repla} in the place of the r.h.s. of \eqref{gate}, we obtain that in the limit of large $N$ the empirical measure $\pi_N(\eta_t)$ converges in weak sense to $\rho(x,t)dx$ satisfying for any $ C^2(\Lambda) $ test function $\psi$
\begin{equation}\label{widro}
\int_{0}^1 \psi(x)\rho(x,t)\,dx-\int_{0}^1\psi(x)\rho(x,0)\,dx=\int_0^t ds\int_{0}^1 A(\rho(x,s))\Delta\psi(x)\,dx\,.
\end{equation}
Equation \eqref{widro} is a weak form of the hydrodynamic equation \eqref{hydrog} with $F=0$
and having a diagonal diffusion matrix with each term in the diagonal equal to
\begin{equation}\label{diff}
D(\rho)= \frac{ d A(\rho)}{d \rho}\,.
\end{equation}
For a mathematical discussion of this issue see \cite{KipLan99} chapter 5. For all the models that we are discussing we have that $h(\xi)=\frac{\xi(0)}{2}$ so that $A(\rho)=\frac{\rho}{2}$. For notational convenience in order to have an unitary diffusion matrix we multiply all the rate of transition by a factor of 2 and correspondingly the diffusion matrix is the identity matrix.

\smallskip

We show now a computation that allows to determine the mobility of the models from the microscopic dynamics.  Let us consider weakly asymmetric models subject to an external field obtained by the discretization \eqref{eq:disF} of a smooth vector field $F$.
Consider the time window $[0,t]$ and we still speed up the process by a $N^2$ factor. The scalar product \eqref{sc} of the flow of mass in this time window with  the discretization $ \bb H $ of a smooth vector field $ H$ is given by
\begin{equation}\label{total-work}
\frac {1}{N^d}\sum_{x\sim y}\mathcal J_t(x,y)\mathbb H (x,y)\,.
\end{equation}
In formula \eqref{total-work} the sum is over unordered nearest neighbor sites. By the antisymmetry of the two vector fields there is no ambiguity in this definition. The factor $N^{-d}$ is due to the fact that the scaling limit normalizes the mass by this factor.
Using \eqref{eq:massmart} and \eqref{eq:KMPEic} we can write \eqref{total-work} up to a neglecting martingale term for large $ N $ as
\begin{equation}\label{work-up}
N^{2-d}\int_0^t\sum_{x\sim y}j^{\mathbb F}_{\xi_s}(x,y)\mathbb H(x,y)\,ds\,.
\end{equation}
For simplicity we consider the one dimensional KMP model with a macroscopic constant external field $F$. Microscopically this corresponds to consider $E=\frac{F}{N}$ in \eqref{eq:KMPEic}. The value $\frac FN$ is indeed the discretized value for a lattice of size $1/N$ corresponding to a constant macroscopic external field $F$.
We introduce the function
\begin{equation}\label{func-g}
g(\xi)=\frac 16\big(\xi(0)^2+\xi(1/N)^2-\xi(0)\xi(1/N)\big)
\end{equation}
that, suitably shifted, multiplies $E$ in the right hand side of \eqref{eq:KMPEic}. In the case of KMPd we have instead to use formula \eqref{eq:KMPdEic}.
We can write \eqref{work-up} in $d=1$ as
\begin{equation}
N\int_0^t\sum_{x\in V_N}j_{\xi_s}\left(x,x+\frac 1N\right)\mathbb H\left(x,x+\frac 1N\right)ds+ F\int_0^t\sum_{x\in V_N}\tau_xg(\xi_s)\mathbb H\left(x,x+\frac 1N\right)ds\,.
\label{primo-passo}
\end{equation}
Since the current $j_{\xi_t}(x,y)$ is gradient, recalling \eqref{eq:newgraKMP}, formula \eqref{primo-passo} becomes after a discrete integration by parts
\begin{equation}\label{secondo-passo}
N\int_0^t\sum_{x\in V_N}\tau_x h(\xi_s)\div \mathbb H(x)\,ds+ F \int_0^t\sum_{x\in V_N}\tau_xg(\xi_s)\mathbb H\left(x,x+\frac 1N \right)\,ds\,.
\end{equation}
From definition \eqref{eq:disF} for $ H $ smooth vector field we find $ \bb H(x,x+\frac 1N)\simeq \frac 1N H(x)+\frac {1}{N^2}\p_x H(x)+o(\frac 1{N^2}) $ and $ \bb H(x,x+\frac 1N)+\bb H(x, x-\frac 1N)\simeq \frac 1{N^2}\,\div H(x) $, therefore
 $N^2\div \mathbb H(x)=\div H(x)$ up to uniform  infinitesimal terms. Applying 
the replacement lemma \eqref{eq:repla} when $ N $ is diverging, with high probability  \eqref{secondo-passo} converges to
\begin{equation}\label{cenone}
\int_0^t ds\int_{0}^1dx\,\left[A(\rho(x,s))\div H(x)+F\sigma(\rho(x,s))H(x)\right]\,,
\end{equation}
where
\begin{equation}\label{finalmente}
\sigma(\rho)=\mathbb E_{\mu^{\lambda[\rho]}_N}\left[g(\eta)\right]\,.
\end{equation}
Formula \eqref{cenone} is a weak form of $\int_0^tds\int_0^1J_F(\rho)\cdot H\,dx$ with
\begin{equation}\label{tipF}
J_F(\rho)=-D(\rho)\nabla\rho +\sigma(\rho)F\,.
\end{equation}
This is the typical current associated to a density profile $\rho$ in presence of an external field $F$.
Recall that for notational convenience we are multiplying the rates by a factor of 2 so that formula \eqref{finalmente}
for ours prototype models gives \eqref{eq:4prot}.

\smallskip

A general identity holding for diffusive systems is the Einstein relation between the transport coefficients and the density of free energy \cite{BDeSGJ-LL15,Spo91}
\begin{equation}\label{Einzwei}
D(\rho)=\sigma(\rho)f''(\rho)\,.
\end{equation}
In \eqref{Einzwei} $f$ is the density of free energy that will be introduced and discussed in subsection \ref{trc}. We have that $f'(\rho)=\lambda[\rho]$ where we recall $\lambda[\cdot]$ is the chemical potential as a function of the density introduced in subsection \ref{subsec:sta}. The Einstein relation can be then written equivalently as
\begin{equation}\label{llaa}
D(\rho)=\sigma(\rho)\lambda'[\rho]\,.
\end{equation}

\bigskip

Here the hydrodynamics was derived with periodic boundary conditions but we remind that it is   still the same when a  boundary driven  version of the system is considered, see \cite{ELS96}. 

\section{Scaling limit of an exclusion process with vorticity}\label{sec:slv}

\bigskip

THIS SECTION WAS NOT IN THE VERSION SUBMITTED TO THE REFEREE BUT IT WAS ALREADY KNOWN BEFORE THE END OF THE PhD AND DISCUSSED DURING THE DEFENSE OF THE THESIS (SLIDES OF THAT DAY ARE AVAILABLE).

\bigskip

Let's consider as lattice the discrete bidimensional torus with vertexes $ V_N=\mathbb{T}^2_N $ and with mesh of size $ \epsilon=1/N $, see definition in section \ref{sec:CTMC}. In this section we discuss heuristically the hydrodynamics limit for exclusion models where the instantaneous current \eqref{eq:istcur} has a decomposition \eqref{hodgefun2} in the form
\begin{equation}\label{eq:localHod2}
j_\eta(x,y)=[\tau_y h(\eta)-\tau_x h(\eta)]+[\tau_{f^+(x,y)}g(\eta)-\tau_{f^-(x,y)}g(\eta)]=j^h_\eta(x,y)+j^g_\eta(x,y),
\end{equation}
with $ h,g $ local function. We are defining $ j^h_\eta(x,y):=\tau_y h(\eta)-\tau_x h(\eta) $ and $ j^g_\eta(x,y):=\tau_{f^+(x,y)}g(\eta)-\tau_{f^-(x,y)}g(\eta) $. Two examples of this class are in subsection \ref{ss:vortmodels}. For example taking an exclusion process with rates  given by
\begin{equation}\label{eq:vortrate}
c_{x,y}(\eta)=\eta(x)(1-\eta(y))+\eta(x)[\tau_{f^+(x,y)}g(\eta)-\tau_{f^-(x,y)}g(\eta)],
\end{equation}
we have  $ j_\eta(x,y)=c_{x,y}(\eta)-c_{y,x}(\eta) $ as in \eqref{eq:localHod2} with $ h(\eta)=-\eta(0) $, note that the first example of subsection \ref{ss:vortmodels} is of this form.

Models with $ j_\eta(x,y) $ as in \eqref{eq:localHod2} can be thought as a generalization of the gradient case $ j_\eta(x,y)=[\tau_y h(\eta)-\tau_x h(\eta)] $, indeed the current is decomposed in an usual gradient part plus an orthogonal gradient part (discrete bidimensional curl).  Because of the presence of this discrete bidimensional curl we use the terminology of "exclusion process with vorticity". The physical phenomena is the one we are describing in next lines. 

When the rates satisfies \eqref{eq:localHod2}, remarkably  the hydrodynamics for the empirical measure works exactly as in section \ref{sec:TC} because 
\begin{equation}\label{eq:divg}
\div j_\eta^g(x)=0,\,\,\,\forall\,\,x\in V_N,
\end{equation}
that is there is a part of the dynamics that doesn't give any macroscopic effect to the hydrodynamics. To observe macroscopically the effect of the discrete curl we have to consider the scaling limits of the current flow  $ J_t(x,y) $ of formula \eqref{eq:currver}, that is what we did in section \ref{sec:TC} after formula \eqref{total-work} to derive the macroscopic current $ J(\rho) $ that will appear in the hydrodynamics
\begin{equation}\label{eq:hydrowJ}
\partial_t\rho = \div (-J(\rho)).
\end{equation}
Proceeding analogously  for the present case, in place of \eqref{secondo-passo} we get  the following
\begin{equation}\label{darimpiaz}
N^{2-d}\int_0^t\left[\sum_{x\in \Lambda_N}\tau_x h(\eta)(x)\div \mathbb H(x) +\sum_{\mathfrak f\in \mathcal F_N}\tau_{\mathfrak f}g(\eta_s)\sum_{(x,y)\in f^\circlearrowleft}\mathbb{H}(x,y)\right]\,,
\end{equation}
where the infinitesimal terms are uniform.
By Taylor expansion, for a $C^1$ vector field $H$ we have
$$
\left\{
\begin{array}{l}
N^2\div \mathbb{H}(x)=\div H(x)+o(1/N)\,,\\
N^2\sum_{(x,y)\in f^\circlearrowleft}\mathbb{H}(x,y)=\nabla^{\perp}\cdot H(z)+o(1/N)\,.
\end{array}
\right.
$$
In the above formula $z$ is any point belonging to the face
while given a $C^1$ vector field $H=(H_1,H_2)$ we used the notation
$\nabla^\perp\cdot H(z):=-\partial_yH_1(z)+\partial_xH_2(z)$.
When $ N $ is diverging, with high probability  \eqref{darimpiaz} converges to
\begin{equation}\label{wHDJ}
\int_0^t ds\int_\Lambda dx\,\left[d(\rho(x,s))\div H(x)+d^\perp(\rho(x,s))\nabla^\perp\cdot H(x)\right]\,.
\end{equation}
Namely we applied the 
the replacement lemmas $ d(\rho)=\mathbb E_{\mu^{ss}_\rho}[h(\eta)] $ and $ d^\perp(\rho)=\mathbb E_{\mu^{ss}_\rho}[g(\eta)]  $, where $ \mu^{ss}_\rho $ is the stationary measure with  profile $ \rho $.
Formula \eqref{wHDJ} is a weak form of $\int_0^tds\int_\Lambda J(\rho)\cdot H\,dx$ with
\begin{equation}\label{tipJ}
J(\rho)=-\nabla d(\rho)-\nabla^\perp d^\perp(\rho)=-D(\rho)\nabla\rho - D^\perp(\rho)\nabla^\perp\rho,
\end{equation}
where $ \nabla^\perp $ is the orthogonal gradient defined as $ \nabla^\perp f=(-\partial_y f,\partial_x f) $.  As we expected from the microscopic argument \eqref{eq:divg} we have that  \eqref{eq:hydrowJ} gives
\begin{equation*}\label{key}
\partial_t\rho=\div (-D(\rho)\nabla\rho - D^\perp(\rho)\nabla^\perp\rho)=\div(-D(\rho)\nabla\rho),
\end{equation*}
that it the hydrodynamics is left unchanged respect to the usual gradient case.
Hence the macroscopic current \eqref{tipJ} can be rewritten 
as
\begin{equation}\label{tipJ2}
J(\rho)=-\left[D(\rho)\begin{pmatrix} 1 & 0\\ 0 & 1 \end{pmatrix}-A(\rho)\begin{pmatrix} 0 & -1\\ 1 & 0 \end{pmatrix}\right]\nabla\rho.
\end{equation}

This formula allows us to underline in the best way the difference respect to the the typical picture  for the macroscopic current $ J(\rho) $ in diffusive interacting particles models, which is  given by the Fick's law
\begin{equation}\label{eq:typpic}
J(\rho)=-D(\rho)\nabla\rho,
\end{equation}
where $ D(\rho) $ is a  positive symmetric diffusion matrix. From the present discussion for the generalized gradient models \eqref{eq:localHod2} this picture is replaced by the following general formulation
\begin{equation}\label{eq:JasymD}
J(\rho)=-\mathcal D(\rho)\nabla\rho, 
\end{equation}
where the diffusion matrix $ \mathcal D (\rho)$ is  positive but not necessarily symmetric. Its symmetric part is $ D(\rho) $ and the asymmetric one is $ D^\perp (\rho) $.

\chapter[Large Deviations in IPS]{Large deviations theory in IPS}\label{ch:LDT}

Large deviation theory is the part of probability  describing
"rare" events when, for examples, a sum of random variables deviates from the values prescribed 
by limit theorems. 
Large deviations are studied in different contexts of probability as Markov chain,
hydrodynamics limit, diffusion, polymers, statistics, information theory, finance, etc. 
Classic reference for large deviations are \cite{deH00,DZ98,Var84}. First we give an idea
of what is a large deviations problem, then we set
rigorously  what is a large deviation principle and  we sketch a typical scheme 
to prove a large deviation principle for some models. Later on we discuss in generality the large deviations problem for the interacting energies-masses/particles systems of chapter \ref{ch:IPSmodels} introducing the macroscopic fluctuation theory. Especially we will give an overview of the stationary large deviations focusing on the Hamiltonian structure on which the macroscopic fluctuation theory is founded by means of an infinite dimensional Hamilton-Jacobi equation.

Some material in the first two introductory sections, in particular the upper and lower bounds   of the mentioned scheme comes from \cite{Ter15}.

\vspace{1cm}

The weak law of large numbers states that for a sequence of i.i.d. random variables we have convergence in probability, i.e.

\begin{theorem}[Weak law of large numbers]
Let $ X_1,X_2, \dots $ be i.i.d random variables such that $ \mathbb{E}|X_1|<+\infty $. Then, for any $ \varepsilon>0 $, holds
\begin{equation}\label{eq:weaklaw}
\mathbb{P}\left[\,\,\left| \frac{X_1+\dots+X_N}{N}-\mathbb{E}(X_1) \right|>\epsilon\,\,\right]\us{N\to\infty}{\longrightarrow}0.
\end{equation}
\end{theorem}

This theorem guarantees convergence but it doesn't tell anything about the speed of convergence when $ N $ is diverging, that is in other terms  how much unlikely is this unlikely event. A large deviation programme is addressing these questions. 
The issue can be reformulated as follows. Denote $ S_N=X_1+\dots+X_N $, where $ X_i $ are the i.i.d. random variables as above. Let $ \bb{P}_N  $  be the measure  induced on $ \bb{R} $ by the random variable $ \frac{S_N}{N} $, from the convergence in probability in  \eqref{eq:weaklaw} the random variable $ \frac{S_N}{N} $ converges weakly to $ \bb E(X_1) $, then from Portmanteau theorem for any measurable $ A\subset \bb R $ such that $ \bb E (X_1)\notin \p A $   we have
\begin{equation*}
\bb{P}_N(A)\us{N\to\infty}{\longrightarrow}
\left\{
\begin{array}{ll}
 1,\quad {\rm if } \quad \bb E(X_1)\in A \\
 0,\quad {\rm if } \quad \bb E(X_1)\notin A \\
\end{array} .
\right.
\end{equation*}
We expect that under suitable assumptions this happens in a form that can be written as
\begin{equation}\label{eq:V(A)}
\bb{P}_N(A)\simeq e^{-\theta_N V(A)},
\end{equation} 
where $ \theta_N $ is the speed of the large deviation and $ V(A) $ is defined by a positive functional $V(\rho)$ such that $V(A)=\us{\rho\in A}{\inf} V(\rho)$, $ V(E(X_1))=0 $ and $ V(\rho\neq E(X_1))> 0 $. In this way $ V(A)=0 $ if $ \bb E(X_1)\in A $ and $ \bb P_N(A)\simeq 1 $, otherwise in general $ V(A)>0 $ under suitable conditions\footnote{If  $ \bb E(X_1)\in\p A $ it can happen that $ V(A)=0 $ also when $ \bb E(X_1)\notin A $.} . Moreover if $ \bb P(A)>0 $ the "rare event is the least unlikely of all unlikely events". In next section we will try to give more motivations to this infimum in the definition of  $ V(A) $. 

\section{A Large Deviation Principle}

Let $ (\Omega,d) $ be a complete separable metric space and $ \bb P_N $ be a sequence of probability measures, all of them defined on the Borel sigma algebra of $ (\Omega,d) $.
\begin{definition}\label{def:LDprin}
We say that $ \{P_N\} $ satisfies a \e{Large Deviation Principle} (LDP) with
\e{large deviation rate functional} $V:\Omega \to [0, \infty]$ and speed $ \theta_N \in \bb R^+ $ such that $ \theta_N\to +\infty $ if
\begin{itemize}
\item[a)]  $ 0 \leq V(\rho) \leq \infty $, for all $ \rho\in\Omega $;
\item[b)] $ V $ is lower semi-continuous;
\item [c)] for any $ c<\infty $, the set $ \{\rho\in\Omega:\, V(\rho)\leq c\} $ is compact in $ (\Omega, d) $;
\item [d)] for any closed set $ C\subset \Omega $
\begin{equation*}
\us{N\to\infty}{\lim\sup}\frac{1}{\theta_N}\log \bb P_N(C)\leq -\us{\rho\in C}{\inf}\,V(\rho)\,\,;
\end{equation*}
\item [e)] for any open set $ A\subset \Omega $.
\begin{equation*}
\us{N\to\infty}{\lim\inf}\frac{1}{\theta_N}\log \bb P_N(A)\geq -\us{\rho\in A}{\inf}\,V(\rho)\,\,.
\end{equation*}
\end{itemize}
\end{definition}

Commonly $ \bb P_N \os{d}{\to} \delta_{\bar\rho} $, $ V(\bar\rho) = 0 $ and $ V(\rho) > 0 $ if $ \rho\neq \bar\rho $, implying
the strong law of large numbers. The speed $ \theta_N $ can be for example $ N $ or $ N^2 $.
Condition
c) implies b), when c) is true case many authors say that $ V $ is a \e{good} rate function but conditions a), b), d) and e) are enough to say to have proved a large deviation principle with rate functional $ V $.
\begin{remark}
Given a set $ A\subset\Omega $, we denote here and in the rest $V(A):=\us{\rho\in A}{\inf}\,V(\rho)$. We indicate with $  \mathring{A}$ and $\bar{A}$  respectively the interior of $ A $ and the closure of $ A $. Observe that if a set $ A $ is such that
\begin{equation}\label{eq:Vcon}
\us{\rho\in \mathring{A}}{\inf}\,V(\rho)= \us{\rho\in A}{\inf}\,V(\rho)=\us{\rho\in \upbar{A}}{\inf}\,V(\rho)
\end{equation}
then
\begin{equation*}\label{key}
\us{N\to\infty}{\lim}\frac{1}{\theta_N}\log \bb P_N(A)= -V(A).
\end{equation*}
\end{remark}
Now let's try to justify the infimum in the rate functional. Consider two disjoint sets $ A_1 $ and $ A_2 $ satisfying \eqref{eq:Vcon}, then \footnote{In the second equality of \eqref{eq:throwaway} we used the fact that  for $ a_N\to +\infty $ and $ b_N,c_N $ two sequence of positive numbers is valid the equality $ \us{N\to\infty}{\limsup}\frac 1 {a_N} \log (b_N+c_N)$$=\max \left\{\us{N\to\infty}{\limsup}\frac 1 {a_N} \log b_N, \us{N\to\infty}{\limsup}\frac 1 {a_N} \log c_N \right\} $.} 

\begin{equation}\label{eq:throwaway}
\begin{array}{ll}
&V(A_1\cup A_2)=-\us{N\to \infty}{\lim}\frac{1}{\theta_N}\log \bb P_N(A_1\cup A_2)\\
&=-\max \left\{\us{N\to \infty}{\lim} \frac{1}{\theta_N}\log\bb P_N(A_1), \us{N\to \infty}{\lim} \frac{1}{\theta_N}\log\bb P_N(A_2) \right\}\\
&=\min\{V(A_1),V(A_2)\}.
\end{array}
\end{equation}
Imaging to iterate this inside the subsets of $ A $ and  we guess $ V(A) $ will be of the form 
\begin{equation*}\label{key}
V(A)=\us{\rho\in A}{\inf}\,V(\rho),
\end{equation*}
so we have the statement (already anticipated)  that "any large deviation is done in the least unlikely of all unlikely ways". Indeed because of the minus in front of the infimum we have that  where $ V(\rho) $ is smaller it means that $ \rho $ is less improbable. 

\section[Large devitions scheme]{Steps in a large deviation proof: a scheme}\label{sec:LD scheme}

Here we describe a general simplified scheme that shows the steps to prove  large deviations for some generic model.

 Let $ \bb P_N $  be a sequence of measures on $ (\Omega,d) $ as before and such that 
\begin{equation*}\label{key}
\bb P_N\os{d}{\to}\delta_{\bar\rho},
\end{equation*}
Given a $ \rho\neq\bar\rho $ we want to find  a rate functional estimating the probability of a small ball $ B_\veps(\rho)  $ of radius $ \varepsilon>0 $ around $ \rho $ like
\begin{equation}\label{eq:tilde(rho)}
\bb P_N(B_\veps(\rho) ) \simeq e^{-\theta_N V(B_\veps(\rho) )}\simeq e^{-\theta_N V(\rho)}.
\end{equation}
First we want to perturb  $ \bb P_N $ in a way to "make typical" a fixed deviation $ \rho$. This is done introducing  some perturbation $ H $ such that for a fixed $ \rho  $  the perturbed measure  $ \bb P_N^H $ satisfies
\begin{equation}\label{eq: P^H}
\bb P^H_N\os{d}{\to}\delta_\rho 
\end{equation}
 and characterizing  the Radon-Nykodym derivative in the following way
\begin{equation}\label{eq:RN}
e^{\theta_N J_H(\rho')}=\frac{d\bb P^H_N}{d\bb P_N}(\rho').
\end{equation}
The factor $ \theta_N $ is the speed of the large deviation.
The perturbation $ H $  "making typical"  $ \rho $  is chosen as $ H=H[\rho] $, hence the functional $ J_H(\rho) $ in \eqref{eq:RN} becomes depending on the only $ \rho $.
For a perturbation $ \hat H $ of the probability $ \bb P_N $ rewriting $ d\bb P_N $ as
\begin{equation*}\label{key}
d\bb P_N= \frac {d\bb P_N} {d\bb P^{\hat H}_N}d\bb P^{\hat H}_N=e^{-\theta_N J_{\hat H}(\rho')}d\bb P^{\hat H}_N ,
\end{equation*}  we will relate the estimate \eqref{eq:tilde(rho)} for $ \bb P_N $  to $ e^{-\theta_N J_{\hat H}(\rho')} $. We can already guess this because for  $ H=H[\rho] $, with the "atypical "$ \rho $ fixed, we have 
\begin{equation}\label{eq:RNes}
\bb P_N(B_\veps(\rho) )=\bb E^H\left(\frac {d\bb P_N} {d\bb P^{ H}_N}\chi_{_{B_\veps(\rho) }}\right)=\bb E^H\left(e^{-\theta_N J_{H[\rho]}(\rho')}\chi_{_{B_\veps(\rho) }}(\rho')\right)\simeq e^{-\theta_N J_{H[\rho]}(B_\veps(\rho))},
\end{equation}
where $\chi_{_{B_\veps(\rho) }}  $ is the characteristic function of a ball $ B_\veps(\rho) $ that we assume with the property \eqref{eq:Vcon}, $ \bb E^H $ is the expectation with respect to $ \bb P^{H}_N $ and  with  $ \varepsilon  $ small enough \eqref{eq:RNes} is  approximated  with
\begin{equation*}\label{key}
\bb P_N(B_\veps(\rho) )\simeq e^{-\theta_N J_{H[\rho]}(\rho)}
\end{equation*} 
for a fixed $ N $. 
From this observation,   we  see that the functional $ V(\rho) $ we are looking for will be defined in terms of a perturbation prescription of $ \bb P_N $ in such a way that  $  J_{\hat H}(\rho) $   have to be optimized over the perturbations $ {\hat H} $ to give as optimizer the best (less expensive in terms of probability) perturbation $ H=H[\rho] $.  The optimization over $ H $ is given by the supremum and $ V(\rho) $ is defined as
\begin{equation}\label{eq:optpre}
V(\rho):=\us {\hat H} {\sup}\, J_{\hat H}(\rho)=J_{H[\rho]}(\rho).
\end{equation}
At this stage we will have found a functional that is a good candidate to satisfy a large deviations principle. In the sense that fixing a $ \rho $ first, what we did in \eqref{eq:RNes}  can be repeated similarly for any measurable set $ A\subset\Omega $ that satisfies \eqref{eq:Vcon} and containing $ \rho $ we have
\begin{equation}\label{eq:P_N(A)}
\mathbb{P}_N(A)\simeq e^{-\theta_N V(A)}.
\end{equation}
Supposing $J_{\hat H}( \rho)  $ lower semi-continuous in $ \rho $ for the $ \hat H $'s in the prescription, the rate functional 
\begin{equation}\label{eq:V=opH}
 V(\rho):= \us{\hat H}{\sup}\,J_{\hat H}(\rho)
\end{equation}
is lower semi-continuous  as desired.  To have a large deviation principle  as in definition \ref{def:LDprin} we need a prescription giving  $ V(\rho) $  such that it is lower semi-continuous  and satisfies d)  and  e).  The scheme for the two bounds comes after in next two subsections.
We highlight few of the troubles one has to deal with in proving rigorous large deviation principles:
\begin{itemize}
\item The set of prescribed perturbations does not produce all the points in the metric space $ (\Omega,d) $. If this is the case one has to show a \e V-density or that the set of not producible points has probability super-exponentially small and use \eqref{eq:throwaway} to throw away this events.
\item The best optimizer in \eqref{eq:optpre} is not $ H[\rho] $, then one has to redefine a new perturbative prescribtion leading to the best optimizers $ H[\rho] $.
\item It could be that, while the measures $ \bb P_N $ is induced by some random element, the functional $ J $ is not a function of that random element. 
\item The set of possible perturbations is very large and instead of looking for some perturbative prescription of $ \bb P_N $ it could be better to face up the problem with direct combinatoric techniques.
\end{itemize}

\subsection{Upper bound d)}
Let be $ C $ a set of $ \Omega $ that we don't specify for the moment if open or closed. Denoting with $ \chi_C$ the characteristic function of a set $ C $ and $ \hat H $ a perturbation of $ \bb P_N $, from inequality
\begin{equation}\label{eq:P_N(C)}
\bb P_N (C)= \bb E_N \left(e^{-\theta_N J_{\hat H}( \rho')}e^{\theta_N J_{\hat H}(\rho')}\chi_C(\rho')\right)\leq\us{\rho'\in C}{\sup}\left\{e^{-\theta_N J_{\hat H}(\rho')}\right\}
\end{equation}
taking the logarithm and limsup it comes out
\begin{equation*}\label{key}
\us{N\to\infty}{\lim\sup}\frac{1}{\theta_N}\log \bb P_N(C)\leq-\inf_{\rho'\in C} J_{\hat H}(\rho').
\end{equation*}
We can still optimize on perturbations, that is taking the inf on $ \hat H $ we find
\begin{equation}\label{eq:supinf}
\us{N\to\infty}{\lim\sup}\frac{1}{\theta_N}\log \bb P_N(C)\leq-\sup_{\hat H}\inf_{\rho'\in C} J_{\hat H}(\rho').
\end{equation}
At this point the validity of d) in \ref{def:LDprin} for any $ C $ closed pass trough the proof of the inversion of sup and inf in \eqref{eq:supinf}, usually this is done using before the minimax lemma (see appendix A.5 in \cite{SR15})  to get d) for  any compact $ K $, i.e. $ \us{N\to\infty}{\lim\sup}\frac{1}{\theta_N}\log \bb P_N(K)\leq-\us{\rho'\in K}{\inf}\us{\hat H}{\sup} J_{\hat H}(\rho') $ and then exponential tightness (see appendix A.5 in \cite{SR15} section 2.3) for $ \bb P_N $ to conclude d)
\begin{equation}\label{eq:upbound}
\us{N\to\infty}{\lim\sup}\frac{1}{\theta_N}\log \bb P_N(C)\leq-\inf_{\rho'\in C}\sup_{\hat H} J_{\hat H}(\rho').
\end{equation}

\subsection{Lower Bound e)}

To obtain the lower bound we start considering an open  set $ A\subset\Omega $ containing the "atypical" $ \rho $ and $ H=H[\rho] $ as in \eqref{eq: P^H}, hence for $ N\to \infty $
\begin{equation}\label{eq: P_N(A)to1}
\bb P^H_N(A)\us{N}{\to} 1.
\end{equation}
Rewriting $ \bb P_N(A) $ analogously to  $ \bb P_N(C) $ in the equality in \eqref{eq:P_N(C)} of the upper bound we have
\begin{equation}\label{eq:logP(A)}
\begin{split}
\frac 1{\theta_N}\log\bb P_N (A)=& \frac 1{\theta_N}\log \bb E_N \left(e^{-\theta_N J_H(\rho')}e^{\theta_N J_H( \rho')}\chi_A(\rho')\right)\\
= &\frac 1{\theta_N}\log\bb E_N\left(e^{-\theta_N J_H(\rho')}\frac {e^{\theta_N J_H(\rho')}\chi_A(\rho')}{\bb P_N^H(A)} \,\,\bb P_N^H(A) \right)\\
= &\frac 1{\theta_N}\log\bb E_N\left(e^{-\theta_N J_H(\rho')}\frac {e^{\theta_N J_H(\rho')}\chi_A(\rho')}{\bb P_N^H(A)} \,\, \right)+\frac 1{\theta_N}\log \bb P_N^H(A),
\end{split}
\end{equation}
inside the round brackets of the last term in \eqref{eq:logP(A)}  we have  the Radon-Nikodym derivative $ \frac {d{\bb Q}^H_N}{d \bb P_N}=\frac{e^{\theta_NJ(\rho',H)}\chi_A}{\bb P^H_N(A)} $ (it is a non-negative function and its mean with respect to $ \bb P_N $ is one) of a measure $ {\bb Q}_N^H $ respect to $ \bb P_N $.
So applying Jensen respect to $ {\bb P}_N^H  $ and using \eqref{eq:RN} we find the inequality
\begin{equation}\label{eq:logP(A)>..}
\frac 1{\theta_N}\log\bb P_N (A)\geq\frac 1{\bb P_N^H(A)}\bb E_N^H\left(-J_H(\rho')\chi_A(\rho') \right)+\frac 1{\theta_N}\log\bb P_N^H(A),
\end{equation}
from \eqref{eq: P_N(A)to1} and \eqref{eq: P^H}  it follows that asymptotically
\begin{equation}\label{eq:logP(A)>J}
\us{N\to\infty}{\lim\inf}\frac 1{\theta_N}\log\bb P_N (A)\geq -J_{H[\rho]}(\rho).
\end{equation}
Here  $ \rho $ and  $ A\ni \rho $ where general, so the same  can be repeated for any $ \rho $ and any open $ A\ni \rho $. Assume that for any $ \rho \in \Omega $ the perturbation prescription is  optimized with the sup, i.e. $ J_{H[\rho]}(\rho)=\us{\hat H}{\sup}J_{\hat H}(\rho)  $, and that  $ J_{H[\rho]}(\rho) $ is lower semi-continuous. Then from \eqref{eq:logP(A)>J} we can prove that
\begin{equation}\label{eq:lowbound}
\us{N\to\infty}{\liminf}\frac{1}{\theta_N}\log \bb P_N(A)\geq-\inf_{\rho'\in A}\sup_{\hat H} J_{\hat H}(\rho').
\end{equation}
Naming $ V(\rho):=\us{\hat H}{\sup} J_{\hat H}(\rho)$, from the two bounds \eqref{eq:upbound} and \eqref{eq:lowbound} we have  the two inequalities  $ \us{N\to\infty}{\lim\sup}\frac{1}{\theta_N}\log \bb P_N(C)\leq -V(C) $ and $ \us{N\to\infty}{\lim\sup}\frac{1}{\theta_N}\log \bb P_N(A)\geq -V(A) $ to conclude  a large deviation principle.

Here it ends this scheme for a "virtual" proof of large deviation, "virtual" because we proved a large deviation principle without a specific model using ad hoc hypothesis to find some common inequality where tools  to procede in the proof are standard.
In next sections we are treating large deviations rate functionals  for interacting particle systems satisfying a scaling limit and considering trajectories $ \rho  $ which are not the solution $ \bar \rho $ of the corresponding hydrodynamic equation.  Respect to the programme of this scheme we will find just  a functional giving \eqref{eq:tilde(rho)} and we are not doing a complete proof of a large deviations principle.

\section{Dynamical large deviations}\label{sec:dLD}

Let $ D([0,T];\Sigma_N) $ be  the space of right continuous with left limits paths from $  [0,T] $ to $ \Sigma_N $ with the Skorohood topology and $ P^\gamma_N $ the distribution of the Markov chain of the considered weakly asymmetric energies-masses/particles interacting model, that is a measure on $ D([0,T];\Sigma_N) $. While let $ D([0,T];\mc M^1(\Lambda)) $ be  the space of trajectories from $ [0,T] $ to the space of positive measure $ \mc M^1(\Lambda) $ with the Skorohood topology and $ \bb P_N^\gamma:= P^\gamma_N\circ\pi_N^{-1} $ the measure induced on  $ D([0,T];\mc M^1(\Lambda)) $ by the empirical measure.   Assuming local equilibrium,  for models like \eqref{eq:fullKMPgen} we can obtain along the lines of sections \ref{sec:SL} and \ref{sec:TC} a scaling limit such that $ \bb P_N $ converges in distribution to a measure concentrated on the unique (weak) solution $ \bar \rho(x,t) $ of 

\begin{equation}\label{eq:drcE2}
\left\{
\begin{array}{ll}
\partial_t\rho(x,t)=\Delta \rho(x,t)-\div \left(\sigma(\rho(x,t)) F(x,t)\right)\,,\\
\rho(x,0)=\gamma(x)\,,\\
\rho(0,t)=\rho_-\,,\\
\rho(1,t)=\rho_+\,,
\end{array}
\right.
\end{equation}
namely  
\begin{equation*}\label{eq:dconKMPE}
\bb P^{\gamma}_N\os{d}{\us{N}{\longrightarrow}}\delta_{\bar\rho} \,\, .
\end{equation*} 
We remember we are taking into consideration the one dimensional case where $ \Lambda_{N}:=\Lambda\cap\frac{1}{N}\mathbb{Z}$ where $\Lambda=(0,1)$ and   the boundary conditions at $ \p\Lambda $ are fixed by  fast  dynamics $ \mc L_{b} $ like that  of \eqref{eq:bKMPgen}. The initial profile $ \gamma $ is  associated to an initial product measure $ \nu_N^{\gamma} $ (e.g. with exponential marginal of parameter $ \gamma(x/N) $), see section 3 in \cite{BGL05} for  details and a discussion of why we choose a random initial configuration instead of a deterministic one. Now we want to know the rate functional $ V(\rho) $  estimating the probability \eqref{eq:tilde(rho)} to observe a trajectory  $ \rho(x,t) $ different from the "typical" one $ \bar\rho(x,t) $ in a bounded time interval $ [0,T] $.  The theory of this problem is  rather well established \cite{BDeSGJ-LL15,KipLan99,BGL05}. To find this large deviation  we  perturb  the dynamics \eqref{eq:bulkxyE} in the following way
\begin{equation}\label{eq:bulkxyE+H}
\mc{L}_{\{x,y\}}^{\hat {\bb H}}f(\xi):=\int\Gamma^{\xi,\mathbb F+\hat{\bb H}}_{x,y}(dj)\big[f(\xi-j\left(\varepsilon^x-\varepsilon^y\right))-f(\xi)\big]\,,
\end{equation}
where $ \hat {\bb H}(x,y) $ is defined respect to the gradient of a smooth scalar field $ \hat H(x,t) $, that is  $ \hat {\bb H}(x,y)=\int^y_x \nabla \hat H(s)ds=\hat H(y)-\hat H(x) $, to be precise it is enough that $ \hat H\in C^2(\Lambda\times [0,T]) $. Moreover $ \hat H $ is such that $ \hat H(0,t)=\hat H(1,t)=0 $ for every $ t\in [0,T] $.  The time variable is going to be omitted in the rest unless necessary.  Let's take $ (x,y)=(x,x+\frac 1N) $ and  observe that for large $ N $ the discrete vector field $ \hat {\bb H}(x,y) $ is still of order $ O(\frac 1N) $ because $ \hat {\bb H}(x,x+\frac 1N)\sim \nabla \hat H(x)\frac 1N $. We denote the distribution of the empirical measure for the related perturbed process \eqref{eq:empmisxi} with $ \bb P^{\gamma,\hat H}_N:= P^{\gamma,\hat H}_N\circ\pi_N^{-1} $, where $  P_N^{\gamma,\hat H} $ is the distribution of the perturbed weakly asymmetric model \eqref{eq:bulkxyE+H}. The process \eqref{eq:bulkxyE+H} is a weakly asymmetric model of the kind \eqref{eq:bulkxyE} where the distribution current for the weakly asymmetric model is $ \Gamma^{\xi,\bb F+\hat{\bb H}}_{x,y}(dj) $, hence along the lines of chapter \ref{ch:HSL} the hydrodynamics scaling limits is
\begin{equation}\label{eq:dconH+F}
\bb P^{\gamma,\hat{ H}}_N\os{d}{\us{N}{\longrightarrow}}\delta_{\bar\rho^{\hat H}}\,\,,
\end{equation}
where $ \bar\rho^{\hat H}$ is the unique (weak) solution of partial differential equation

\begin{equation}\label{eq:drcE+H}
\left\{
\begin{array}{ll}
\partial_t\rho(x,t)=\Delta \rho(x,t)-\div \left(\sigma(\rho(x,t)) \left(F(x,t)+\nabla \hat H(x,t)\right)\right)\,,\\
\rho(x,0)=\gamma(x)\,,\\
\rho(0,t)=\rho_-\,,\\
\rho(1,t)=\rho_+\, .
\end{array}
\right.
\end{equation}
Therefore  the correct perturbation\footnote{Note that this kind of perturbation doesn't produce all the possible atypical profile, but the general class in which we are interest in.} to observe a fixed  profile $ \rho(x,t) $ , that is $ \pi_N(\xi_t)\sim \rho (x,t) $, is the solution $ H(x,t) $ for all $ t\in[0,T] $ of the Poisson equation 

\begin{equation}\label{eq:PoiH}
\left\{
\begin{array}{ll}
\div\big(\sigma(\rho(x,t))\nabla H(x,t)\big)=-\partial_t \rho(x,t)+\Delta \rho(x,t)-\div\big(\sigma(\rho(x,t))F(x,t)\big)\\
H(0,t)=H(1,t)=0.
\end{array}
\right.
\end{equation}
With this choice of $ H $ formula \eqref{eq:dconH+F} becomes $ \bb P^{\gamma,H}_N\os{d}{\us{N}{\longrightarrow}}\delta_{\rho^{ H}} $ and to derive a large deviation principle \eqref{eq:optpre} estimating the probability \eqref{eq:P_N(A)} that $ \pi_N(\xi_t) $ stay close to $ \rho (x,t)$, i.e. 
\begin{equation}\label{eq:A}
A=\{\pi_N(\xi_t)\sim\rho(x,t),\,t\in[0,T] \},
\end{equation}
we have to compute the Radon-Nykodim derivative  $ \frac{d\bb P^{\gamma}_N}{d\bb P_N^{\gamma,\hat H} }(\xi)$.  
First we  need to compute the Radon-Nykodim derivative $ \frac{d\bb P^{\xi_0}_N}{d\bb P_N^{\xi_0,\hat H} }(\xi)$ for an initial deterministic configuration $ \xi_0 $ by a standard computation in the theory of jump Markov processes like in chapter 10 and section 7 of appendix 1 in \cite{KipLan99}, for a computation in the KMP case see also \cite{BGL05}. Along the lines of subsection 3.2 in \cite{BGL05} it comes out that for an initial deterministic configuration it is
\begin{equation*}\label{eq:dynldt}
 \bb P^{\xi_0}_N(A)\simeq e^{-NJ_{[0,T]}(\rho)} 
\end{equation*}
where  $J_{[0,T]}(\rho)$ is
\begin{equation}\label{eq:JdLD}
J_{[0,T]}(\rho)=\frac 14\int^T_0dt\int_\Lambda dx\,\sigma(\rho)(\nabla  H(x,t))^2
\end{equation}
with $ H $ solving \eqref{eq:PoiH}.
Since we are considering a random initial configuration we should  consider also the contribution to the dynamic large deviation coming from the initial condition adding to \eqref{eq:JdLD} an extra functional which would take  into account this fact, for this argument see subsection 3.2 of \cite{BGL05}. Here we don't report it for question of space and because for our purposes is not relevant. The time interval $ [0,T] $ is just for notational choice, we could chose any finite interval $ [T_0,T_1] $. 

\section[Stationary large deviations]{Stationary large deviations theory}\label{sec:Stldt}

In the 70's  Freidlin and Wentzell \cite{FW12} studied the problem of large fluctuations in dynamical systems subject to weak noise, the logarithm of the stationary distribution becomes proportional to a  rate functional called the quasi-potential, which is characterized
through a pure variational problem.

In interacting particle systems it has been developed an analogous framework to study stationary  large deviations of macroscopic quantities like particle density or current. Stationary means that we consider    large deviations when the mass is distributed according to  the invariant measure. This theory is known as  
Macroscopic Fluctuation Theory \cite{BDeSGJ-LL15} that can be understood intuitively as a generalization of the
Freidlin-Wentzell theory to particles systems where  a weak noise proportional to $ 1/ \sqrt{N^d} $ (with $ N^d  $  size of the lattice and $ d $ the dimension)  is added to the hydrodynamics.
In this context the quasi-potential 
appears as a natural  extension of the availability to non-equilibrium particle systems. Both from a physical point of view and a technical one, a very important   issue is to actually compute the quasi-potential from its variational characterization for non-equilibrium systems for which detailed balance does not hold while $ N $ goes to infinity.

A formal bridge between the two theories can be  produced  looking at  their respective fluctuation problems
with  path integral approach for their Lagrangian reformulation and applying  a saddle-point approximation. This allows to have a common physical point of view of them.
Even if this approach is only formal because a mathematical meaning of non-linear fluctuating
hydrodynamics is still lacking, the results obtained are in complete agreement with the ones obtained through
rigorous probabilistic methods in all cases where a comparison is possible.
\vspace{0.6cm}

In this section for interacting particles systems with scaling limit \eqref{eq:drcE2} and  $ F(x,t)=E $, we denote with $ \mu_{E,N}$ the invariant measure of the dynamics\footnote{We are dealing with irreducible dynamics in such a way to have a unique invariant measure.} and with $ \bb P_{\mu_{E,N}}:=\mu_{N,E}\circ\pi^{-1}_N $ the stationary distribution of the empirical measure. Moreover we consider $ C(\Lambda) $ and its dual of positive continuous linear functionals $ C^*_+(\Lambda) $ with the $ \ast $-weak topology, namely $ \rho_n\in C^*_+(\Lambda)  $ is such that $ \rho_n\us{n}{\to}\rho $ iff for each test function $\phi\in C(\Lambda)  $ we have $ \rho_n(\phi)\us{n}{\to}\rho(\phi) $. Because of $ \Lambda $ is a compact Hausdorff space,  Riesz-Markov theorem tell us that for every $ \rho $ there is a positive regular Borel measure  $ \mu_\rho $ such that $ \rho(\phi)=\int_\Lambda\phi d\mu_\rho $ and viceversa, so for each $ \rho\in L^1(\Lambda) $ it is $ \rho(\phi)=\int_\Lambda \phi d\mu_\rho=\int_\Lambda dx\phi\rho  $  and in this case we define  the spatial integration
\begin{equation}\label{eq:scalarL2}
\langle \rho,\phi\rangle:=\int_{\Lambda}^{}dx\rho(x)\phi(x) 
\end{equation}
 which is a scalar product in $ L^2(\Lambda) $.

\subsection{Freidlin-Wentzell theory}

Consider  a finite dimensional dynamical system $ \dot x_t=B(x_t) $ with $ x\in \bb R^n $ and the random dynamical system defined by the Ito stochastic differential equation
 \begin{equation}\label{eq:noise}
\dot X_t = B(X_t)+\sqrt{2\veps}k\dot w_t,
\end{equation}
where $ k$ is a  $ n\times n $ non-singular constant matrix and $ \dot w_t\in\bb R^n $ is a vector of white  Gaussian noises $ w^i_t $ with zero mean and covariance $ \bb E(w_t^i w_{t'}^j)=\delta^{ij}\delta(t-t') $. Moreover we denote $ K=k\cdot k^T $. Assume $ B $ a Lipschitz vector fields and  with a unique globally attractive equilibrium ($ B(\bar x)=0$) point $ \bar x $.
The unique invariant measure $ \mu_\veps(dx) $ solves the stationary Fokker-Planck equation corresponding  to \eqref{eq:noise}. We want to consider a  non-reversible\footnote{Under suitable assumptions of the matrix $ k $ (for example $ k\, \bb I $ with $ k $ a real constant and $ \bb I $ the identity matrix), if $ B(x)=-\nabla U(x) $ the process \eqref{eq:noise} is reversible. In this case the invariant measure is $ \mu_\veps(dx)= \frac 1{Z_\veps}e^{-\veps^{-1}U(x)}dx$, where $ Z_\veps $ is a normalization factor.} process, in this case computations are difficult and the invariant measure in general is not known. 

In the limit of weak noise $ \veps\to 0 $,
Freidlin-Wentzell proved \cite{FW12} that  the probability to observe in a finite time interval $ [T_0,T_1] $ a trajectory $ x_\veps(t) $  of   \eqref{eq:noise} close to  a trajectory $ x(t) $ different from the deterministic solution $ \bar x(t) $ of $ \dot x=B(x) $  with  initial deterministic condition $ x(T_0)=x_0 $ is given by 
\begin{equation}\label{eq:dynWF}
\bb P^{x_0}  (X(t)\sim x(t), t\in[T_0,T_1])\simeq e^{\veps^{-1} J_{[T_0,T_1]}(x(t)) }
\end{equation}
where the  rate functional $ J_{[T_0,T_1]}(x(t)) $ is 
\begin{equation}\label{eq:FW ac}
J_{[T_0,T_1]}(x)=\frac 14\int^{T_1}_{T_0} dt[\dot x-B(x)]\cdot K^{-1}[\dot x-B(x)].
\end{equation}
The rate functional \eqref{eq:FW ac} comes from a dynamical large deviation principle \cite{FW12}, in the case of interacting particles systems this kind of problem was discussed in the previous section \ref{sec:dLD} with the help of the theory of Markov process and soon we will see the resemblance  between  \eqref{eq:JdLD} and \eqref{eq:FW ac}. Now we are interested in the stationary large deviations, namely the problem of "rare" events after the systems relaxed into the stationary state $ \mu_\veps $, therefore we want the rate functional $ V(x) $ estimating
\begin{equation}\label{eq:stdevFW}
\bb P_{\mu_\veps}(X \sim x)\simeq e^{-\veps^{-1}V(x)},
\end{equation}
where $ x\neq\bar x $ and $ \bb P_{\mu_\veps} $ is the probability distribution in the stationary states $ \mu_\veps $. Freidlin and Wentzell proved that the quasi-potential defined as
\begin{equation}\label{eq:q-p FW}
W(x):=\us{T>0}{\inf}\,\,\, \us{\hat x\in\mc{T}_{x,T}}{\rm inf}J_{[-T,0]}(\hat x),
\end{equation}
where $ J_{[-T,0]} $ is \eqref{eq:FW ac} and $ \mc{T}_{x,T} $ is the set of trajectories connecting the equilibrium point $ \bar x $ to $ x $, i.e.
\begin{equation*}\label{eq:pathset1}
\mc{T}_{x,T}:=\{\hat x(t): \hat x(-T)=\bar x,\hat x(0)=x \},
\end{equation*}
gives $ V(x)=W(x) $ for the described scenario where there is a unique globally attractive equilibrium point $ \bar x $ . For  a short review with a physical flavour see \cite{BGN16}. In \eqref{eq:noise} we chose a constant matrix but everything is generalizable to a non constant matrix $ k(x) $.

\subsection{A mechanical point of view}
We want to give a physical motivation to  \eqref{eq:q-p FW}. First we observe  that the functional \eqref{eq:FW ac} is positive and  vanishes on the solutions of $ \dot x=B(x) $, in other words it is optimized. Therefore fixing the extremals of a   trajectory from $ x(T_0)=x_0 $ to $ x(T_1)=x_{1} $ the functional  $ J_{[T_0,T_1]}(x) $ of \eqref{eq:FW ac} can be thought as the  action functional of the Lagrangian
\begin{equation}\label{eq:actionFW}
L(x,\dot x)=[\dot x-B(x)]\cdot K^{-1}(x)[\dot x-B(x)],
\end{equation}
which in general will have other optimizers (possibly other minimizers). With a  path-integral approach going back to
Onsager and Machlup \cite{MO53} the transition probability $ \bb P(x(T_0)=x_0,x(T_1)=x_1) $  from a state $ x_0 $ to a state $ x_1 $  is written as the path integral denoted with $ \mc D(\cdot) $
{\small\begin{equation}\label{eq:pathint}
\bb P(x(T_0)=x_0,x(T_1)=x_1)\simeq\frac 1{Z_\veps}\int \mc D(x(t))\delta(x(T_0)=x_1,x(T_1)=x_{1})e^{-\veps^{-1}J_{[T_0,T_1]}(x)}
\end{equation}  }
where $ Z_\veps $ is a normalization function and the integration  is on the space of trajectories $ x(t) $ with the constraint  that  $ x(T_0)=x_0$ and $x(T_1)=x_{1}  $. Since $ k $ is non-singular $ K $ is a positive matrix and $ L >0 $  if $ \dot x\neq B(x) $, then contributions to the integral
are greater  as much as  closer they are to a minimizing stationary points of $ L(x,\dot x) $  (remember we are in the limit of weak noise). The exponential factor plays then  the
role of the probability density in the space of paths, where the estimate for  $\bb P(X(t)\sim x(t), t\in[T_0,T_1]) $ contributes in \eqref{eq:pathint} with 
\begin{equation}\label{eq:actionest}\bb P(X(t) \sim x(t),t\in[T_0,T_1])\simeq e^{-\veps^{-1}J_{[T_0,T_1]}(x)}.
\end{equation}
In \eqref{eq:pathint} we  take $ [T_0,T_1]=[-T,0] $ and $ x(-T)=\bar x, \,\, x(0)=x$, then considering the limit for $ T\to +\infty $ we have \begin{equation}\label{eq:Qp fluc}
\bb P(x(-T)=\bar x,\,x(0)=x)\us{T\to+\infty}{\rightarrow}\bb P_{\mu_\veps}(X \sim x).
\end{equation}
So in the weak limit noise $ \veps \to 0 $ \eqref{eq:stdevFW} is obtained with a saddle point approximation ($ \veps $ is very small) to the right hand side of \eqref{eq:pathint} where dominant contribution in the path integral is the one from the path minimizing $ J_{[-\infty,0]}(x) $.

\subsection{Macroscopic fluctuation theory}\label{subsec:MFT} This theory is a variational characterization of stationary fluctuations in the framework of interacting particles systems we presented so far. Here, we are developing it for weakly asymmetric interacting systems with scaling limits like \eqref{eq:drcE2} where $ F(x,t)=E $.
We can think about it as a field theory version of the Freidlin-Wentzell theory. In the last one randomness came from a noise added to an ordinary differential equations, while in our case randomness is intrinsic in the process defining the empirical measure and its diffusive equation \eqref{gate} describing the macroscopic behaviour.
To relate the two theories, we can imagine to consider the problem of stationary large deviations for a stochastic pde's as \eqref{eq:drcE2} where the first equation with $ F(x,t)=E $ is formally substituted by 
\begin{equation*}\label{eq:pdenoise}
\partial_t\rho(x,t)=\Delta \rho(x,t)-\div \sigma(\rho(x,t))  E+"\,\,noise\propto\frac{1}{\sqrt N}".
\end{equation*}
But the real similarity between the two theories is given by their respective dynamical large deviations. Indeed defining the Lagrangian 
\begin{equation}\label{eq:Lagrangia}
\mc L(\rho,\partial_t\rho):= \left\langle\, \nabla H,\sigma(\rho)\nabla H\,\right\rangle
\end{equation}
where $ \langle \cdot,\cdot\rangle $ is the spatial integration defined in \eqref{eq:scalarL2} and $ H $ solves \eqref{eq:PoiH},
the dynamical rate functional $ J_{[0,T]}(\rho) $ in \eqref{eq:JdLD}   can be introduced  as
\begin{equation}\label{eq:MFTact}
J_{[0,T]}(\rho)=\frac 14\int^T_0dt\,  \mc L(\rho,\partial_t\rho)  .
\end{equation}
Using \eqref{eq:PoiH} to rewrite the gradient field $ \nabla H $ as
\begin{equation}\label{eq:nablaH}
\nabla H=\frac 1{\sigma(\rho)}\div^{-1}\left(\p_t\rho-\Delta\rho+\div (\sigma E)\right)
\end{equation}
it appears clear that  also in this case we can  interpret the dynamical large deviations rate functional as the action of a Lagrangian. We can then proceed with the same  heuristic argument we did in the Freidlin-Wentzell theory, in the present case  $ (\rho,\p_t\rho) $ play the role of $ (x,\dot x)$, to evaluate the stationary large deviations
\begin{equation}\label{eq:flusta}
\bb P_{\mu_{N,E}}(\pi_N\sim\rho)\simeq e^{-NV_E(\rho)},
\end{equation}
so the large deviation rate functional $ V_E(\rho) $ is computed as the quasi-potential
\begin{equation}\label{eq:QP-MFT}
W_E(\rho):=\us{T>0}{\inf}\,\,\, \us{\hat \rho\in\mc{T}_{\rho,T}}{\rm inf}J_{[-T,0]}(\hat \rho),
\end{equation}
where $ J_{[-T,0]} (\hat\rho)$ is \eqref{eq:MFTact} and $ \mc{T}_{\rho,T} $ is the set of paths density connecting the stationary space dependent profile $ \bar \rho  $ to the space fluctuation $ \rho $, i.e.
\begin{equation}\label{eq:pathset2}
\mc{T}_{\rho,T}:=\{\hat \rho(x,t): \hat \rho(x,-T)=\bar \rho(x),\hat \rho(x,0)=\rho(x),\hat\rho(\p_\pm \Lambda,t)=\rho_\pm \text{ for all } t \},
\end{equation}
where $ \partial_\pm $ is the boundary operator giving the left boundary (-) and the right boundary (+) of $ \Lambda=[0,1] $.
Note that in \eqref{eq:pathset2} we underlined the space-time dependences of the path profiles. The paths in $ \mc T_{\rho,T} $ must also satisfy the boundary condition $ \rho(t,\p_\pm\Lambda)=\rho_\pm $, differently for  our boundary dynamics it can be shown \cite{BDeSGJ-LL03} that $ J_{[-T,0]}(\rho)=+\infty $ if the path $ \rho $ does not respect it. Under general conditions for boundary driven exclusion precesses in \cite{BG04,Fa09} they proved that 
\begin{equation}\label{eq:QP=V_E}
W_E(\rho)=V_E(\rho).
\end{equation}
That is what we derived with the path integral approach and results by the general arguments in \cite{BDeSGJ-LL02}.
\begin{remark}
We can  compute the asymptotic large deviations of the invariant  measure solving the dynamical variational problem  \eqref{eq:QP-MFT} without entering into the details of the invariant measure (generally  difficult to compute and not known).
\end{remark}
\begin{remark}
Inside \eqref{eq:QP-MFT} and \eqref{eq:pathset2} one can guess the Onsager-Machlup principle, that is  the macroscopic reversibility  for equilibrium states formulated as 
\begin{quotation}
\e{" the most probable trajectory creating the fluctuation is the time reversal of the most probable relaxation trajectory"}.
\end{quotation}
Since we are dealing with general boundary condition with $ \rho_+>\rho_- $ and we added an external field $ E $, unless $ E $ is a special value (see section \ref{sec:stasol}) such that the stationary current  $ J_E $ defined in \eqref{jeff} is zero , here we are in non-equilibrium stationary states   and this principle will be extended with an extra-condition later. In addition, we remark that it is not necessary having microscopic reversibility (detailed balance) for the microscopic dynamics to have the Onsager-Malchup principle that we call \e{macroscopic reversibility}, in subsection \ref{ss:Gladyn} we already gave an example of this phenomena due to fact that going from the microscopic to the
macroscopic scale a lot of information is lost and irreversibilities at a small scale
may be erased when taking macroscopic average see remark \ref{r:MacRev} and \cite{GJLV97}. 
\end{remark}

\subsection{Infinite dimensional Hamilton-Jacobi equation } Associated to the Lagrangian \eqref{eq:Lagrangia} the systems of coordinates $ (\rho,\partial_t\rho) $ inherited a Hamiltonian structure that allow us to derive the Hamilton-Jacobi equation correspondent to the minimization of the action  \eqref{eq:QP-MFT}.  The Hamiltonian to compute is 
\begin{equation}\label{eq:Hamiltonian}
\mc{H(\rho,\omega)}:=\underset{\zeta}{\rm sup} \left\{\langle \omega,\zeta\rangle-\mc{L}(\rho,\zeta)  \right\}.
\end{equation}
Let $ \delta $ be the variation of a functional $ \mc F $ and $ \frac{\delta\mc F}{\delta\rho} $ its functional derivative. Using $\delta\zeta=-(\sigma(\rho)\delta H_x)_x$ the conjugated momentum  $\omega=\frac{\delta \mc{L}}{\delta\zeta}$ turns out to be exactly $H$, therefore, expressing $\zeta(\rho,H)$ via the perturbed hydrodynamic  $\p_t\rho=\Delta\rho-\nabla\left(\sigma(\rho)\left(E+\nabla H\right) \right) $, the explicit form of the  $\mc{H(\rho,H)}=\left\{\langle H,\zeta\rangle-\mc{L}(\rho,\zeta)  \right\} _{\zeta=\zeta(\theta,H)}$ is 
\begin{equation}\label{eq:Hamiltonian2}
\mc{H}(\rho,H)=\left\langle \nabla H,  \sigma(\rho) \nabla H\right\rangle + \left\langle  H, \Delta\rho -E\,\nabla\sigma(\rho) \right\rangle,
\end{equation}
We normalize the quasi-potential in such a way  that $ V_E(\rho)\geq 0 $ and $ V_E(\bar \rho)=0 $.
  By a classical result\footnote{For the trajectories $ x(t) $ of a Lagrangian $ L(x,\dot x) $ connecting $ x(-T)=\bar x $ to $ x(0)=x $  the function $ V(x)=\us{T>0}{\inf}\,\,\, \us{\hat x(t)\in\mc{T}_{x,T}}{\rm inf}J_{[-T,0]}(\hat x(t))=\us{T>0}{\inf}\,\,\, \us{\hat x(t)\in\mc{T}_{x,T}}{\rm inf}\int_{-T}^0dt\,L(\hat x,	\dot {\hat x}) $ solves the Hamilton-Jacobi equation $  H(x,\nabla V(x))=c $, where $ H  $  is the hamiltonian and the constant $ c $ can be fixed as $ c=H(\bar x,\nabla V(\bar x)) $, see \cite{Arn89}. For example in the Freidlin-Wentzell case \eqref{eq:actionFW} for $ \bar x $ such that $ B(\bar x)=0 $ we have $ V(\bar x)=0 $ and then $ \nabla V(\bar x)=0 $. So the constant $ c $ can be fixed as $ c=H(\bar x,0)=0$. } in analytic mechanics and by conservation of energy  $ V_E(\rho) $ solves by the infinite Hamilton-Jacobi equation  $ \mc H\left(\rho,\frac {\delta V_E}{\delta \rho}\right)=\mc H(\bar \rho,0) =0 $, i.e.
\begin{equation}\label{eq:HJeq}
\left\langle \nabla\frac{\delta V_E}{\delta \rho},  \sigma(\rho) \nabla \frac{\delta V_E}{\delta \rho}\right\rangle + \left\langle  \frac{\delta V_E}{\delta \rho}, \Delta\rho -E\, \nabla\sigma(\rho) \right\rangle=0 
\end{equation}
with the boundary conditions $ \left. \frac{\delta V_E}{\delta\rho}(\hat \rho)\right|_{(\p_{\pm}\Lambda)}  = 0 $ with $ \hat\rho(\p_\pm \Lambda)=\rho_\pm  $. This equation is the keystone of the problem, indeed solving it we will identify the quasi-potential.

In next chapter we will look for the solution of the Hamilton-Jacobi equation \eqref{eq:HJeq} that identifies the quasi-potential $ V_E $.

\chapter[Totally asymmetric limit]{Totally asymmetric limit for models of heat condutcion}\label{ch:TALD}

In this chapter we characterize the quasi-potential \eqref{eq:QP-MFT} as solution of the Hamilton-Jacobi equation \eqref{eq:HJeq}. This is done in various steps. Before we compute the quasi-potential at equilibrium, then we will find a functional depending on an auxiliary variable and that solves the Hamilton-Jacobi equation. This functional will be identified with the quasi-potential after an optimization  over the auxiliary variable. The optimization will be according to the values of the external field $ E $. Later on, the adjoint hydrodynamics will characterize the  trajectories of the quasi-potential \eqref{eq:QP-MFT} with a generalized Onsager-Machlup principle.

After this, we study the asymptotic limit of the quasi-potential when the external field is large. This is done, like in \cite{BDeSGJ-LL10,BDeSGJ-LL11,BGL09}, studying first the limit of $\mathcal G_E$ and then solving a corresponding variational problem.
In this chapter we indicate the stationary state $ \bar\rho $ with a subscript $ \bar\rho_E $ because we underline  its dependence on $ E $   when the totally asymmetric limit $ E\to\infty $ is considered.

To study  large deviations in the totally asymmetric limit we need to understand the behaviour of the solution of the hydrodynamics  at the stationary state, that is when $ \p_t\rho=0 $. For our purpose it is enough to know the behaviour of   $ \frac{J_E(\bar\rho_E)}{E} $  in the totally asymmetric limit, this is done in section \eqref{sec:stasol}. The explicit stationary state $ \bar\rho_E $ for the KMP model is derived with a direct computation in section \ref{sec:KMPsta} of the appendix.

\section{Stationary solutions and currents}\label{sec:stasol}

Before to enter into the computation of the quasi-potential $ V_E $ is useful to study the stationary solutions. This will characterize a critical value of the external field $ E $ denoted $ E^* $ such that the systems is at equilibrium where the typical current $ J_E $ is zero. Moreover it will allow us to split
our study in three cases, namely when  in \eqref{jeff} the current $ J_E>0 $ for $ E>E^* $, when it is $ J_E=0 $ and when $ J_E<0 $ for $ E<E^* $.
In one dimension with boundary
conditions $\rho_\pm$, the stationary solution $\bar\rho_E$ of the hydrodynamic equation \eqref{hydrog} with a constant external field $F(x,t)=E$ is obtained as the solution of
\begin{equation}\label{steq}
\left\{
\begin{array}{l}
\Delta\rho-E\nabla\sigma\left(\rho\right)=0\,, \\
\rho(0)=\rho_-\ , \rho(1)=\rho_+\,.
\end{array}
\right.
\end{equation}
Recalling the typical current \eqref{jeff}, equation \eqref{steq} can be written as $\nabla\cdot J_E(\rho)=0$. In one dimension this implies that the typical current in the stationary state is spatially constant.
We are interested in the asymptotic behaviour of the stationary solution of the hydrodynamic equation in the limit of a large external field. The asymptotic behaviour of the solution can be obtained either by a direct computation or using the general theory (\cite{Se00} chapter 15). According to this the limiting value is the stationary solution of the conservation law obtained removing the second order derivative term and with Bardos Leroux N\'{e}d\'{e}lec boundary conditions. This should be also the stationary solution of the hydrodynamic equation for the asymmetric models discussed in section \ref{sec:Amodels} \cite{Bah12,PS99}. For our aims it is enough a weaker result. We use the fact that the unique solution of \eqref{steq} is monotone and the asymptotic behaviour of the current for large fields.

Equation \eqref{steq} can be integrated obtaining
\begin{equation}\label{ei}
\nabla\rho-E\sigma(\rho)=-J_E
\end{equation}
where $J_E$ is the integration constant that coincides with $J_E(\bar\rho_E)$ the typical current in the stationary state.

The monotonicity of $\bar\rho_E$ follows by the fact that if there is a non constant solution $\tilde\rho$ of the equation in \eqref{steq} such that $\nabla\tilde\rho(y)=0$ for some $y\in[0,1]$ then we have two different solutions to the Cauchy problem determined by the conditions $\nabla\rho(y)=0$ and $\rho(y)=\tilde\rho(y)$. One is $\tilde\rho$ itself and the other one is the constant one.

Since the solution $\rho$ in \eqref{ei} is monotone the integration constant $J_E $ is determined imposing the validity of the boundary conditions by
\begin{equation}\label{condju}
\int_{\rho_-}^{\rho_+}\frac{d\rho}{E\sigma(\rho)-J_E}=1\,.
\end{equation}
The left hand side of \eqref{condju} is monotone on $J_E$ that can be uniquely fixed for any choice of $\rho_\pm$ and $E$. Once $J_E$ has been fixed $\bar\rho_E$ is uniquely obtained by a direct integration of \eqref{ei}.
We distinguish the stationary states according to the sign of the stationary current. For any choice of $\rho_\pm$ there exists an external field $E^*$ for which the typical value of the current in the stationary state vanishes. This field is obtained selecting $J_{E^*}=J_{E^*}(\bar\rho_{E^*})=0$ in \eqref{condju} and using \eqref{llaa}
\begin{equation}\label{estar}
E^*=\lambda[\rho_+]-\lambda[\rho_-]\,.
\end{equation}
As the intuition suggests if we have a field $E>E^*$ then $J_E(\bar\rho_E)>0$ while $J_E(\bar\rho_E)<0$ for a field $E<E^*$.

\smallskip

To study the asymptotic behaviour of the current for large fields it is convenient to introduce the variable $\alpha =\frac 1E$ and the function $\mathcal E(\alpha):=\alpha J_{\frac 1\alpha}$. Condition \eqref{condju} becomes
\begin{equation}\label{condjue}
\alpha\int_{\rho_-}^{\rho_+}\frac{d\rho}{\sigma(\rho)-\mathcal E(\alpha)}=1\,.
\end{equation}
For any  $\alpha\neq 0$ the value $\mathcal E(\alpha)$ cannot belong to the interval $\{\sigma(\rho)\,, \rho\in[\rho_-,\rho_+]\}$ because otherwise the integral on the left hand side of \eqref{condjue} is divergent. Moreover when $\alpha <0$
then we need to have $\mathcal E(\alpha)\geq \sigma(\rho)$ for any $\rho$ while if $\alpha>0$ we get $\mathcal E(\alpha)\leq \sigma(\rho)$ for any $\rho$. This follows by the fact that otherwise the sign of the integral in \eqref{condjue} is not positive. When $|\alpha|\to 0$ the value of $\mathcal E(\alpha)$ cannot stay far from the interval $\{\sigma(\rho)\,, \rho\in[\rho_-,\rho_+]\}$ since otherwise the equality \eqref{condjue} cannot be satisfied. Since depending on the sign of $\alpha$ we have that $\mathcal E(\alpha)$ is always above or below the interval we deduce that
\begin{equation}\label{bgn}
\left\{
\begin{array}{l}
\lim_{\alpha\uparrow 0} \mathcal E(\alpha)=\lim_{E\to -\infty}J_E/E=\max_{\rho\in[\rho_-,\rho_+]}\sigma(\rho)\,, \\
\lim_{\alpha\downarrow 0} \mathcal E(\alpha)=\lim_{E\to +\infty}J_E/E=\min_{\rho\in[\rho_-,\rho_+]}\sigma(\rho)\,.
\end{array}
\right.
\end{equation}

\section[Quasi-potential for asymmetric models]{Quasi-potential for weakly asymmetric models}\label{sec:QPweas}

The problem we are interested in is the computation of the large deviation rate functional for the empirical measure  when the mass is distributed according to the invariant measure $\mu_{N,E}$  and the hydrodynamics is \eqref{hydrog} with a constant external field $ F=E $. In general the computation of the invariant measure in the non reversible case is a difficult problem. We give a description of the invariant measure at a large deviations scale. This asymptotic behaviour is described by the associate rate functional by \eqref{eq:flusta}. This computation is following along the lines of previous subsection \ref{subsec:MFT}.

In \cite{BGL05} it was shown that for all the boundary driven one dimensional
symmetric models \eqref{eq:KMPbulkgen} having constant diffusion and quadratic mobility it is possible to compute the corresponding non local quasi-potential. We show that this is possible for the same class of models also in the case of a weak constant asymmetry \eqref{eq:bulkxyE}. The corresponding quasi-potential is still non local and has a structure similar to the one of weakly asymmetric exclusion \cite{BGL09,BDeSGJ-LL11,ED04}.
 
We consider only the cases with $c_2\neq 0$ in $ \sigma(\rho)=c_2\rho^2+c_1\rho+c_0 $. The cases $c_2=0$ correspond to special models (zero range, Ginzburg-Landau) that have a local quasi-potential and can be studied directly.

\subsection{The reversible case}\label{trc}
In the case of reversible models (see subsection \ref{subsec:sta}) the computation of the quasi-potential is direct and we do
not need to solve the variational problem \eqref{eq:QP-MFT}. This is due to the fact that the invariant measure is product and the corresponding large deviations rate functional can be computed directly as the Legendre transform of a scaled cumulant generating function \cite{SR15,Tou09}. Let us show this computation in a more general framework.  Consider a family of probability measures $\mu^\lambda$
on $\mathbb R$ depending on the real parameter $\lambda$ of the form \eqref{eq:mu_lambda}. We call
\begin{equation}\label{prr}
P_\lambda(\phi)=\log \int_{\mathbb R}\mu^\lambda(d x)e^{\phi x}
\end{equation}
its cumulant  generating function. Using the expression \eqref{eq:mu_lambda} we have that
\begin{equation}\label{ca}
P_\lambda(\phi)=P(\phi+\lambda)-P(\lambda)
\end{equation}
where $P(\cdot)=\log Z(\cdot)$. We call
\begin{equation}\label{llp}
f_\lambda(\alpha):=\sup_{\phi}\left\{\alpha\phi-P_\lambda(\phi)\right\}
\end{equation}
the Legendre transform of $P_\lambda$. Using \eqref{ca} we obtain
\begin{equation}\label{serio}
f_\lambda(\rho)=f(\rho)+P(\lambda)-\lambda\rho\,,
\end{equation}
where $f(\cdot)$ is the Legendre transform of $P$ and, in the case of models with equilibrium product measures, it is called the density of free energy. Recall that $\rho[\lambda]$ and $\lambda[\rho]$ are the monotone functions determining respectively the density as a function of the chemical potential and the chemical potential as a function of the density. We have that $\lambda[\rho]=f'(\rho)$, $P'(\lambda)=\rho[\lambda]$. By the Legendre duality
we obtain
\begin{equation}\label{freng}
f_\lambda(\rho)=f(\rho)-f\big(\rho[\lambda]\big)-f'\big(\rho[\lambda]\big)\big(\rho-\rho[\lambda]\big)\,.
\end{equation}
The density of free energy $f$ satisfies the Einstein relation \eqref{Einzwei} and we obtain
\begin{equation}\label{lef}
f(\rho)=\left\{
\begin{array}{ll}
-\log\rho & \textrm{KMP}\,,\\
\rho\log\rho-(1+\rho)\log(1+\rho) & \textrm{KMPd}\,,\\
\rho\arctan\rho-\frac 12\log(1+\rho^2) & \textrm{KMPx}\,.
\end{array}
\right.
\end{equation}
Consider a slowly varying product measure
$\mu_N^{\lambda(\cdot)}=\prod_{x\in\Lambda_N}\mu^{\lambda(x)}(d\xi(x))$ where $\lambda(\cdot)$ is a continuous function on $\Lambda$. Let $g:\Lambda\to \mathbb R$
be a continuous test function. We can compute
\begin{eqnarray}\label{ep}
&& \mathcal P(g)=\lim_{N\to +\infty}\frac{1}{N^d}\log \int_{\mathbb R^{\Lambda_N}}\prod_{x\in\Lambda_N}\mu^{\lambda(x)}(d\xi(x))
e^{N^d\int_{\Lambda} g d\pi_N(\xi)}\nonumber \\
& &=\lim_{N\to +\infty}\frac{1}{N^d}\sum_{x\in \Lambda_N}P_{\lambda(x)}(g(x))=\int_\Lambda P_{\lambda(x)}(g(x))dx\,.
\end{eqnarray}
A general theorem \cite{SR15,Tou09} implies that the large deviations rate functional $\mathcal V_{\lambda}$ for $\pi_N(\xi)$ when the configuration is distributed according to the slowly varying product measure with marginals $\mu^{\lambda(x)}$ is
obtained as the Legendre transform of \eqref{ep}, i.e.
\begin{eqnarray}\label{compiti}
& &\mathcal V_{\lambda}(\rho)=\sup_{g}\left\{\int_\Lambda \rho(x)g(x)dx-\mathcal P(g)\right\}=\int_\Lambda f_{\lambda(x)}(\rho(x))dx\nonumber\\
& &=\int_\Lambda \Big[f(\rho(x))-f\big(\rho[\lambda(x)]\big)-f'\big(\rho[\lambda(x)]\big)\big(\rho(x)-\rho[\lambda(x)]\big)\Big]\,dx\,.
\end{eqnarray}
In the case of $\lambda(\cdot)=\lambda$ constant we obtain the rate functionals for the empirical measure associated to product of independent random variables having distribution $\mu^\lambda$. This means that the densities of free energy in \eqref{lef} encode the distribution of the microscopic system at equilibrium. The distribution $\mu^\lambda$ can be obtained by $f(\rho)$  with a Legendre transform and a double sided inverse Laplace transform. For the first two expressions in \eqref{lef} we obtain respectively exponentials and geometric distributions. For the third expression in \eqref{lef} we could not find a closed analytic expression after these two operations and this is the main reason why we have not an explicit microscopic model for the KMPx dynamics.

Consider a model having rates of transition \eqref{eq:pertEnrate} where $\mathbb F$ is the discretization of $F=\nabla \psi$ with $\psi$ is a smooth function on the domain $\Lambda$ such that $\psi|_{\partial \Lambda_N}=\lambda$.
By the general result on non homogeneous reversible models in Section \ref{subsec:sta} we have that the invariant measure is of product type slowly varying and indeed coincides with $\mu_N^{\psi(\cdot)}$. The quasi-potential can be computed using the approach above described and we have $V_{\nabla \psi}(\rho)=\mathcal V_\psi(\rho)$.

So in this subsection we got a general expression for the stationary rate functional in the case of equilibrium states, namely  when   $ E=E^* $ in \eqref{estar}. Therefore in a stationary state  the current $ J_E^* $ in \eqref{ei} is zero, see section \ref{sec:stasol}.  

\subsection {Solution of the Hamilton-Jacobi equation}

In subsection \eqref{subsec:MFT} we saw that associated to the variational problem \eqref{eq:QP-MFT} there is the 
infinite dimensional infinite dimensional Hamilton-Jacobi equation \eqref{eq:HJeq}, that we rewrite as
\begin{equation}\label{HJ}
\int_{\Lambda}\left[\nabla \frac{\delta V_E(\rho)}{\delta \rho}\cdot \sigma(\rho)\nabla\frac{\delta V_E(\rho)}{\delta \rho}+
\frac{\delta V_E(\rho)}{\delta \rho}\div\Big(\nabla \rho-E \sigma(\rho)\Big)\right]dx=0\,.
\end{equation}   
We remember that the rate functional $V_E$ that coincides with the quasi-potential $W_E$ is a solution to \eqref{HJ} from the arguments in subsection \eqref{subsec:MFT}.
We show that it is possible to find the relevant solutions of this equation for all the weakly asymmetric one dimensional models discussed here. The cases with zero external field unitary diffusion matrix and quadratic mobility were discussed in \cite{BGL05}. The symmetric and weakly asymmetric cases with unitary diffusion matrix and quadratic concave mobility ($ c_2<0 $) correspond essentially to exclusion models and have been already discussed in \cite{BDeSGJ-LL11, BGL09, ED04}. Here we complete the class of solvable models discussing the cases of unitary diffusion, quadratic and convex mobilities and in presence of a constant external field $E$. A similar computation for a model of oscillators having the same dynamic rate functional as the KMP model has been done in \cite{Bernpc}. In \cite{Bertpc} it is considered the KMP case and its totally asymmetric limit.

We show in this section how to find solutions of the Hamilton-Jacobi equation \eqref{HJ}. The existence and possibly uniqueness of $ \phi^\rho $ with boundary conditions $ \phi_\pm $  will be postponed to appendix \ref{sec:exsist} and proved only for the case of mobility $ \sigma(\rho) $ with real roots.                                                                                                                                                                                                                                                                                                                                                                                            
Later on we will discuss
more precisely the relevant variational problems  corresponding to the different values of the external field, this   issue is not discussed in full detail since a complete analysis  requires a long discussion. The variational problems are however very similar to the corresponding ones for the exclusion process and we refer to \cite{BDeSGJ-LL10,BDeSGJ-LL11,BGL09} for the details.
Following the approach of \cite{BDeSGJ-LL02,BDeSGJ-LL01,BDeSGJ-LL15, BGL05} we search for a solution of \eqref{HJ}  of the form
\begin{equation}\label{theform}
\frac{\delta \mc G_E(\rho(x),\phi(x))}{\delta \rho}=f'(\rho(x))-f'(\phi(x))\,,
\end{equation}
where $\phi$ has to be determined by the equation \eqref{HJ}. By the general theory \cite{BDeSGJ-LL15}  the space profile $\phi $ has to satisfies the boundary conditions \begin{equation}\label{eq:bcEL}
\left\{\begin{array}{ll}
\phi_-=\phi(0)=\rho_-\\
\phi_+=\phi(1)=\rho_+
\end{array}\right..
\end{equation}
We write the generic quadratic mobility as $\sigma(\rho)=c_2\rho^2+c_1\rho+c_0$
for suitable constants $c_i$. We rememeber we are considering only the cases with $c_2\neq 0$.
We insert \eqref{theform} into \eqref{HJ} and use the quadratic expression of the mobility with  some manipulations like in \cite{BDeSGJ-LL01,BDeSGJ-LL02,BDeSGJ-LL15,BGL05}. Since $\rho$ and $\phi$ satisfy the same boundary conditions, after one integration by parts with boundary terms disappearing  we obtain
\begin{equation}\label{primoc}
\int_\Lambda\left[\nabla\big(f'(\phi)-f'(\rho)\big)\sigma(\rho)\nabla f'(\phi)\right]\,dx +E\int_\Lambda\left(f'(\phi)-f'(\rho)\right)\nabla \sigma(\rho)\,dx=0\,.
\end{equation}
The first term in \eqref{primoc} can be developed as follows.
First we compute the derivatives and add and subtract suitable terms getting
\begin{equation}\label{secondoc}
\int_\Lambda\left[\nabla\left(\phi-\rho\right)\frac{\nabla\phi}{\sigma(\phi)}+
\left(\sigma(\rho)-\sigma(\phi)\right)\left(\frac{\nabla\phi}{\sigma(\phi)}\right)^2\right]\,dx\,.
\end{equation}
Then we integrate by parts the first term in \eqref{secondoc} and use the identity
\begin{equation}\label{formula}
\sigma(\rho)-\sigma(\phi)=(\rho-\phi)(c_2(\rho+\phi)+c_1)
\end{equation}
obtaining
\begin{equation}\label{terzoc}
\int_\Lambda\left[\frac{\rho-\phi}{\sigma(\phi)}\Delta\phi+c_2\left(\frac{\rho-\phi}{\sigma(\phi)}\right)^2\left(\nabla \phi\right)^2\right]\, dx\,.
\end{equation}
For the second term in \eqref{primoc} we integrate by parts and use again \eqref{formula} obtaining
\begin{equation}\label{quartoc}
E\int_\Lambda\left[\frac{(\phi-\rho)\nabla\phi}{\sigma(\phi)}\big(c_2(\rho+\phi)+c_1\big)\right]\,dx\,.
\end{equation}
Putting together \eqref{terzoc} and \eqref{quartoc} we obtain that the Hamilton Jacobi equation can be written as, 
\begin{equation}\label{onestep}
\int_{\Lambda}\frac{(\rho-\phi)}{\sigma^2(\phi)}\Big[\Delta\phi\sigma(\phi)+c_2(\nabla\phi)^2(\rho-\phi)-
E\sigma(\phi)\nabla\phi\left(c_2(\rho+\phi)+c_1\right)\Big]dx=0\, ,
\end{equation}
a possible way of solving the above equation is to impose that the term inside squared parenthesis is zero, that is  
\begin{equation}\label{eq:[]=0}
\Delta\phi\sigma(\phi)+c_2(\nabla\phi)^2(\rho-\phi)-
E\sigma(\phi)\nabla\phi\left(c_2(\rho+\phi)+c_1\right)=0.
\end{equation}
Its solutions are denoted with $ \phi^\rho $.
Explicitly we define the functional $ \mc G_E(\rho,\phi) $ as
\begin{equation}\label{effeciu}
\mathcal G_E(\rho,\phi):=\int_\Lambda\Big[f(\rho)-f(\phi)- f'(\phi)(\rho-\phi)\Big]dx + \mathcal R(\phi)\,,
\end{equation}
where we want the functional $ \mc R (\phi) $ such that  its stationary solutions $ \phi^\rho $, i.e. satisfying
\begin{equation}\label{eq:deltaG/deltarho}
\frac{\delta \mathcal G_E(\rho,\phi^\rho)}{\delta \phi}=0,
\end{equation}
are solutions also of \eqref{eq:[]=0}. Introducing the functional 
\begin{equation}\label{arrr}
\mathcal R(\phi)=\int_\Lambda \frac{1}{c_2E\sigma(\phi)}\Big[(\nabla\phi-E\sigma(\phi))\log|\nabla\phi-E\sigma(\phi)|-\nabla\phi\log|\nabla\phi|\Big]\,,
\end{equation}
with a long but straightforward computation we have 
\begin{equation}\label{perfi}
\frac{\delta \mc G_E}{\delta \phi}=\frac{\phi\nabla\phi}{\sigma(\phi)(\nabla\phi-E\sigma(\phi))}-\frac{\rho}{\sigma(\phi)}-\frac{\Delta\phi}{c_2\nabla\phi(\nabla\phi-E\sigma(\phi))}
+\frac{E(c_2\phi+c_1)}{c_2(\nabla\phi-E\sigma(\phi))}.
\end{equation}
If  the right hand side of $ \eqref{perfi} $ is zero then  \eqref{eq:[]=0} is satisfied. Hence a critical point 
$ \phi^\rho $  of $ \mc G_E $, i.e. satisfying \eqref{eq:deltaG/deltarho},    solves also \eqref{eq:[]=0}.

In this way we obtain that the functional of $\rho$ defined by $\mathcal G_E(\rho, \phi^\rho)$ solves the Hamilton Jacobi equation \eqref{HJ}. The case $E=0$ in \cite{BGL05} can be recovered as a limit.
\begin{remark}
Observe that the first term in \eqref{effeciu} corresponds to an equilibrium rate functional with typical density profile $\phi$ (see \eqref{compiti}). Namely  the ansatz in \eqref{theform} was 
\begin{equation}\label{eq.ansatz}
\mc G_E(\rho,\phi)=\mc V_{eq}(\rho,\phi)+\mc R(\phi)
\end{equation}
where $ \mc  V_{eq}(\rho,\phi):=\mc V_\lambda(\rho) $ with $ \phi=\rho[\lambda] $ and  $ \mc R(\phi) $ to be determined as we have seen.
So $\frac{\delta \mathcal G_E}{\delta \rho}=f'(\rho)-f'(\phi)$ and the functional derivative of  new term $ \mc{R}(\phi) $  is the right hand side of  \eqref{perfi}. 
\end{remark}

\section{Optimization}\label{sec:opt}
Since in general the critical points $ \phi^\rho $ are not unique we discuss more in detail the specific identification of the relevant $\phi^\rho$ that gives the quasi-potential for our models, namely in this section we investigate the right optimization  over $ \phi $ that gives $\mathcal G_E(\rho, \phi^\rho)$, in addition we calculate the normalization constant $ K_E $  such that $ V_E(\bar\rho_E)=0 $. This constant can be computed either as follows or imposing $ \mc G_E(\bar\rho_E,\phi^{\bar\rho_E})=0 $ once we note that $ \phi^{\bar\rho_E}=\bar\rho_E $ from \eqref{eq:[]=0}. We distinguish the 3 cases $E=E^*$, $E<E^*$ and $E>E^*$. For the two nonequilibrium cases ($ E\neq E^* $) we will determined the sign of the moduli in \eqref{arrr} and with this choice \eqref{eq:[]=0} and \eqref{perfi} will be equivalent.

\subsection{Case $E=E^*$}
This is a special case that corresponds to a model that has no current in the stationary state $J_{E^*}(\bar\rho_E)=0$. This is the condition of macroscopic reversibility (see end of subsection \ref{subsec:MFT}) that corresponds microscopically to the inhomogeneous reversible product measure discussed in subsection \ref{subsec:sta}. In this case the quasi-potential $V_{E^*}$ is local and can be computed both microscopically like in section \ref{trc} that macroscopically using \eqref{HJ}. We obtain that $V_{E^*}(\rho)=\mathcal V_{\lambda\left[\bar\rho_E\right]}(\rho)$ and $\lambda\big[\bar\rho_E(x)\big]$  linearly interpolates $\lambda[\rho_-]$ and $\lambda[\rho_+]$ when $x\in[0,1]$.

\subsection{Case $E<E^*$}\label{subsec:E<E^*}

For some computations it is more convenient to use the variable $\psi=\lambda[\phi]$, that is a one-to-one change of variable. We remember that   instead of discuss the general case we consider  the three prototype models.
Recall that in the case of the KMP model this change of variables corresponds to $\psi=-\frac 1\phi$.
First we discuss the KMP case, showing later how to modify the computations to cover also the other cases.
We consider the functional \eqref{effeciu} with a fixed determination of the signs of the moduli in \eqref{arrr}. This is enough to identify the correct solution. We write  the functional in terms of the variable $\psi$. The choice of the sign of the two logarithmic terms has to be $+$ since otherwise the function $\psi$ cannot satisfy the boundary conditions. Consequently we have to restrict the domain of definition of $\mathcal G_E$. The constant $ K_E $  fixes the normalization. In terms of $\psi$ the functional becomes
\begin{eqnarray}\label{eg}
\mathcal G_E(\rho,\psi)&=&\int_\Lambda\left[\left(\frac{\nabla\psi}{E}-1\right)\log\left(\nabla\psi-E\right)
-\frac{\nabla\psi}{E}\log \left(\nabla\psi\right)\right]dx\nonumber \\
& +&\int_\Lambda\left[-\rho\psi +\log \left(-\frac{\psi}{\rho}\right)-1\right]dx+K_E
\end{eqnarray}
where the constant $K_E$ is
\begin{equation}\label{Kekmp}
K_E=\log\left(-J_E\right)+\frac 1E\int_{\rho_-}^{\rho_+}\frac{d \rho}{\sigma(\rho)}\log\left(1-\frac{\sigma(\rho)E}{J_E}\right)\,.
\end{equation}
The functions $\psi$ that we are considering belong to
\begin{equation}\label{domps}
\mathcal F_E:=\left\{\psi\in C^1(\Lambda)\,:\,\nabla\psi\geq \max\{E,0\}\,,\, \psi(0)=-\frac{1}{\rho_-}\,, \psi(1)=-\frac{1}{\rho_+}\right\}\,.
\end{equation}
For a $\psi \in \mathcal F_E$ the functional \eqref{eg} is well defined. Formula \eqref{eg} is not well defined in the special case $E=0$. This case corresponds to the symmetric KMP process and has been already discussed in \cite{BGL05}. It is possible to obtain the corresponding functional for this special case as a limit of \eqref{eg} when $E\to 0$.
The constant $K_E$ has been fixed in such a way that $\inf_{\rho,\psi\in \mathcal F_E}\mathcal G_E(\rho,\psi)=0$. We discuss later in the general framework this point.

For any fixed $\rho$ the functional $\mathcal G(\rho,\cdot)$ is neither concave nor convex. Its critical points are determined by the Euler-Lagrange equation
\begin{equation}\label{elkmmm}
\frac{\Delta \psi}{\nabla\psi(E-\nabla\psi)}+\frac{1}{\psi}=\rho\,.
\end{equation}
We define the functional
\begin{equation}\label{esed1}
S_E(\rho)=\inf_{\psi\in \mathcal F_E} \mathcal G_E(\rho, \psi)\,.
\end{equation}
We can identify the quasi-potential $W_E=V_E$ with the infimum \eqref{esed1}, i.e. $V_E=W_E=S_E$. This is based on the interpretation of $\mathcal G_E$ as the pre-potential in the Hamiltonian framework obtained in subsection \ref{subsec:MFT} interpreting \eqref{eq:JdLD} as a Lagrangian action \cite{BDeSGJ-LL10,BDeSGJ-LL11,BDeSGJ-LL15}.
The pre-potential is defined on the unstable manifold for the Hamiltonian flow relative to a suitable equilibrium point  associated to the stationary solution $\bar\rho_E$. The value $\mathcal G_E(\rho,\psi)$ coincides with the value of the pre-potential when the pair $(\rho,\psi)$ belongs to the unstable manifold. Since the unstable manifold can be characterized by the stationary condition \eqref{elkmmm} we can then consider simply the infimum in \eqref{esed1} since all the critical points are belonging to the unstable manifold. For these values of the field $E$ it is not possible to show that there is uniqueness in the minimizer in \eqref{eq:QP-MFT} and equivalently in \eqref{esed1}. In correspondence the unstable manifold is not a graph and there is the possibility to have Lagrangian phase transitions. We are not going to discuss the details of these arguments and we refer to \cite{BDeSGJ-LL10,BDeSGJ-LL11} for the analogous computation in the case of the simple exclusion process.

For the KMPd model we use again the change of variables $\psi=\lambda[\phi]=\log\frac{\phi}{\phi+1}$ and the corresponding functional $\mathcal G_E$ has a form very similar to \eqref{eg}. In particular we give a general form that works for all the three models that we are considering and that depends on the density of free energy and the transport coefficients. In particular  \eqref{eg} can be obtained as a special case. The general form is constituted by three terms. The first one depends only on $\nabla \psi$ and coincides with the first term in \eqref{eg}. The second\footnote{ Adding $ \frac 1c_2 $ in front $ \log \sigma(\rho[\phi]) $ in \eqref{egdd} the expression is valid for general $ (c_2,c_1,c_0) $.} one can be written in general as
\begin{equation}\label{egdd}
\int_\Lambda\left[f(\rho)-f\left(\rho[\psi]\right)-\psi\big(\rho-\rho[\psi]\big)-\log\sigma\big(\rho[\psi]\big)\right] \,dx\,,
\end{equation}
$ \rho[\psi] $ is the inverse of $ \psi=\lambda[\phi] $. The third term is the constant \eqref{Kekmp}.
When $\rho$ is fixed the general form of the Euler Lagrange equation for $\mathcal G_E(\rho,\cdot)$ can be written in the general form
\begin{equation}\label{eld}
\frac{\Delta\psi}{\nabla\psi(E-\nabla\psi)}+\rho[\psi]-\sigma'\big(\rho[\psi]\big)=\rho\,.
\end{equation}
These general expressions work also for the KMPx model. The change of variable is again $\psi=\lambda[\phi]=\arctan \phi$.

The functional space for the $\psi$ is like \eqref{domps} with just a difference in the boundary values. In particular for KMPd we have $\psi(0)=\log\frac{\rho_-}{1+\rho_-}$ and $\psi(1)=\log\frac{\rho_+}{1+\rho_+}$ while for KMPx we have $\psi(0)=\arctan\rho_-$ and $\psi(1)=\arctan\rho_+$. For the same reason just explained in the KMP case, the optimization for $ \mc G_E(\rho,\psi) $ over $ \psi $ is done taking the infimum.

\medskip

To compute the global infimum of $\mathcal G_E$, that is relevant for the determination of the normalizing constant, it is convenient to minimize first in $\rho$ keeping fixed $\psi$. The stationary condition that corresponds to a minimum is $\lambda[\rho]=\psi$ and in correspondence the term \eqref{egdd} reduces to
$-\int_\Lambda\log\sigma\big(\rho[\psi]\big) \,dx$.
We minimize now over $\psi$ and obtain the stationary condition
\begin{equation}\label{psdc}
\frac{\Delta\psi}{\nabla\psi(E-\nabla\psi)}=\sigma'\big(\rho[\psi]\big)\,.
\end{equation}
Using the change of variable $\rho=\rho[\psi]$ equation \eqref{psdc} becomes the stationary equation 
\begin{equation*}\label{key}
\left\{\begin{array}{ll}
\nabla \rho -E\nabla \sigma(\rho)=0\\
\rho(0)=\rho_-,\rho(1)=\rho_+
\end{array},
\right.
\end{equation*}
that has an unique solution. We obtain then $\inf_{\rho,\psi\in\mathcal F_E}\mathcal G_E(\rho,\psi)=\mathcal G_E\left(\bar\rho_E,\lambda\big[\bar\rho_E\big]\right)$ and imposing that this value is zero we get the general formula \eqref{Kekmp}.

\subsection{Case $E>E^*$}\label{subsec:E>E^*}
Also in this case we use the variable $\psi=\lambda[\phi]$. We consider the functional \eqref{effeciu} in terms of this new variable. In this case the sign of the modulus in the second logarithmic term in \eqref{arrr} is still $+$ while the first one has to be fixed as $-$. This is because the values of the field are different and this is the choice that allows to satisfy the boundary conditions for $\psi$. Consequently we have to restrict the functions considered. In the case of the KMP model we have
\begin{eqnarray}\label{zanna}
\mathcal G_E(\rho,\psi)&=&\int_\Lambda\left[\left(\frac{\nabla\psi}{E}-1\right)\log(E-\nabla\psi)
-\frac{\nabla\psi}{E}\log \nabla\psi\right]dx\nonumber \\
& +&\int_\Lambda\left[-\rho\psi +\log \left(-\frac{\psi}{\rho}\right)-1\right]dx+K_E\,,
\end{eqnarray}
where the constant $K_E$ is
\begin{equation}\label{Kekmp2}
K_E=\log(J_E)+\frac 1E\int_{\rho_-}^{\rho_+}\frac{d \rho}{\sigma(\rho)}\log\left(\frac{\sigma(\rho)E}{J_E}-1\right)\,.
\end{equation}
The function $\psi$ is a function belonging to the set
\begin{equation}\label{doveleo}
\mathcal F_E:=\left\{\psi\in C^1(\Lambda)\,:\, 0\leq\nabla\psi\leq E\,,\, \psi(0)=-\frac{1}{\rho_-}\,,\psi(1)=-\frac{1}{\rho_+}\right\}\,.
\end{equation}
We have that the function $\left(\frac{\alpha}{E}-1\right)\log(E-\alpha)
-\frac{\alpha}{E}\log \alpha$ is concave when $\alpha \in [0,E]$ and also the function $\log (-\alpha)$
defined on the negative real line. The concavity of the functions implies the concavity of the functional $\mathcal G_E(\rho,\cdot)$ for any fixed $\rho$. Like in \cite{BGL09} there exists in $\mathcal F_E$ an unique critical point of $\mathcal G_E(\rho,\cdot)$ that is then a maximum. This is obtained as the unique solution in $\mathcal F_E$ to the Euler-Lagrange equation
\begin{equation}\label{ELK}
\frac{\Delta\psi}{\nabla\psi(\nabla\psi-E)}+\frac{1}{\psi}=\rho\,.
\end{equation}
We define the functional
\begin{equation}\label{esed}
S_E(\rho)=\sup_{\psi\in \mathcal F_E} \mathcal G_E(\rho, \psi)=\mathcal G_E(\rho,\psi^\rho)
\end{equation}
where $\psi^\rho$ is the maximizer solving \eqref{ELK}. Again we have $S_E=V_E=W_E$.
In this case the uniqueness of the solution of \eqref{ELK} is related to the fact that there is an unique critical point
for the variational problem related to the quasi-potential \eqref{eq:QP-MFT}.
Also in this case $\mathcal G_E$ can be interpreted as the pre-potential in an Hamiltonian framework but in this case for any $\rho$ there is an unique $(\rho,\psi)$ belonging to the unstable manifold. We have then not to minimize over the different points of the unstable manifold. It turns out that we have instead to maximize because the unstable manifold is exactly characterized by the stationary condition \eqref{ELK} and the critical point corresponds to a maximum by concavity \cite{BGL09}.

Also in this case we can obtain similar expressions of $\mathcal G_E$ for the models KMPd and KMPx and write a general expression for it.
Like in the previous case we have that $\mathcal G_E$ is composed by the sum of 3 terms. The first one depending only on $\nabla \psi$ coincides with the first term in \eqref{zanna}. The second term has a general form that coincides with \eqref{egdd}. The additive constant has the form \eqref{Kekmp2}. Also for these models the functional $\mathcal G_E(\rho,\cdot)$ is concave, so we optimize it taking the supremum.

\section{The adjoint hydrodynamics}\label{sec:adjhy}

The original problem consisted in the minimization \eqref{eq:QP-MFT}. To identify $ W_E$ with $S_E$ we still have to show that an optimizer $ \phi^\rho $  constrains a correspondent minimizer $ \rho^*(x,t) $ of \eqref{eq:QP-MFT}  to follow a trajectory satisfying $ \us{t\to-\infty}{\lim}\rho^*(x,t)=\bar\rho_E(x) $. With this we complete the identification of  $ W_E $
with $ S_E $. We use the original variable $ \phi $.
For a solution 
\begin{equation}\label{eq:solV_E=G_e}
 V_E(\rho)= \mc G_E(\rho,\phi^\rho)
\end{equation} 
of the Hamilton-Jacobi equation \eqref{eq:HJeq}, hence satisfying \eqref{theform} and	\eqref{eq:deltaG/deltarho},
we want to show that a minimizer $\rho^*\in \mc T_{\rho,T}$ in \eqref{eq:QP-MFT} is a fluctuation that connect the stationary state $\bar\rho_E(x)$ to the large deviation profile $\rho(x)$. Denoting with $ B(\rho,\phi) $ the left hand side of \eqref{eq:[]=0},  we want to minimize $J_{[-T,0]}(\hat \rho)$ asking 
\begin{equation}\label{eq:minJ}
\left\{\begin{array}{ll}
     \frac{\delta V_E}{\delta \rho}(\hat\rho)=f'(\hat\rho)-f'(\phi), & \qquad (x,t)\in\Lambda\times(0,+\infty),\\
    \hat\rho(x,0)=\rho(x), & \\
	 B(\hat\rho,\phi)=0, & \qquad (x,t)\in\Lambda\times(0,+\infty),\\
	\hat\rho(\p_\pm\Lambda,t)=\phi(\p_\pm\Lambda,t) =\rho_\pm. &
\end{array}
  \right.
\end{equation}
to verify that $\underset{T\rightarrow+\infty}{\lim}\rho^*(-T)=\bar\rho_E$. From \eqref{eq:nablaH} and  \eqref{eq:MFTact} we write  $J_{[-T,0]}(\hat\rho)$ as 
{\small\[ 
\begin{split}
J_{[-T,0]}(\hat\rho)=&\frac{1}{4}\int_{-T}^0 dt \left\langle \div^{-1}\left( \p_t\hrho-\Delta\hat\rho+E\nabla\sigma(\hat\rho)+\nabla\left(\sigma\nabla\frac{\delta V_E}{\delta\rho}\right)(\hat\rho) -\nabla\left(\sigma\nabla\frac{\delta V_E}{\delta\rho}\right)(\hrho)  \right)\right.,\\
&\left. \frac{1}{\sigma(\rho)}\div^{-1}\left( \p_t\hrho-\Delta\hrho+E\nabla\sigma(\hrho)+\nabla\left(\sigma\nabla\frac{\delta V_E}{\delta\rho}\right)(\hrho) -\nabla\left(\sigma\nabla\frac{\delta V_E}{\delta\theta}\right)(\hrho)  \right] \right\rangle,
\end{split} 
 \] }
using  $\left\langle\nabla\frac{\delta V_E}{\delta\rho}(\hrho),\sigma(\hrho)\nabla\frac{\delta V_E}{\delta\rho}(\hrho)\right\rangle=-\left\langle\nabla\frac{\delta V_E}{\delta\rho}(\hrho),\Delta\hrho-E\nabla\sigma(\hrho)\right\rangle $ from the Hamilton-Jacobi equation \eqref{eq:HJeq}  and $\frac{dV_E}{dt}(\hrho)=\left\langle \p_t\hrho,\frac{\delta V_E}{\delta\rho}(\hrho)\right\rangle$  we arrive at
{\small \begin{equation}\label{eq:Jasyn}
\begin{split}
J_{[-T,0]}(\hrho)=& V_E(\hrho(x))-V_E(\hrho(-T))+\frac 14\int_{-T}^0 dt \left\langle \div^{-1}\left( \p_t\hrho-\Delta\hrho+E\nabla\sigma(\hrho)+\nabla\left(\sigma\nabla\frac{\delta V_E}{\delta\rho}\right)(\hrho)  \right) \right.,\\
&\left. \frac{1}{\sigma(\rho)}\div^{-1}\left( \p_t\hrho-\Delta\hrho+E\nabla\sigma(\hrho)+\nabla\left(\sigma\nabla\frac{\delta V_E}{\delta\rho}\right)(\hrho)  \right) \right\rangle,
\end{split} 
\end{equation} }
since $ \sigma(\rho)$ is  strictly positive  a minimizer $\rho^*$ must be such that $\p_t\rho^*=\D\rho^*-E\nabla\sigma(\rho^*)-\nabla\left(\sigma \n\frac{\delta V_E}{\delta\rho}\right)(\rho^*)=-\D\rho^*-E\n\sigma(\rho^*)+\n\left(\frac{\sigma(\rho^*)}{\sigma(\phi)}\n\phi \right)$.  Denoting with $\rho_r^*(x,t):=\rho^*(x,-t)$ the time reversed of the optimal path $ \rho^* $, along the lines of \cite{BDeSGJ-LL02}, with a long computation one shows that the following two systems of PDE's are equivalent
\begin{equation}\label{eq:adjdyn}
\left\{
  \begin{split}
     &\p_t\rho^*=-\D\rho^*-\div\left(\sigma(\rho^*)\left(E-\frac{\n\phi}{\sigma(\phi)} \right)\right)\\
    &\rho^*(x,0)=\rho(x) \\
	 &B(\rho^*,\phi)=0\\
	&\phi(\p_\pm\Lambda,t)=\rho^*(\p_\pm\Lambda,t) =\rho_{\pm}
  \end{split}
  \right.
\Longleftrightarrow\,\,\,\,
\left\{
  \begin{split}
    &\p_t\phi = \D\phi -E\n\sigma(\phi)\\
    &\rho_r^*(x,0)=\rho(x)\\
    &B(\rho_r^*,\phi)=0\\
   &\phi(\p_\pm\Lambda,t)=\rho^*_r(\p_\pm\Lambda,t)=\rho_\pm\\
  \end{split}
   \right.,
\end{equation}
where in the right hand side equations system $ \phi $ solves an autonomous equation $ \phi_t=\mc D(\phi) $ that coincide with the original hydrodynamics, moreover from the initial condition  we have $\phi(x,0)=\phi^\rho$. From the third equation of the right hand side system we have 
\begin{equation}\label{eq:rho^*_r=..}
\rho_r^*=\frac{c_2\phi\n\phi(\n\phi+E\sigma(\phi))}{c_2\n\phi(\n\phi-E\sigma(\phi))}+
\frac{\sigma(\phi)(c_1E\n\phi -\D\phi)}{c_2\n\phi(\n\phi-E\sigma(\phi))}.
\end{equation}
When the time $ t  $ goes to infinity   $\underset{t\rightarrow+\infty}{\lim}\phi_t=0$, so the autonomous equation for $ \phi $ becomes $\D\phi=E\nabla(\sigma(\phi))$. So from \eqref{eq:rho^*_r=..}, using this limit identity we have that  $\underset{t\rightarrow+\infty}{\lim}\rho^*(-t)=\underset{t\rightarrow+\infty}{\lim}\rho_r^*(t)=\underset{t\rightarrow+\infty}{\lim}{\phi}(t)=\bar\rho_E$. 
Then a minimizer  $\rho^*$ individuated by \eqref{eq:adjdyn} is a trajectory from the stationary state  $\bar\rho_E$ to the space fluctuation $\rho$ and   a trajectory minimizing the cost of a fluctuation is the time reversal of the relaxation trajectory $ \rho^*_r $ to equilibrium according to the adjoint hydrodynamics, that is the dynamics obtained substituting $ \rho^* $ with $ \rho^*_r $ in the first equation of the left hand side of \eqref{eq:adjdyn}. Therefore the Onsager-Machlup principle is generelized to stationary non-equilibrium  states as
\medskip
\begin{quotation}
\e{"In a stationary non-equilibrium state the spontaneous emergence of a macroscopic fluctuation takes place most likely following a trajectory which is the time reversal of the relaxation path according to the adjoint hydrodynamics"}.
\end{quotation}

\bigskip

It follows that  $\underset{
T\rightarrow+\infty}{\lim}{V_E(\rho^*(-T))}=V_E(\bar\rho_E)=0$, hence from \eqref{eq:Jasyn} we get 
\begin{equation}\label{eq:Jlim}
 \underset{T\rightarrow+\infty}{\lim}J_{[-T,0]}(\rho^*)=V_E(\rho).
\end{equation}
If the stationary point $ \phi^\rho $ is unique  we can define just one solution \eqref{eq:solV_E=G_e} and the identification $ S_E=W_E $ it is straightforward, indeed there is an unique critical point for the variational problem related to the quasi-potential \eqref{eq:QP-MFT}. This was the situation for  $ E>E^* $ in subsection \ref{subsec:E>E^*}. If  there is more than one stationary point $ \phi^\rho $ as in subsection \ref{subsec:E<E^*}, the optimization of $ \mc G_E(\rho,\phi) $ over $ \phi $ will select the less expensive solution $ V_E(\rho)=\mc G_E(\rho,\phi^\rho) $ of the Hamilton-Jacobi equation, this will correspond to the best minimizers $ \rho^* $ of \eqref{eq:QP-MFT} because of the underlying Hamiltonian structure of the problem as explained in  subsection \ref{subsec:E<E^*}.

\bigskip


At this point we have all what we need to compute the totally asymmetric limit of the quasi-potential $ W_E=S_E $  derived in  section \ref{sec:opt}.  We consider the KMP, KMPd and KMPx models.

\section[The case $E\to -\infty$]{Totally asymmetric limit: $E\to -\infty$}\label{sec:TAL-}

We study first the limit $E\to -\infty$ of the auxiliary functional $\mathcal G_E$. Since we are interested in the limiting value we can assume that $E<0$ (instead of $ E<E^* $) and it is convenient to add and subtract the term $\log (-E)$ in \eqref{eg}. We add this factor to the first term of \eqref{eg} that becomes $\int_\Lambda s\left(-\frac{\nabla\psi}{E}\right)dx$ with
\begin{equation}\label{vitam}
s(\alpha):=\alpha\log\alpha-(1+\alpha)\log(1+\alpha)\,.
\end{equation}
Since $\us{\alpha\downarrow 0}{\lim}s(\alpha)=0$ this term is converging to zero in the limit of large and negative field. The factor that we subtract is inserted in the additive constant $K_E$ that becomes
\begin{equation}\label{riu1}
\log\frac{J_E}{E}+\frac 1E\int_{\rho_-}^{\rho_+}\frac{d \rho}{\sigma(\rho)}\log\left(1-\frac{\sigma(\rho)E}{J_E}\right)\,.
\end{equation}
The asymptotic behaviour of \eqref{riu1} can be easily understood since several terms depend just on the ratio $\frac{J_E}{E}$ whose behavior in the limit for large fields is given by \eqref{bgn}. In particular in this case the second term in \eqref{riu1} converges to zero while the first one converges to $2\log \rho_+$.

The second term in \eqref{eg} does not depend on $E$ and remains identical in the limit. We obtained then that
$\mathcal G_{-}:=\lim_{E\to -\infty}\mathcal G_E$ is given by
\begin{equation}\label{defgmeno1}
\mathcal G_-(\rho,\psi)=\int_0^1\left[-\rho\psi +\log \left(-\frac{\psi}{\rho}\right)-1\right]dx+2\log\rho_+\,.
\end{equation}
The function $\psi$ in \eqref{defgmeno1} has to belong to the set
\begin{equation}\label{stpiz}
\mathcal F_-:=\left\{\psi\in C^1(\Lambda)\,:\, \nabla\psi\geq 0\,, \psi(0)=-\frac{1}{\rho_-},\psi(1)=-\frac{1}{\rho_+},\right\}\,.
\end{equation}
We can obtain (see \cite{BDeSGJ-LL11,BGL09}) $V_-(\rho):=\lim_{E\to-\infty}V_E(\rho)$ as
\begin{equation}
V_-(\rho)=S_-(\rho):=\inf_{\psi\in \mathcal F_-}\mathcal G_-(\rho,\psi)\,.
\end{equation}
Since the function
$\psi\to-\rho\psi +\log \left(-\frac{\psi}{\rho}\right)$ is decreasing when
$\psi \in\left[\frac{-1}{\rho_-},\frac{-1}{\rho_+}\right]$, the minimum  of the integrand in \eqref{defgmeno1} over the $\psi$ for a fixed $\rho$ is obtained for  $\psi=-\frac{1}{\rho_+}$. This means that we can construct a minimizing sequence in $\mathcal F_-$ approximating a function that takes the value $-\frac{1}{\rho_-}$ at $0$ and then immediately jumps to the value $-\frac{1}{\rho_+}$. We then obtain
\begin{equation}\label{maleo}
V_-(\rho)=\inf_{\psi\in \mathcal F_-} \mathcal G(\rho,\psi)=\int_\Lambda \left[\frac{\rho}{\rho_+}-\log\frac{\rho}{\rho_+}-1\right]dx
\end{equation}
that is the large deviation rate function for masses distributed according to a product of exponentials distributions of parameter $\frac{1}{\rho_+}$. This is  the large deviations rate functional for the invariant measure of the version 1 of the totally asymmetric KMP dynamics discussed in section \ref{sec:Amodels} if we invert the direction of the asymmetry there.

\smallskip

For the KMPd model the asymptotic behavior is quite similar. The limiting value of $\frac{J_E}{E}$ when $E\to -\infty$ coincides always with $\rho_+(1+\rho_+)$.
 This implies that the limiting value of the constant term \eqref{riu1} is equal to $\log(\rho_+(1+\rho_+))$. The term depending only on $\nabla \psi$ is the same as in the KMP case and it is converging to zero. The term \eqref{egdd} does not depend on $E$ and is preserved in the limit. We obtain in the limit
\begin{equation}\label{limgd-}
\mathcal G_-(\rho,\psi)=\int_0^1\left[f(\rho)+\log(1-e^\psi)-\psi(\rho+1)\right]dx +\log(\rho_+(1+\rho_+))\,,
\end{equation}
and $\psi$ belongs to
\begin{equation}\label{stpizd}
\mathcal F_-:=\left\{\psi\in C^1(\Lambda)\,:\, \nabla\psi\geq 0\,, \psi(0)=\log\frac{\rho_-}{1+\rho_-},\psi(1)=\log\frac{\rho_+}{1+\rho_+},\right\}\,.
\end{equation}
As before in the interval $\left[\log\frac{\rho_-}{1+\rho_-},\log\frac{\rho_+}{1+\rho_+}\right]$ the function $\psi\to\log(1-e^\psi)-\psi(\rho+1)$ is decreasing  and then a minimizing sequence in $\mathcal F_-$ is converging
to a function that is equal to $\log\frac{\rho_-}{1+\rho_-}$ at $0$ and then is constantly equal to $\log\frac{\rho_+}{1+\rho_+}$. We obtain \cite{BDeSGJ-LL11,BGL09} that $V_-=\lim_{E\to -\infty}V_E$ is obtained by
\begin{equation}\label{evaid}
V_-(\rho)=S_-(\rho):=\inf_{\psi\in \mathcal F_-} \mathcal G_-(\rho,\psi)=\int_\Lambda\left[\rho\log\frac{\rho}{\rho_+}-(1+\rho)\log\frac{1+\rho}{1+\rho_+}\right]dx
\end{equation}
that is the large deviation rate functional for a product measure with marginals given by geometric distributions of
parameter $\frac{1}{1+\rho_+}$.

\smallskip

The limiting behavior of the KMPx model is instead quite different and exhibits a behavior similar to the exclusion process.
Recalling \eqref{bgn}, for the KMPx model the limiting value of $\frac{J_E}{E}$ when $E\to -\infty$  is given by $\max\{1+\rho_-^2,1+\rho_+^2\}$. Let us call
\begin{equation}\label{m}
\bar\rho:=\left\{
\begin{array}{ll}
\rho_- & \textrm{if}\ |\rho_-|\geq |\rho_+|\,, \\
\rho_+ & \textrm{if}\ |\rho_-|< |\rho_+|\,.
\end{array}
\right.
\end{equation}
The limiting functional $\mathcal G_-$ in this case becomes
\begin{equation}\label{arrd}
\mathcal G_-(\rho,\psi)=\int_0^1\left[\rho\arctan\rho-\frac{\log(1+\rho^2)}{2}-\rho\psi-\frac{\log(1+(\tan\psi)^2)}{2}\right]dx
+\log(1+\bar\rho^2)\,,
\end{equation}
and the function $\psi$ has to belong to the set
\begin{equation}\label{stpizx}
\mathcal F_-:=\left\{\psi\in C^1(\Lambda)\,:\, \nabla\psi\geq 0\,, \psi(0)=\arctan(\rho_-),\psi(1)=\arctan(\rho_+)\right\}\,.
\end{equation}
Since the function $\psi\to-\log(1+(\tan\psi)^2)$ is concave
the functional $\mathcal G_-(\rho,\cdot)$ is also concave.
To compute $V_-(\rho)=S_-(\rho)$ we need to minimize  $\mathcal G_-(\rho,\cdot)$ over the convex set \eqref{stpizx}. The infimum is
then realized on the extremal points of a suitable closure of that convex set. The extremal elements are the function of the form
\begin{equation}\label{step}
\psi^y(x)=\arctan(\rho_-)\chi_{[0,y)}(x)+\arctan(\rho_+)\chi_{[y,1]}(x)\,,\qquad y\in(0,1]\,.
\end{equation}
This means piecewise constant functions jumping from $\arctan(\rho_-)$ to $\arctan(\rho_+)$ at the single point $y$. Let us define the functional
\begin{eqnarray}\label{dinotte}
& &\tilde{\mathcal G}_-(\rho,y):=\mathcal G_-(\rho,\psi^y)=\int_0^1\left[\rho\arctan\rho
-\frac 12\log(1+\rho^2)\right]dx\nonumber \\
& &-\arctan(\rho_-)\int_0^y\rho(x)\, dx-\arctan(\rho_+)\int_y^1\rho(x)\, dx\nonumber \\
& &-\frac{y}{2}\log(1+\rho_-^2)-\frac{1-y}{2}\log(1+\rho_+^2)+\log(1+\bar\rho^2)\,.
\end{eqnarray}
This is the functional $\mathcal G_-$ computed in correspondence of the extremal elements.
The infimum of $\mathcal G_-(\rho,\cdot)$ coincides with the infimum of the function of one real variable
$\tilde{\mathcal G}_-(\rho,\cdot)$ on the interval $[0,1]$. Consequently we have
$$
V_-(\rho)=S_-(\rho)=\inf_{y\in(0,1]}\tilde{\mathcal G}_-(\rho,y)\,.
$$
The stationary condition for $\tilde{\mathcal G}_-(\rho,\cdot)$ is
\begin{equation}\label{stgt}
\rho(y)\left[\arctan\rho_+-\arctan\rho_-\right]=\frac 12 \log\frac{1+\rho_-^2}{1+\rho_+^2}\,.
\end{equation}
The second derivative establishing if a critical point is a local minimum or a local maximum
is given by $\nabla\rho(y)\left[\arctan\rho_+-\arctan\rho_-\right]$. To find the global minimum we have to consider also the values of the functions on the two extrema of the interval. Like for the exclusion process it is possible to have that the minimum
is obtained in more than one single point $y$ and this phenomenon is suggesting the possibility of the existence of dynamic phase transitions \cite{BDeSGJ-LL10} for finite and large enough negative fields.

\section[The case $E\to +\infty$]{Totally asymmetric limit:  $E\to +\infty$}\label{sec:TAL+}

We study first the limit when $E\to +\infty$ of the auxiliary functionals $\mathcal G_E$ \cite{BDeSGJ-LL11,BGL09}. We are in the $ E>E^* $ regime. We start with the KMP model. As before it is convenient to add and subtract the term $\log E$. We add this factor to the first term of \eqref{zanna} that becomes $\int_\Lambda \tilde s\left(\frac{\psi_x}{E}\right)dx$ with
\begin{equation}
\tilde s(\alpha)=-\alpha\log\alpha-(1-\alpha)\log(1-\alpha)
\end{equation}
Since $\lim_{\alpha\downarrow 0}\tilde s(\alpha)=0$ this term is converging to zero in the limit of large and positive field. The factor that we subtract is inserted in the additive constant $K_E$ and we obtain
\begin{equation}\label{riu}
\log\frac{J_E}{E}+\frac 1E\int_{\rho_-}^{\rho_+}\frac{d \rho}{\sigma(\rho)}\log\left(\frac{\sigma(\rho)E}{J_E}-1\right)\,.
\end{equation}
The asymptotic behaviour of \eqref{riu} can be easily understood since several terms depend just on the ratio $\frac{J_E}{E}$ whose behaviour in the limit for large fields is given by \eqref{bgn}. In particular in the KMP case the second term in \eqref{riu} converges to zero while the first one converges to $2\log \rho_-$.

The second term in \eqref{zanna} does not depend on $E$ and remains identically in the limit. We deduce that
$\lim_{E\to +\infty}\mathcal G_E=\mathcal G_{+}$ defined as
\begin{equation}\label{defgmeno}
\mathcal G_+(\rho,\psi)=\int_0^1\left[-\rho\psi +\log \left(-\frac{\psi}{\rho}\right)-1\right]dx+2\log\rho_-\,.
\end{equation}
Since we are considering the limit of a large positive field the bound from above on the derivative in \eqref{doveleo} disappears
and the function $\psi$ in \eqref{defgmeno} belongs to the set
\begin{equation}\label{dovelea3}
\mathcal F_+:=\left\{\psi\in C^1(\Lambda)\,:\, \nabla\psi\geq 0\,, \psi(0)=-\frac{1}{\rho_-}\,,\psi(1)=-\frac{1}{\rho_+}\right\}\,.
\end{equation}
Since the function
$\psi\to-\rho\psi +\log \left(-\frac{\psi}{\rho}\right)$ is decreasing when
$\psi \in\left[-\frac{1}{\rho_-},-\frac{1}{\rho_+}\right]$ we can compute the supremum  of \eqref{defgmeno} over the $\psi$ for a fixed $\rho$ obtaining a minimizing sequence converging to a function that assumes the value $-\frac{1}{\rho_-}$ on all the interval except the point $1$ where it assumes the value $-\frac{1}{\rho_+}$. We then obtain
\begin{equation}\label{maleo+}
V_+(\rho)=S_+(\rho):=\sup_{\psi\in \mathcal F_+} \mathcal G_+(\rho,\psi)=\int_0^1 \left[\frac{\rho}{\rho_-}-\log\frac{\rho}{\rho_-}-1\right]dx
\end{equation}
that is the large deviation rate function for masses distributed according to a product of exponentials of parameter $\frac{1}{\rho_-}$. This is exactly the large deviations rate functional for the invariant measure of the version 1 of the totally asymmetric KMP dynamics discussed in subsection \ref{subsec:tkmp1dis}.
For version 2 in section \ref{subsec:tkmp2dis},  conjecture that the large deviations rate functional for the empirical measure
when particles are distributed according to the invariant measure of the original model is the
same of the corresponding one associated to a product of exponentials. A direct microscopic
computation of this rate functional would be very interesting.

\smallskip

For the KMPd model the limiting value of $\frac{J_E}{E}$ when $E\to +\infty$ coincides always with $\rho_-(1+\rho_-)$.
The limiting behavior for this model is very similar to the KMP. In particular the limiting value of the constant term \eqref{riu} is equal to $\log(\rho_-(1+\rho_-))$. The term depending only on $\nabla \psi$ is the same as the KMP case and it is converging to zero. The term \eqref{egdd} does not depend on $E$ and is preserved in the limit. Writing explicitly the terms we obtain in the limit
\begin{equation}\label{limgd}
\mathcal G_+(\rho,\psi)=\int_0^1\left[f(\rho)+\log(1-e^\psi)-\psi(\rho+1)\right]dx +\log(\rho_-(1+\rho_-))
\end{equation}
The function $\psi$ in \eqref{limgd} belongs to the set
\begin{equation}\label{dovelea2}
\mathcal F_+:=\left\{\psi\in C^1(\Lambda)\,:\, \nabla\psi\geq 0\,, \psi(0)=\log\frac{\rho_-}{1+\rho_-}\,,\psi(1)=\log\frac{\rho_+}{1+\rho_+}\right\}\,.
\end{equation}
Since in the interval $\left[\log\frac{\rho_-}{1+\rho_-},\log\frac{\rho_+}{1+\rho_+}\right]$ the function $\psi\to\log(1-e^\psi)-\psi(\rho+1)$ is decreasing,
the supremum to compute $V_+=S_+$ is realized on a function $\psi$ that is constantly equal to $\log\frac{\rho_-}{1+\rho_-}$. We obtain
\begin{equation}\label{evaid-}
V_+(\rho)=S_+(\rho):=\sup_{\psi\in \mathcal F_+} \mathcal G_+(\rho,\psi)=\int_0^1\left[\rho\log\frac{\rho}{\rho_-}-(1+\rho)\log\frac{1+\rho}{1+\rho_-}\right]dx
\end{equation}
that is the large deviation rate functional for a product measure with marginals given by geometric distributions of
parameter $\frac{1}{1+\rho_-}$.

\smallskip

For the KMPx model the limiting value of $\frac{J_E}{E}$ when $E\to +\infty$ has 3 different possible values.
When $0\leq\rho_-\leq\rho_+$ then the limiting value is $1+\rho_-^2$. When $\rho_-\leq0\leq\rho_+$ then the limiting value is $1$. When $\rho_-\leq \rho_+\leq 0$ then the limiting value is $1+\rho_+^2$. Let us define
\begin{equation}\label{eq:barho+}
\bar\rho:=\left\{
\begin{array}{ll}
\rho_- & \textrm{if}\ 0\leq\rho_-\leq\rho_+\,, \\
0 & \textrm{if}\  \rho_-\leq0\leq\rho_+\,,\\
\rho_+ & \textrm{if}\ \rho_-\leq \rho_+\leq 0\,.
\end{array}
\right.
\end{equation}
With this definition we have  $\frac{J_E}{E}\to\sigma\left(\bar\rho\right)$.
The limiting functional $\mathcal G_+$ in this case becomes
\begin{equation}\label{arr}
\mathcal G_+(\rho,\psi)=\int_0^1\left[\rho\arctan\rho-\frac{\log(1+\rho^2)}{2}-\rho\psi-
\frac{\log(1+(\tan\psi)^2)}{2}\right]dx+\log(1+(\bar\rho)^2)\,.
\end{equation}
and the function $\psi$ belongs to
\begin{equation}\label{corri}
\mathcal F_+:=\left\{\psi\in C^1(\Lambda)\,:\, \nabla\psi\geq 0\,, \psi(0)=\arctan\rho_-\,,\psi(1)=\arctan\rho_+\right\}\,.
\end{equation}
We have $V_+(\rho)=S_+(\rho)=\sup_{\psi\in \mathcal F_+} \mathcal G_+(\rho,\psi)$ and the maximizer $\psi_M^\rho$ such that $V_+(\rho)=\mathcal G_+\left(\rho,\psi_M^\rho\right)$ can be described as follows. Let $H(x)=-\int_0^x\rho(y)\,dy$ and $G=co(H)$ its convex envelope. We define
\begin{equation}\label{lafi}
\phi_M^\rho(x)=\left\{
\begin{array}{ll}
\rho_- & \textrm{if}\ \nabla G(x)\leq\rho_- \\
\rho_+ & \textrm{if}\ \nabla G(x)\geq\rho_+ \\
\nabla G(x) & \textrm{otherwise}\,.
\end{array}
\right.
\end{equation}
We have that the maximizer is $\psi_M^\rho=\arctan\phi_M^\rho$. To prove this statement it is convenient to write back \eqref{arr} in terms of the variable $\phi$ related to $\psi$ by $\phi=\tan\psi$. We need to prove that $\phi_M^\rho$
is the minimizer of
\begin{equation}\label{irafunesta}
\inf_{\phi\in \tilde{\mathcal F}_+} \int_0^1 \left[\rho \arctan\phi+\frac 12\log \left(1+\phi^2\right)\right]dx=
\inf_{\phi\in \tilde{\mathcal F}_+} \mathcal B(\rho,\phi)\,,
\end{equation}
where the last equality defines the functional $\mathcal B$ and
\begin{equation}\label{nes}
\tilde{\mathcal F}_+:=\left\{\phi\in C^1(\Lambda)\,:\, \nabla\phi\geq 0\,, \phi(0)=\rho_-\,,\phi(1)=\rho_+\right\}\,.
\end{equation}
To identify the minimizer $\phi_M^\rho$ in \eqref{irafunesta} we can adapt the argument in \cite{DLS03}. Since $H(0)=G(0)$ and $H(1)=G(1)$ with an integration by parts we obtain
\begin{equation}\label{partii}
\int_0^1\left(\rho+\nabla G\right)\arctan(\phi)\, dx=\int_0^1\left(H-G\right)\nabla \arctan(\phi)\, dx\geq 0\,.
\end{equation}
The last inequality follows by the fact that $\nabla \phi\geq 0 $, the function $\arctan$ is increasing and by definition $G\leq H$. We have that the second term of \eqref{partii} is the integral of the product of two non negative terms and consequently it is non negative. Inequality \eqref{partii} can be written as $\mathcal B(\rho,\phi)\geq \mathcal B(-\nabla G,\phi)$.

The derivative with respect to $\phi$ of the function $\alpha \arctan\phi+\frac 12\log \left(1+\phi^2\right)$ is given by $\frac{\alpha+\phi}{1+\phi^2}$ and is negative for $\phi<-\alpha$ and positive for $\phi>-\alpha$. Since the function $\phi(x)$ takes values only on the interval $[\rho_-,\rho_+]$ we have that
$$
\inf_{\phi\in [\rho_-,\rho_+]}\left\{\alpha \arctan\phi+\frac 12\log \left(1+\phi^2\right)\right\}
$$
is obtained at $\rho_-$ when $\alpha\geq -\rho_-$, it is obtained at $\rho_+$ when $\alpha\leq -\rho_+$ and it is obtained at $-\alpha$ when $\alpha\in [-\rho_+,-\rho_-]$. This fact implies that we have the inequality $\mathcal B(-\nabla G,\phi)\geq \mathcal B(-\nabla G,\phi_M^\rho)$. The last fact that remains to prove is that it holds the equality $\mathcal B(-\nabla G,\phi_M^\rho)= \mathcal B(\rho,\phi_M^\rho)$. This follows by the following argument. We split the interval $[0,1]$ on intervals where either $-\nabla G=\rho$ or $\nabla G$ is constant. On an interval where $-\nabla G=\rho$ we have obviously the same contribution. On an interval $[a,b]$ where $\nabla G$ is constant correspondingly also $\phi_M^\rho$ is constant. Moreover the constant value of $\nabla G$ coincides with $\frac{\int_a^b\rho(y)\,dy}{a-b}$. Since $\phi_M^\rho$ is constant on the interval the contribution coming from intervals of this type still coincide and we get the equality.
Summarizing we have the following chain
\begin{equation}\label{chaina}
\mathcal B(\rho,\phi)\geq \mathcal B(-\nabla G,\phi)\geq \mathcal B(-\nabla G,\phi_M^\rho)= \mathcal B(\rho,\phi_M^\rho)
\end{equation}
that implies that $\phi_M^\rho$ is the minimizer in \eqref{irafunesta}. Since $\psi=\arctan \phi$ we obtain that the maximizer to compute $V_+(\rho)=S_+(\rho)$ is given by $\psi_M^\rho=\arctan\phi_M^\rho$.

\bigskip

The space profile $ \bar\rho $ in \eqref{m} and \eqref{eq:barho+} might be proved to be the stationary profile for the KMPx case in the strong asymmetric limit.

\clearpage
\pagestyle{fancy}
\fancyhf{}
\fancyhead[RE]{\sffamily \fontsize{10}{12} \selectfont \nouppercase{\rightmark}}
\fancyhead[RO]{\sffamily \fontsize{10}{12} \selectfont \nouppercase{\rightmark}}
\fancyhead[LE]{\sffamily \fontsize{10}{12} \selectfont \nouppercase{\leftmark}}
\fancyhead[LO]{\sffamily \fontsize{10}{12} \selectfont \nouppercase{\leftmark}}
\cfoot{\footnotesize \thepage}
\renewcommand{\headrulewidth}{0.01cm}

\chapter{Perspectives}\label{ch:pers}

We will discuss open problems from previous chapters and  possible future investigations from the theory we developed. Additionally we present other few applications. The discussion is going to be split in a perspectives section relative to part \ref{p:mt} and another one relative to part \ref{p:mact}.

\section{Perspectives: part 1}

Our functional Hodge decomposition in section \ref{sec: HD}   offers a splitting of the instantaneous current \eqref{eq:istcur} and \eqref{eq:istce-m} with a simple explicit form in terms of translated functions depending only on the configurations. Other splitting in general are not explicit \cite{VY97}. This decomposition is interesting in the context of   non-gradient interacting systems because in this case 
the theory develops  from an orthogonal decomposition of the current \cite{VY97} and the
transport coefficients (section \ref{sec:TC}) have a variational representation in terms of the Green-Kubo formulas
\cite{Spo91,KipLan99}. Recently a finite dimensional
approximation of the transport coefficients   has been analysed in \cite{AKM17}. It would be interesting to understand if our decomposition is helpful to compute the transport coefficient in relation to these results. 

In chapter \ref{ch:SNS} we generated in dimension one and two local dynamics for discrete particles models when  a short range Gibbs translational invariant measure  is given.  Similar constructions could be done for other state spaces, for example for the case of energies-masses models. In particular it could be interesting consider long range interaction and non translational invariant measure. Additionally also constructing local dynamics and defining Markov rates from non local cycles, i.e. moving the particles on the entire torus $ \mathbb{T}^n_N $, are computations that merit attention \cite{PSS17}.

The construction with cycles of section \ref{sec:Dfreefl} was used to generated dynamics for a given target measure.  But we could  try to understand if   a representation in terms of cycles can be useful for the problem of the invariant measures  when the dynamics is assigned.
\vspace{0.8cm}

Here we propose a different application that can be developed from the cyclic decomposition construction of sections \ref{sec:Dfreefl} and \ref{sec:applications}.

\subsection{Invariant stationary states under perturbations far away from equilibrium} We consider an exclusion process  \eqref{eq:EPjump} on $ {\bb T}^d_N $ and perturb the rates as in \ref{ss:wapmod}, that is we have $ c^{\bb F}_{x,y}(\eta)=c_{x,y}(\eta)e^{\bb F(x,y)} $ and $ c^{\bb F}_{y,x}(\eta)=c_{y,x}(\eta)e^{-\bb F(x,y)} $ for $ (x,y)\in E_N $. The flow \eqref{defbas} of the perturbed process  on the transition graph $ (\Sigma_N,E) $ is
\begin{equation}\label{eq:perflow}
Q^{\bb F}(\eta,\eta^{x,y}):=Q(\eta,\eta^{x,y})e^{\bb F(x,y)} \text{ for all $ (x,y)\in E_N $},
\end{equation}
where $ Q(\eta,\eta^{x,y}) $ is the flow of the unperturbed process $ c_{x,y}(\eta) $. Let $ \tau_x $ be  the translation defined in section \ref{sec:CTMC}. We show a perturbative application of the theory of section \ref{sec:Dfreefl}.

\begin{proposition}\label{p:pertu}
Let $ \mu $ be a translational invariant reversible measure for the dynamics $ c_{x,y}(\eta) $. Moreover  the instantaneous current $ j_\eta(x,y) $ is gradient, i.e. $ j_\eta(x,y)=\tau_y h(\eta)-\tau_x h(\eta) $ for all $ (x,y)\in E_N $ with $ h:\Sigma_N\to\bb R $. If the perturbation $ \bb F (x,y)$ is such that the discrete vector field $ \phi(x,y):=e^{\bb F (x,y)}-e^{-\bb F (x,y)}$ is a divergence free field, i.e. $ \mathrm{div } \phi(x)=0 $ for all $ x\in\bb T^d_N $, then $ \mu $ is invariant with respect to the dynamics $ c^{\bb F}_{x,y}(\eta) $.
\end{proposition}

\begin{proof}
From the translational invariance of $ \mu $ and the gradient condition on $ j_\eta (x,y) $ we have
that on the graph $ (\Sigma_N,E) $ the flow-current defined as 
 \begin{equation}\label{eq:flowcurr}
 J^Q_\eta(x,y):=Q(\eta,\eta^{x,y})-Q(\eta^{y,x},\eta)\end{equation}
 is gradient respect to the function $ \tilde h(\eta):=\mu(\eta)h(\eta) $ in the sense that $ J^Q_\eta(x,y)=\tau_y \tilde h(\eta)-\tau_x\tilde h(\eta) $.
The invariant condition for $ \mu $ with respect to $ c^{\bb F}_{x,y}(\eta) $ is also
\begin{equation}\label{eq:invQ^F}
\us{(x,y)\in E_N}{\sum}[Q^{\bb F}(\eta,\eta^{x,y})-Q^{\bb F}(\eta^{y,x},\eta)]=0,
\end{equation}
using the reversibility $ Q(\eta,\eta^{x,y})=Q(\eta^{y,x},\eta) $ and the antisymmetry of the field $ \bb F $ equation \eqref{eq:invQ^F} becomes
\begin{equation}\label{eq:invQ^F2}
\us{(x,y)\in E^+_N}{\sum}J^Q_\eta(x,y)\phi(x,y)=0,
\end{equation}
where $ E_N^+ $ is set of positive edges of $ \mathbb{T}^d_N $, see notation in section \ref{sec:disop}. From the gradient condition for $J^Q_\eta(x,y)  $ respect to $ \tilde h $ and doing an integration by part we obtain the equivalent condition 
\begin{equation}\label{eq:invQ^F3}
\us{x\in\bb T^d_N}{\sum}\tau_x\tilde h \div \phi (x)=0.
\end{equation}
The proof is concluded since equation \eqref{eq:invQ^F3} is true if and only if $ \div \phi(x)=0 $ for all $ x\in \bb T^d_N $.
\end{proof}

\begin{remark}
Of course the conditions of proposition \ref{p:pertu} are strong, but it is interesting because   we are not asking the perturbation $ \bb F $ to be small, that is we are  perturbing  an equilibrium state (namely when we have reversibility)  to obtain a stationary state arbitrary far away from equilibrium. The most simple case but also  meaningful from the physical point of view is   a constant field $ \bb F(x,y)=F $ for each $ (x,y)\in E_N^+ $.  This result was obtained also in \cite{BDeSGJ-LL07}.
\end{remark}

\section{Perspectives: part 2}\label{s:per2}

In section \ref{sec:QPweas} we derived a large deviation functional for  weakly asymmetric interacting systems with a diffusive scaling with constant diffusion coefficient and quadratic mobility. For the reasons explained in subsection \ref{trc} in the paragraph just after formula \eqref{compiti} we don't know a KMPx microscopic dynamics having an hydrodynamics \eqref{eq:drcE} with mobility $ \sigma(\rho)=\rho^2+1 $. Finding a model with these transport coefficients would be useful  to complete the picture and  because of the macroscopic phenomenology of this hypothetical model is different from the one of the other models as explained at end of section \ref{sec:SL}.
\vspace{0.8cm}

The macroscopic fluctuation theory (MFT)\cite{BDeSGJ-LL15} we presented in chapter \ref{ch:LDT}  considers  a strong interaction
with the sources (reservoirs) that fix the values of the density at the boundaries. Namely it  is a fluctuation theory for the bulk dynamics, as explained in section \ref{sec:Stldt}, because we have considered fast boundary dynamics like \eqref{eq:bKMPgen} that thermalise the parameter $ \rho $ at the boundaries simulating thermostats much larger than the system as we observed in the introduction to part \ref{p:mact}.
But other situations are possible. We would like to apply systematically the macroscopic fluctuation theory for a more general class of models. In particular we plan to study systems with
different interactions with the boundary sources, for example weak interactions with the boundaries can be modelled  as slow boundaries  that are  obtained through a different rescaling of the boundary dynamics, see  \cite{FGN16}. In this case the
density is no more fixed and there will be a contribution to the large deviations cost coming from
the boundary. Accordingly the Hamilton Jacobi equation has to be modified and in particular it
will contain an additional term. All the structure of the Macroscopic Fluctuation Theory has to be
modified accordingly.
More generally it is possible to consider models having different boundary sources \cite{DPT11} and inhomogeneous systems having slow/fast bonds/sites \cite{BDL10}.  For some of these models a dynamic large deviations principle is already available \cite{FN16}.

In \cite{BDeSGJ-LL06} the macroscopic fluctuation theory is extended to the joint fluctuation current-density. This theory is developed for diffusive particles models where it is true the Fick's law \eqref{eq:typpic}. In section \ref{sec:slv} we have seen that we have the more general scenario \ref{eq:JasymD}, therefore as future work the theory has to be investigated for this more general case.

In this thesis we  had a purely variational approach to the problem of large deviations. Another perspective is that   to study large deviations for stationary nonequilibrium states of interacting particle systems through some combinatorial representations of
the invariant measures, like for example in \cite{DS05} for the totally asymmetric exclusion process and in
\cite{Kin69} for the simple exclusion process. In particular the representation of \cite{Kin69} seems to be useful for
a direct computation of the quasi-potential considering different scaling of the interaction with
the boundaries.

\facciatabianca
\cleardoublepage
\pagestyle{fancy}
\fancyhf{}
\rhead{\sffamily  \selectfont \nouppercase{\rightmark}}
\lhead{\sffamily  \nouppercase{\leftmark} }
\cfoot{\footnotesize \thepage}
\renewcommand{\headrulewidth}{0.01cm}

\appendix
\addappheadtotoc
\addtocontents{toc}{\protect\setcounter{tocdepth}{0}}

\chapter{Complements to part 1}\label{a:mic}

Here we show   a functional Hodge decomposition for  vertex functions and a consequent  construction to generate invariant dynamics, that is equivalent to the construction in subsection \ref{ss:Glort2}.

\section[Hodge decomposition for 0-forms]{Functional Hodge decomposition for discrete 0-forms}\label{ss:Hod0}

From the results of  section \ref{sec: HD} we can derive a functional Hodge decomposition for vertex functions (discrete 0-zero forms) on $ \mathbb T^n_N $ in general dimension $ n $, on the torus we need to consider only the set of vertexes $ V_N $ and the set of edges $ E_N $.
If  in \eqref{eq:dhodd} of theorem \ref{th:HT} we take $ k=0 $ we have the Hodge decomposition
\begin{equation}\label{e:HD0}
\Omega^0=\delta\Omega^1\oplus\Omega^0_H,
\end{equation}
where $\Omega^0$ and  $\Omega^1 $ are defined as at the beginning of  section \ref{sec: HD}. The scalar product in $ \Omega^0 $ is defined as 
\begin{equation}\label{e:sp0}
\langle h,g\rangle=\us{x\in V_N}{\sum}h(x)g(x)\quad\forall\,h,g\in\Omega^0.
\end{equation}
Observe  that in one dimension to decompose a discrete vector fields or a vertex function is the same because of  the duality between edges and vertexes, indeed the decomposition is $ \Omega^1=d\Omega^0\oplus\Omega^1_H$. But in higher dimensions the decomposition for vertexes function is always \eqref{e:HD0}, namely $ \delta \Omega^1 $ is the $ N^n-1 $ dimensional space  of the vertex function $ h:V_N\to\bb R $ such that  
\begin{equation}\label{e:delta}
h(x)=\delta\vphi(x)=\us{y:(y,x)\in E_N}{\sum} \vphi(y,x),
\end{equation}
where $ \vphi:E_N\to\bb R $ is a discrete vector field in $ \Omega^1 $. While $ \Omega^0_H $ is the one dimensional space  of constant vertex function generated by
\begin{equation}\label{e:I0}
\bb I(x)=1 \text{ for all } x\in V_N.
\end{equation}
As we already did in  section \eqref{sec: HD}, when there is a dependence on $ \eta\in\Sigma_N $ we use the terms discrete forms and discrete vector fields considering a fixed configuration $ \eta $.
We denote with $ g(\eta,x) $  a general function $ g:V_N\times\Sigma_N\to \bb R$ and we say that it is \e{translational covariant} when for any $ z\in V_N $  it is
\begin{equation}\label{e:covfun}
g(\eta,x)=g(\tau_z\eta,x+z) \text{ for all } (x,\eta)\in  V_N\times\Sigma_N.
\end{equation}
In subsection \ref{ss:delta1to0}
we defined the discrete divergence $\div=-\delta  $, so for any $ f(\eta,x)\in \delta\Omega^0 $ we can write $ f(\eta,x)=\div \vphi_\eta (x)=\us{y:(x,y)\in E_N}{\sum} \vphi_\eta(x,y) $ for some $ \vphi_\eta\in \Omega^1 $. 
The decomposition for vertex functions is as follow.
\begin{theorem}\label{t:vertex}
Let $g(\eta,\cdot)$  be a translational covariant vertex function. Then there exists a function $h(\eta)$ and a translational invariant function $c(\eta)$ such that
\begin{equation}\label{eq:g=D+c}
g(\eta,x)=\Delta\tau_x h(\eta)+c(\eta)\,,
\end{equation}
where $ \Delta  f(\eta,x):=f(\eta,x+e_1)+\dots+f(\eta,x+e_n)+f(\eta,x-e_1)+\dots+f(\eta,x-e_n)-n\,f(\eta,x) $ for $ f(\eta,\cdot)\in\Omega^0 $.
The function $c$ is uniquely identified and coincides with
\begin{equation}\label{eq:c(eta)}
c(\eta)=\frac1N\sum_{x\in V_N} g(\eta,x).
\end{equation}
The function $h$ is uniquely identified up to an arbitrary additive translational invariant function and solves the Poisson equation
\begin{equation}\label{eq:h(eta)0}
\Delta \tau_x h(\eta)=g(\eta,x)-c(\eta).
\end{equation}
\end{theorem}
\begin{proof}
From \eqref{e:HD0} we  know that \begin{equation}\label{e:div+c}
 g(\eta,x)=\div\vphi_\eta(x) + c(\eta)
\end{equation}
 for some $ \vphi_\eta\in\Omega^1 $ and $ c(\eta)\in\Omega^0_H $. Summing over all $ x\in V_N $ on both sides we obtain that $ c(\eta) $ is uniquely determined by \eqref{eq:c(eta)} and  translational invariant. Since $ g(\eta,x) $ is translational covariant we can prove with the same argument of theorem \ref{belteo2}  that the $ \eta $ dependent discrete vector field $ \vphi_\eta $ can be chosen translational covariant too. From the Hodge decomposition for discrete vector fields \eqref{eq:HD1} we have that $ \vphi_\eta=d h(\eta)+\delta\psi(\eta)+\vphi^1_{H}(\eta) $ for some $ h(\eta)\in \Omega^0 $, $ \psi(\eta)\in \Omega^2 $ and $ \vphi^1_H(\eta)\in\Omega^1_H $, taking the divergence on both sides  we have that
\begin{equation}\label{eq:divdiv}
 \div\vphi_\eta(x)=\div dh(\eta,x)
\end{equation}  
for any $ x\in V_N $ because of the property $ \delta\circ\delta=0 $ (see \eqref{eq:deltaodelta}) and the definition of harmonic form (see theorem \ref{th:HT}). Since $ \vphi_\eta $ is covariant we can apply the remark \ref{re:gradpart} to have that $ dh(\eta,x,y)$ is of the form $dh(\eta,x,y)=\tau_{y} h(\eta)-\tau_x h(\eta) $ for any $ (x,y)\in E_N $, where $ y=x\pm e_i $ for some $ i\in 1,\dots, n $. Plugging \eqref{e:div+c} in \eqref{eq:divdiv}
and from last observation we find the  Poisson equation 
\begin{equation}\label{e:Poi}
\Delta \tau_x h(\eta)=g(\eta,x)-c(\eta)
\end{equation}
that for a fixed $ \eta $ has always a solution.
From \eqref{re:gradpart} it follows that $ h $ is identified up to an arbitrary additive translational invariant function.
\end{proof}  

We will use this result  to generate invariant Glauber dynamics respect to  a given measure in  general dimension $ n $.

\section[Glauber dynamics]{General n-dimensional Glauber dynamics}
So far we constructed Glauber dynamics passing trough discrete vector fields, as in subsection \ref{ss:Glort1} and \ref{ss:Glort2}. But using the functional Hodge decomposition for 0-discrete forms (namely vertex functions) of theorem \ref{t:vertex} we have a  natural construction in any dimension. 
Indeed  we can interpret  the invariant condition \eqref{step32} in terms of the scalar product \eqref{e:sp0} for vertex functions and it becomes
\begin{equation}\label{eq:ort0}
\us{x\in V_N}{\sum} g(\eta,x)=\langle g,\bb I\rangle=0,
\end{equation}
where $ g(\eta,x):=\tau_x g(\eta) $ with $ g(\eta) $ as in \eqref{carsoli} and $ \bb I $  defined as in \eqref{e:I0}.
This is an orthogonality condition between $ g(\eta)\in \Omega^0 $ and $ \bb I\in \Omega^0_H $. Therefore we have \eqref{eq:ort0} if and only if $ g(\eta)\in \delta\Omega^1 $. Since $ g(\eta) $ is translational covariant, from theorem \ref{t:vertex} we obtain that the stationary condition is satisfied if and only if there exist a function $ h:\Sigma_N\to\bb R $ such that solves the Poisson equation $  \Delta\tau_x h(\eta)=g(\eta,x)$ for all $ x\in V_N $
where $ \Delta\tau_x h(\eta)=\tau_{x+e_1}h(\eta)+\dots+\tau_{x+e_n}h(\eta)+\tau_{x-e_1}h(\eta)+\dots+\tau_{x-e_n}h(\eta)-n\,\tau_{x}h(\eta) $. Therefore it has to be $ g(\eta)=\Delta h(\eta) $.
To generate an invariant Glauber dynamics we have to fix the function $ h $ in such a way that
$$
\frac{\Delta h(\eta)}{\left(1-\frac{\eta(0)}{p}\right)}
$$
has a domain $D$ such that $D\cap \{0\}=\emptyset$. A general family of functions that satisfies this constraint is given by
\begin{equation}\label{eq:hGla}
h(\eta)=(1-\eta(-e_1))\dots(1-\eta(-e_n))(1-\eta(0))(1-\eta(e_1))\dots(1-\eta(e_n))\tilde h(\eta),
\end{equation}
where $ D(\tilde h)\cap\{-e_n,\dots,-e_1,0,e_1,\dots,e_n\} =\emptyset$. Under these assumptions we can find the positive solutions $c^\pm$ using again \eqref{posol}. Writing only the non reversible contribution we have that $c^+(\eta)$ is equal to
$\left[
\frac{\Delta h(\eta)}{\left(1-\frac{\eta(0)}{p}\right)}
\right]_+$ 
while $c^-(\eta)$ is given by $c^-(\eta)=\frac{1-p}{p}\left[
-\frac{\Delta h(\eta)}{\left(1-\frac{\eta(0)}{p}\right)}
\right]_+$.

\chapter{Complements to part 2}\label{a:appA}

In this appendix   we treat some technical questions  related to chapter \ref{ch:TALD}.   In section \ref{sec:exsist} we  discuss the existence of a solution $ \phi^\rho $ to the Euler-Lagrange equation \eqref{eq:[]=0} with boundary condition $ \phi(0)=\rho_- $, $ \phi(1)=\rho_+ $ when models in the classes of KMP, KMPd and SEP are considered, i.e.  $ \sigma(\rho) $ with real roots. In section \ref{sec:KMPsta}  we derive the explicit stationary solution $ \bar\rho_E $ to \eqref{steq} for the KMP model with a direct computation. We consider $ \Lambda=(0,1] $ and $ \rho(1)=\rho_+>\rho_- =\rho(0)$ and define $ \p_\pm\Lambda=\{0,1\} $.

\section[Existence of the solution]{Existence of the solution to the Euler-Lagrange equation}\label{sec:exsist} Let $C^1(\Lambda)$ be  the Banach space of continuously differentiable functions $f:\Lambda\rightarrow \mathbb{R}$ endowed with the norm $||f||_{C^1}:=\us{x\in\Lambda}{\sup}\{|f(x)|+|f'(x)|\}$. Let's consider a bounded fluctuation $\rho$, namely $||\rho||_{\infty}<+\infty$. The existence of a solution $ \phi^\rho $ to the Euler-Lagrange 
\begin{equation}\label{eq:exEL}
\left\{
\begin{array}{ll}
\Delta\phi\sigma(\phi)+c_2(\nabla\phi)^2(\rho-\phi)-
E\sigma(\phi)\nabla\phi\left(c_2(\rho+\phi)+c_1\right)=0,\\
\phi(0)=\phi_-=\rho_-,\,\phi(1)=\phi_+=\rho_+
\end{array}\right.
\end{equation}
needs  different dissertations for different  models. We remind that \eqref{eq:exEL} characterizes the stationary points of the functional $ \mc G_E(\rho,\phi) $ in \eqref{effeciu}.  

In section \ref{sec:opt} we studied the solution of the Hamilton-Jacobi equation $ \mc G(\rho,\phi^\rho) $ distinguishing  the cases $ E>E^* $ and $ E<E^* $ to discuss the quasi-potential in the nonequilibrium context. We did it using the change of variable $ \psi=\lambda[\phi] $. Since it is invertible, the optimization of $ \mc G_E(\rho,\cdot) $ doesn't change going back to the variable $ \phi $. So $ S_E(\rho)=\us{\phi\in\mc D_E}{\sup\,}\mc G(\rho,\phi)$ for $ E>E^*$  and  $S_E(\rho)=\us{\phi\in\mc D_E}{\inf}\mc G(\rho,\phi)$ for $ E<E^* $, where $ \mc D_E:=\left\{\phi\in C^1(\Lambda):\phi=\rho[\psi],\,\psi\in\mc F_E\right\}$ and $ \mc{F_E} $ was defined in \eqref{doveleo}. for $ E>E^*  $ and \eqref{domps} for $ E<E^* $. Specifically    we have respectively
\begin{equation}\label{eq:D_E>}
\mc D_E:=\{f\in C^1(\Lambda):f(\p_\pm\Lambda)=\phi_\pm, \nabla f (x)>0, \nabla f (x)<E\sigma(f) \} \text{ if } E>E^* 
\end{equation}  
and 
 \begin{equation}\label{eq:D_E<}
 \mc D_E:=\{f\in C^1(\Lambda):f(\p_\pm\Lambda)=\phi_\pm, \nabla f (x)>0, \nabla f(x)>E\sigma(f) \} \text{ if } E<E^* .
 \end{equation} 
Note that both for $ \psi\in\mc F_E $ and $ \phi\in\mc D_E $ we don't ask to be $ C^2(\Lambda) $ because to prove existence of the solution of the Euler-Lagrange equation in an integral-differential form to be $ C^1(\Lambda) $ is enough. Under the conditions on a function $ f\in \mc D_E $, the  functional derivative \eqref{perfi} of $ \mc G_E(\rho,\phi) $ and the first equation of \eqref{eq:exEL} (i.e. the original Euler-Lagrange equation)  are equivalent.

For convenience of visual compactness in next formulas instead of $ \nabla $ and $ \Delta $  we are using  respectively a subscript $ x $ and a subscript $ xx $ to indicate the spatial derivatives. We remember that the mobility of the model we consider has to be positive $ \sigma(\rho)>0 $ for physical reasons.

\subsection{KMP and KMP-dual "like" models}
Following the scheme introduced in \cite{DLS01} and developed in \cite{BDeSGJ-LL03}, \cite{BGL05} and \cite{BGL09} we prove the existence of the solution to the Euler-Lagrange equation \eqref{eq:exEL}  for $\sigma(\rho)=c_0+c_1\rho+c_2\rho^2$ where $c_2>0$, $\sigma(\rho)$ has two real roots and $\rho>\alpha_+$ where $\alpha_+$ where is the most right root of $ \sigma(\rho) $.  We rewrite equation \eqref{eq:exEL} as an  integral-differential equation, that is
\begin{equation}\label{K eq}
\left\{\begin{array}{ll} \phi_{xx}=\phi_x\mathcal{H_\rho(\phi)} \\
\phi(0)=\phi_-,\,\phi(1)=\phi_+
 \end {array}\right.
\end{equation}
with $\mc{H_{\rho}(\phi)}:=c_2\phi_x\frac{\phi-\rho}{\sigma(\phi)}+E(c_2(\rho+\phi)+c_1)$. If $\phi\in C^2(\Lambda)$ is such that $\phi_x$ is a solution of (\ref{K eq}), then it solves also the integral-differential equation 
\begin{equation}\label{integro EL}
\phi(x)=\phi_{-}+(\phi_+-\phi_-)\frac{\int^x_{0}dy \exp \int^y_{0}\mc{H_{\rho}(\phi)}(z) dz }{\int^1_{0}dy \exp \int^y_{0}\mc{H_{\rho}(\phi)}(z) dz }.
\end{equation}
On the other hand if $\phi$ is solution of (\ref{integro EL}) then $\phi\in C^1(\Lambda)$, $\phi_x>0$, $\phi_{xx}$ exists in almost all $x$ and it follows that (\ref{K eq}) holds almost  everywhere. Therefore we define the integral-differential operator $\mc{K_{\rho}}:\mc{D}_{E}\rightarrow C^1(\Lambda)$ given by 
\begin{equation}
\mc{K}_{\rho}(f)(u):=\phi_{-}+(\phi_+-\phi_-)\frac{\int^x_{0}dy \exp \int^y_{0}\mc{H_{\rho}(f)}(z) dz }{\int^1_{0}dy \exp \int^y_{0}\mc{H_{\rho}(f)}(z) dz },
\end{equation}
with $\mc{H_{\rho}(f)}:=c_2 f_x\frac{f-\rho}{\sigma(f)}+E(c_2(\rho+f)+c_1)$. For the moment let's consider $\sigma(x)=c_1x+c_2x^2$ with $c_1\geq0,c_2>0$ and later we are extending the results to all the models where $\sigma(x)$ has two real roots and $c_2>0$. 

In place of  considering  the problem (\ref{K eq}) we want to prove the existence of a fixed point $f_0$ for $\mc{K}_{\rho}(\cdot)$, that is $f_0$ such that $\mc{K}_{\rho}(f_0)=f_0$. To this purpose we are using the Schauder-Tychonov fixed point theorem.
\begin{theorem}(Schauder-Tychonov)
Let $X$ a locally convex space, $\mc{B}\subset X$ convex and compact. Each continuous function $f:\mc{B}\rightarrow \mc{B}$ has a fixed point in $\mc{B}$, namely it exists  $ x_0\in\mc{B}$ such that $f(x_0)=x_0$.
\end{theorem}
\begin{proof}
see Dunford-Schwartz \cite{DS58} chapter V.
\end{proof}

In next proposition we will apply the Schauder-Tychonov theorem to the convex compact set $\mc{B}_E$ defined as 
\begin{equation}\label{inv set}
\mc{B}_E:=\{f\in C^1(\Lambda):b_E\leq f_x(x)\leq B_E\}\subset\mc{D}_E,
\end{equation} 
where $ b_E:=\us{f\in \mc D_E}{\min}b_E(f) $,  $ B_E:=\us{f\in \mc D_E}{\max}B_E(f) $,
\[b_E(f):=\left\{\begin{array}{ll}
\left(\phi_+-\phi_-\right)\left(\frac{\sigma(\phi_-)}{\sigma(\phi_+)}\right)^{\frac{\sigma(\phi_+)}{m_{f_x}}} \exp({-c_2||\rho||_\infty c^1_E E})  \text{ if } E>E^*\\
\\
\left(\phi_+-\phi_-\right)\left(\frac{\sigma(\phi_-)}{\sigma(\phi_+)}\right)\exp(-c_2||\rho||_\infty c^2_E|E|) \text{ if } E<E^*
\end{array} 
\right. \] 
and
\[B_E(f):=\left\{\begin{array}{ll}
\left(\phi_+-\phi_-\right)\left(\frac{\sigma(\phi_+)}{\sigma(\phi_-)}\right)^{\frac{\sigma(\phi_+)}{m_{f_x}}} \exp(c_2||\rho||_\infty c^1_E E) \text{ if } E>E^*\\
\\
\left(\phi_+-\phi_-\right)\left(\frac{\sigma(\phi_+)}{\sigma(\phi_-)}\right)\exp(c_2||\rho||_\infty c^2_E |E|)  \text{ if } E<E^*
\end{array} 
\right.,\bigskip \] 
with $c^1_E:=\left(1-\frac{m_{f_x}}{E\sigma(\phi_+)}\right)$, $c^2_E:=\left(1+\frac{M_{f_x}}{|E|\sigma(\phi_-)}\right)$, $m_{f_x}:=\us{y\in\Lambda}{\min\,}f_x(y) $ and $M_{f_x}:=\us{y\in\Lambda}{\max\,}f_x(y)$.

\begin{proposition}\label{existence KMP}
If $\sigma(x)=c_1x+c_2x^2$ with  $c_1\geq0$ and $c_2>0$, for each positive $\rho\in L^{\infty}(\Lambda)$, $\mc{K}_{\rho}$ is a continuous map on $\mc{D_E}$ and $\mc{K}_{\rho}(\mc{D_E})\subset\mc{B}_E$. The set $\mc{B}_E$ is a convex compact subset of $C^1(\Lambda)$, hence exists a fixed point $f_0=\mc{K}_{\rho}(f_0)\in \mc{B}_E$. 
\end{proposition}

\begin{proof}
The map $\mc K_{\rho}$ is continuous on $\mc{D_E}$  because   it is composition of  continuous functions. Computing $\mc K_{\rho}$ in $\p_\pm\Lambda=\{0,1\}$ it is immediate $\mc K_{\rho}(f)(\p_\pm\Lambda)=\phi_{\pm}$. We are considering $\rho>0$, using the condition of $f\in \mc{D}_E$, the equalities $2c_2ff_x+c_1f=\sigma_x(f)$, $\us{x\in\Lambda}{\inf}\,\sigma(f)=\sigma(\phi_-)$, $\us{x\in\Lambda}{\sup}\,\sigma(f)=\sigma(\phi_+)$ and the fact that $M_{\phi_x}$ and $m_{\phi_x}$ are finite we have 
\[ \begin{array}{ll}
\log_x(\sigma(f))\leq \mc{H}_\rho(f)\leq\frac{\sigma(\phi_+)}{m_{f_x}}\log_x(\sigma(f))+c_2||\rho||_\infty c^1_E E\text{ for } E>E^*,\\
\frac{\sigma(\phi_-)}{M_{f_x}}\log_x(\sigma(f))-c_2||\rho||_\infty c^2_E |E|\leq \mc{H}_\rho(f)\leq\log_x(\sigma(f)) \text{ for } E<E^*.
\end{array} \] 
This inequalities implies that 
\begin{equation}\label{ineq K}
b_E\leq (\mc{K}_{\rho}(f))_x\leq B_E 
\end{equation}
for all $x\in\Lambda$,
then  $\mc{K}_{\rho}(\mc{D_E})\subset\mc{B}_E$.

Taking a convex combination of two functions it is straightforward verifying that $\mc{B}_E$ is convex.  Compactness comes from Ascoli-Arzel\`{a} theorem, namely $\mc{B}_E$ has to be closed, limited and equicontinuous\footnote{A subset $\mc{B}$ of a  metric space $(X,d)$ is equicontinuous if $\forall \epsilon>0$ exists $\delta_{\epsilon}$ such that $d(x,y)<\delta_{\epsilon}$ implies $|f(y)-f(x)|<\epsilon$ for all $f\in\mc{B}$.}. $\mc{B}_E$ is closed by definition, it is also limited because if $f\in B_E$ is a continuous function on the compact set $\Lambda$  and $b_E\leq f_x(x)\leq B_E$ then $\us{f\in  \mc{B}_E}{\sup}||f||_{C^1}<C$ for some constant $C>0$, finally using Lagrange theorem it is easy to get equicontinuity. At this point we can apply Schauder-Tychonov because $\mc{K}_{\rho}(\mc{B_E})\subset\mc{K}_{\rho}(\mc{D_E})$ and we already proved $\mc{K}_{\rho}(\mc{D_E})\subset\mc{B}_E$. This concludes the proof.
\end{proof}

From proposition \ref{existence KMP} it is possible to conclude the existence of a solution of (\ref{eq:exEL}) for all mobilities with two real roots and $c_2>0$.
 
\begin{corollary}\label{traslazione}
If  the mobility is a second order polynomial $\sigma(x)=c_0+c_1x+c_2x^2$ with two real roots and $c_2>0$, then (\ref{K eq}) has at least almost everywhere  a solution in $\mc{D}_E$  for each  $\rho\in L^{\infty}(\Lambda)$ such that $\rho>\alpha_+$, where $\alpha_+$ is the most right root of the polynomial $ \sigma(\rho) $.
\end{corollary}

\begin{proof}
All second order polynomials with $ c_2>0 $ and two real roots can be obtained from the ones of proposition \ref{existence KMP} by an appropriate translation. Namely, given $\tilde{\sigma}(x)=\tilde{c}_2x^2+\tilde{c}_1x+\tilde{c}_0$ this can be obtained as $\tilde{\sigma}(x)=\sigma(x+k)$ for some translation constant $k$ such that $\tilde{c}_2=c_2$, $\tilde{c}_1=2c_2k+c_1$ and $\tilde{c}_0=c_2k+c_1k$. We want to show that for an arbitrary profile $\tilde{\rho}\in L^{\infty}(\Lambda)$ such that $\tilde{\rho}>\alpha_+$,  equation (\ref{K eq}) with mobility $\tilde{\sigma}(x)$   has a solution almost everywhere. Considering  $\phi$    solution of (\ref{K eq}) with $ \sigma(\rho) $ as in proposition \ref{existence KMP} and $ \rho $ positive profile in $L^{\infty}(\Lambda)$, we  define $\tilde{\phi}:=\phi-k$ and $\tilde{\rho}:=\rho-k$ and  with few computations we have $\tilde{\phi}_{xx}-\tilde{\phi}_x\mc{H}_{\tilde{\rho}}(\tilde{\phi})=\phi_{xx}-\phi_x\mc{H}_{\rho}(\phi)=0$. Since we choosed $\rho$ arbitrary, $\tilde{\rho}$ it is in turn. In addition $\tilde{\rho}>-k$ with  $-k=\alpha_+$  because $\tilde{\sigma}(x)=\sigma(x+k)$. Lastly $\tilde{\phi}\in \mc{D}_E$ comes from $\frac{d^n\tilde{\phi}}{dx^n}=\frac{d^n\phi}{dx^n}$ and $\sigma(\tilde{\phi})=\sigma(\phi)$.
\end{proof}

\begin{corollary} 
If $\rho\in C^1(\Lambda)$ then $\phi=\phi^\rho\in C^2(\Lambda)\cap\mc{D_E}$ and (\ref{K eq}) holds everywhere.
\end{corollary}

\begin{proof}
we consider the ODE $\left\{\begin{array}{ll}  {\Phi}_x=F_{\rho}(\phi,\Phi)\\
{\phi}_x=\Phi\\
\phi(-1)=\phi_-  \end{array}\right. .$, where  $\Phi:={\phi}_x$, $\Phi_x:=F_{\rho}(\phi,\Phi)$ and $F_{\rho}(\phi,\Phi)$ is a Lipschitz continuous field. By   Cauchy-Lipschitz theorem  solution (\ref{integro EL}) is   $C^2(\Lambda) $ and solves \eqref{K eq}  everywhere.
\end{proof}

\subsection{SEP "like" models} The existence for the asymmetric SEP, namely $\sigma(\rho)=\rho-\rho^2$, it is already proved in \cite{BGL09}, so we refer to this work for the proof. If $\sigma(\rho)=c_1\rho+c_2\rho^2$ with $c_2<0$, $c_1\geq0$ and $\rho>0$ this can be treated mutatis mutandis.  The existence for the most general case $\sigma(\rho)=c_0+c_1\rho+c_2\rho^2$, where $c_2<0$ and $\sigma(\rho)$ with two real roots,  can be deduced by translation as in corollary. \ref{traslazione}.

\subsection{KMPx "like" models} For this model also the existence of a solution to the Euler-Lagrange is a difficult issue and we didn't prove it. The difficulties arise from the fact that  $ \sigma(\rho) $ does not have real roots and we can't consider profiles $ \rho $ with a defined sign. So it does not hold the previous argument leading to suitable inequalities to apply Schauder-Tychonov. For the moment we are not proving the existence of $ \phi^\rho $.

\section{KMP stationary solution}\label{sec:KMPsta}
In this section we solve explicitly the equation \eqref{steq} when $\p_t{\rho}=0$ for the KMP model. We define the costant $\j:=\rho_x-E\rho^2$, it depends on, i.e.  $\j=\j(E)$, and it is the current $J_E$ in \eqref{ei} with opposite sign, i.e. $ j(E)=-J_E $. Integrating $\j=(\rho_+-\rho_-)-E\int_{\Lambda}dx\:\rho^2$, therefore $ \sign (\j)=\left\{ \begin{array}{ll} 1 \text{ if } E<E*\\ -1 \text{ if } E>E^*\\ \end{array}\right.$ where\footnote{If $E=E^*$ then  $J_E=0$, since the particle flow due to the boundary sources and the one due to the external field balance  one has a reversible  dynamics, see section \ref{trc} and \ref{sec:stasol}} $E^*$ can be rewritten  as $E^*=\frac{\rho_+-\rho_-}{\int_{\Lambda} dx \:\rho^2}$. 

First we want to study the limit $E\rightarrow +\infty$, the opposite case $E\rightarrow -\infty$ will be solved by symmetry as explained later.  We can assume $E>E^*$, integrating  $\rho_x-E\rho^2=\j$  by  separation of variables \[\frac{1}{2\sqrt{e}}\log\frac{|1+\sqrt{e}\rho|}{|1-\sqrt{e}\rho|}=\j x+\k, \text{ where } e:=\frac{E}{|\j|} \text{ and } \k=\k(E). \]
We know $\rho$ is a  profile from  $\rho_-$ to $\rho_+$ and to obtain the explicit form we have to invert the previous equality, a priori we can't do  assumptions on $\rho_x$ so we need to consider both cases $\sign(1+\sqrt{e}\rho)=\sign(1-\sqrt{e}\rho)$ and $\sign(1+\sqrt{e}\rho)\neq\sign(1-\sqrt{e}\rho)$. In the first case
$\rho=\frac{1}{\sqrt{e}}\tanh (\j \sqrt{e} x+\k\sqrt{e})$ which is a strictly decreasing profile because $\j<0$, therefore this solution is impossible since $\rho_+>\rho_-$. While in the second case we have a  strictly increasing profile, namely  
\begin{equation}\label{st profile}
\rho(x)=\frac{1}{\sqrt{e}}\coth(\sqrt{e}(\j x+\k)),
\end{equation}
that is compatible with our boundary conditions. Since $\k=-jx+\frac{1}{2\sqrt{e}}\log\frac{\sqrt{e}\rho+1}{\sqrt{e}\rho-1}\quad $ for all $x$ and $ \rho>0 $, taking $x=1$ we have $\k>-\j$ and  there is a   vertical asymptote in $x=-\frac{\k}{\j}$ greater than one, i.e. $ -\frac{\k}{\j}>1 $. Let $I(E):=\int_{\Lambda}dx\,\rho^2$, asymptotically  $-\j,\k\us{E}{\sim} I(E)\cdot E$ and then $\us{E\rightarrow+\infty}{\lim}-\frac{\k}{\j}=1^+$.

To deduce  the proof of  the strong asymmetric limit of the stationary solution $ \bar\rho_E $ for $ E\to-\infty $ from that of $ E\to+\infty $ we will use next lemma.

\begin {lemma}\label{le:sym}
If $\rho(x)=\rho(\rho_-,\rho_+,E,J_E,x)$ is solution of (1):   $\left\{\begin{array} {l}
\rho_x-E\rho^2=-J_E\\ 
\rho(\p_\pm\Lambda)=\rho_{\pm}
\end{array}\right.$, 
then $\gamma(x)=\frac{1}{\rho(\rho_-,\rho_+,E,J_E,x)}$ is solution of (2):
$\left\{\begin{array} {l}
\gamma_x-J_E\gamma^2=-E\\ 
\gamma(\p_\pm\Lambda)=\frac{1}{\rho_{\pm}}
\end{array}\right.$. In particular    $\rho\left(\frac{1}{\rho_-},\frac{1}{\rho_+},J_E,E,x\right)=\frac{1}{\rho(\rho_-,\rho_+,E,J_E,x)}$.
\end{lemma}
\begin{proof}
Defining $\gamma(x):=\frac{1}{\rho(x)}$, we get the equation (2) and by definition the boundary conditions are the desired ones. Equation (2) has the same form of (1) where $E$ is interchanged with $J_E$. Then $\gamma(x)=\rho\left(\frac{1}{\rho_-},\frac{1}{\rho_+},J_E,E,x\right)$ is solution of (2) and by  definition of $\gamma(x)$ the lemma is proved.
\end{proof}

\begin{proposition}
In the limits $E\rightarrow +\infty$ and $E\rightarrow -\infty$ the stationary profile is respectively
\begin{equation}\label{KMP limit profile}
\rho (x)=\left\{\begin{array}{ll}\rho_- \text{ if } x\in [0,1)\\
\rho_+ \text{ if } x=1
\end{array}\right.
\qquad \text{ and } \qquad
\rho (x)=\left\{\begin{array}{ll}\rho_- \text{ if } x=0\\
\rho_+ \text{ if } x=(0,1].
\end{array}\right.
\end{equation}
\end{proposition}

\begin{proof} First we sketch the proof for  the case $E\rightarrow +\infty$. Rewriting $-\j x+\k=\j(1-x)+\frac{1}{2\sqrt{e}}\log\frac{\sqrt{e}\rho_+ +1}{\sqrt{e}\rho_- -1}$ and using $\us{E\rightarrow+\infty}{\lim}-\frac{j}{E}=\us{E\rightarrow+\infty}{\lim}I(E)$ one computes that $\us{E\rightarrow+\infty}{\lim}\rho(x)=\us{E\rightarrow+\infty}{\lim}\sqrt{I(E)}$ for all $x\neq1$.  From the boundary conditions $\us{E\rightarrow+\infty}{\lim}\rho(0)=\rho_-=\us{E\rightarrow+\infty}{\lim}\sqrt{I(E)}$ and  one can verify $\us{E\rightarrow+\infty}{\lim}\rho(1)=\us{E\rightarrow+\infty}{\lim}\sqrt{I(E)}\cdot H(E)=\rho_+$, where $H(E)$ is a finite positive constant such that $\us{E\rightarrow+\infty}{\lim}H(E)=\frac{\rho_+}{\rho_-}$.

We prove the proposition for $E\rightarrow-\infty$ by symmetry  taking  advantage of  lemma \ref{le:sym}.
If $\tilde{\rho}(x)=\tilde{\rho}(\rho_+,\rho_-,\tilde{E},J_{\tilde{E}},x)$ solves $\left\{\begin{array} {l}
\tilde{\rho}_x-\tilde{E}\tilde{\rho}^2=-J_{\tilde{E}}\\ 
\tilde{\rho}(\p_\pm\Lambda)=\rho_{\mp}
\end{array}\right.$ where\footnote{From the physical point of view is the same to consider $\left\{\begin{array} {l}
\rho_x-E\rho^2=-J_{E}\\ 
\rho(\p_\pm\Lambda)=\rho_{\pm}
\end{array}\right.$. } $\tilde{E}=-E$ then $\rho(x)=\tilde{\rho}(\rho_+,\rho_-,\tilde{E},J_{\tilde{E}},1-x)$ solves $\left\{\begin{array} {l}
\rho_x-E\rho^2=-J_E\\ 
\rho(\p_\pm\Lambda)=\rho_{\pm}
\end{array}\right.$ with  $J_E=-J_{\tilde{E}}$. From  previous lemma \ref{le:sym} it follows that  $\tilde{\rho}(x)=\tilde{\rho}(\rho_+,\rho_-,\tilde{E},J_{\tilde{E}},x)=\frac{1}{\tilde{\rho}\left(\frac{1}{\rho_+},\frac{1}{\rho_-},J_{\tilde{E}},\tilde{E},x\right)}$, since   $\tilde{E}$ and $J_{\tilde{E}}$ have the same sign  $\underset{E\rightarrow-\infty}{\lim}\tilde{\rho} (x)=\left\{\begin{array}{ll}\rho_- \text{ if } x=0\\
\rho_+ \text{ if } x=(0,1]\end{array}\right.$ and $\underset{E\rightarrow-\infty}{\lim}\rho(x)$ is as in (\ref{KMP limit profile})  .
\end{proof}

\bigskip

So have found that the stationary solution $ \bar\rho_E $ is as \eqref{st profile} and in the strong asymmetric limit becomes \eqref{KMP limit profile}.


\facciatabianca
\cleardoublepage

\pagestyle{fancy}
\fancyhf{}
\rhead{}
\lhead{}
\cfoot{ \footnotesize \thepage}
\renewcommand{\headrulewidth}{0pt}


\end{document}